\documentclass[english,thm-restate,final]{lipics-v2021}
\PassOptionsToPackage{capitalize}{cleveref}

\pdfoutput=1 
\hideLIPIcs  

\graphicspath{{./Figures_new/}}

\usepackage{amsmath}
\usepackage{amsthm}
\usepackage{amssymb}
\usepackage{mathtools}

\usepackage{cite}
\usepackage{graphicx}
\usepackage{subcaption}
\usepackage{xspace}

\usepackage{pdfpages}
\usepackage{pgfplots}
\usepackage{pst-3dplot}
\usepackage{animate}
\usepackage[ruled,noend]{algorithm2e}
\usepackage{cleveref}

\newcommand{\R}{\mathbb{R}}

\newcommand{\RRR}{\mathbb{R}^3}

\newcommand{\eps}{\varepsilon}

\newcommand*{\Rom}[1]{\uppercase\expandafter{\romannumeral #1\relax}}

\newcommand{\vd}[1]{\text{VD$_{#1}$}}
\newcommand{\nvd}{\mathrm{NVD}}
\newcommand{\fvd}{\mathrm{FVD}}

\newcommand{\gmap}{\mathrm{\Gamma}}

\renewcommand{\t}[3]{T(\ell_{#1},\ell_{#2},\ell_{#3})}

\usepackage{mdframed} 				

\newcommand{\evanthiac}[1]{}

\newcommand{\zeyuc}[1]{}
\newcommand{\deleted}[1]{}

\title{The Voronoi Diagram of Four Lines in $\RRR$} 

\author{Evanthia Papadopoulou}
{Faculty of Informatics, Università della Svizzera italiana (USI), Lugano, Switzerland}
{evanthia.papadopoulou@usi.ch}
{https://orcid.org/0000-0003-0144-7384}{}

\author{Zeyu Wang}
{Faculty of Informatics, Università della Svizzera italiana (USI), Lugano, Switzerland}
{zeyu.wang@usi.ch}
{https://orcid.org/0009-0004-4207-198X}
{} 

\authorrunning{E. Papadopoulou and Z. Wang} 

\Copyright{Evanthia Papadopoulou and Zeyu Wang} 


\category{} 

\funding{Supported 
  by the Swiss National Science Foundation (SNF),  project
  200021E$\_$201356}

\acknowledgements{We thank Dr.~Martin Suderland for early discussions and valuable comments related to Sections~\ref{sec:preliminaries},~\ref{sec:basics}, and for two \textsc{Matlab}-based visualisation tools, which allowed us to visualize $\Gamma(\fvd(L))$, Fig.~\ref{fig:one_nvd_region}, and Fig.~\ref{fig:NVD1}.
We also thank SoCG anonymous reviewers for comments that helped improve the presentation of this paper, the general position assumption, and \cref{lem:generalposition-notangentialintersection}.}

\nolinenumbers 

\EventEditors{Hee-Kap Ahn, Michael Hoffmann, and Amir Nayyeri}
\EventNoEds{3}
\EventLongTitle{42nd International Symposium on Computational Geometry
(SoCG 2026)}
\EventShortTitle{SoCG 2026}
\EventAcronym{SoCG}
\EventYear{2026}
\EventDate{June 2--5, 2026}
\EventLocation{New Brunswick, NJ, USA}
\EventLogo{socg-logo.pdf}
\SeriesVolume{367}
\ArticleNo{XX}     

\begin{document}

\nolinenumbers
\maketitle

\begin{abstract}
We consider the Voronoi diagram of lines in $\mathbb{R}^3$ under the Euclidean metric, and give a full classification of its structure in the base case of four lines in general position. We first show that the number of vertices in the Voronoi diagram of four lines in general position is always even, between 0 and 8, and all such numbers can be realized. We identify a key structure for the diagram formation, called a \emph{twist}, which is a pair of consecutive intersections among trisector branches; only two types of twists are possible, so-called \emph{full} and \emph{partial} twists. A full twist is a purely local structure, which can be inserted or removed without affecting the rest of the diagram. Assuming no full twists, the nearest and the farthest Voronoi diagrams of four lines, each have 15 distinct topologies, which are in one-to-one correspondence; the two-dimensional faces are all unbounded, and the total number of vertices is at most six. The unbounded features of the farthest diagram, encoded in a two-dimensional spherical map, are also in one-to-one correspondence. The identified topologies are all realizable. Any Voronoi diagram of four lines in general position in $\mathbb{R}^3$ can be obtained from one of these topologies by inserting full twists; each twist induces a bounded face of exactly two vertices in both the nearest and farthest diagrams. We obtain the classification by an exhaustive search algorithm using some new structural and combinatorial observations of line Voronoi diagrams.
\end{abstract}

\section{Introduction}

Voronoi diagrams are fundamental space partitioning structures in Computational Geometry. Given a set $S$ of $n$ objects in some space, called sites, the \emph{nearest} (respectively, \emph{farthest}) Voronoi diagram of $S$ decomposes the underlying space into regions that have the same nearest (resp., farthest) site. In this paper, we consider Voronoi diagrams of lines in $\R^3$ under the Euclidean metric, focusing on the elementary, yet fundamental base case of four lines.

The Euclidean Voronoi diagram of point sites in $\R^{d}$ is a classical geometric partitioning structure that is generally well-understood. It has complexity $O(n^{\lceil \frac{d}{2} \rceil})$ and can be computed in $O(n \log n + n^{\lceil \frac{d}{2} \rceil})$ time~\cite{Edelsbrunner1986,Chazelle1991,Klee1980}; these bounds are tight. The results also hold for certain polyhedral norms such as the $L_\infty$ or $L_1$ norms~\cite{Boissonnat1998,Icking2001}. For the Euclidean farthest Voronoi diagram of point sites, exact worst-case bounds in the range of $\Theta(n^{\lceil \frac{d}{2} \rceil})$  have also been reported by Seidel~\cite{Seidel1987}.

For sites more general than points, however, Voronoi diagrams in $\R^d$, $d\geq 3$, have not yet been well-understood. These diagrams get complicated because their bisecting surfaces are curved, triggering complicated algebraic descriptions of their features such as their trisector curves. The combinatorial complexity of the classical Voronoi diagram of $n$ lines (or line segments) in $\R^3$ is a well-known outstanding open problem~\cite{Mitchell2001}. There is a gap of an order of magnitude between the $\Omega(n^2)$ lower bound~\cite{Aronov2002} and the only known $O(n^{3+\eps})$ upper bound by Sharir~\cite{Sharir1994}. This gap extends to the Voronoi diagram of lines in~$\R^d$,~$d>3$, where the known upper bound is $O(n^{d+\eps})$~\cite{Sharir1994}. For $d\geq 4$, an $\Omega(n^{\lceil 2d/3 \rceil})$ lower bound has recently been reported by Glisse~\cite{glisse2021lower}. The upper bounds are derived from general complexity results on envelopes of algebraic functions in $\R^{d+1}$, under the general framework~\cite{Edelsbrunner1986} for studying Voronoi diagrams in $\R^d$ through the \emph{arrangement} of distance functions of the given sites in $\R^{d+1}$. Sharper complexity bounds have been established in certain special cases such as lines with $c$ constant orientations~\cite{Koltun2002} and parallel half-lines~\cite{Aurenhammer2017}, and for norms induced by  convex polyhedra of constant complexity~\cite{Chew1998, Koltun2002a, Aurenhammer2021}.

The Voronoi diagram of only three lines in $\R^3$ was addressed by the authors of~\cite{Everett2009} who proved the fundamental properties of its structure. The trisector of three pairwise skew lines was shown to consist of four unbounded curves, which are roots of a nonsingular quartic, or a nonsingular cubic and a nonintersecting line. This is the first, most basic building block necessary in computing any Voronoi diagram of lines or polyhedra in $\R^3$, but it contains no Voronoi vertices. In this paper, we address the next basic building block,  which is the Voronoi diagram of four lines in $\R^3$, where a minimal number of Voronoi vertices appears.

The unbounded features of the Euclidean order-$k$, $1\leq k\leq n{-}1$, Voronoi diagrams of lines and line segments in $\R^d, d\geq 3$, were studied in~\cite{Barequet2023}. They are encoded in a map on a \emph{sphere of directions}, called the Gaussian map, which has complexity $O(\min\{k, n-k\}n^{d-1})$; for the farthest Voronoi diagram the Gaussian map has complexity $\Theta(n^{d-1})$.

The Voronoi diagram of lines in $\R^3$ can be computed in $O(n^{3+\epsilon})$ time~\cite{Agarwal1994} as part of a more general technique to compute the lower envelope of algebraic functions in $\R^4$. Using the envelope package of \href{https://www.cgal.org}{CGAL}, a numerically robust algorithm for computing the Voronoi diagram of lines in $\R^3$ has been reported in~\cite{Hemmer2010}. Algorithms for computing approximations of the Voronoi diagram also exist, see e.g.,~\cite{Etzion1999,Dey2002}.


\subparagraph{Contribution.}
In this paper, we prove structural properties of the Voronoi diagram of four lines in $\mathbb{R}^{3}$, building upon the previous work~\cite{Everett2009} on three lines. The case of four lines is fundamental, and is required for understanding any Voronoi diagram of lines or polyhedra in $\mathbb{R}^{3}$. We first show that the number of vertices in the Voronoi diagram of four lines in general position is always even, between 0 and 8, and all such numbers can be realized. We then identify basic structures in the arrangement of trisectors, called \emph{twists}, which are critical for the formation of the diagram and the resulting number of its Voronoi vertices. A twist is a pair of consecutive trisector intersections along some trisector components (called branches, see \cref{def:twist1}). We show that there are only two types of twists: \emph{full twists} that involve four different trisector branches, and \emph{partial twists} that involve six trisector branches, two of which are incident to both vertices of the twist (see \cref{def:twist}). A full twist has a \emph{local} structure, in the sense that adding or removing one does not affect the rest of the diagram. It results in a bounded face with exactly two vertices in both the nearest and the farthest diagrams. Since full twists only have a local effect, we first assume that no full twists exist in the trisector arrangement, and fully classify the Voronoi diagrams of four lines in this case. We prove the following theorems.

\begin{theorem}\label{thm:topologywithoutfulltwist}
The Voronoi diagram of four lines in general position in $\mathbb{R}^{3}$ (both the nearest and farthest diagram) has 15 distinct topologies, assuming there are no full twists among the trisectors. The distinct topologies of the nearest and the farthest diagrams are in one-to-one correspondence. Further, they are in one-to-one correspondence with the unbounded features of the farthest diagram. To obtain more than two Voronoi vertices, twists are required. The two-dimensional faces in both diagrams are unbounded and the total number of vertices is at most 6, see \cref{table1}. In the nearest Voronoi diagram, all vertices are connected in a single component of the 1-skeleton.

\begin{table}[h]
    \centering
    \begin{tabular}{c|c|c|c|c|c}
        \hline
        Number of vertices & 0 & 2 & 4 & 6 & 8\\
        \hline
        Number of topologies & 1 & 3 & 5 & 6 & 0 \\
        \hline
    \end{tabular}    
    \caption{Distinct topologies for each possible number of Voronoi
      vertices, assuming no full twists.}\label{table1}
\end{table}
\end{theorem}

Additional structural properties are listed in Observation~\ref{observation}. 

\begin{theorem}\label{thm:topology}
The nearest (resp.\ farthest) Voronoi diagram of four lines in general position in $\mathbb{R}^{3}$ can be obtained from one of the 15 topologies of \cref{thm:topologywithoutfulltwist} by inserting the diagram's full twists. Inserting a full twist to a topology creates one new bounded face with two vertices in both the nearest and farthest diagram; it also splits an existing face in two faces. The total number of vertices, including full-twists, is at most 8.
\end{theorem}

To obtain this classification, we devise an exhaustive search algorithm that enumerates possible \emph{configurations} of projected trisector pairs. By ignoring the geometry and focusing on the structure, we encode a bisector and the two trisectors it contains as a \emph{configuration} of eight curves in the plane, where each trisector has four branches. We identify necessary conditions under which a configuration is valid, 
and use them to create filters that eliminate impossible cases. We also create filters for six-tuples of configurations, where each tuple corresponds to the six bisectors of four lines. After an exhaustive search, which eliminates impossible cases by the derived necessary conditions, we are left with only 15 distinct configuration six-tuples. Each tuple is shown to correspond to a distinct topology and the enumeration process is shown to be complete. All 15 tuples and their corresponding topologies are realizable as we demonstrate by concrete examples.
These results also provide an algebraically simple method to identify the structure of the Voronoi diagram.

In addition, we establish some simple new combinatorial results on the nearest and farthest Voronoi diagrams of $n$ lines, expressing their features as a function of the number of Voronoi vertices (see \cref{thm:VEFCnlines}). We use these results for $n=4$ to obtain the aforementioned filters. For the farthest Voronoi diagram of $n$ lines, we show that the region of each line has exactly $n{-}1$ three-dimensional cells, extending a result of~\cite{Barequet2023} on the total number of cells.

The exhaustive search has been implemented as a program, available on GitHub~\cite{zeyu_program_4VLD}, which performs the actual search.

\section{Preliminaries}\label{sec:preliminaries}

Let $L = \{\ell_{1},\ldots,\ell_{n}\}$ be a set of $n$ lines in general position in $\R^{3}$. By general position, we mean the following: (1)~lines are pairwise skew, and no three lines are parallel to a common plane; (2)~no sphere can be tangent to four lines with coplanar tangency points; (3)~no sphere is tangent to five lines; (4)~no four lines have direction vectors whose representations on the sphere of directions are cocircular.

We denote by $d(x,y)$ the Euclidean distance between two points $x,y \in \mathbb{R}^3$. The distance $d(x,\ell)$ from a point~$x \in \mathbb{R}^3$ to a line~$\ell \in L$ is defined as \(d(x,\ell) = \min\{d(x,y) \mid y \in \ell\}\). The $i$-sector of $i$ lines in $\R^3$, $i\in \{2,3\}$, is the locus of points equidistant from the $i$ lines. For $i=2$, a bisector $B(\cdot,\cdot)$ is a hyperbolic paraboloid; for $i=3$, a trisector $T(\cdot,\cdot,\cdot)$ is a quartic consisting of four unbounded branches~\cite{Koltun2002,Everett2009}. Two trisectors are said to be \emph{related} if they involve exactly four lines, i.e., they are of the form $T(\ell_{i},\ell_{j},\ell_{k}),T(\ell_{i},\ell_{j},\ell_{l})$, with $k\neq l$. 

The \emph{nearest Voronoi diagram} of $L$, denoted $\nvd(L)$, subdivides $\R^{3}$ into maximal regions, each consisting of points closer to one line than to any other line. More generally, for $1 \leq k leq n-1$, the \emph{order-$k$ Voronoi diagram}, denoted $\vd{k}(L)$, subdivides $\mathbb{R}^3$ into maximal regions that have the same $k$ nearest lines. The case $k=n{-}1$ gives the \emph{farthest Voronoi diagram}, denoted $\fvd(L)$. Each diagram is a 3D cell complex, whose \emph{features} are vertices, edges, faces (2D), and cells (3D). Unless otherwise stated, we call 2D faces simply \emph{faces} and 3D cells simply \emph{cells}.

The general position assumptions imply the following properties. By assumption~(1), the topology of trisectors is unique, \cite{Everett2009}~Theorem~1. By~(2), related trisectors intersect transversely (see \cref{lem:generalposition-notangentialintersection}). By~(3), vertices of the nearest and farthest Voronoi diagrams have degree 4. By~(4), the asymptotes of related trisectors do not coincide (see \cref{lem:middlebranch}).

\begin{restatable}{lemma}{lemgeneralpositionnotangentialintersection}\label{lem:generalposition-notangentialintersection}
If two related trisectors (associated with four lines) intersect tangentially, then there exists a sphere tangent to the four lines such that the four contact points are coplanar. 
\end{restatable}

\begin{proof}
  Assume that the related trisectors $\t{1}{2}{3}$ and $\t{1}{2}{4}$ intersect tangentially at point $p$.
  Let $p_i$ denote the point on line $\ell_i$, $i=1,2,3,4$, that is  closest  to $p$.
  Let $v$ be a tangent direction of $T(\ell_1,\ell_2,\ell_3)$ at $p$; then $v$ is also a tangent direction of $T(\ell_1,\ell_2,\ell_4)$ at $p$.
Since $T(\ell_1,\ell_2,\ell_3)=B(\ell_1,\ell_2)\cap B(\ell_1,\ell_3)$, the tangent direction $v$ is parallel to $\overrightarrow{n_1}\times \overrightarrow{n_2}$, where $\overrightarrow{n_1}$ (resp.\ $\overrightarrow{n_2}$) is the normal to $B(\ell_1,\ell_2)$ (resp.\ $B(\ell_1,\ell_3)$) at $p$. Similarly, $v$ is parallel to $\overrightarrow{n_1}\times \overrightarrow{n_3}$, where $\overrightarrow{n_3}$ is the normal to $B(\ell_1,\ell_4)$ at $p$. Since $B(\ell_1,\ell_2)$ bisects $\ell_1,\ell_2$ and $p_1,p_2$ are the closest points to $p$ on $\ell_1$ and $\ell_2$, we have $\overrightarrow{n_1} \parallel \overrightarrow{p_1 p_2}$. Similarly, $\overrightarrow{n_2} \parallel\overrightarrow{p_1 p_3}$ and $\overrightarrow{n_3} \parallel \overrightarrow{p_1 p_4}$. Hence, the four points $p_1,p_{2},p_{3},p_4$ are coplanar. This completes the proof.
\deleted{
Since $T(\ell_1,\ell_2,\ell_3)=B(\ell_1,\ell_2)\cap B(\ell_1,\ell_3)$, the tangent direction $v$ is parallel to $\overrightarrow{n_1}\times \overrightarrow{n_2}$, where $\overrightarrow{n_1}$ (resp.\ $\overrightarrow{n_2}$) is the normal to $B(\ell_1,\ell_2)$ (resp.\ $B(\ell_1,\ell_3)$) at $p$. Similarly, $v$ is parallel to $\overrightarrow{n_1}\times \overrightarrow{n_3}$, where $\overrightarrow{n_3}$ is the normal to $B(\ell_1,\ell_4)$ at $p$.
Since $B(\ell_1,\ell_2)$ bisects $\ell_1,\ell_2$ and $p_1,p_2$ are the closest points to $p$ on $\ell_1$ and $\ell_2$, we have $\overrightarrow{n_1} \parallel \overrightarrow{p_1 p_2}$. Similarly, $\overrightarrow{n_2} \parallel\overrightarrow{p_1 p_3}$ and $\overrightarrow{n_3} \parallel \overrightarrow{p_1 p_4}$. Hence, the four points $p_1,p_{2},p_{3},p_4$ are coplanar. This completes the proof.
}
\end{proof}

For a cell complex $M$, let $B$ be a ball large enough to intersect any $i$-dimensional feature of $M$, for $0\leq i\leq 3$, in one connected component. Let $\Gamma$ be the boundary of $B$. The intersection $M\cap \Gamma$ encodes the unbounded features of $M$ into a two-dimensional map on $\Gamma$, denoted $\gmap(M)$. In particular, the unbounded edges, unbounded 2D faces, and unbounded 3D cells of $M$ correspond exactly to the vertices, edges and faces, respectively, of $\gmap(M)$. 

Next, we give a summary of known elementary properties for the Voronoi diagram
of four lines.
Four lines induce four trisectors and six bisectors. When two trisectors intersect, the other two trisectors necessarily pass through the same intersection points. These intersections are Voronoi vertices in the order-$k$ Voronoi diagram, $k\in\{1,2,3\}$. Voronoi vertices lie on all six bisectors. A trisector consists of four unbounded curves, called branches. Trisector branches are partitioned by the Voronoi vertices into arcs, 
which are order-$k$ Voronoi edges for some $k$. 
Each trisector lies in three bisectors, and each bisector contains exactly two trisectors.
A bisector is partitioned by the two trisectors it contains into faces, where each face belongs
to an order-$k$ Voronoi diagram, for some $k$. 
The nearest and farthest Voronoi diagrams of four lines share no edges or faces, by definition.
In contrast, both diagrams share their edges with the order-2 Voronoi diagram. 



\subsection{Algebraic framework and geometric reinterpretation}\label{sec:algebra}

We recall the parametrization of lines from~\cite{Everett2009}. WLOG, let the coordinate system satisfy the following: $\ell_1$ is defined by point  $(0,0,1)$ and vector $(1,a,0)$, $\ell_2$ is defined by point $(0,0,-1)$ and vector $(1,-a,0)$, for some $a\in\mathbb{R}$. For $i\geq 3$, line $\ell_{i}$ is defined by point $(b_i,c_i,0)$ and vector $(d_i,e_i,1)$, with $b_i,c_i,d_i,e_i\in\mathbb{R}$. Then, the bisector $B(\ell_1,\ell_2)$ has the following formulation: 
\begin{equation*}
	B(\ell_1,\ell_2)=\{(x,y,z) \mid (x,y)\in\R^2, z = \frac{-a}{1+a^2}\cdot xy\}.
\end{equation*}              
There is a natural homeomorphism $\Pi:B(\ell_1,\ell_2)\to\mathbb{R}^2$, given by $\Pi(x,y,z)=(x,y)$, that projects the bisector onto the Euclidean plane. We call $\Pi(B(\ell_{1},\ell_{2}))$ the \emph{projected bisector}. For any trisector $T=T(\ell_{1},\ell_{2},\ell_{i})$ on $B(\ell_1,\ell_2)$, we call $\Pi(T)$ the \emph{projected trisector}. The following lemma summarizes key properties of the Voronoi diagram of three lines from~\cite{Everett2009}, including~\cite{Everett2009}~Propositions~14 and~17. We restate them in the following lemma in terms of projected trisectors on projected bisectors, as needed for analyzing four lines in \cref{sec:twists,sec:method}. Refer to \cref{fig:branch}.

\begin{restatable}[\cite{Everett2009}]{lemma}{lemmiddlebranch}\label{lem:middlebranch}
A projected trisector $T$ consists of four unbounded branches. It has two vertical and two horizontal asymptotes. There is a unique branch that admits exactly one asymptote, called the \emph{middle branch}. The middle branch partitions a projected bisector into two regions: one containing a single branch and the other containing two branches; denote the single branch as the \emph{U branch}. Furthermore, the following hold:
\begin{enumerate}
\item the middle branch of $T$ is the same on all three bisectors that $T$ lies on;
\item the three branches of $T$, other than the middle branch, each becomes a U branch on exactly one of the three bisectors that $T$ lies on;
\item the four branches of $T$ partition the bisector into five faces, two of which belong to the $\nvd$ and three belong to the $\fvd$ of the three lines;
\item if the asymptote of the middle branch is vertical, then all branches are $y$-monotone; otherwise, all branches are $x$-monotone. 
\end{enumerate}

\begin{figure}[h]
    \centering
    \includegraphics[draft=false]{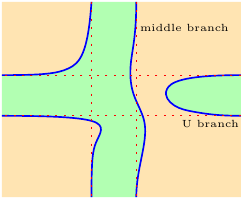}
    \caption{A projected trisector on a projected bisector. Branches are shown in blue and the asymptotes are shown in red. The $\nvd$ faces are shown in green and the $\fvd$ faces in orange. }\label{fig:branch}
\end{figure}

\end{restatable}

\section{Some combinatorial results on Voronoi diagrams of lines in $\R^3$}\label{sec:basics}

We first analyze the asymptotes of a projected trisector, building upon the results of~\cite{Everett2009}. The following lemma uses the parametrization notation of \cref{sec:algebra}.

\begin{restatable}{lemma}{lemtrisector}\label{lem:trisector}
The vertical and horizontal asymptotes of a projected trisector $\Pi (T(\ell_1, \ell_2, \ell_i)) $, on the projected bisector $\Pi(B(\ell_1, \ell_2))$, are roots of quadratic polynomials whose coefficients depend only on the parameters $a, d_{i}, e_{i}$ that define the direction of the line $\ell_i$. Consequently, the asymptotes of the projected trisector are invariant under translation of the line $\ell_i$. 
\end{restatable}

\begin{proof}
  Let $p = (x, y, z)$ 
  be a point on $T=T(\ell_{1},\ell_{2},\ell_{i})$; then $d(p,\ell_{1})=d(p,\ell_{2})$ and $d(p,\ell_{1})=d(p,\ell_{i})$. The first equation yields $z=-\frac{a}{1+a^{2}}xy$, by our choice of coordinates. We substitute this expression for $z$ in the second equation. Consequently, the projection of $T$ onto the $xy$-plane satisfies \(P(x,y):=A(x)y^{2}+B(x)y+C(x)=0\), where $A(x),B(x),C(x)$ are quadratic in $x$ with coefficients depending on $a,b_{i},c_{i},d_{i},e_{i}$. The explicit expressions are:
$A(x)=a^2 + a^4 + a^2 d_i^2 + a^4 d_i^2 - e_i^2 - a^2 e_i^2 + 2 a e_i x + 2 a^3 e_i x - a^2 x^2$, $B(x)=-2 c_i - 4 a^2 c_i - 2 a^4 c_i - 2 c_i d_i^2 - 4 a^2 c_i d_i^2 - 2 a^4 c_i d_i^2 + 2 b_i d_i e_i + 4 a^2 b_i d_i e_i + 2 a^4 b_i d_i e_i - 2 a b_i d_i x - 2 a^3 b_i d_i x - 2 a c_i e_i x - 2 a^3 c_i e_i x - 2 d_i e_i x - 4 a^2 d_i e_i x - 2 a^4 d_i e_i x + 2 a d_i x^2 + 2 a^3 d_i x^2$, and $C(x)=-1 - 2 a^2 - a^4 + c_i^2 - d_i^2 - 2 a^2 d_i^2 - a^4 d_i^2 + c_i^2 d_i^2 + 2 a^2 c_i^2 (1 + d_i^2) + a^4 c_i^2 (1 + d_i^2) - 2 b_i c_i d_i e_i - 4 a^2 b_i c_i d_i e_i - 2 a^4 b_i c_i d_i e_i - e_i^2 - 2 a^2 e_i^2 - a^4 e_i^2 + b_i^2 (1 + e_i^2) + 2 a^2 b_i^2 (1 + e_i^2) + a^4 b_i^2 (1 + e_i^2) - 2 b_i x - 4 a^2 b_i x - 2 a^4 b_i x + 2 c_i d_i e_i x + 4 a^2 c_i d_i e_i x + 2 a^4 c_i d_i e_i x - 2 b_i e_i^2 x - 4 a^2 b_i e_i^2 x - 2 a^4 b_i e_i^2 x + x^2 + a^2 x^2 - a^2 d_i^2 x^2 - a^4 d_i^2 x^2 + e_i^2 x^2 + a^2 e_i^2 x^2$.

Any real vertical asymptote $x=x_0$ must satisfy $A(x_0)=0$. The coefficients of $A(x)$ are polynomials in $a,d_{i},e_{i}$, which depend solely on the directions of the lines. Similarly, $P(x, y)$ can also be expressed as $A'(y)x^{2}+B'(y)x+C'(y)$, where $A'(y),B'(y),C'(y)$ are quadratic in $y$ with coefficients depending on $a,b_{i},c_{i},d_{i},e_{i}$. Any real horizontal asymptote $y=y_0$ must satisfy $A'(y_0)=0$, where $A'(y) = -a^4 d_i^2+2 a^3 d_i y-a^2 d_i^2+a^2 e_i^2-a^2 y^2+a^2+2 a d_i y+e_i^2+1$, whose coefficients depend solely on the directions of the lines. This completes the proof. 
\end{proof}

The number of vertices in the Voronoi diagram of four lines is known to be at most 8~\cite{Koltun2002}. We extend this observation and show that the number of vertices is always even.

\begin{restatable}{lemma}{evennumbervertices}\label{lem:evennumbervertices}
The number of vertices in a Voronoi diagram of four lines in general position is always even, between 0 and 8, and  all numbers can be attained. 
\end{restatable}

\begin{proof}
By the general position assumption, item~(4), related trisectors have different asymptotes. This is because, if two related trisectors have the same asymptote, then the four corresponding direction vectors have cocircular representations on the sphere of directions (with the center of the circle being the direction of the asymptote). Since Voronoi vertices appear on all six bisectors, it suffices to show that the number of vertices on one projected bisector is even. All calculations are modulo $2$.

Consider a projected bisector with two projected trisectors colored red and blue. Consider a large circle $C$. Let $r_1,\ldots,r_8,b_1,\ldots,b_8$ be the intersection points of the red and blue projected trisectors with $C$. Let $C$ be large enough so that the order of these points on $C$ is the same as the order of the intersections between $C$ and the asymptotes of the projected trisectors. Let $P$ be the circular order of the points. We name the red points corresponding to one trisector branch $r_{2i-1}$ and $r_{2i}$, for $i\in\{1,2,3,4\}$; similarly for the blue points.

Consider a red branch with endpoints $r_1,r_2$ and a blue branch with endpoints $b_1,b_2$. Then, these two branches intersect precisely $k$ times, where $k$ is the number of blue endpoints that are between $r_1$ and $r_2$ in $P$ (recall that calculation is under modulo $2$). Summing over all blue branches, the number of vertices on this red branch equals the number of blue endpoints between $r_1$ and $r_2$ in $P$. Denote this number by $x_{1,2}$. Hence, the total number of vertices is $x=x_{1,2}+x_{3,4}+x_{5,6}+x_{7,8}$.

We make the following observation: if a red point and an adjacent blue point are swapped in $P$, $x$ changes by 1. When $P_0=\{r_1,\ldots,r_8,b_1,\ldots,b_8\}$, we have $x=0$ (note that $P_0$ cannot be realized by real projected trisectors, but our proof does not depend on it). It suffices to show that any given order $P$ can be transformed to $P_0$ by an even number of swaps. We show this through an intermediate step $P_1$ defined in \cref{fig:proofevenvertices}. One can check that in $12$ swaps, one can move $b_1,b_8,r_1,r_8$ from their positions in $P_0$ to those in $P_1$. Hence, it remains to show that any order $P$ given by two related trisectors can be transformed to $P_1$ by an even number of swaps. 

\begin{figure}[h]
	\centering
	\includegraphics[draft=false]{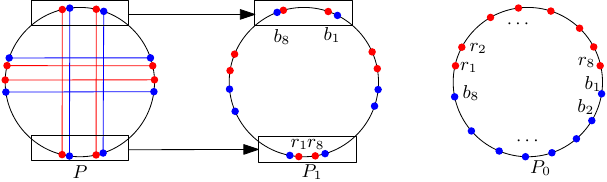}
	\caption{Illustration for the proof of \cref{lem:evennumbervertices}. }\label{fig:proofevenvertices}
\end{figure}

By symmetry of the vertical asymptotes, we can see that the number of swaps needed for the top four points is the same as the number of swaps needed for the bottom four points, in order to transform from $P$ to $P_1$. Similarly for the left and right points that correspond to the horizontal asymptotes. Hence, the total number of swaps from $P$ to $P_1$ is even. 

It was shown by Koltun and Sharir~\cite{Koltun2002} that the Voronoi diagram of four lines has at most 8 vertices. For all even numbers between 0 and 8, there are lines in general position that give the corresponding number of vertices. See Appendix~\ref{app:topology} for examples with vertices between 0 and 6. The following parameters $(a,b_{3},c_{3},d_{3},e_{3},b_{4},c_{4},d_{4},e_{4})=(10,14,-6,-9,-5,-19,5,13,11)$ give four lines that have 8 Voronoi vertices, which can be verified numerically by showing that the system of equations $d(p, \ell_1)=d(p, \ell_i)$, for $i\in \{2, 3, 4\}$ and $p\in \R^3$, has exactly 8 real roots. This completes the proof. 
\end{proof}

\begin{remark}
  Following the discussion in \cref{sec:algo}, a simpler proof can be obtained as follows.
  The number of Voronoi vertices equals the number of intersections of two projected trisectors on a projected bisector, which is captured by a configuration. 
  If a configuration has no twists, i.e., if every pair of branches intersect at most once, then such a configuration is ``simple'', see Step~1 in Phase~1. All simple configurations are generated by the exhaustive search algorithm, and they all have an even number of intersections.
  Any other configuration differs from its corresponding simple configuration by an even number of intersections. This completes the alternative proof.
\end{remark}

The farthest Voronoi diagram of a set $L$ of $n$ lines in general position is known to have $n^{2}{-}n$ distinct 3D cells, which are all unbounded, and do not form tunnels~\cite{Barequet2023}. The following lemma follows from~\cite{Barequet2023} and extends this result to the cells of individual lines. 

\begin{restatable}{lemma}{thmfvdnlines}\label{thm:fvdnlines}
In the farthest Voronoi diagram of $n$ lines in general position, $n\geq 2$, the Voronoi region of each line consists of exactly $n{-}1$ unbounded 3D cells.
\end{restatable}

\begin{proof}
  The argument follows the proofs of \cite{Barequet2023}~Theorems 5.11 and 5.12, together with the following observation. For a line $\ell_{i}$, consider the set of directions on the unit sphere $\mathbb{S}$ orthogonal to $\ell_{i}$. These directions form a great circle $C_{i}$. Consider $C_{1}$, which intersects each other great circle $C_{i}$ in exactly two points, at the ``vertices of anomaly'' (as defined in~\cite{Barequet2023}). These two points are antipodal. For each pair of vertices of anomaly, exactly one point separates a face of $\ell_{1}$ in $\gmap(\fvd(L))$, and the other separates a face of $\ell_{i}$ in $\gmap(\fvd(L))$. Hence, $C_{1}$ is divided into exactly $n-1$ arcs.
  The remaining steps proceed as in the proof of Theorem 5.12 in~\cite{Barequet2023}.
\end{proof}
 
In $\nvd(L)$, where $L$ is a set  of $n$ lines in general position, 3D Voronoi regions are connected, and each region is topologically an infinite cylinder with two unbounded ends, assuming $n\geq 3$. Thus, we can derive the following results for $\nvd(L)$. The farthest counterpart follows from \cite{Barequet2023}~Theorem 5.12.

\begin{restatable}{lemma}{lemVEFCnlinesunbounded}\label{lem:VEFCnlinesunbounded}
The map $\gmap(\nvd(L))$ has $4n{-}4$ vertices, $6n{-}6$ edges, and $2n$ faces. The map $\gmap(\fvd(L))$ has $2n^2{-}2n{-}4$ vertices, $3n^2{-}3n{-}6$ edges, and $n^2{-}n$ faces. The numbers indicate the unbounded Voronoi edges, faces, and cells, respectively, in each diagram.
\end{restatable}

\begin{proof}
By the general position assumption, item~(4), each vertex in the unbounded map $\gmap(\nvd(L))$ (resp.\ $\gmap(\fvd(L)$)) is incident to exactly three edges, and each edge is incident to two vertices. There are $n$ 3D cells in $\nvd(L)$, each of which is an infinite topological cylinder with two unbounded ends. Hence, $\gmap(\nvd(L))$ has $2n$ faces. By \cite{Barequet2023}~Theorem 5.12, $\gmap(\fvd(L))$ has $n^{2}-n$ faces. Both $\gmap(\nvd(L))$ and $\gmap(\fvd(L))$ are planar graphs, so Euler’s formula applies. The stated results then follow easily. 
\end{proof}

Note that the \emph{Gaussian map} on the sphere of directions from~\cite{Barequet2023} can be regarded as the limit of the unbounded features map on $\Gamma$, as $\Gamma$ goes to infinity. Combining \cref{lem:VEFCnlinesunbounded} and some counting arguments, we derive the following result, which we use for $n=4$ in \cref{sec:method}.

\begin{restatable}{theorem}{thmVEFCnlines}\label{thm:VEFCnlines}
Let $V_{N}$ (resp.\ $V_{F}$) denote the number of vertices of $\nvd(L)$ (resp.\ $\fvd(L)$).
\begin{itemize}
    \item $\nvd(L)$ has $2V_{N}{+}2n{-}2$ edges, $V_{N}{+}3n{-}3$ faces, and $n$ cells.
    \item $\fvd(L)$ has $2V_{F}{+}n^2{-}n{-}2$ edges, $V_{F}{+}2n^2{-}2n{-}3$ faces, and $n^2{-}n$ cells.
\end{itemize}
\end{restatable}

\begin{proof}
Let $V$, $E$, and $F$ denote the numbers of vertices, edges, and faces of the diagram under consideration.
	
Consider the $\nvd(L)$. For a 3D cell $c$, let $V_c, E_c, F_c$ be the numbers of vertices, edges, and faces of $c$. Since $c$ is an infinite topological cylinder with two unbounded ends, its intersection with the large sphere $\Gamma$ consists of two disjoint faces. Consequently, Euler's formula yields the equation $V_c-E_c+F_c = 0$. Summing over all 3D cells of $\nvd(L)$ yields $4V-3E+2F=0$, since each vertex/edge/face is incident to $4/3/2$ cells.
	
Let $E_0,E_1,E_2$ be the numbers of edges that are incident to two vertices, incident to one vertex, and not incident to any vertex, respectively. By counting the vertex-edge-incidences, we get $2E_0 + E_1 = 4V$. By \cref{lem:VEFCnlinesunbounded}, we get $E_1+2E_2 = 4n-4$. Combining with $E=E_0+E_1+E_2$ gives $E=2V+2n-2$ and $F=V+3n-3$.
	
Now consider $\fvd(L)$. Each 3D cell $c$ corresponds to a single face of $\gmap(\fvd(L))$~\cite{Barequet2023}. Similarly, Euler's formula gives the equation $V_c-E_c+F_c = 1$. Summing up over the $n^2{-}n$ 3D cells of $\fvd(L)$ yields $4V-3E+2F=n^{2}-n$. The rest of the proof is similar to  $\nvd(L)$. Finally, we get $E = 2V+n^2-n-2$ and $F = V+2n^2-2n-3$.    
\end{proof}

\section{Twists}\label{sec:twists}

From this section onward, we focus on a set $L$ of four lines in general position. The four lines give rise to four trisectors that form the trisector system of $L$, and six bisectors.

\begin{definition}\label{def:twist1}
A pair of trisector intersections $(v_{1},v_{2})$ is called a \emph{twist}, if there exist two trisectors $T_{1},T_{2}$, that both have a single branch passing through $v_{1},v_{2}$ consecutively (i.e., no other vertices exist between $v_{1}$ and $v_{2}$ on either $T_{1}$ or $T_{2}$); see \cref{fig:twist}. 
\end{definition}

\begin{definition}\label{def:twist}
A twist $(v_{1},v_{2})$ is called a \emph{full twist}, if there are exactly four trisector branches passing through $v_1$ and $v_2$ consecutively (with no other vertices between $v_1, v_2$ on any branch); see \cref{fig:twist}~right. A twist is called a \emph{partial twist}, if there are exactly two trisectors with a single branch passing through both $v_1$ and $v_2$, and two trisectors with exactly two branches each, one passing through $v_1$ and the other through $v_2$; see \cref{fig:twist}~left.
\end{definition}

\begin{figure}[h]
    \centering
    \begin{minipage}[b]{0.48\textwidth}
        \centering
        \includegraphics{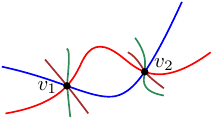}
    \end{minipage}
    \hfill
    \begin{minipage}[b]{0.48\textwidth}
        \centering
        \includegraphics{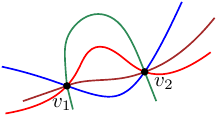}
     \end{minipage}
\caption{Left: a partial twist. Right: a full twist. Different colors denote different trisectors. }\label{fig:twist}
\end{figure}

The proof of the following lemma is deferred to the next section.

\begin{restatable}{lemma}{lemnothreetwist}\label{lem:no3twist}
A twist is either a full twist or a partial twist. 
\end{restatable}

\subsection{Full twists}\label{sec:fulltwist}

We show that the Voronoi diagram (both nearest and  farthest) exhibits a unique local structure around any full twist; see \cref{fig:fulltwist3d} for an illustration.

\begin{restatable}{lemma}{lemfulltwist}\label{lem:fulltwist}
Among the four bounded trisector arcs incident to a full twist, exactly two lie in $\nvd(L)$ and the other two lie in $\fvd(L)$, defining one bounded face in each diagram called a full-twist face. The arcs \emph{alternate} between the nearest and the farthest diagrams, alternate in the sense that the two full-twist faces intersect each other. 
\end{restatable}

\begin{proof}
Consider a local section of the trisector $T(\ell_{1},\ell_{2},\ell_{3})$. It is the intersection of three bisectors, $B(\ell_{1},\ell_{2})$, $B(\ell_{1},\ell_{3})$, and $B(\ell_{2},\ell_{3})$, which intersect transversely~\cite{Everett2009}. Around this local section, the $\nvd$ and $\fvd$ portions of the bisectors alternate, following exactly the structure shown in \cref{fig:3DVD}.

\begin{figure}[h]
    \centering
    \begin{minipage}[b]{0.48\textwidth}
        \centering
        \includegraphics[width=\textwidth]{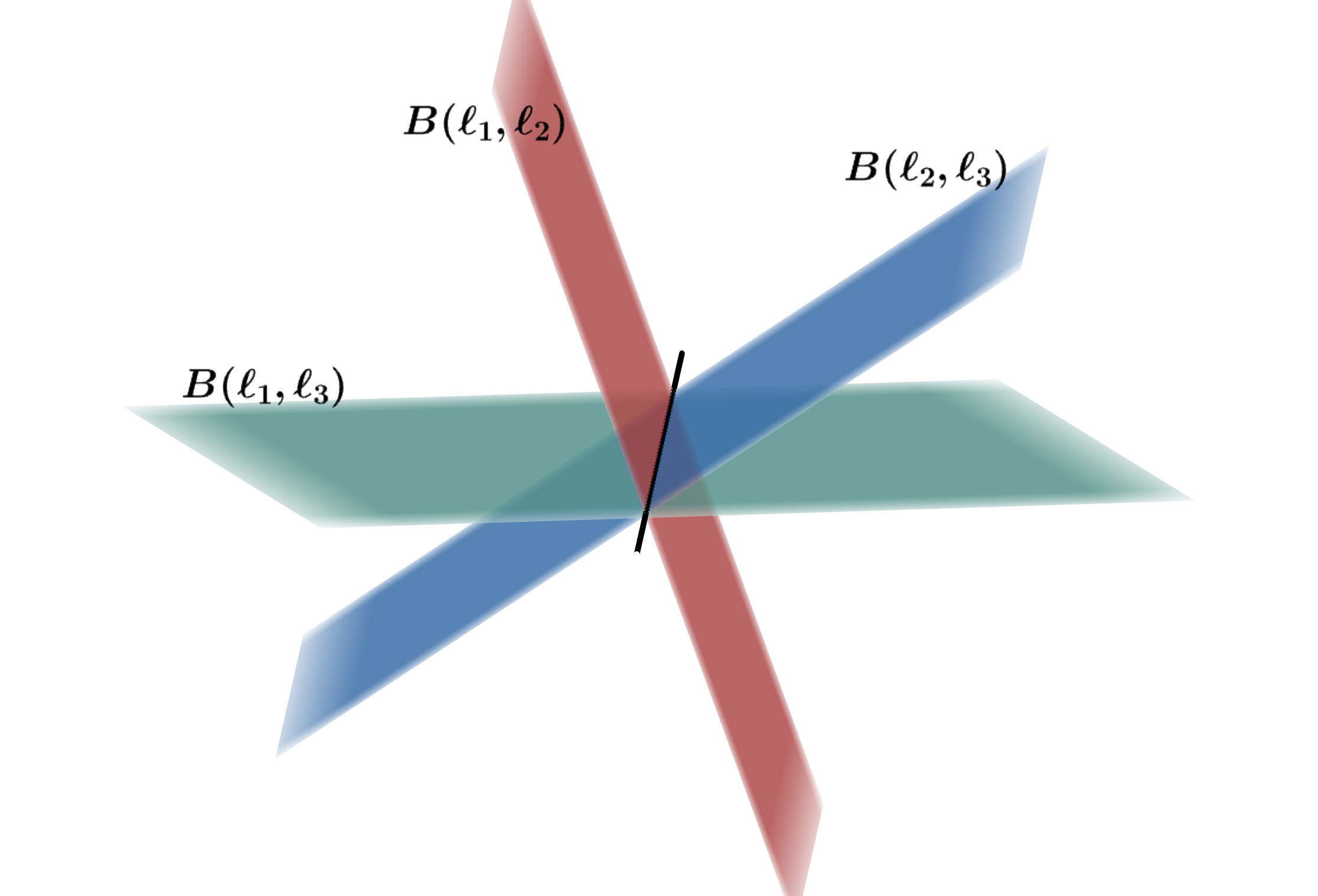}
    \end{minipage}
    \hfill 
    \begin{minipage}[b]{0.48\textwidth}
        \centering
        \includegraphics{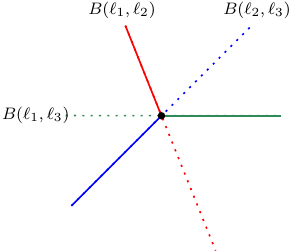}
     \end{minipage}
\caption{Left: the three bisectors intersect transversely. Right: view along the trisector; solid rays indicate the $\nvd$ portions, dotted rays indicate $\fvd$ portions. The rays that belong to the $\nvd$ and the $\fvd$ alternate when rotating around the center. }\label{fig:3DVD}
\end{figure}

Let $(v_1,v_2)$ be a full twist, see \cref{fig:fulltwist3d}.
Since the incident trisector arcs are all related, each pair of them bounds a bisector face in $\vd{k}(L)$ for some $k$, $1\leq k\leq 3$.
We will show that two of the incident trisectors arcs belong to $\nvd(L)$, and two belong to $\fvd(L)$; furthermore, they alternate.

Assume for the sake of contradiction that at least three of the trisector arcs incident to $v_1,v_2$ belong to $\nvd(L)$ (resp. $\fvd(L)$).
Then the three bounded arcs give rise to three bounded faces in $\nvd(L)$ (resp. $\fvd(L)$), which together form a bounded 3D cell in $\nvd(L)$ (resp. $\fvd(L)$).
This  contradicts the fact
that the 3D cells of these diagrams are unbounded.

Hence, among the four bounded arcs incident to the full twist vertices, exactly two arcs belong to $\nvd(L)$, defining an $\nvd$ face, and two belong to $\fvd(L)$, defining an $\fvd$ face.
Assume for the sake of contradiction that the $\nvd$ and $\fvd$ arcs do not alternate; wlog,
let the red and green arcs be in $\nvd(L)$ and the other two arcs be in $\fvd(L)$, see \cref{fig:fulltwist3d}~left.
Consider the face $F_1$ (resp. $F_2$) bounded by the red and blue (resp.\ green and brown) arcs in \cref{fig:fulltwist3d}~middle. Both faces belong to $\vd{2}(L)$ since they are incident to both $\nvd$ and $\fvd$ edges.
The two faces $F_1,F_2$ intersect at a curve that is not an edge of the diagram (which can be proven formally using the fact that related bisectors intersect transversely), see \cref{fig:fulltwist3d}~middle, deriving a contradiction.
\end{proof}

\begin{figure}[h]
    \centering
    \begin{minipage}[b]{0.32\textwidth}
        \centering
        \includegraphics[width=0.8\textwidth]{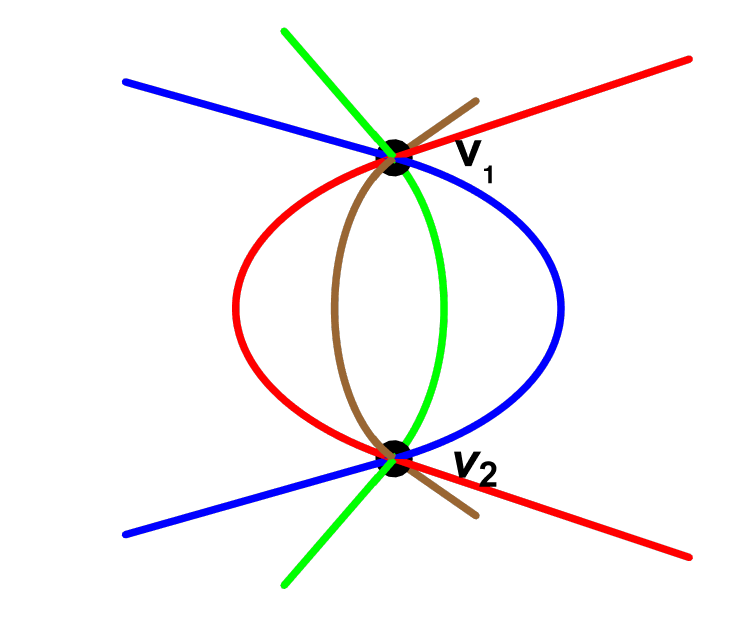}
    \end{minipage}
    \hfill
    \begin{minipage}[b]{0.32\textwidth}
        \centering
        \includegraphics[width=0.8\textwidth]{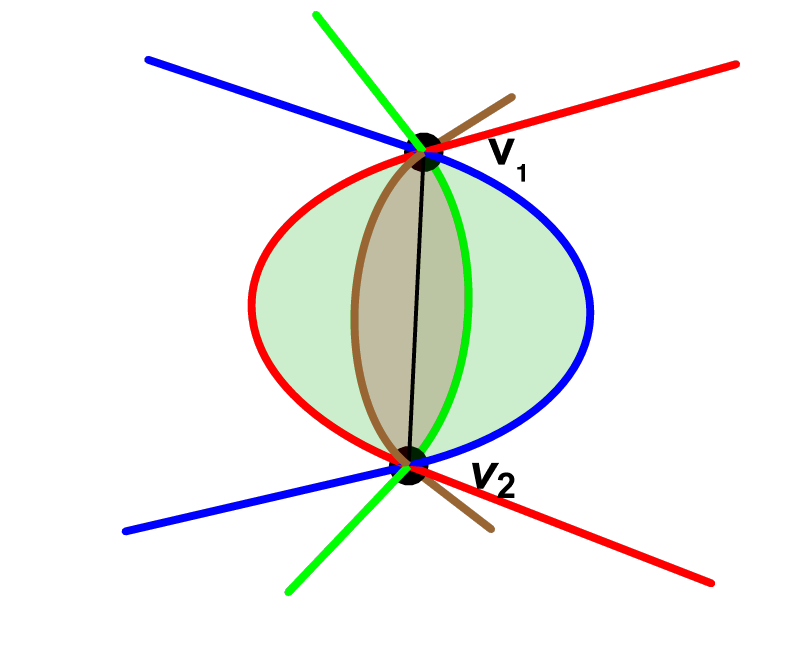}
     \end{minipage}
     \hfill
    \begin{minipage}[b]{0.32\textwidth}
        \centering
        \includegraphics[width=0.8\textwidth]{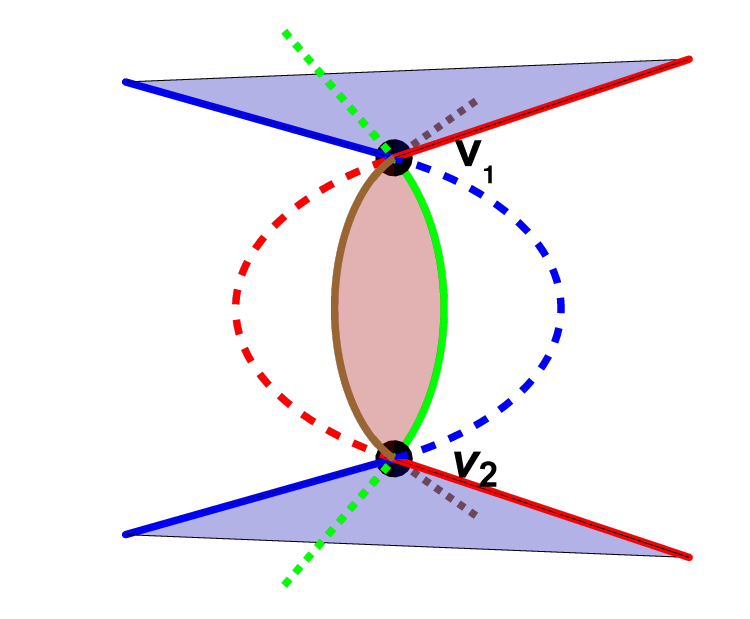}
     \end{minipage}
     \hfill
\caption{Left: a 3D view of a full twist. Middle: illustration for the proof of \cref{lem:fulltwist}. Faces $F_1$ and $F_2$ are highlighted, intersecting at the black segment which is not an edge of the diagram.  Right: the unique structure of the $\nvd(L)$ around the full twist. }\label{fig:fulltwist3d}
\end{figure}

The converse of the above lemma is also true, which will be shown in \cref{lem:partialtwist}. Next, we describe a local modification on a Voronoi diagram associated with a full twist.

\subparagraph{Local modification.} Refer to \cref{fig:fulltwistdynamic}. Start with an $\nvd$ face incident to two edges from different trisectors (e.g., the blue face on $B(\ell_{1},\ell_{2})$ incident to the red and blue edges). Move the interior of the edges closer until they intersect twice (at vertices $v_{1},v_{2}$). Consequently, each of the two original edges is split into three pieces, and the middle arcs no longer belong to the $\nvd$. The original $\nvd$ face is split into two faces. Two new bounded edges from the other two trisectors appear in the diagram (the green and brown edges in \cref{fig:fulltwistdynamic}), which form a full-twist face in the $\nvd$. The cells of $\ell_{3}$ and $\ell_{4}$ touch each other along the full-twist face. The local modification applies to the $\fvd$ in the same way.

\begin{figure}[h]
    \centering
    \includegraphics[draft=false]{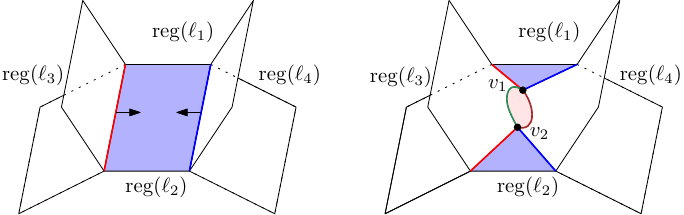}
    \caption{Left: local $\nvd$/$\fvd$ structure; right: the modified diagram after adding a full twist. }\label{fig:fulltwistdynamic}
\end{figure}

We say that the diagram in \cref{fig:fulltwistdynamic}~right is obtained from the diagram in \cref{fig:fulltwistdynamic}~left by ``adding a full twist''. The local operation can be reversed by ``removing a full twist''. Since the effect of a full twist is purely local, we first focus on the global structure, by restricting attention to Voronoi diagrams whose trisector arrangements have no full twists.

\deleted{
  With a slight abuse of terminology, we say that the diagram in \cref{fig:fulltwistdynamic}~right is obtained from the one on the left by ``adding a full twist''. The local operation can be reversed by ``removing a full twist''. Since the effect of a full twist is purely local, we first focus on the global structure, by restricting attention to Voronoi diagrams whose trisector arrangements have no full twists.
  }

\subsection{Partial twists and further properties}

We illustrate the local structure of the Voronoi diagrams around twists, and prove that any twist in the trisector arrangement must be either full or partial.

\begin{restatable}{lemma}{lemtrisystemone}\label{lem:trisystem1}
Let $(v_{1},v_{2})$ be a twist and let $(e_{1},e_{2})$ be the two trisector arcs incident to both $v_1,v_2$, as shown in \cref{fig:twistlemma}. Assume that a third trisector has two distinct branches, one of which passes through  $v_{1}$ and the other through $v_{2}$. Then, one of $e_1,e_2$ belongs to $\nvd(L)$ and the other belongs to $\fvd(L)$; the face bounded by $e_1,e_2$ belongs to $\vd{2}(L)$.
\end{restatable}

\begin{figure}[h]
        \centering
        \includegraphics{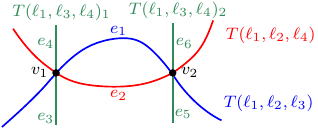}
        \caption{A local section of a trisector system, illustrating the setup of \cref{lem:trisystem1}. }\label{fig:twistlemma}
\end{figure}

\begin{proof}
  Refer to \cref{fig:twistlemma}. We distinguish two branches of trisector $T(\ell_1,\ell_3,\ell_4)$ as ${T(\ell_1,\ell_3,\ell_4)}_1$ and ${T(\ell_1,\ell_3,\ell_4)}_2$. Note that one of the edges $e_3, e_4$ in ${T(\ell_1,\ell_3,\ell_4)}_1$ is an $\nvd$ edge and the other is an $\fvd$ edge. Analogously for the edges $e_5, e_6$ in ${T(\ell_1,\ell_3,\ell_4)}_2$. Hence, wlog, we assume that $e_4,e_6$ are edges of $\nvd(L)$. 

  Assume for the sake of contradiction that $e_{1},e_{2}$ are both in $\nvd(L)$. Then, the face on $B(\ell_1, \ell_2)$ bounded by $e_{1}$ and $e_{2}$ is also in $\nvd(L)$. Further, there is a face of $\nvd(L)$ on $B(\ell_1,\ell_3)$ incident to $e_1, e_4, e_6$, and a face of $\nvd(L)$ on $B(\ell_1,\ell_4)$ incident to $e_2, e_4, e_6$. Hence, there are two faces $f_1$ and $f_2$ in $\nvd(\{\ell_1, \ell_3, \ell_4 \})$, on $B(\ell_1,\ell_3)$ and $B(\ell_1,\ell_4)$, respectively, that are incident to both branches ${T(\ell_1,\ell_3,\ell_4)}_1 $ and $ {T(\ell_1,\ell_3,\ell_4)}_2$. By~\cite{Everett2009}, the Voronoi region of a line in the $\nvd$ of three lines is invariant, which is shown schematically in \cref{fig:one_nvd_region}-middle. To satisfy the aforementioned properties, faces $f_1$ and $f_2$ must be the top and bottom faces in \cref{fig:one_nvd_region}-middle, and the two branches ${T(\ell_1,\ell_3,\ell_4)}_1 $ and $ {T(\ell_1,\ell_3,\ell_4)}_2$ must be the green branches in the same figure. Adding line $\ell_2$ to $\nvd(\{\ell_1, \ell_3, \ell_4 \})$ creates vertices $v_1,v_2$ on the branches illustrated by green in \cref{fig:one_nvd_region}-right, 
  with incident edges $e_1,e_2$ on the top and bottom faces, and a face bounded by $e_1, e_2$ in  $\nvd(L)$, see \cref{fig:one_nvd_region}-right. This new face splits the 3D cell of $\ell_{1}$ into two components, which gives a contradiction as the 3D cells of $\nvd(L)$ are topological infinite cylinders. 

\begin{figure}[h]
    \centering
    \begin{minipage}[b]{0.31\textwidth}
        \centering
        \includegraphics[width=\textwidth]{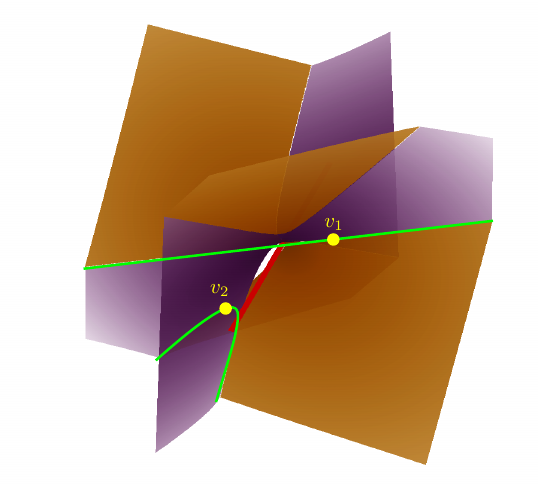}
    \end{minipage}
    \hfill
    \begin{minipage}[b]{0.33\textwidth}
        \centering
        \includegraphics[width=\textwidth]{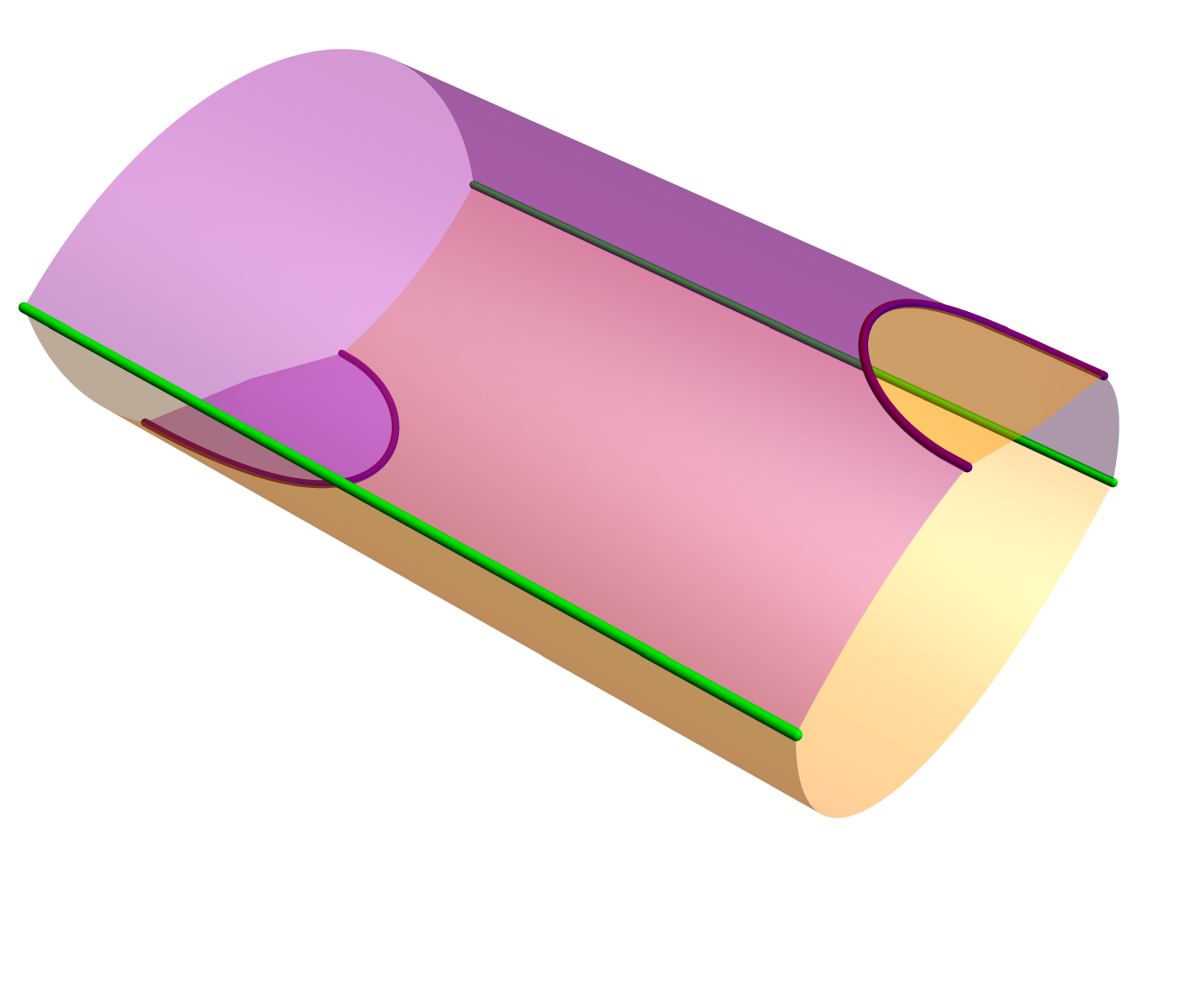}
     \end{minipage}
     \hfill
     \begin{minipage}[b]{0.33\textwidth}
        \centering
        \includegraphics[width=\textwidth]{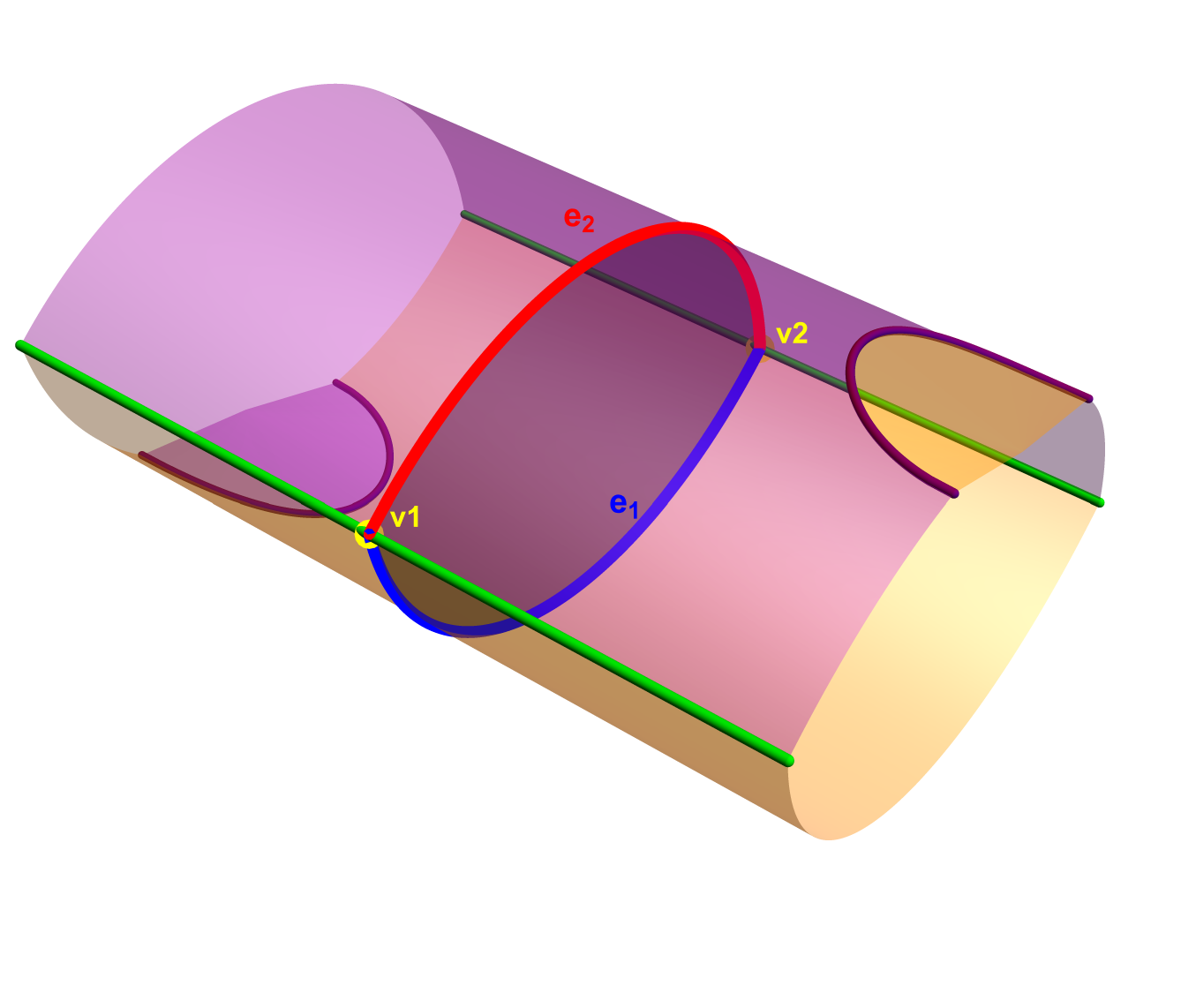}
     \end{minipage}
     \hfill
    \caption{The region of $\ell_{1}$ in $\nvd(\ell_1,\ell_3,\ell_4)$, with the two green trisector branches highlighted. Left: a real instance. Middle: a schematic representation. The purple faces belong to $B(\ell_1, \ell_4)$ and the yellow faces belong to $B(\ell_1, \ell_3)$. Right: the result of inserting $\ell_2$. }\label{fig:one_nvd_region}
\end{figure}

Next, assume that $e_{1},e_{2}$ are both in $\fvd(L)$. Analogously, we can show that there are two faces in $\fvd(\{\ell_1, \ell_3, \ell_4 \})$, on $B(\ell_1,\ell_3)$ and $B(\ell_1,\ell_4)$, respectively, that are incident to both trisector branches ${T(\ell_1,\ell_3,\ell_4)}_1 $ and $ {T(\ell_1,\ell_3,\ell_4)}_2$. However, such a pair of trisector branches does not exist, by \cref{lem:middlebranch}~(items~2 and~3). Thus, one of $e_1, e_2$ is an $\nvd$ edge and the other is an $\fvd$ edge. Consequently, the face bounded by $e_1$ and $e_2$ lies in $\vd{2}(L)$. 
\end{proof}

\begin{corollary}\label{cor:partialtwist}
For any partial twist, the face bounded by the two trisector arcs incident to both vertices of the partial twist belongs to the order-$2$ Voronoi diagram $\vd{2}(L)$.
\end{corollary}

We can now use \cref{lem:trisystem1} to prove \cref{lem:no3twist}.

\begin{proof}[Proof of \cref{lem:no3twist}]
  Let $(v_1, v_2)$ be a twist and assume for the sake of contradiction that it is neither partial nor full. Suppose that three trisector branches pass through $v_1, v_2$ consecutively. Since $(v_1, v_2)$ is not a full twist, the fourth trisector $T_4$,  must either have (a) two branches through $v_1$ and $v_2$, or (b) a single branch through $v_1$ and $v_2$ without $v_1,v_2$ being consecutive. 
  In case (a) (\cref{fig:3twist} left), by the pigeonhole principle, at least two of the three bounded edges incident to the twist must lie in $\nvd(L)$ or in $\fvd(L)$, and this contradicts \cref{lem:trisystem1}.
Thus case (a) cannot occur.
For case (b) we refer to the trisectors by their colors in \cref{fig:3twistproof} and show that case (b) cannot occur either.

\begin{figure}[h]
    \centering
    \includegraphics{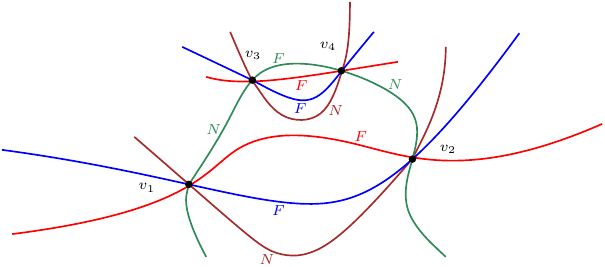}
    \caption{Labeling of the trisector system for type (b) in the proof of \cref{lem:no3twist}. The label `N' (resp.\ `F') on an edge indicates an $\nvd$ (resp.\ $\fvd$) edge. }\label{fig:3twistproof}
\end{figure}

Consider the green trisector in  \cref{fig:3twistproof} that has two vertices $v_3$ and $v_4$ between $v_1$ and $v_2$.
Then $v_3$ and $ v_4$ must lie on the same red (resp.\ blue and brown) trisector branch, otherwise that trisector would self-intersect. Hence, the trisector system matches \cref{fig:3twistproof} (ignoring the edge labels for now).
Consider the two bounded green arcs $(v_{1},v_{3})$ and $(v_{2},v_{4})$. They belong to the same diagram (either $\nvd$ or $\fvd$) since they are both incident to the green arc $(v_{3},v_{4})$. Assume that they are both $\nvd$ edges. Then the green arc $(v_{3},v_{4})$ is an $\fvd$ edge. By the same reasoning as in the proof of \cref{lem:fulltwist}, among the three bounded arcs incident to $v_{1}$ and $v_{2}$ (red, blue, and brown arcs), two belong to $\fvd(L)$ and one belongs to $\nvd(L)$. WLOG, let the brown arc belong to $\nvd(L)$ and the blue, red arcs belong to $\fvd(L)$. Consider the bisector containing  the green and blue trisectors, and the bounded face with vertices $v_{1}, v_{3},v_{4},v_{2}$. Since the face is bounded by edges from both $\nvd(L)$ and $\fvd(L)$, it is a $\vd{2}$ face. Hence, the blue edge $(v_{3},v_{4})$ is an $\fvd$ edge. Similarly, the red edge $(v_{3},v_{4})$ is an $\fvd$ edge. The four bounded edges incident to $(v_{3},v_{4})$ contradict \cref{lem:fulltwist}, since three of them  lie in $\fvd(L)$.
The case where $(v_{1},v_{3})$ and $(v_{2},v_{4})$ are  both $\fvd$ edges is analogous.
Hence,  case (b) cannot occur. 

\begin{figure}[h]
    \centering
    \includegraphics{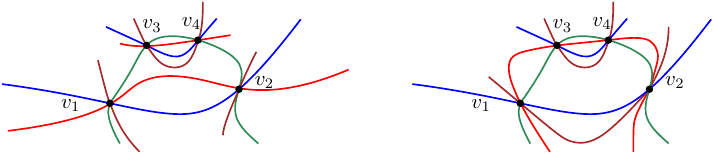}
    \caption{Additional impossible cases. }\label{fig:3twistproofadditional}
\end{figure}

Some additional cases are shown in \cref{fig:3twistproofadditional}. They can all be shown impossible using the same proof idea (inspecting the labels of edges on different faces). This completes the proof. 
\end{proof}

Case (b) of the above proof shows that  ``nested full twists'', as shown in \cref{fig:3twist}~right, are impossible. Consequently, full twists are either isolated or sequential.

\begin{figure}[h]
    \centering
    \begin{minipage}[b]{0.48\textwidth}
        \centering
        \includegraphics{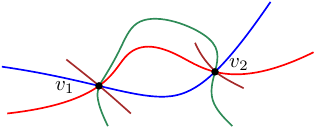}
    \end{minipage}
    \hfill
    \begin{minipage}[b]{0.48\textwidth}
        \centering
        \includegraphics{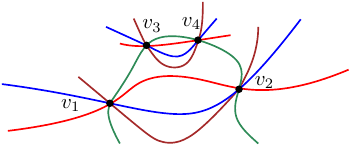}
     \end{minipage}
     \hfill
\caption{Left: a trisector system corresponding to (a) in the proof of \cref{lem:no3twist}. Right: two ``nested full twists'', corresponding to (b) in the proof of \cref{lem:no3twist}. Both are impossible.}\label{fig:3twist}
\end{figure}

\begin{restatable}{lemma}{lempartialtwist}\label{lem:partialtwist}
In  $\nvd(L)$ (resp. $\fvd(L)$), if there is a bounded face with exactly two edges and two vertices, then these vertices belong to a full twist.
\end{restatable}
\begin{proof}
If the two bounded edges both lie in $\nvd(L)$ (resp. $\fvd(L)$), then each of the remaining two trisectors (other than those of the two edges) must have exactly one branch passing through the two vertices, by \cref{lem:trisystem1}. If both remaining branches pass through the two vertices consecutively, then these vertices form a full twist. Otherwise, for at least one of the branches, there is an even number of vertices between the two vertices of the twist, which contradicts \cref{lem:no3twist}.
\end{proof}

\cref{lem:partialtwist} provides a concrete criterion to detect full twists by inspecting two projected trisectors (on a single bisector). It is used to classify Voronoi diagrams without full twists.

We finish the section with another technical lemma that generalizes \cref{lem:trisystem1}. Its proof uses similar ideas to \cref{lem:trisystem1}.

\begin{restatable}{lemma}{lemtrisystemtwo}\label{lem:trisystem2}
  Assume that related trisectors interact as shown in \cref{fig:trisys2}. 
  Then one of the edges $e_1$ and $e_7$ belongs to $\nvd(L)$ and the other belongs to $\fvd(L)$. Consequently, the bounded face on $B(\ell_1,\ell_2)$ incident to $v_1,v_3,v_4,v_2$ belongs to $\vd{2}(L)$,
  and the vertices $v_3,v_4$ form a full twist.
\end{restatable}
\begin{figure}[h]
        \centering
        \includegraphics{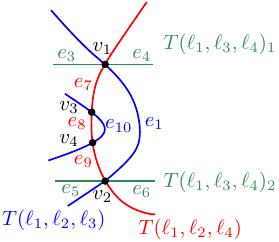}
        \caption{A local trisector arrangement for \cref{lem:trisystem2}. }\label{fig:trisys2}
\end{figure}
\begin{proof}
For simplicity, we refer to the trisectors by their colors in \cref{fig:trisys2}. Assume to the contrary that $e_1$ and $e_7$ (and thus $e_9$ as well) both belong to  $\nvd(L)$ (resp. $\fvd(L)$). Similarly to the proof of \cref{lem:trisystem1}, there is a face of $\nvd(\{\ell_1, \ell_3, \ell_4\})$ (resp. $\fvd(\{\ell_1, \ell_3, \ell_4\})$) on $B(\ell_1,\ell_3)$ that is incident to both trisector branches ${T(\ell_1,\ell_3,\ell_4)}_1$ and ${T(\ell_1,\ell_3,\ell_4)}_2$. We claim that on bisector $B(\ell_1,\ell_4)$, the two trisector branches ${T(\ell_1,\ell_3,\ell_4)}_1$ and ${T(\ell_1,\ell_3,\ell_4)}_2$ do not bound a common face in $\nvd(\{\ell_1, \ell_3, \ell_4\})$ (resp. $\fvd(\{\ell_1, \ell_3, \ell_4\})$).

Assume to the contrary that they do.
For the $\fvd$, the claim can be shown following the same logic as in the proof of \cref{lem:trisystem1}: such a pair of trisector branches ${T(\ell_1,\ell_3,\ell_4)}_1$ and ${T(\ell_1,\ell_3,\ell_4)}_2$ does not exist, by \cref{lem:middlebranch}~(items~2 and~3).

For the $\nvd$, we can also use a similar idea as in the proof of \cref{lem:trisystem1}: such a pair of trisector branches can only be the two green branches shown in \cref{fig:trisys2proof1}, because of the unique structure of the Voronoi diagram of three lines~\cite{Everett2009}.
Consider $\nvd(L)$ as the result of adding line $\ell_2$ to the $\nvd(\{\ell_1, \ell_3, \ell_4\})$, which creates new vertices and edges shown in \cref{fig:trisys2proof1}. Consider the face $f$ bounded by $e_8$ and $e_{10}$ on $B(\ell_1, \ell_2)$. By inspecting the bisector $B(\ell_1, \ell_2)$ with the red and blue trisectors in \cref{fig:trisys2}, it follows that every point in the interior of $f$ is closest to $\ell_3$, However, this face $f$ is contained in the cell of $\ell_1$ in $\nvd(\{\ell_1, \ell_3, \ell_4\})$. This gives a contradiction.

\begin{figure}[h]
        \centering
        \includegraphics[width=0.6\textwidth]{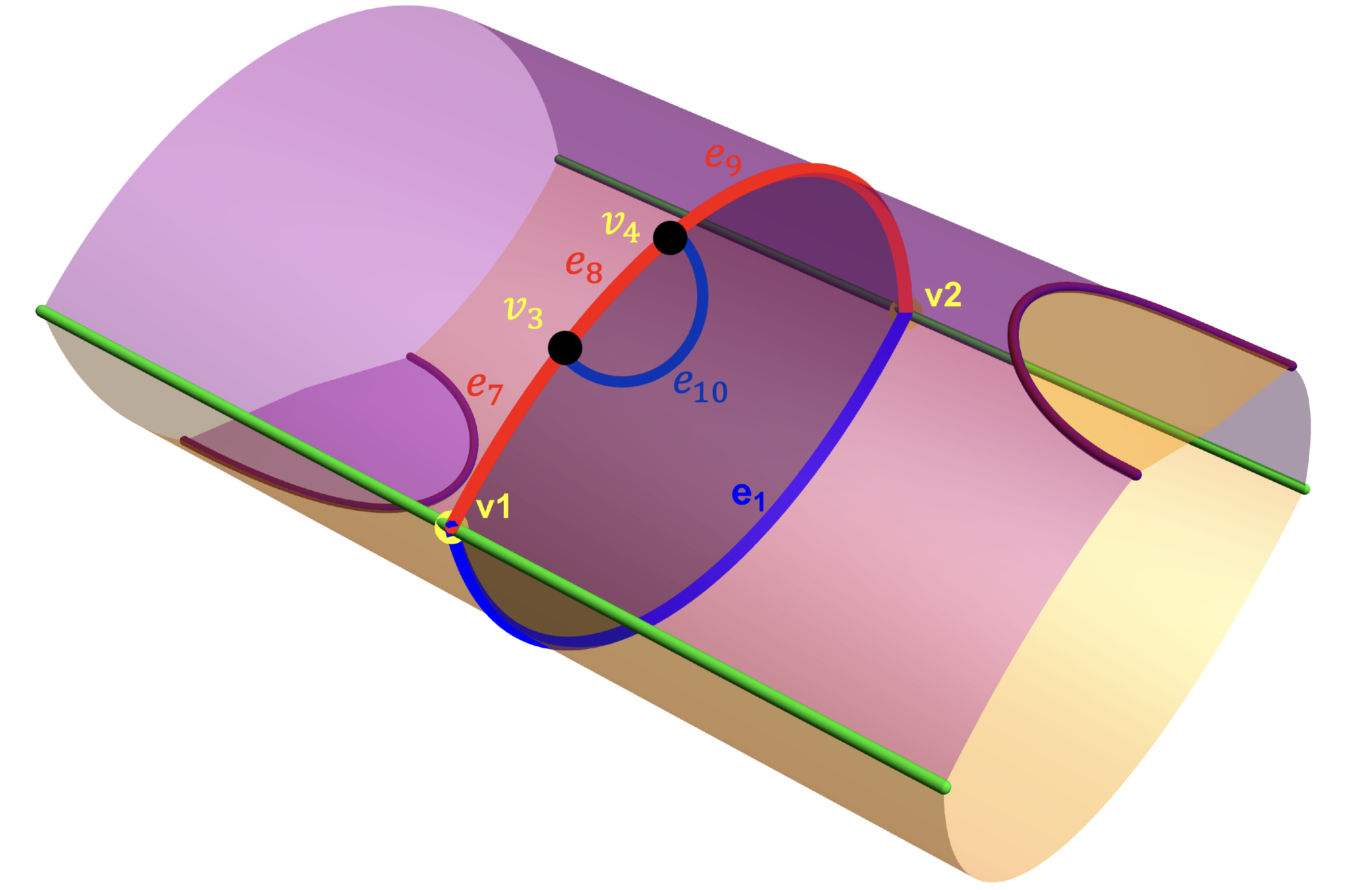}
        \caption{Illustration for the proof of \cref{lem:trisystem2}: a schematic 3D view of the cell of $\ell_1$ in $\nvd(\{\ell_1, \ell_3, \ell_4\})$ with some additional edges and vertices. The purple faces belong to $B(\ell_1, \ell_4)$ and the yellow faces belong to $B(\ell_1, \ell_3)$.}\label{fig:trisys2proof1}
\end{figure}

We have shown that the two green branches ${T(\ell_1,\ell_3,\ell_4)}_1$ and $ {T(\ell_1,\ell_3,\ell_4)}_2$ bound a common $\fvd$ face on $B(\ell_1,\ell_3)$, but not on $B(\ell_1,\ell_4)$. For the remaining proof, we assume that both $e_1$ and $e_7$ are edges of $\fvd(L)$. The case of $\nvd(L)$ can be shown analogously. 
By \cref{lem:middlebranch}, item 3, on $B(\ell_1,\ell_3)$, there is only one $\fvd$ face that is bounded by two distinct trisector branches, namely the middle branch and the U branch. Thus, the two green branches ${T(\ell_1,\ell_3,\ell_4)}_1$ and $ {T(\ell_1,\ell_3,\ell_4)}_2$ are the middle and U branch respectively on $B(\ell_1,\ell_3)$.
By \cref{lem:middlebranch} again, the bisector configurations of $B(\ell_1,\ell_3)$ and $B(\ell_1,\ell_4)$ are partially depicted in \cref{fig:trisys2proof}.

Consider the fourth trisector $T(\ell_2, \ell_3, \ell_4)$. It must have two different branches passing through $v_1$ and $ v_3$  (otherwise by \cref{lem:no3twist}, $v_1$ and $ v_3$ form a full twist which is a contradiction). Similarly, $T(\ell_2, \ell_3, \ell_4)$ must have two different branches passing through $v_2$ and $ v_4$, and two different branches passing through $v_3$ and $ v_4$. By \cref{lem:middlebranch}~(items~1 to~3), it is not hard to see that such branches cannot exist.

Consequently, one of the edges $e_1$ and $e_7$ belongs to $\nvd(L)$ and the other belongs to $\fvd(L)$. The face on $B(\ell_1,\ell_2)$ bounded by $v_1,v_3,v_4,v_2$ is incident to both $e_1$ and $e_7$, therefore it belongs to $\vd{2}(L)$. Hence, the face on $B(\ell_1,\ell_2)$ bounded by $e_8$ and $e_{10}$ either belongs to $\nvd(L)$ or $\fvd(L)$. In either case, it implies that the vertices $v_3$ and $v_4$ form a full twist, by \cref{lem:partialtwist}. This completes the proof.
\end{proof}

\begin{figure}[h]
    \centering
    \includegraphics{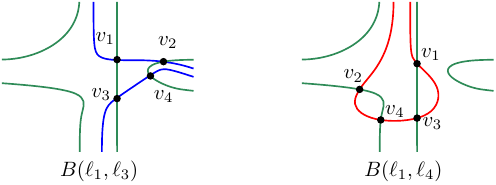}
    \caption{Illustration for the proof of \cref{lem:trisystem2}.}\label{fig:trisys2proof}
\end{figure}

\section{Identifying the Voronoi topologies by exhaustive search}\label{sec:method}

In this section, we present an exhaustive search algorithm to classify the Voronoi diagram of four lines, assuming no full twists. To this end, we use the notion of a \emph{configuration} as an abstraction for a projected bisector and the two relevant trisectors it contains.

We call a planar curve with four non-intersecting components \emph{trisector-like}, if it satisfies the properties of the projected trisector of \cref{lem:middlebranch}, in terms of the asymptotes and the middle trisector branch. Each trisector-like curve partitions the plane into faces, which we label as $\nvd$ or $\fvd$ according to \cref{lem:middlebranch}, item 3. If a trisector-like curve is indeed the projected trisector of a set of three lines $L$, its faces are projections of $\nvd(L)$ and $\fvd(L)$.

We define a \emph{configuration} on the plane as the overlay of two trisector-like curves. A configuration  partitions the plane into faces, which we label as $\nvd$, $\fvd$, or $\vd{2}$. The faces labeled as $\nvd$ (resp. $\fvd$) are the common intersections of the respectively labeled faces of the two trisector-like curves; 
any other face is labeled $\vd{2}$. A configuration is called \emph{realizable} if there exist four lines $L$ that define a bisector and two related trisectors whose projection has the same structure as the configuration. If a configuration is realizable, then the projected faces of $\nvd(L)$ (resp. $\fvd(L)$ and $\vd{2}(L)$) on that bisector  correspond exactly to the configuration faces labeled as $\nvd$ (resp. $\fvd$ and $\vd{2}$). For example, the configuration in \cref{fig:config} right is realizable by the projected bisector and trisectors on the left.

\begin{figure}[h]
    \centering
    \begin{minipage}[b]{0.44\textwidth}
        \centering
        \includegraphics[width=0.5\textwidth]{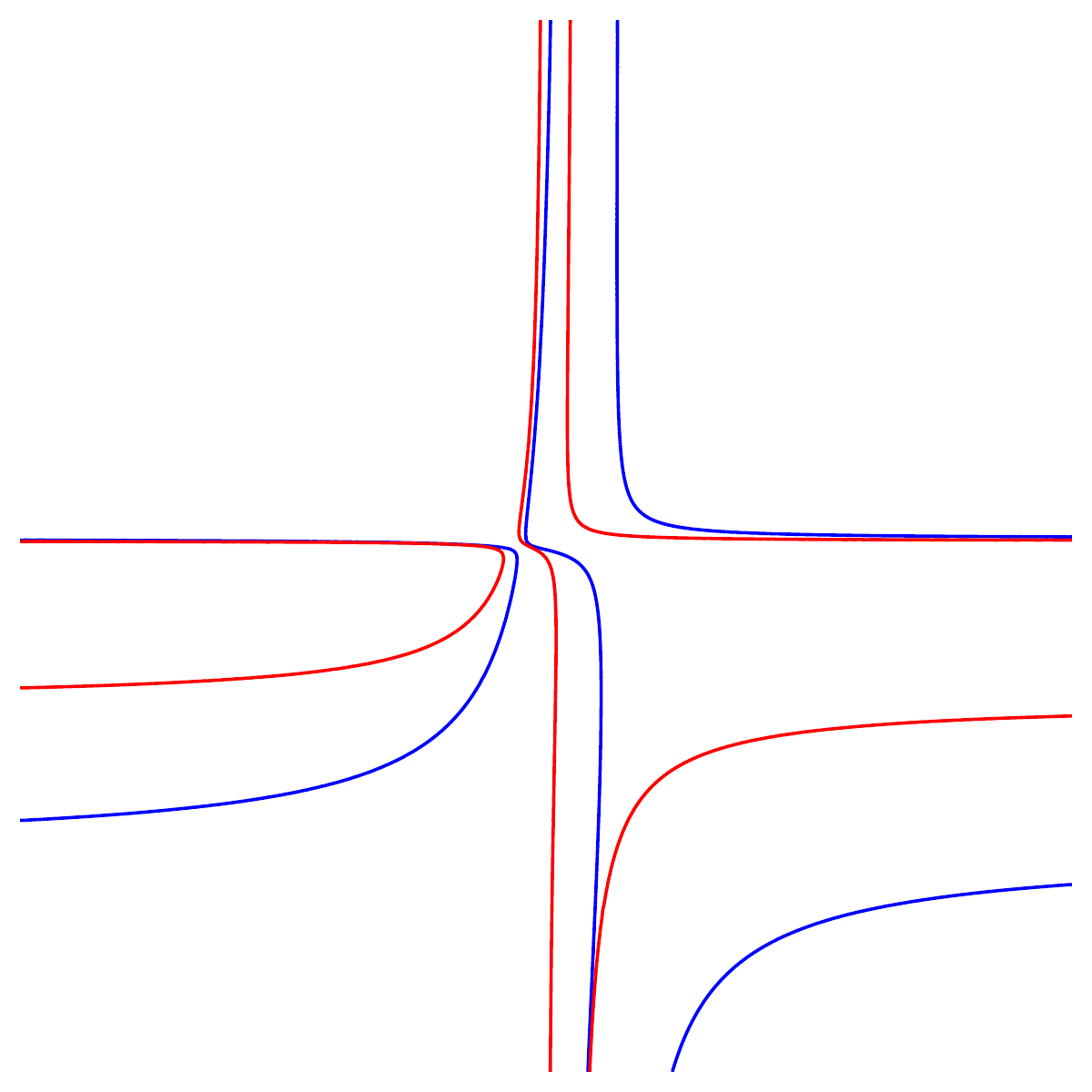}
    \end{minipage}
    \hfill
    \begin{minipage}[b]{0.54\textwidth}
        \centering
        \includegraphics{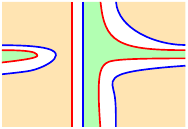}
     \end{minipage}
\caption{Left: a real projected bisector with two projected trisectors on it. Right: a configuration. Green faces are $\nvd$ faces, orange faces are $\fvd$ faces, and white faces are $\vd{2}$ faces. }\label{fig:config}
\end{figure}

\subsection{Properties of configurations}\label{sec:config}

We give necessary conditions that restrict how trisectors and their asymptotes interact. These conditions rule out many unrealizable configurations. By \cref{lem:VEFCnlinesunbounded}, the map $\gmap(\nvd(L))$ has $12$ vertices that indicate the unbounded Voronoi edges.

\begin{restatable}{lemma}{lemfourlineunbounded}\label{lem:4line-unbounded}
Given four lines $L$ in general position, one of the following holds:
\begin{enumerate}
	\item one trisector induces 6 vertices in $\gmap(\nvd(L))$ and each other trisector induces 2;
	\item one trisector induces 0 vertices in $\gmap(\nvd(L))$ and each other trisector induces 4.
\end{enumerate}
\end{restatable}

\begin{proof}

\begin{figure}[h]
	\centering
	\includegraphics{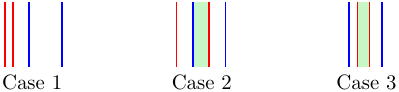}
	\caption{Three cases for the unbounded parts of two vertical asymptotes. The unbounded $\nvd$ face (when present) is shown in green. In Case 1, there is none; in Cases 2 and 3, there is one. Consequently, in Case 1, the red and blue trisectors induce 0 vertices to $\gmap(\nvd(L))$, while in Cases 2 and 3, they induce 4 vertices. }\label{fig:asymptotes}
\end{figure}

Consider $\gmap(\nvd(L))$. It has $12$ vertices which correspond to the unbounded edges of $\nvd(L)$. Fix one projected bisector with the two projected trisectors it contains. Each projected trisector has two vertical asymptotes. Up to symmetry, there are three cases, shown in \cref{fig:asymptotes}. By \cref{lem:middlebranch}, any unbounded $\nvd$ face lies between the two red asymptotes and
the two blue asymptotes.
Horizontal asymptotes behave analogously. Hence, on one bisector, the two trisectors together contribute $0$, $4$, or $8$ vertices to $\gmap(\nvd(L))$. Let $n_1,n_2,n_3,n_4$ be the number of vertices of $\gmap(\nvd(L))$ that are induced by the four trisectors. Then \(\sum_{i=1}^4 n_i=12\), and \(n_i+n_j\in\{0,4,8\}\), for any $i\neq j$.
Up to permutation, the only solutions are $(0,4,4,4)$ and $(2,2,2,6)$. This completes the proof.
\end{proof}

A bisector is called an $(i,j)$-bisector if the two trisectors it contains induce $i$ and $j$ vertices in $\gmap(\nvd(L))$. As a corollary of \cref{lem:4line-unbounded}, we obtain the following classification of bisectors.

\begin{corollary}\label{lem:4line-unboundedbisector}
Given four lines $L$ in general position, one of the following holds:
\begin{enumerate}
    \item there are three $(2,2)$-bisectors and three $(2,6)$-bisectors (with respect to $\gmap(\nvd(L))$);
    \item there are three $(0,4)$-bisectors and three $(4,4)$-bisectors  (with respect to $\gmap(\nvd(L))$).
\end{enumerate}
\end{corollary}

\begin{restatable}{lemma}{lemimpossibleconfig}\label{lem:impossibleconfig2}
Consider a bisector with two projected trisectors $T_{i},T_{j}$. Let $x_i^-, x_i^+$ ($x_{i}^{-}<x_{i}^{+}$) (resp.\ $y_{i}^{-}, y_{i}^{+}$) be the two vertical (resp.\ horizontal) asymptotes of $T_{i}$; analogously for $T_{j}$. It is impossible that simultaneously $x_{i}^{-}<x_{j}^{-}<x_{j}^{+}<x_{i}^{+} $ and $ y_{j}^{-}<y_{i}^{-}<y_{i}^{+}<y_{j}^{+}$. 
\end{restatable}

\begin{proof}
By the formula in the proof of \cref{lem:trisector}, the vertical asymptotes of a trisector $T_{i}$ are $x_{i}^{\pm}=\frac{e_{i}A}{a}\pm \sqrt{A\Delta_{i}}$, where $A=1+a^{2}$ and $\Delta_{i}=d_i^2+e_i^2+1$. Thus, $x_{i}^{-}<x_{j}^{-}<x_{j}^{+}<x_{i}^{+}$ is equivalent to \(\sqrt{\Delta_j}-\sqrt{\Delta_i}<\frac{(e_i-e_j)\sqrt{A}}{a}<\sqrt{\Delta_i}-\sqrt{\Delta_j}\). Similarly, the horizontal asymptotes of a trisector $T_{i}$ are $y_{i}^{\pm}=\frac{d_{i}A\pm \sqrt{A\Delta_{i}}}{a}$. Therefore, $y_{j}^{-}<y_{i}^{-}<y_{i}^{+}<y_{j}^{+}$ is equivalent to \(\sqrt{\Delta_i}-\sqrt{\Delta_j}<(d_i-d_j)\sqrt{A}<\sqrt{\Delta_j}-\sqrt{\Delta_i}\). It is not hard to verify that the two relations contradict each other. This completes the proof.
\end{proof}

\begin{figure}[h]
        \centering
        \includegraphics{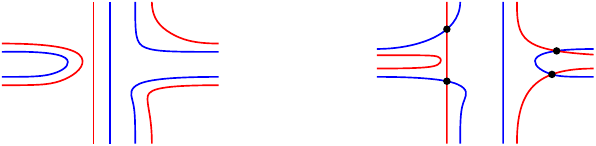}
        \caption{Left: a $(0,8)$-bisector, impossible by \cref{lem:4line-unboundedbisector}. Right: impossible by \cref{lem:impossibleconfig2}.}\label{fig:impossibleconfig}
\end{figure}

We apply \cref{lem:4line-unboundedbisector,lem:impossibleconfig2} to filter out unrealizable configurations such as those illustrated in \cref{fig:impossibleconfig}.

We say that two related trisectors (on a bisector) are \emph{parallel} (denoted ``$\cdot \parallel \cdot $''), if they are both $x$-monotone or both $y$-monotone; otherwise, they are not parallel (``$\cdot \nparallel \cdot $''). Parallelism of trisectors is transitive. The following is used as a filter in \cref{sec:algo}.

\begin{restatable}{lemma}{lemparallel}\label{lem:parallel}
For any three pairwise related trisectors $T_{i},T_{j},T_{k}$, if $T_{i}\parallel T_{j}$ and $T_{j}\parallel T_{k}$, then $T_{i}\parallel T_{k}$; if $T_{i}\nparallel T_{j}$ and $T_{j}\nparallel T_{k}$, then $T_{i}\parallel T_{k}$; see \cref{fig:parallel}.
\end{restatable}

\begin{figure}[h]
    \centering
    \includegraphics[draft=false]{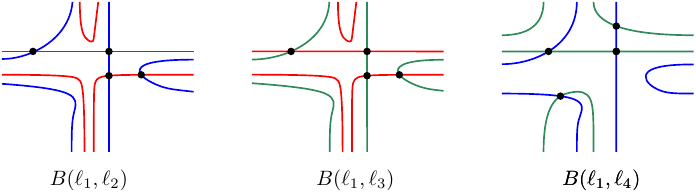}
    \caption{Three bisectors and trisectors; $T(\ell_{1},\ell_{2},\ell_{3})$ is red,  $T(\ell_{1},\ell_{2},\ell_{4})$ is blue, $T(\ell_{1},\ell_{3},\ell_{4})$ is green. They are not realizable by \cref{lem:parallel}. }\label{fig:parallel}
\end{figure}

\begin{proof}
WLOG, we may assume that $T_i, T_j, T_k$ are $T(\ell_{1},\ell_{2},\ell_{3})$, $T(\ell_{1},\ell_{2},\ell_{4})$, and $T(\ell_{1},\ell_{3},\ell_{4})$ respectively. Consider $\gmap(\fvd(L))$ and the following function that projects away from the line~$\ell_{1}$: for $x\in T$, let $y_x$ be the nearest point on $\ell_{1}$ to $x$, let $\pi: T\to \Gamma$ be the projection such that $\pi(x)=\overrightarrow{y_x x}\cap \Gamma$. Consider $\pi(M)$, where $M$ is the middle branch of $T$. Since $M$ has only one asymptote (\cref{lem:middlebranch}), $M\cap \Gamma$ are two points that are antipodal in the limit when $\Gamma\to\infty$. To be precise, let $\{x_{1},x_{2}\}=M\cap \Gamma$, then $\lim\limits_{\Gamma \to \infty}x_{1}+x_{2}=0$. For the rest of the proof, all statements hold under the assumption that $\Gamma\to\infty$. 

The intersections of $\pi(T(\ell_{1},\ell_{2},\ell_{3}))$ and $\pi(T(\ell_{1},\ell_{2},\ell_{4}))$ are in one-to-one correspondence with the intersections of $T(\ell_{1},\ell_{2},\ell_{3})$ and $T(\ell_{1},\ell_{2},\ell_{4})$. Two middle trisector branches are parallel if and only if they intersect an even number of times (including 0), equivalently, their projections intersect an even number of times. It is not hard to verify that this property is transitive.
\end{proof}

Next, we consider \emph{$6$-tuples of configurations}, where each configuration corresponds to one bisector associated with four lines. A configuration tuple is realizable if and only if there exists a set of four lines that realizes every configuration in the tuple. In a realizable  $6$-tuple, each trisector-like curve (12 in total) represents a trisector $T(\ell_i, \ell_j, \ell_k)$ whose representation must be consistent with all the elementary properties listed in \cref{sec:preliminaries} and \cref{lem:middlebranch} (by viewing a configuration as a bisector and a trisector-like curve as a trisector).
Note that there are exactly three trisector-like curves that represent $T(\ell_1,\ell_2,\ell_3)$, and they lie on the configurations of $B(\ell_1,\ell_2)$, $B(\ell_1,\ell_3)$, and $B(\ell_2,\ell_3)$.

\subsection{A search algorithm, proving
  \cref{thm:topologywithoutfulltwist,thm:topology}}\label{sec:algo}

To classify the topology of the Voronoi diagram of four lines without full twists, we use the following exhaustive search algorithm.

\begin{description}
\item[\textbf{Phase 1: Generate configurations.}] We generate all configurations with at most 8 vertices, excluding those that lead to full twists. Phase~1 works in two steps. 
  
  In Step~1, we generate simple configurations that have no twists. By \cref{lem:middlebranch}, such a configuration
  is uniquely determined by the asymptotes and the choice of the middle branches.
  We construct the simple configurations by enumerating all asymptote arrangements and all choices of middle branches of  the two trisector-like curves. We label the faces of each configuration as described above.  
  
In Step~2, we add twists to simple configurations, as many as possible, provided that the number of resulting vertices is  at most 8. 
To add a twist, we select two configuration edges that bound an existing face, fix their endpoints, and intersect their interior twice more. To avoid creating full twists, we only consider edge pairs incident to a face with label $\vd{2}$ (see \cref{fig:addtwist}~(a)-(b)). This is sufficient because intersecting any other pair of edges produces a full twist, as shown in the proof of \cref{thm:topologywithoutfulltwist}.

\begin{figure}[h]
    \centering
    \includegraphics[draft=false]{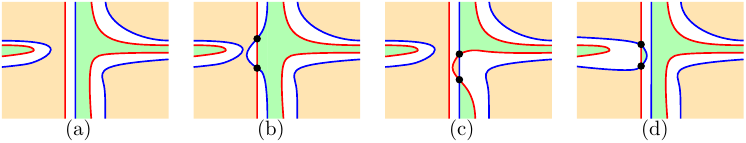}
    \caption{(a) A simple configuration. Faces labeled $\nvd$ (resp. $\fvd$) are shown in green (resp.\ orange).\ (b) Add a twist by intersecting two edges bounding a $\vd{2}$ face (in white).\ (c), (d) Add a twist by intersecting two edges bounding an $\nvd$ or $\fvd$ face; both lead to a full twist. }\label{fig:addtwist}
\end{figure}

\item[Phase 2: Filter out configurations.]  We filter out any configurations of Phase 1 that are not realizable by applying filters derived from \cref{lem:4line-unboundedbisector,lem:impossibleconfig2}. 

\item[Phase 3: Filter out unrealizable configuration tuples.] We generate all possible configuration $6$-tuples from the configurations that remain at the end of Phase~2. We filter out unrealizable tuples by applying filters derived from the elementary properties of \cref{sec:preliminaries}, \cref{lem:4line-unboundedbisector,lem:parallel,thm:VEFCnlines}, and  obtain the set $\mathcal{P}$ of configuration tuples that survive these filters. 
Since the search space is large, we implement the filters in a program~\cite{zeyu_program_4VLD}, which outputs 15 remaining configuration tuples at the end of Phase~3.  
\end{description}

Each surviving configuration tuple in $\mathcal{P}$ can be realized by a set of four lines. We give an example for each one in Appendix~\ref{app:topology}. Each configuration tuple in $\mathcal{P}$ corresponds to a unique Voronoi diagram topology as shown in the proof of \cref{thm:topologywithoutfulltwist}. Then, straightforward inspection of the tuples and their topologies reveals the following properties.

\begin{observation}\label{observation}
Assuming no full twists, the following hold for $\nvd(L)$ and $\fvd(L)$.
\begin{enumerate}
\item Topologies with 0 or 2 vertices have no partial twists in their trisector system. 
\item Topologies with 4 or more vertices have at least one partial twist.
\item No topology has 8 vertices; hence 8 Voronoi vertices require the presence of full twists.
\item The 2D faces of both $\nvd(L)$ and $\fvd(L)$ are all unbounded.
\item Each distinct $\fvd$ topology has a distinct map of unbounded features $\gmap(\fvd(L))$.
\item There are only two distinct maps $\gmap(\nvd(L))$, each corresponding to a case of \cref{lem:4line-unbounded}. 
\item A partial twist involves at least two middle trisector branches, each incident to exactly one vertex of the partial twist.
\item The 1-skeleton of $\nvd(L)$ has a single connected component that contains all vertices. This does not hold for $\fvd(L)$.
\item  The edges of $\fvd(L)$ are unbounded except those incident to both vertices of a partial twist.
\end{enumerate}
\end{observation}

Before proving \cref{thm:topologywithoutfulltwist}, we show that Step~1 in Phase~1 works correctly.

\begin{lemma}\label{lem:simpleconfig}
There is a way to enumerate the set of all configurations with no twists, i.e., the simple configurations, such as the one in Step~1 in Phase~1 of the exhaustive search algorithm.
\end{lemma}
\begin{proof}
For simplicity, let the two trisector-like curves in a configuration be red and blue respectively. It is easy to see that a configuration with no twists is uniquely determined by (1) the arrangement of the red and blue asymptotes (both horizontal and vertical) and (2) the choices of the red and blue middle branches. 
WLOG, we fix the blue trisector-like curve to have the shape  illustrated in \cref{fig:branch}, i.e., the asymptote of its middle branch is the right vertical asymptote and its shape is vertically symmetric. 

The horizontal (resp. vertical) asymptotes of the red trisector-like curve have four (resp. six) different possible arrangements with respect to the blue asymptotes, shown in \cref{fig:horizontal,fig:vertical} respectively. 

\begin{figure}[h]
	\centering
	\includegraphics{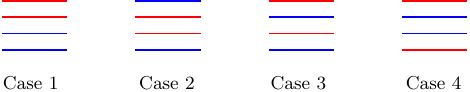}
	\caption{Four horizontal asymptote arrangements. }\label{fig:horizontal}
\end{figure}

\begin{figure}[h]
	\centering
	\includegraphics{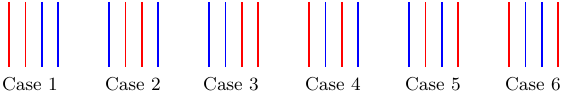}
	\caption{Six vertical asymptote arrangements. }\label{fig:vertical}
\end{figure}

Next, we determine the asymptote of the middle branch of the red trisector-like curve, which has four choices: upper horizontal, lower horizontal, left vertical, and right vertical. This gives a systematic way to enumerate all simple configurations, which is adopted in Step~1.
\end{proof}

As an example of the enumeration process above, if we choose horizontal asymptote arrangement (Case 1), vertical asymptote arrangement (Case 2), and left vertical asymptotes for the red middle branch, we obtain the simple configuration shown in \cref{fig:base_config}. Further, the proof shows that the number of simple configurations is upper bounded by $96(=4\times 6\times 4)$. In fact, the number is smaller as some of them are symmetric to each other by rotation, mirror symmetry, or swap of colors. 

\begin{figure}[h]
	\centering
	\includegraphics{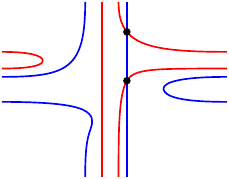}
	\caption{A simple configuration with two vertices.}\label{fig:base_config}
\end{figure}

Next, we explain how to construct the Voronoi diagram from a configuration tuple output at the end of Phase~3.
For this purpose, we assume that the edges and vertices in each configuration of the tuple are labeled in a way consistent with the filtering conditions of Phase~3. See Appendix~\ref{app:topology} for the 15 labeled configuration tuples.

\begin{lemma}\label{lem:uniquetopology}
Given a labeled configuration tuple, there is at most one matching topology for the nearest Voronoi diagram of four lines. Similarly for the farthest counterpart.
\end{lemma}
\begin{proof}
We show how to reconstruct the diagram uniquely (up to mirror symmetry) given the label information of the configuration tuple. First, we may assume that all edges are oriented. If an edge is incident to at least one vertex, then the orientation is given by the label of the vertex. Otherwise, the edge is a whole trisector branch, in which case locally around the edge, the diagram is the Voronoi diagram of three lines, which has a fixed structure~\cite{Everett2009}. Hence, an edge in a labeled configuration tuple has a fixed orientation.

Start with an arbitrary edge and the three faces incident to it. Glue the faces along the edge while respecting its orientation. This determines a unique local structure up to mirror symmetry (see the proof of \cref{lem:fulltwist} and \cref{fig:3DVD}). The local Voronoi regions around the faces are fixed. We then build the diagram incrementally, adding one face in each step.

Assume that we have a partially constructed diagram $D$. Let $f$ be a new face that shares an edge $e$ with $D$. Then there is a unique way to attach $f$ along $e$ that respects the orientation of $e$ (see the proof of \cref{lem:fulltwist} and \cref{fig:3DVD}).

By~\cite{Barequet2023}~Theorem~3.7, the farthest Voronoi diagrams of lines have connected 2-skeletons. The nearest counterpart also has connected 2-skeletons, which can be easily shown by induction where the base case of three lines follows from~\cite{Everett2009}. Therefore, there is a sequence of faces that incrementally build the whole diagram. Consequently, the resulting diagram is unique up to mirror symmetry.
\end{proof}

\begin{proof}[Proof of \cref{thm:topologywithoutfulltwist}]

Phase~1, Step~1 generates all the simple configurations that have no twists, by \cref{lem:simpleconfig}. Next, we argue that Step 2 correctly considers only the edges incident to faces labeled $\vd{2}$ in order to generate all possible partial twists.
  
Consider two arbitrary edges $e_1$ and $e_2$ in a simple configuration from Step 1. If $e_1$ and $e_2$ are incident to a common face labeled $\nvd$ (resp. $\fvd$), intersecting them twice creates a bounded face with two vertices labeled $\fvd$ (resp. $\nvd$), see \cref{fig:addtwist}~(c) and (d) respectively. By \cref{lem:partialtwist}, such a face yields a full twist. 
Assume now that $e_1$ and $e_2$ are not incident to a common configuration face. To intersect them, there must be at least one intermediate edge $e_3$ such that we first intersect $e_1$ and $e_3$ twice, after which $e_1$ and $e_2$ may become adjacent, and then we can intersect $e_1,e_2$  twice. 
This is captured by \cref{fig:trisys2}. By \cref{lem:trisystem2}, the twist between $e_1$ and $e_2$ is a full twist.
The case when there is more than one intermediate edge between $e_1$ and $e_2$ can be handled analogously.
%
Hence, Step 2 correctly considers edges that are only incident to faces labeled as $\vd{2}$, and this is enough to generate all the relevant configurations.

Phases 2 and 3 use filters derived from the necessary conditions on the diagrams that have already been proved. Thus, the output set $\mathcal{P}$ of Phase~3 contains every configuration tuple that may be realizable. Examples in Appendix~\ref{app:topology} verify that the configuration tuples in $\mathcal{P}$ are realizable.

Each configuration tuple fixes the vertices, edges, and faces of all six configurations. By \cref{lem:uniquetopology}, these features can be glued uniquely to form the Voronoi diagram. Thus, there are 15 topologies for each of $\nvd(L)$ and $\fvd(L)$, which are in one-to-one correspondence with the configuration 6-tuples, and thus they are in one-to-one correspondence to each other. The correspondence naturally extends to $\vd{2}(L)$. The remaining assertions of the theorem follow from inspecting the 15 configuration tuples listed in Appendix~\ref{app:topology}. This completes the proof.
\end{proof}

As already mentioned, all configuration tuples that survive Phase~3 are realizable. This is crucial for the proof of \cref{thm:topology}, and further implies that the necessary conditions of the filters are also sufficient.

\begin{proof}[Proof of \cref{thm:topology}]
Consider four lines and their four trisectors. If there are no full twists in the trisector system, then the Voronoi diagram falls into one of the 15 topologies of \cref{thm:topologywithoutfulltwist}. Assume that there are some full twists; by \cref{lem:no3twist} they cannot be nested. Any full twist can be removed locally by the reverse local modification described in \cref{sec:fulltwist}. Furthermore, if a configuration tuple satisfies the filters of the exhaustive search algorithm, then the configuration tuple derived after the removal of a full twist also does. Hence, the resulting configuration tuple, after removing all full twists, is one of the 15 configuration tuples in $\mathcal{P}$, all of which are realizable. This completes the proof. 
\end{proof}

\section{Examples and concluding remarks}\label{sec:example}

\subparagraph{Example 1: Voronoi diagram of four lines with 0
  vertices.}
Assuming that there are 0 vertices, there are only two configurations (type~(a) and type~(b)) that satisfy \cref{lem:4line-unboundedbisector,lem:impossibleconfig2}; see \cref{fig:0vertex}. Only one 6-tuple of configurations satisfies \cref{lem:4line-unboundedbisector}: three configurations of type~(a) and three of type~(b).

\begin{figure}[h]
  \centering
  \includegraphics[draft=false]{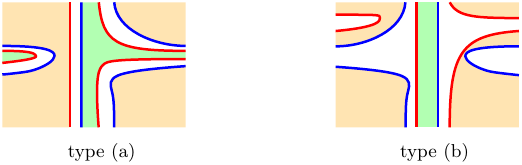}
  \caption{Two configurations with 0 vertices. Faces with label $\nvd$, $\fvd$, and $\vd{2}$ are in green, orange, and white, respectively. Type (a) gives a $(2,6)$-bisector and type (b) gives a $(2,2)$-bisector. }\label{fig:0vertex}
\end{figure}

An $\nvd$ example is shown in \cref{fig:NVD1} left. Of the four Voronoi regions, one is an unbounded triangular prism with three side faces (the region of the orange line in \cref{fig:NVD1} left); the other three have identical structure, which is schematically illustrated in \cref{fig:NVD1} right. The farthest Voronoi diagram is depicted schematically in \cref{fig:FVD} left.

\begin{figure}[h]
    \centering
    \begin{minipage}[b]{0.48\textwidth}
        \centering
        \includegraphics[width=\textwidth]{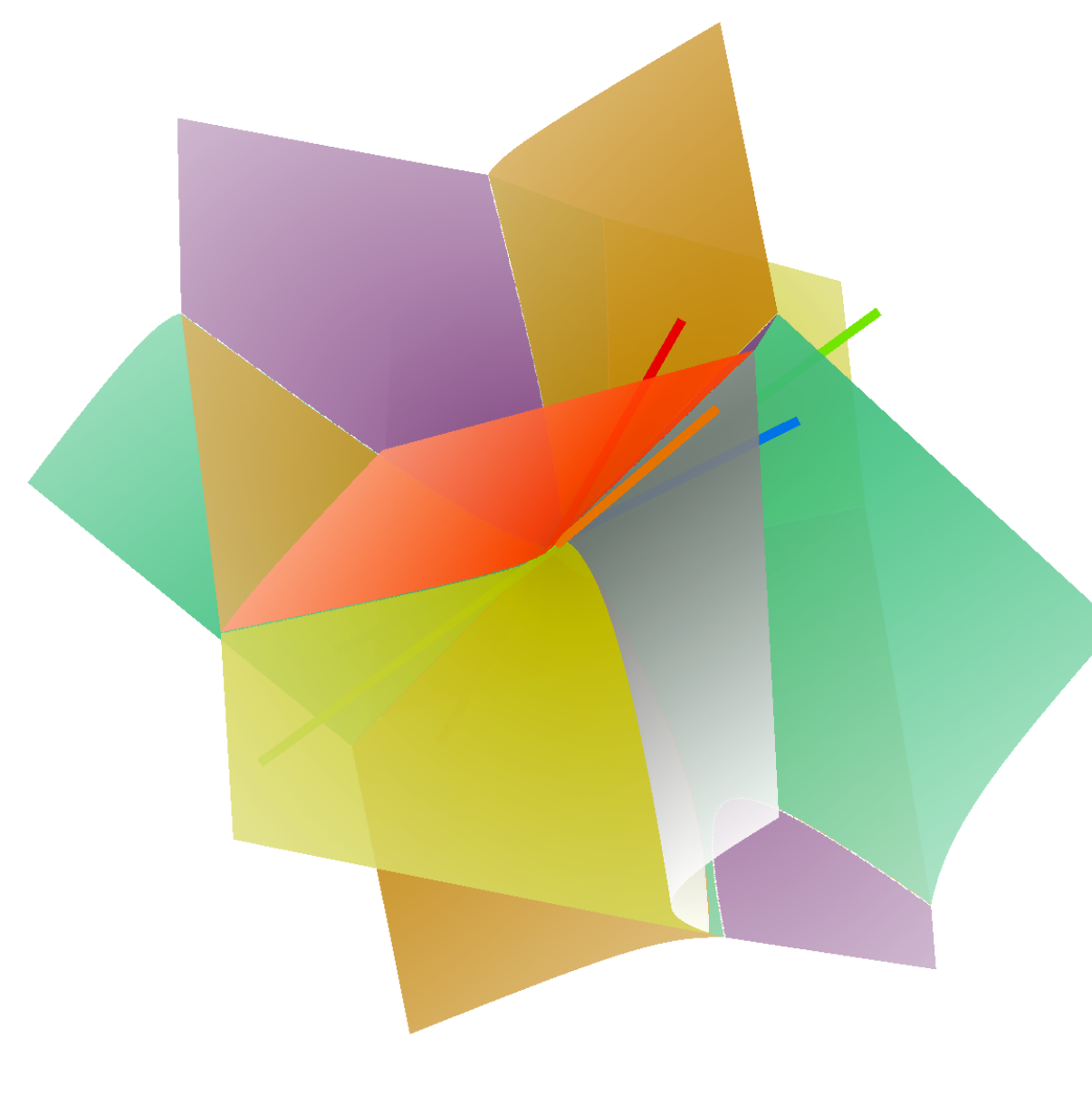}
    \end{minipage}
    \hfill
    \begin{minipage}[b]{0.48\textwidth}
        \centering
        \includegraphics[width=\textwidth]{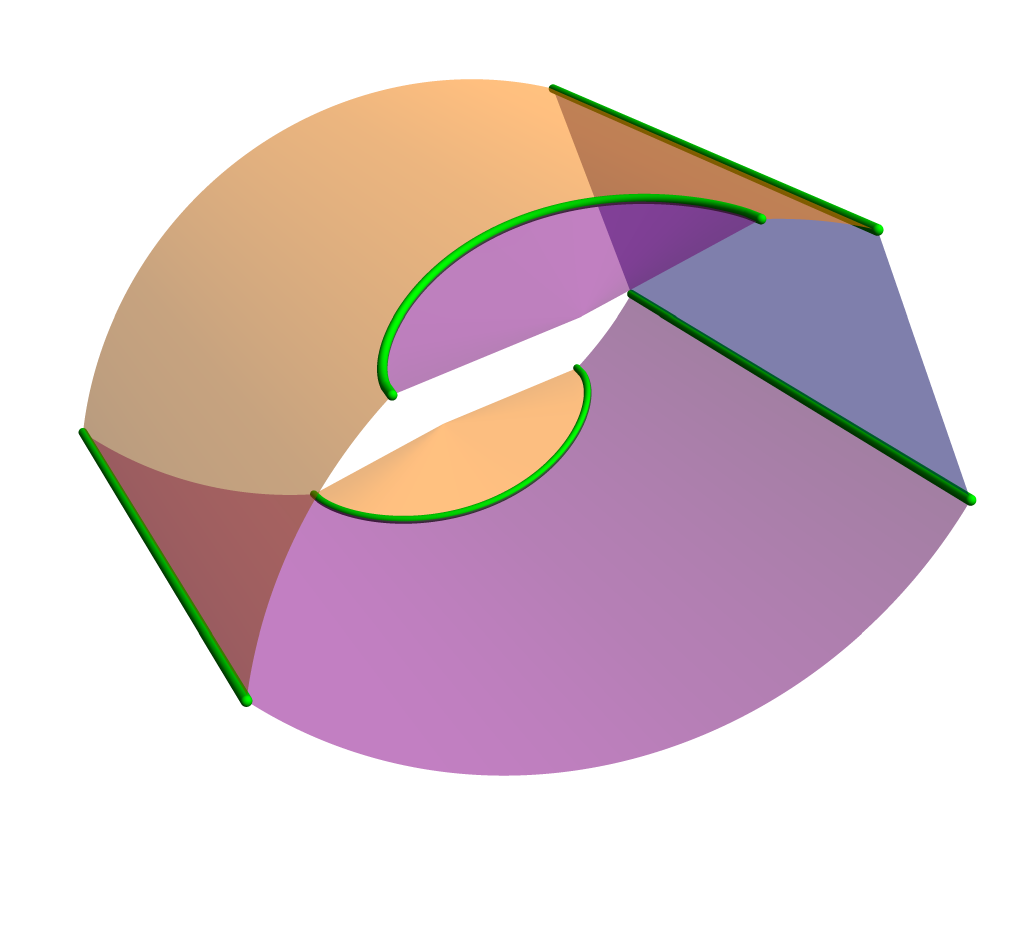}
     \end{minipage}
\caption{Left: $\nvd(L)$ with 0 vertices, faces in the same color are from the same bisector. Right: schematic illustration of a Voronoi region of $\nvd(L)$. Compare this with a region in the Voronoi diagram of three lines shown in \cref{fig:one_nvd_region}~middle. }\label{fig:NVD1}
\end{figure}

\subparagraph{Example 2: Voronoi diagram of four lines with 6
  vertices.}
We show an example where the 6-tuple of configurations is symmetric: it consists of three configurations of type~(c) and three of type~(d), shown in \cref{fig:example2}. The $\fvd$ is shown schematically in \cref{fig:FVD} right.  

\begin{figure}[h]
        \centering
        \includegraphics[draft=false]{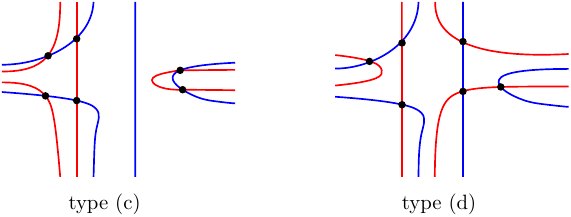}
        \caption{Two configuration types involved in Example 2. }\label{fig:example2}
\end{figure}

\begin{figure}[h]
    \centering
    \begin{minipage}[b]{0.47\textwidth}
        \centering
        \includegraphics[width=0.8\textwidth]{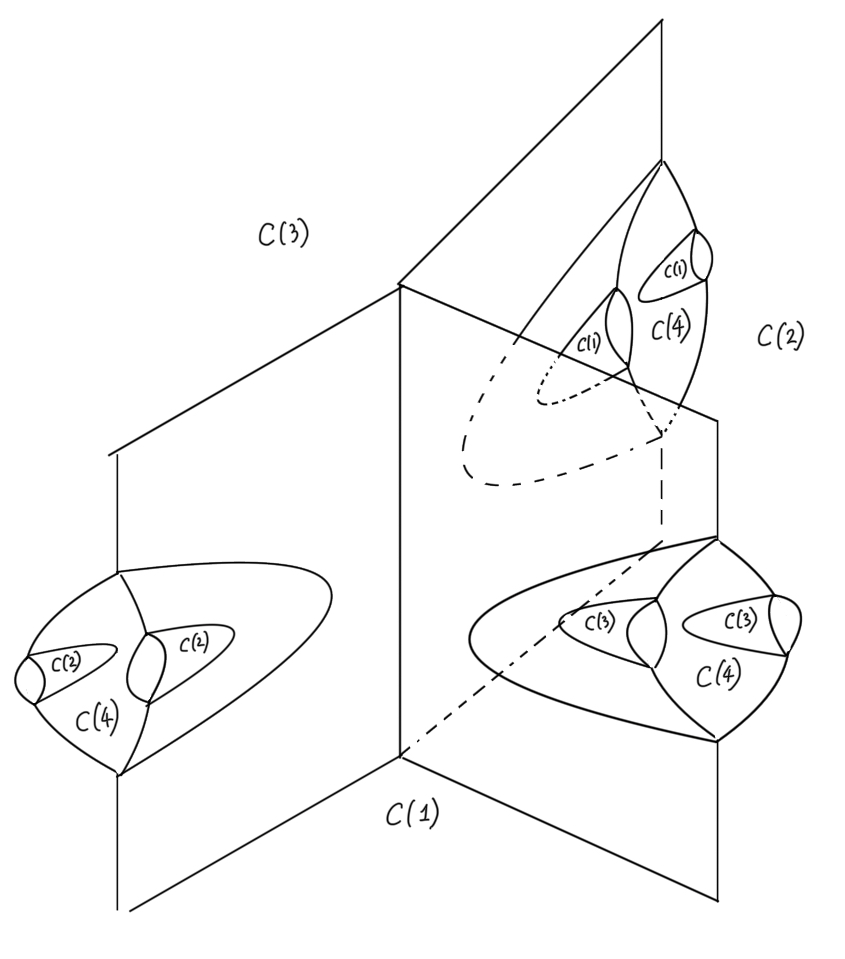}
    \end{minipage}
    \hfill
    \begin{minipage}[b]{0.49\textwidth}
        \centering
        \includegraphics[width=\textwidth]{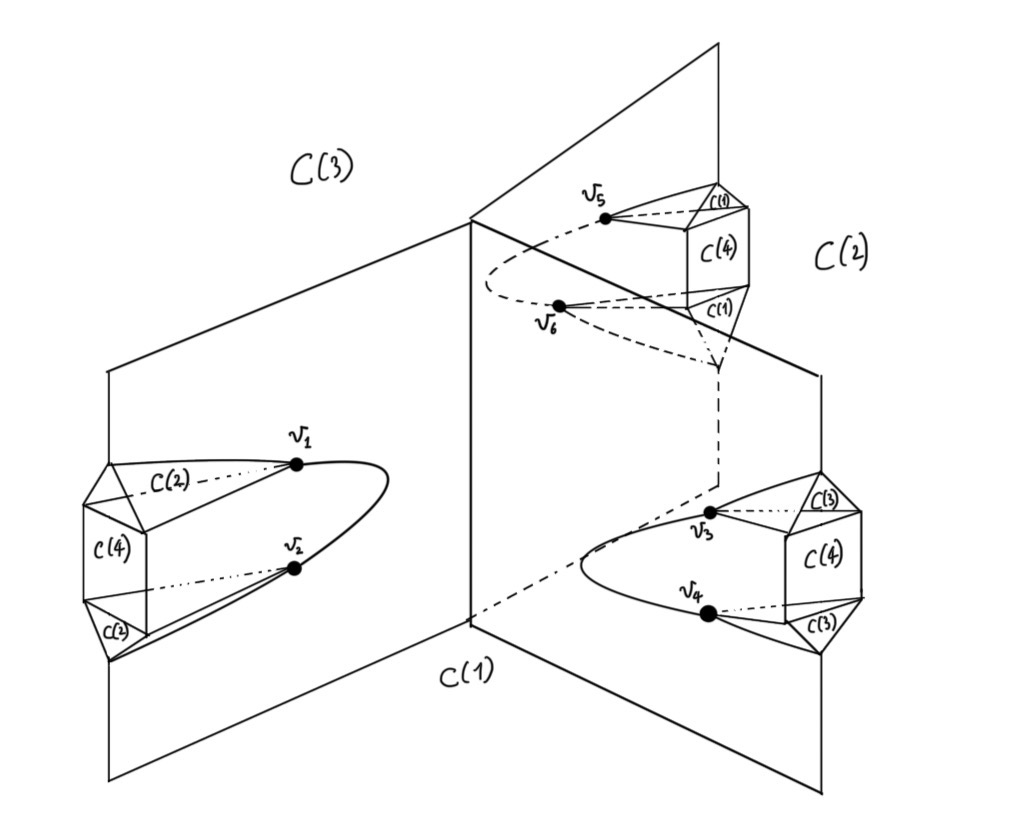}
     \end{minipage}
\caption{Left: schematic of the $\fvd(L)$ with 0 vertices. Right: schematic of the $\fvd(L)$ with 6 vertices of Example 2. Each label $C(i)$ denotes a cell in the farthest region of $\ell_i$.}\label{fig:FVD}
\end{figure}

Our classification results yield an algebraically simple method to compute the structure of the Voronoi diagram of four lines. We start by computing the Gaussian map of $\fvd(L)$~\cite{Barequet2023}, which is simple algebraically. We then extract $\gmap(\fvd(L))$ by a minor adjustment (concerning the so-called vertices of anomaly of~\cite{Barequet2023}). From $\gmap(\fvd(L))$, we can directly identify the matching base topologies of $\nvd(L)$ and $\fvd(L)$. Next, we decide if full twists exist by comparing the number of vertices in the base topology with the number of projected trisector intersections. The latter equals the number of real roots of a degree 8 univariate polynomial, which we compute via Sturm’s theorem. If equal, there are no full twists, and the base topology provides the answer; otherwise, we need to locate the full twists. To this end, we derive the real configurations of each bisector (by applying the separation of four branches of a trisector given in~\cite{Everett2009} to obtain (1) the corresponding simple configurations as in Phase 1, and (2) the number of intersections between all pairs of trisector branches). We compare with the configurations of the base topology, revealing the vertices of full twists, and add them to the base topology via the local modification. All operations can be done in low algebraic degrees, the most expensive being the location of full twists (whose degree is in the range of 8). 

\subparagraph{}
In conclusion, we have obtained a complete characterization of Voronoi diagrams of four lines in general position in $\mathbb{R}^{3}$. This offers a concrete foundation for future research on the Voronoi diagrams of $n$ lines.


\bibliography{2026SoCG}

\newpage

\appendix

\part*{Appendix}

\section{Topologies of the Voronoi diagram of four lines without full
  twists}\label{app:topology}

We describe the 15 topologies of the Voronoi diagram of four lines without full twists. For each topology, we list the following:
\begin{itemize}
\item the configuration tuple of the six projected bisectors and their
  trisectors;
\item a set $L$ of four lines realizing the configuration tuple;
\item the map of unbounded features $\Gamma(\fvd(L))$.
\end{itemize}

Note that we do not show the 3D Voronoi diagram figures for the real examples. The real diagrams are curved and complex, see \cref{fig:NVD1} left for the simplest example with 0 vertices. 
So we feel that such figures do not provide much additional information to the diagrams. 

To illustrate that the examples we provide indeed realize the
configuration tuple, we also show the real projected bisectors of
these lines, each containing two projected trisectors. By comparing
the real projected bisectors (with their trisectors) and the
configurations, we finish the verification of the examples.

For practical reasons, we may not show all the trisector branches or vertices on some real projected bisector examples. For example, a vertex might not be shown because including it would cause the figure to zoom out much further and the result would be much harder to read. However, we make sure that all missing features can be recovered easily manually, following the basic properties of a projected trisector in \cref{lem:middlebranch}.

\subsection{Topology \Rom{1}: 0 vertices}

The configuration tuple that induces topology \Rom{1} is shown in \cref{fig:topo1comb}. A set of lines that realizes this tuple, hence this topology, has the following parameters: $(a,b_{3},c_{3},d_{3},e_{3},b_{4},c_{4},d_{4},e_{4})=(-9, 2,-2, 1,-3,-1, 1,-1,-5)$. \cref{fig:topo1verify} shows the six projected bisectors of these four lines, which match the configurations shown in \cref{fig:topo1comb}. The $\gmap(\fvd(L))$ is shown in \cref{fig:gmapfvd1}.

\begin{figure}[h]
        \centering
        \includegraphics[page=1,width=\textwidth]{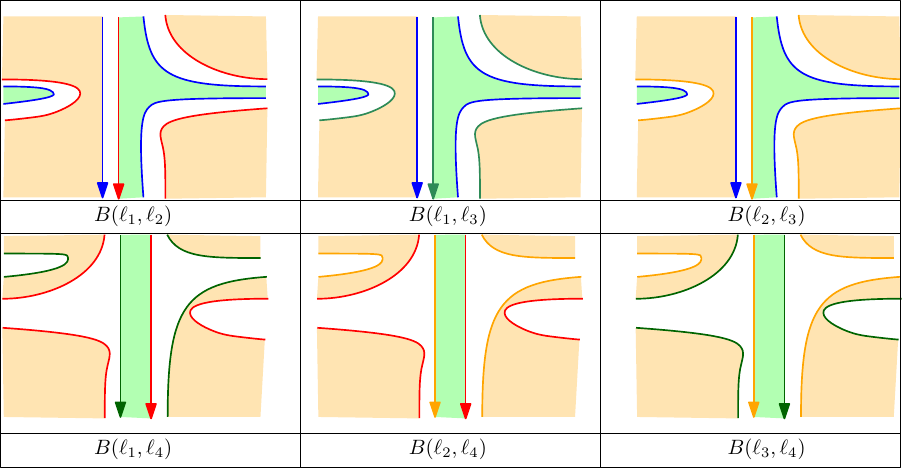}
        \caption{Combination that forms topology \Rom{1}. }\label{fig:topo1comb}
\end{figure}

\begin{figure}[!h]
\centering
\minipage{0.03\textwidth}
(a)	
\endminipage
\minipage{0.3\textwidth}
\includegraphics[width=\textwidth]{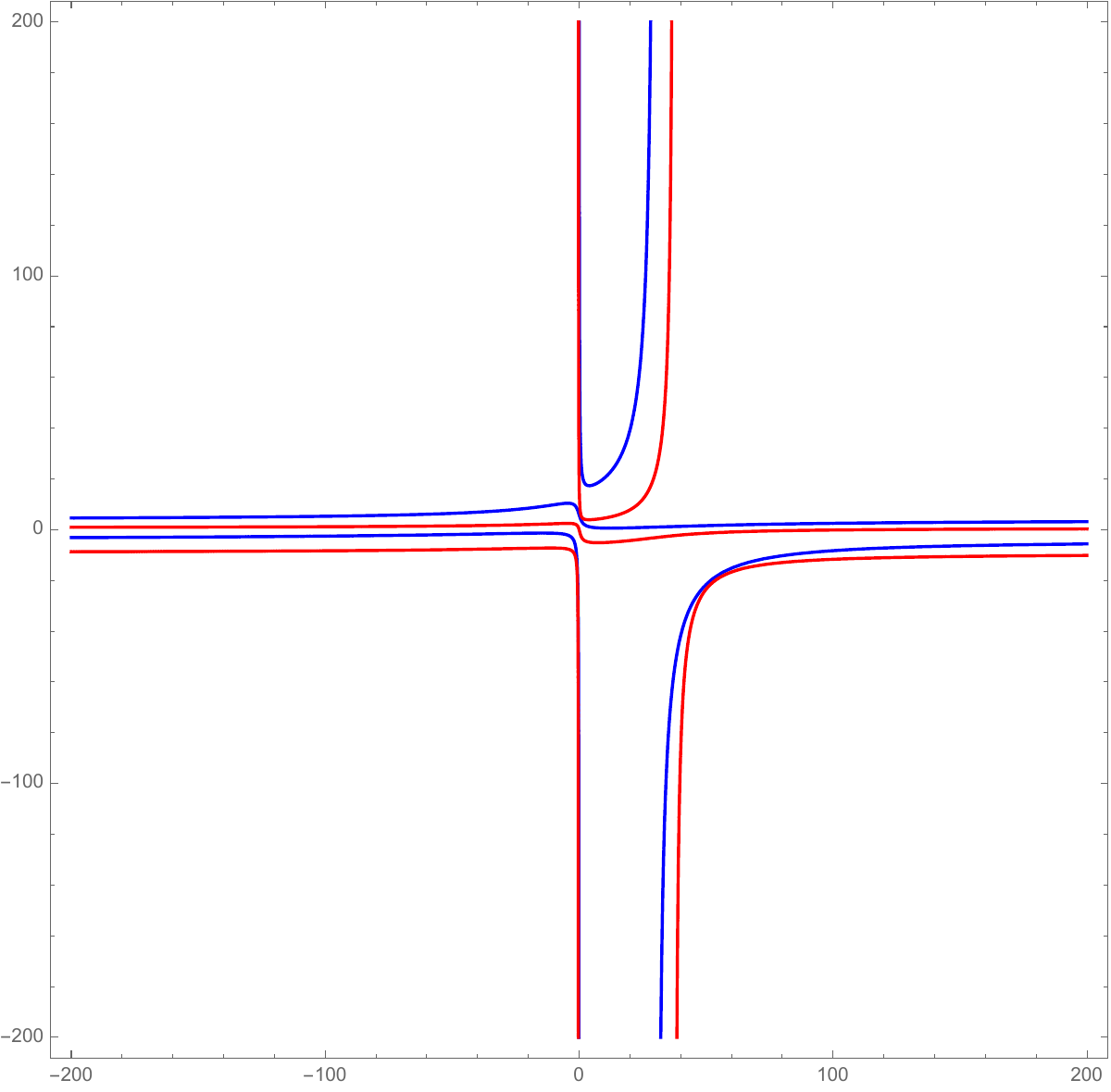}
\endminipage\hfill
\minipage{0.03\textwidth}
(b)	
\endminipage
\minipage{0.3\textwidth}
\includegraphics[width=\textwidth]{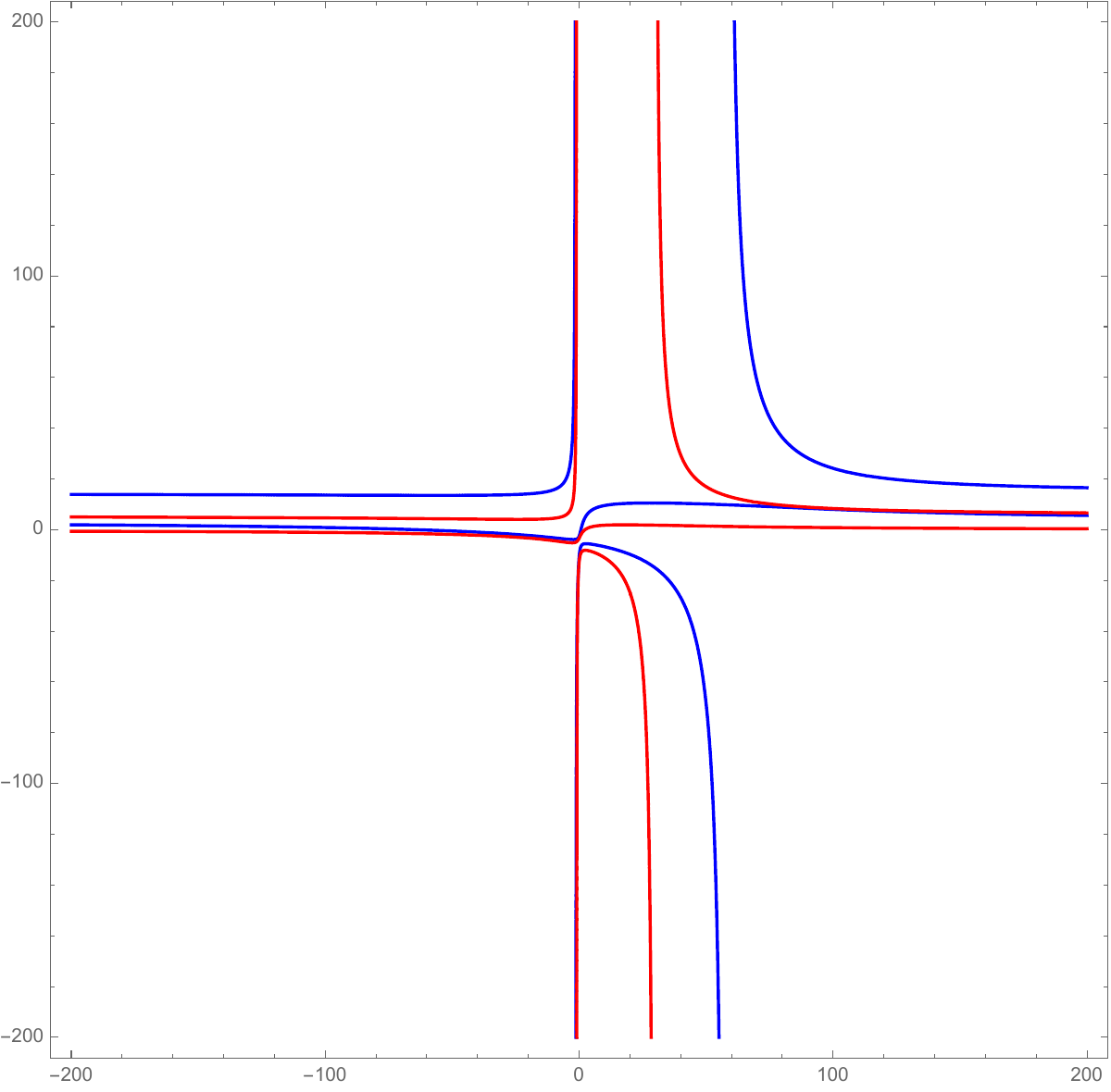}	
\endminipage\hfill
\minipage{0.03\textwidth}
(c)	
\endminipage
\minipage{0.3\textwidth}
\includegraphics[width=\textwidth]{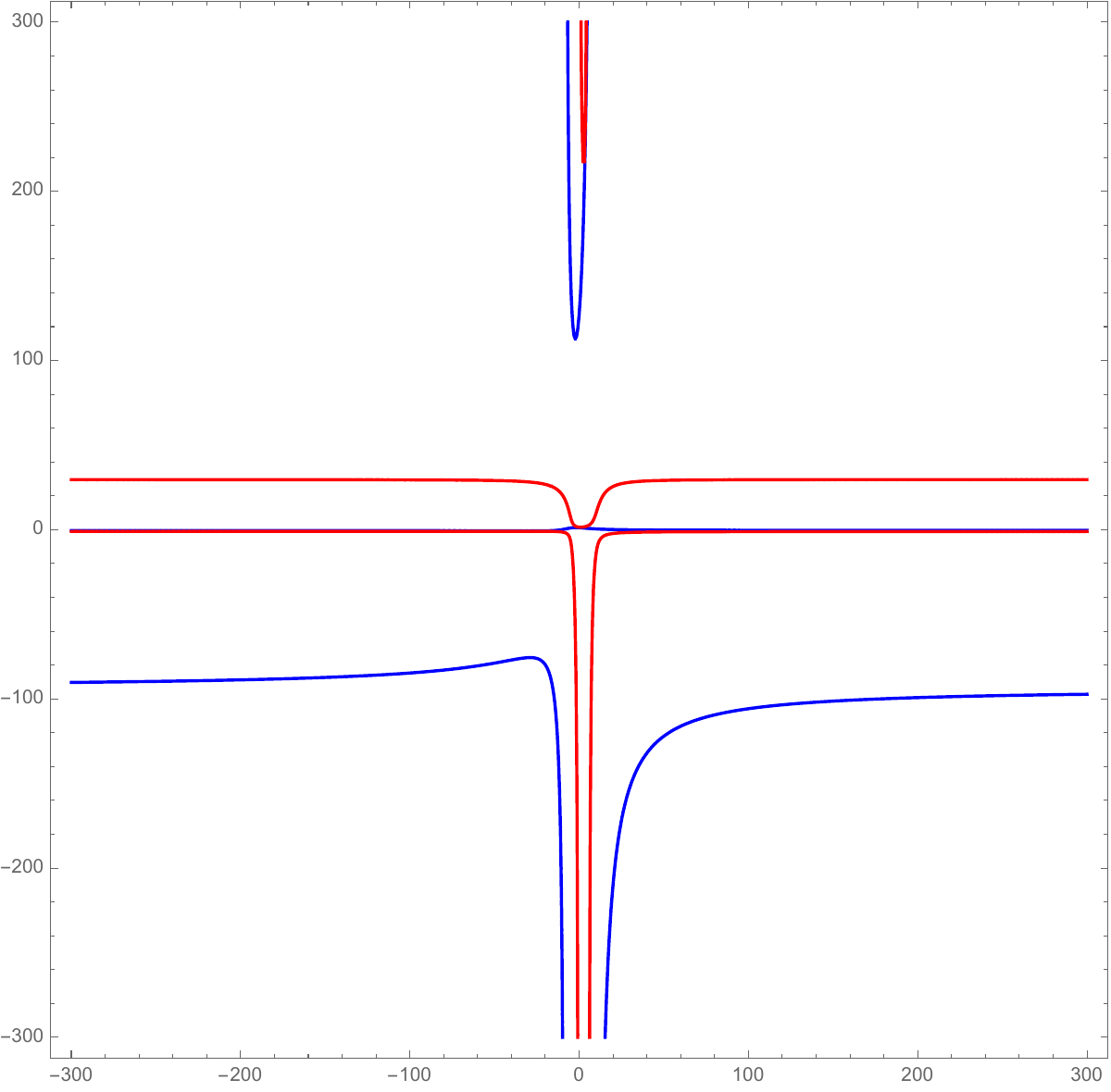}
\endminipage%
\newline
\centering
\minipage{0.03\textwidth}
(d)	
\endminipage
\minipage{0.3\textwidth}
\includegraphics[width=\textwidth]{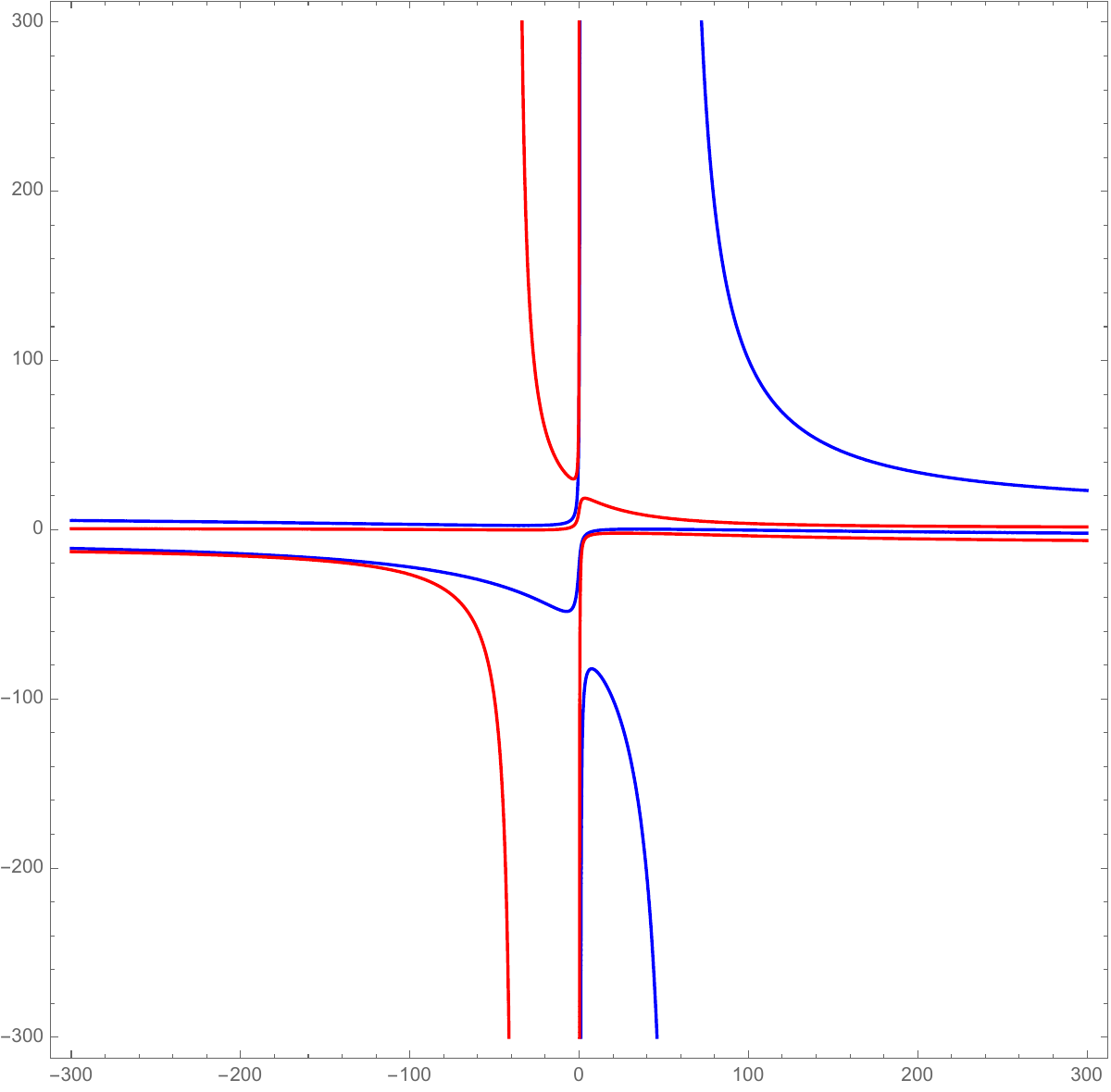}
\endminipage\hfill
\minipage{0.03\textwidth}
(e)	
\endminipage
\minipage{0.3\textwidth}
\includegraphics[width=\textwidth]{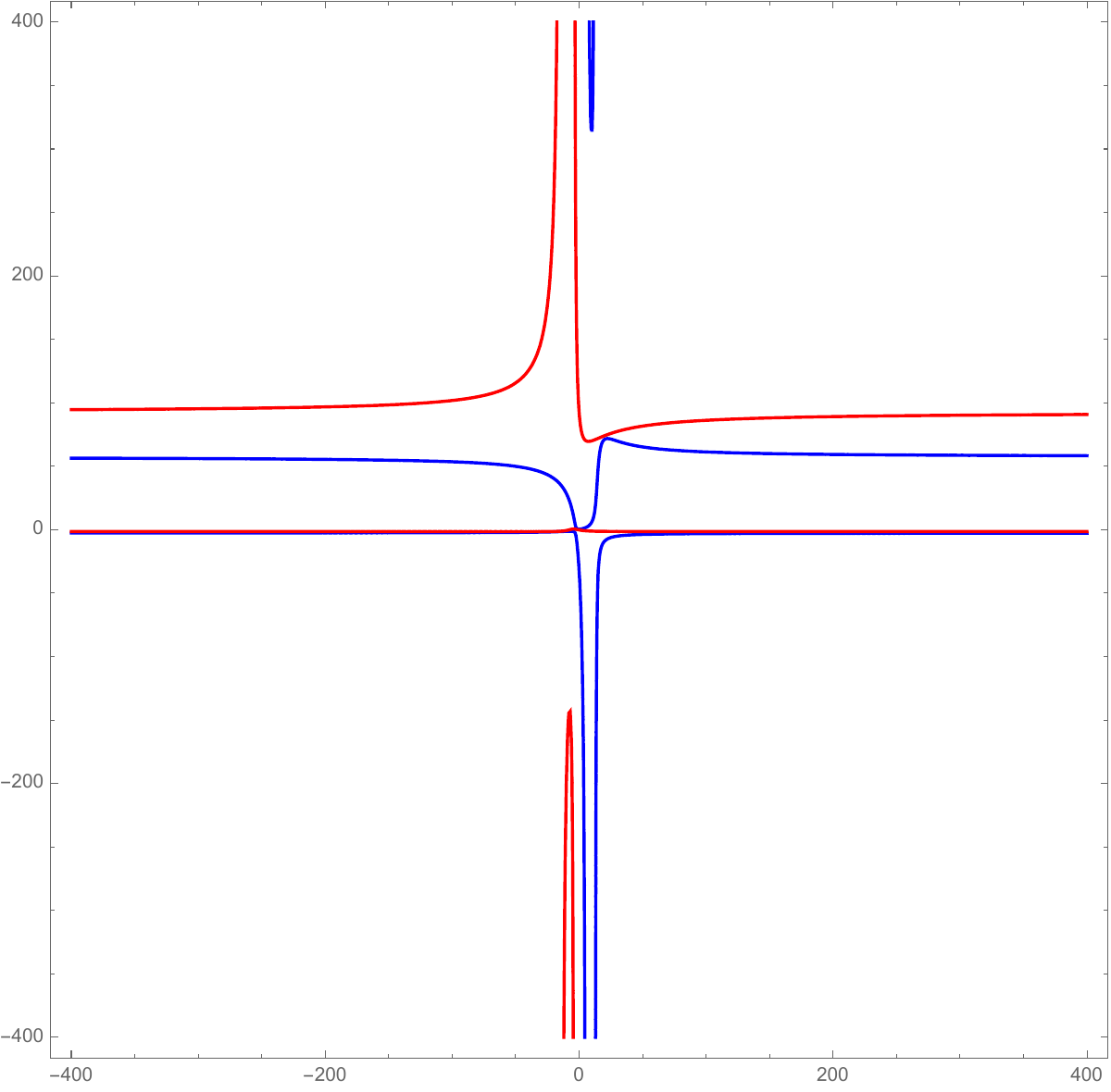}
\endminipage\hfill
\minipage{0.03\textwidth}
(f)	
\endminipage
\minipage{0.3\textwidth}
\includegraphics[width=\textwidth]{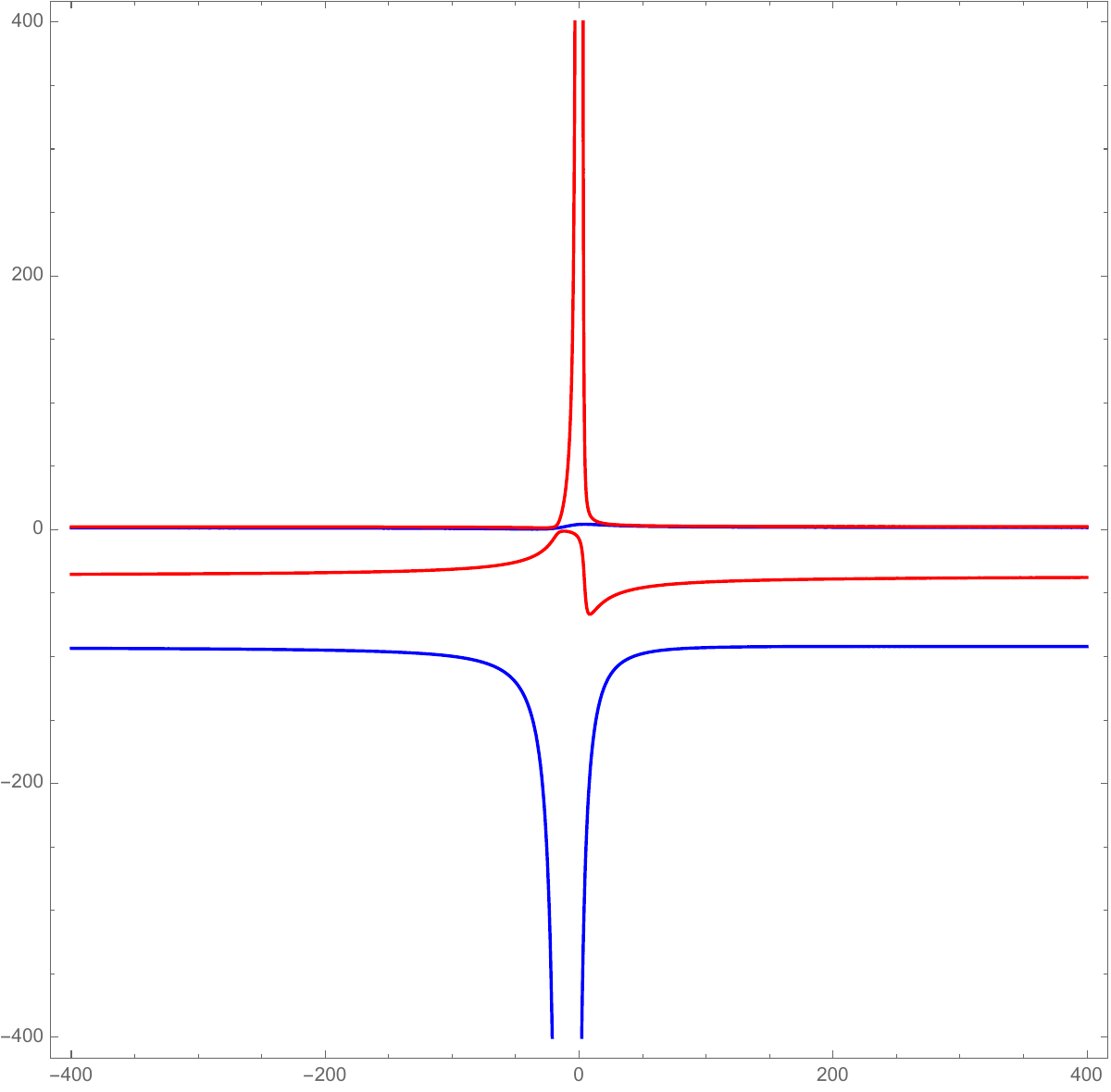}
\endminipage%
\caption{Six projected bisectors when the four lines have the following parameters: $(a,b_{3},c_{3},d_{3},e_{3},b_{4},c_{4},d_{4},e_{4})=(-9,2,-2,1,-3,-1,1,-1,-5)$. They match the configurations shown in \cref{fig:topo1comb} and form topology \Rom{1}. }\label{fig:topo1verify}
\end{figure}

\begin{figure}[h]
    \centering
    \begin{minipage}[b]{0.48\textwidth}
        \centering
        \includegraphics[width=0.9\linewidth]{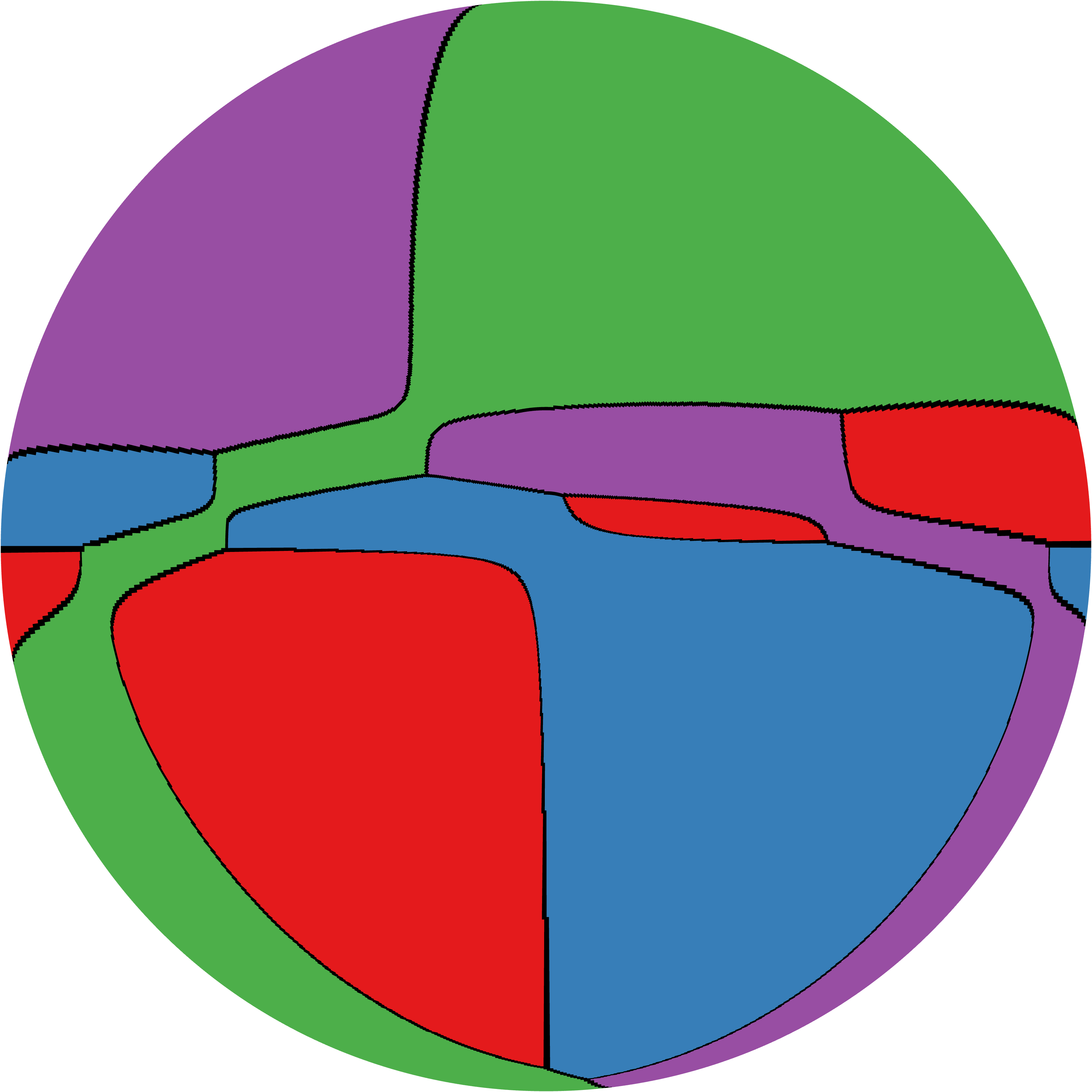}
    \end{minipage}
    \hfill
    \begin{minipage}[b]{0.48\textwidth}
        \centering
        \includegraphics[width=0.9\linewidth]{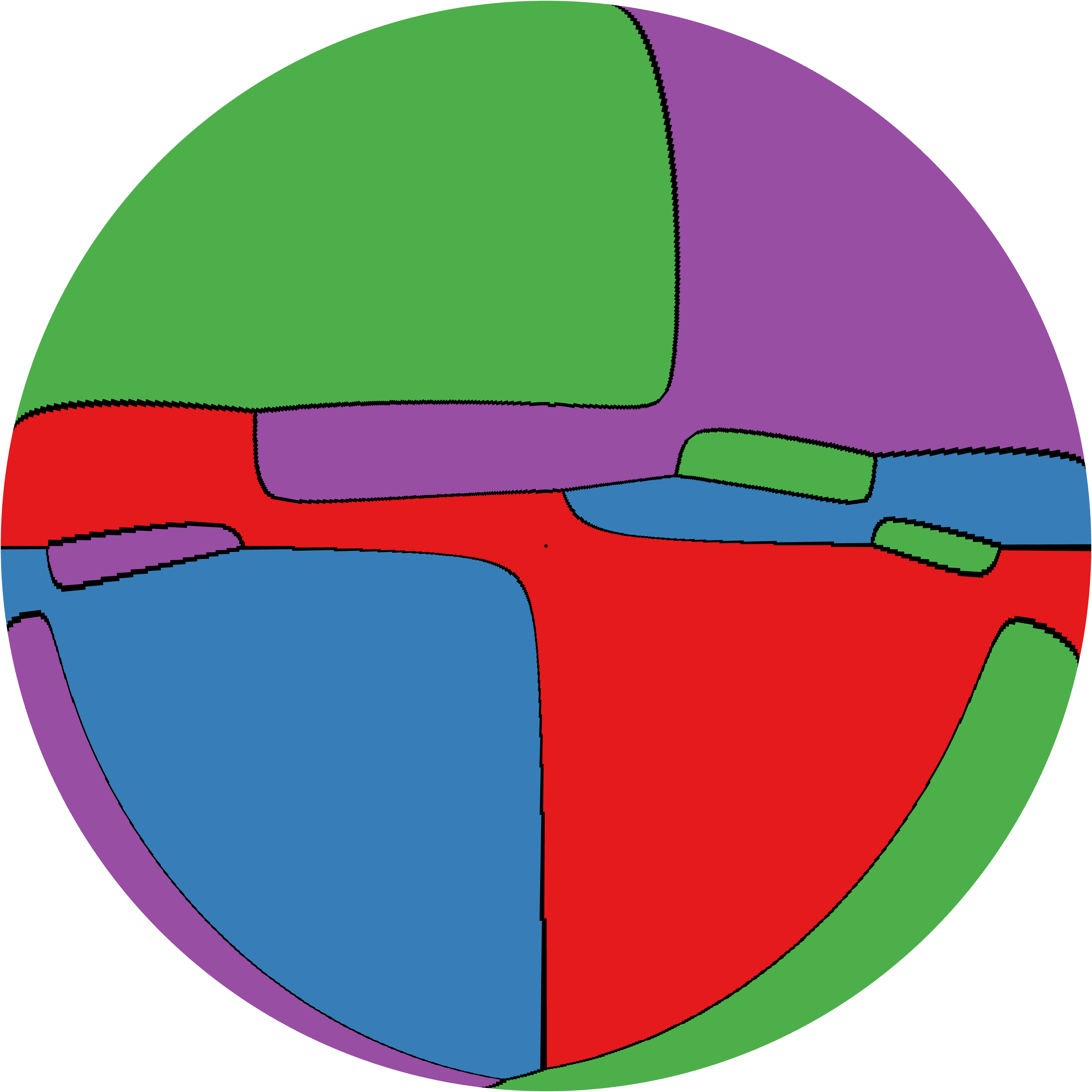}
     \end{minipage}
\caption{Top and bottom view of $\gmap(\fvd(L))$ of topology \Rom{1}. }\label{fig:gmapfvd1}
\end{figure}

\subsection{Topology \Rom{2}: 2 vertices}

The configuration tuple that induces topology \Rom{2} is shown in \cref{fig:topo2comb}. A set of lines that realizes this tuple, hence this topology, has the following parameters: $(a,b_{3},c_{3},d_{3},e_{3},b_{4},c_{4},d_{4},e_{4})=(-3,-16,-12,-3, 9,-10, 1, 9,-11)$. \cref{fig:topo2verify} shows the six projected bisectors of these four lines, which match the configurations shown in \cref{fig:topo2comb}. The $\gmap(\fvd(L))$ is shown in \cref{fig:gmapfvd2}.

\begin{figure}[h]
        \centering
        \includegraphics[page=2,width=\textwidth]{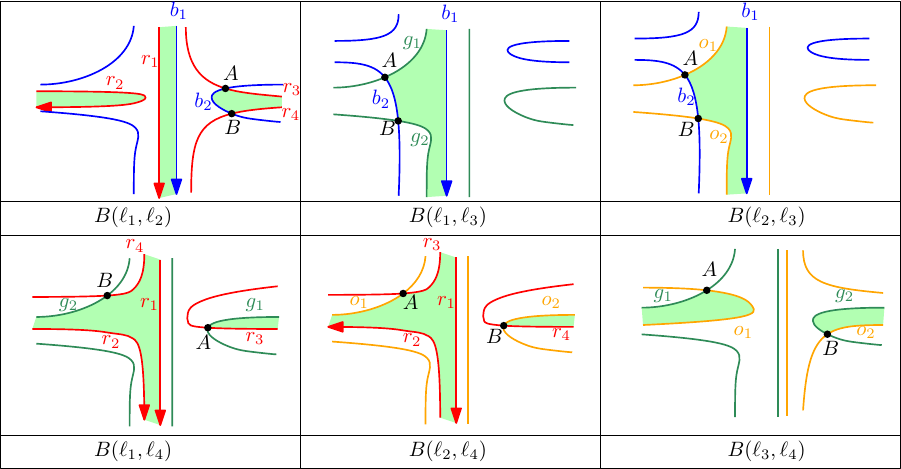}
        \caption{Combination that forms topology \Rom{2}. }\label{fig:topo2comb}
\end{figure}

\begin{figure}[!h]
\centering
\minipage{0.03\textwidth}
(a)	
\endminipage
\minipage{0.3\textwidth}
\includegraphics[width=\textwidth]{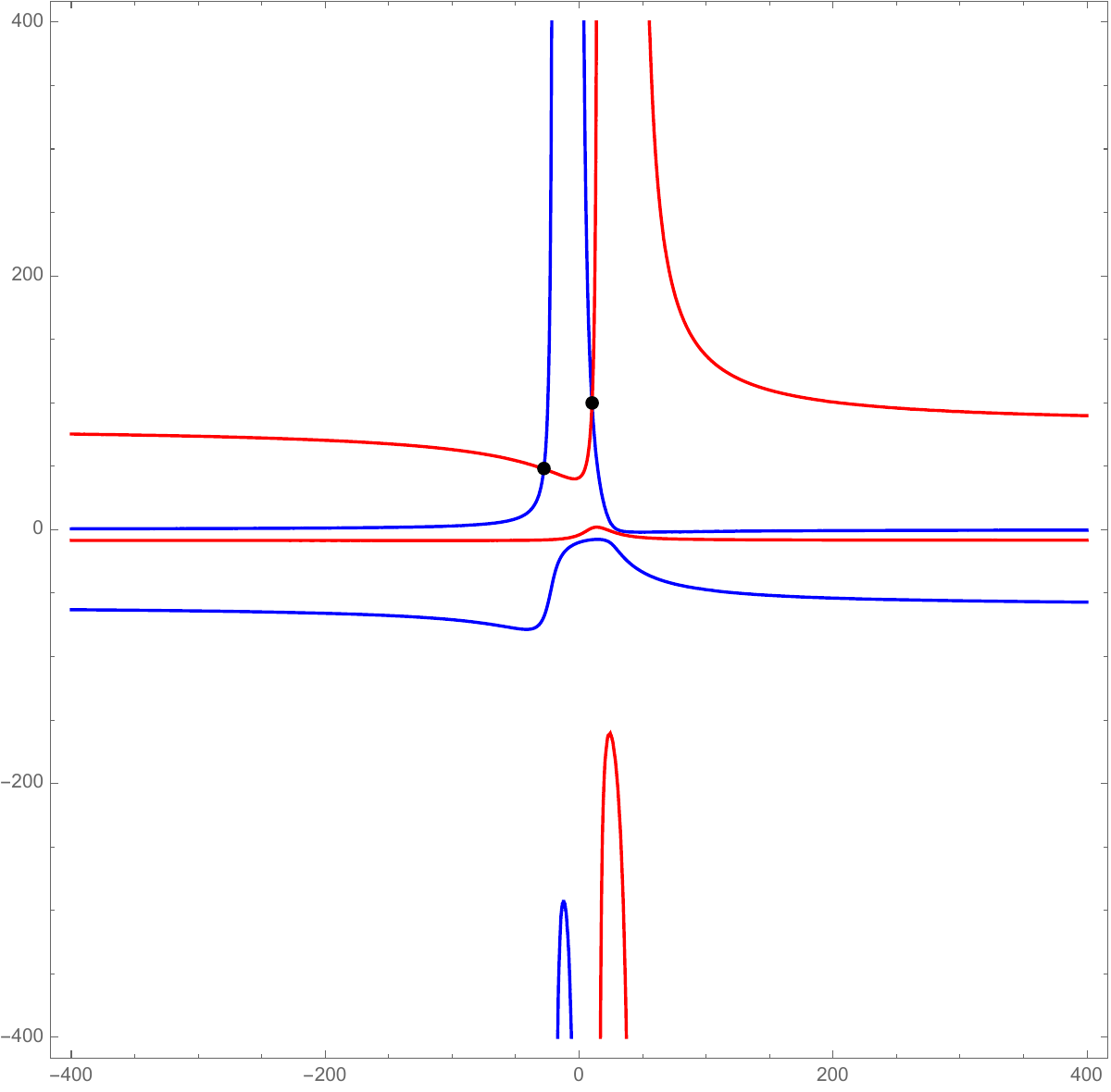}
\endminipage\hfill
\minipage{0.03\textwidth}
(b)	
\endminipage
\minipage{0.3\textwidth}
\includegraphics[width=\textwidth]{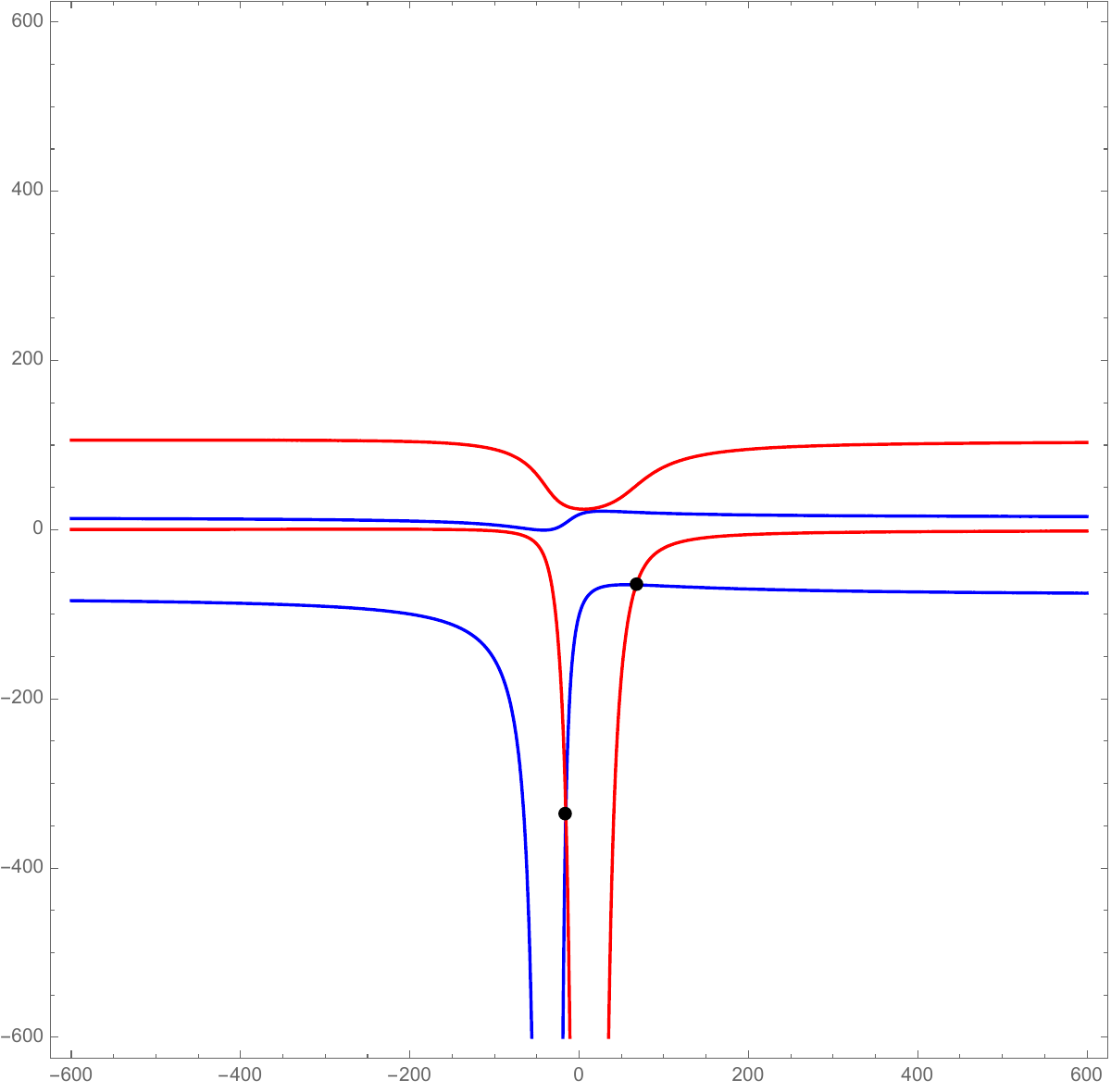}	
\endminipage\hfill
\minipage{0.03\textwidth}
(c)	
\endminipage
\minipage{0.3\textwidth}
\includegraphics[width=\textwidth]{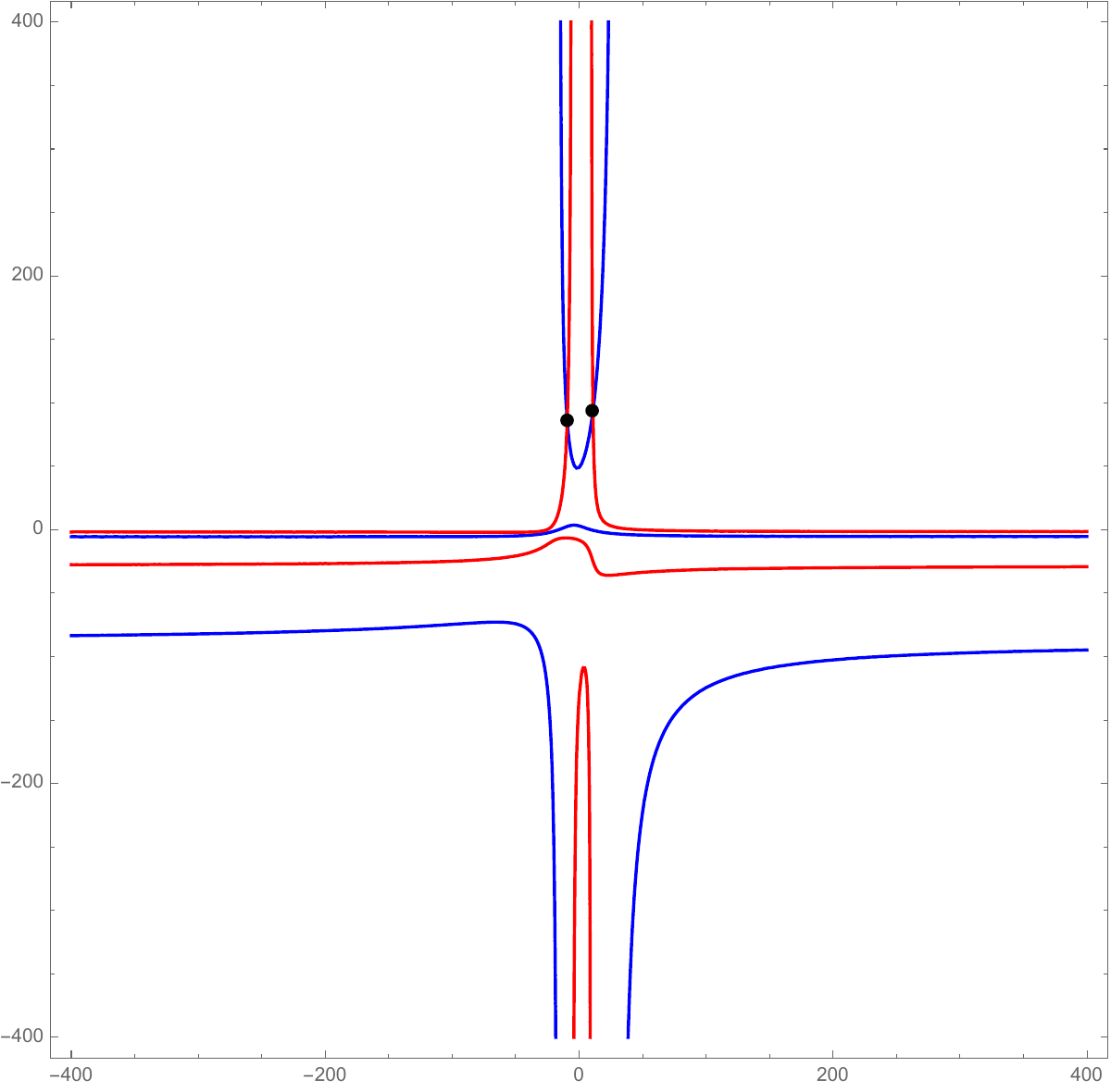}
\endminipage%
\newline
\centering
\minipage{0.03\textwidth}
(d)	
\endminipage
\minipage{0.3\textwidth}
\includegraphics[width=\textwidth]{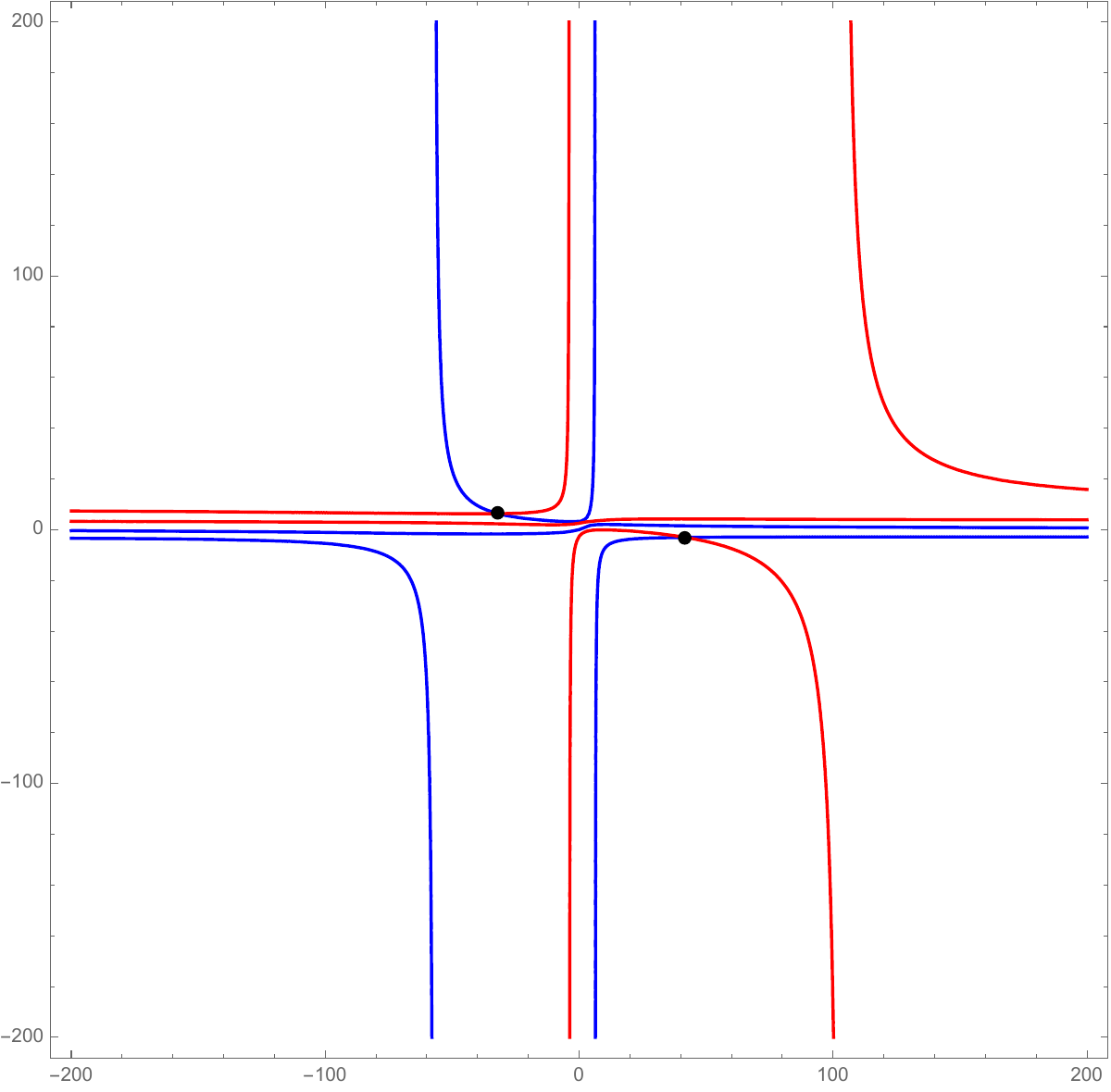}
\endminipage\hfill
\minipage{0.03\textwidth}
(e)	
\endminipage
\minipage{0.3\textwidth}
\includegraphics[width=\textwidth]{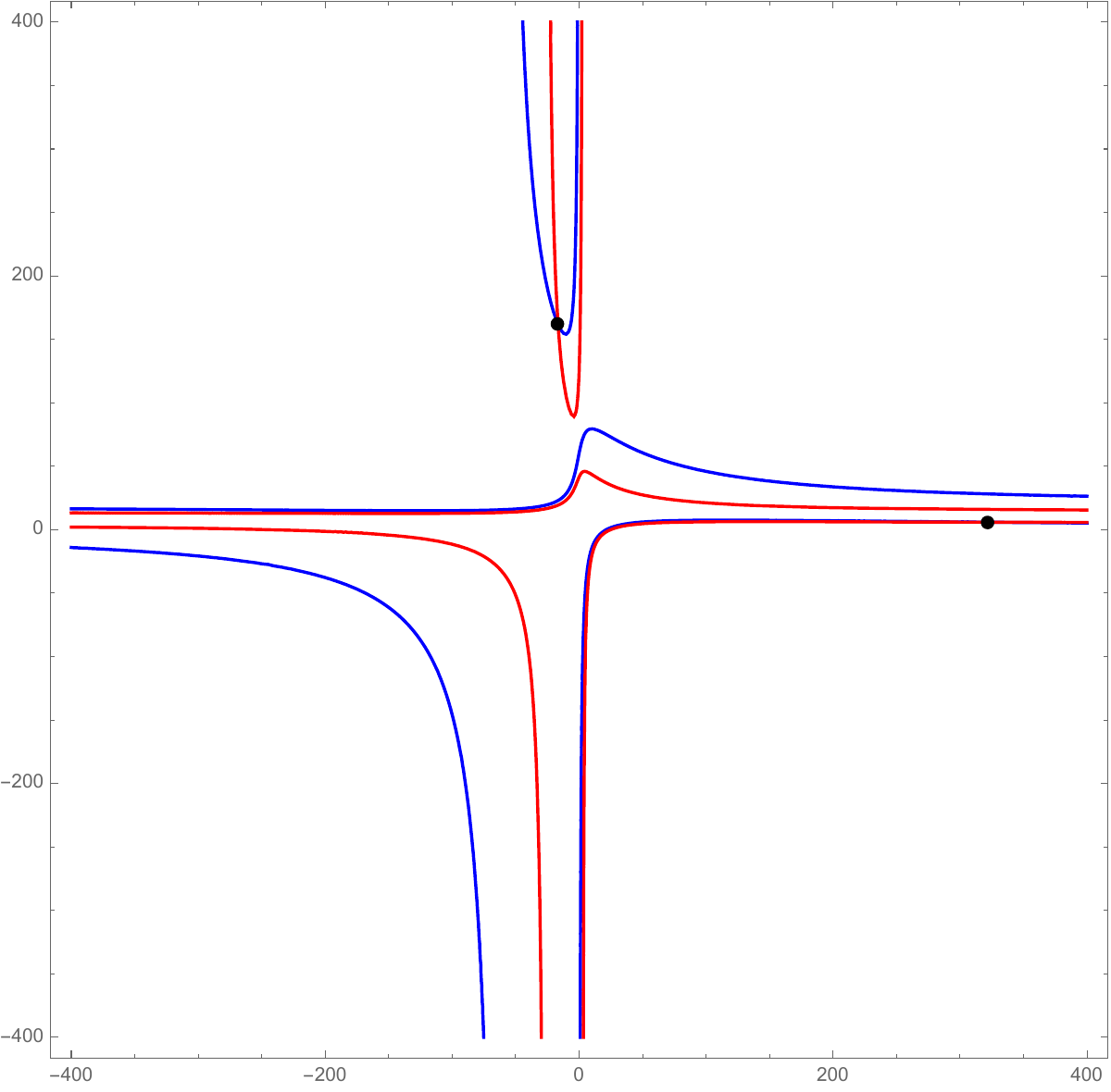}
\endminipage\hfill
\minipage{0.03\textwidth}
(f)	
\endminipage
\minipage{0.3\textwidth}
\includegraphics[width=\textwidth]{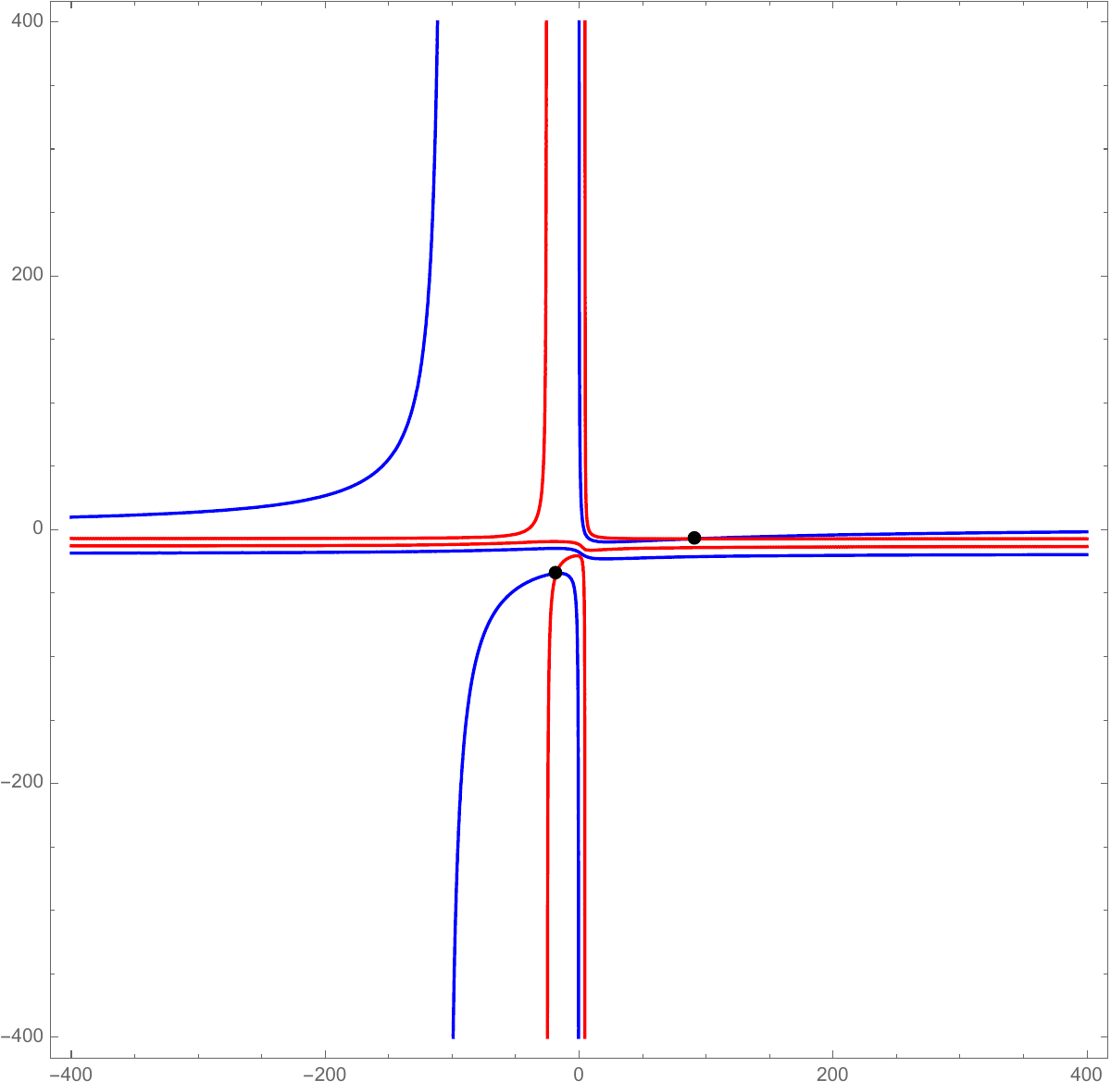}
\endminipage%
\caption{Six projected bisectors when the four lines have the following parameters: $(a,b_{3},c_{3},d_{3},e_{3},b_{4},c_{4},d_{4},e_{4})=(-3,-16,-12,-3,9,-10,1,9,-11)$. They match the configurations shown in \cref{fig:topo2comb} and form topology \Rom{2}. }\label{fig:topo2verify}
\end{figure}

\begin{figure}[h]
    \centering
    \begin{minipage}[b]{0.48\textwidth}
        \centering
        \includegraphics[width=0.9\linewidth]{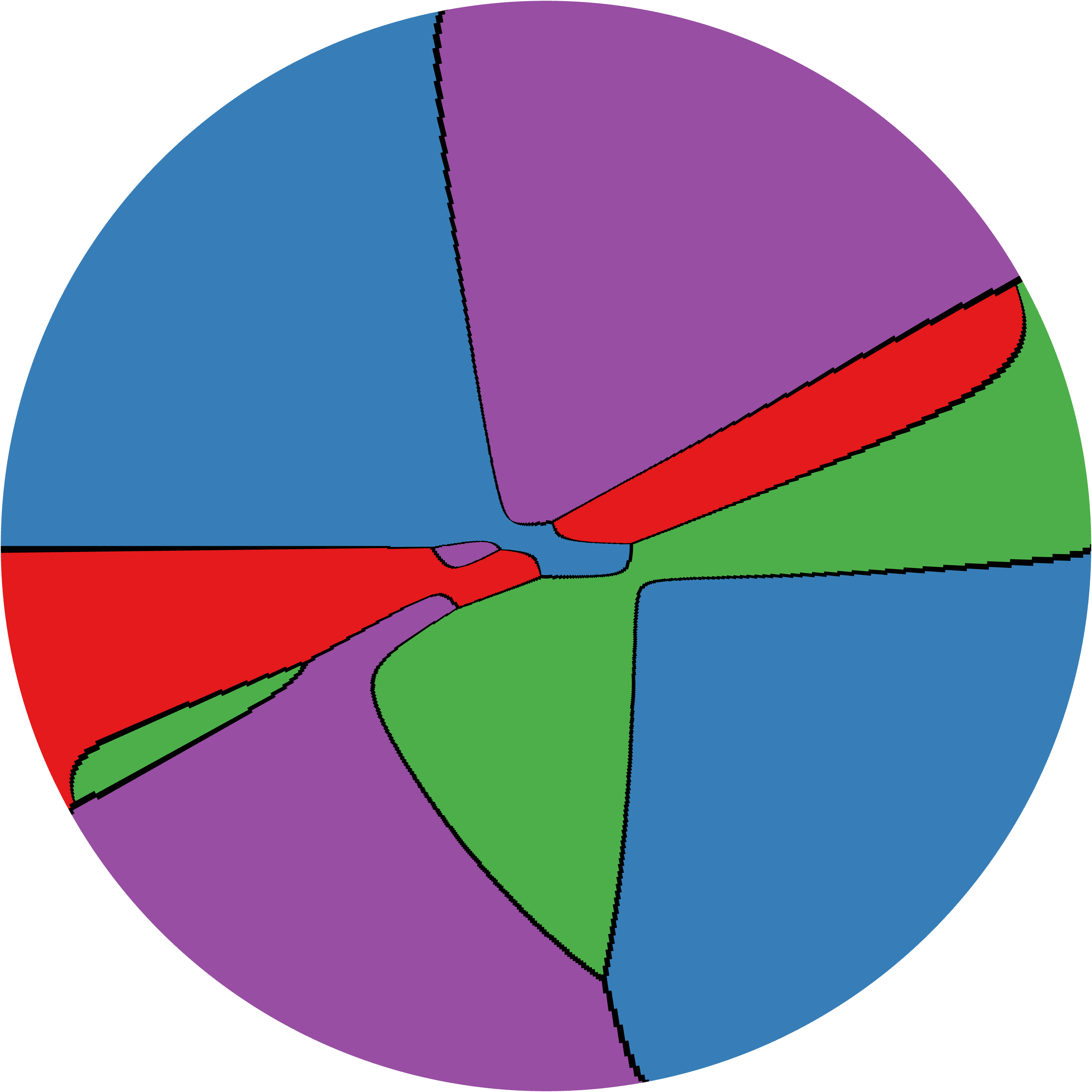}
    \end{minipage}
    \hfill
    \begin{minipage}[b]{0.48\textwidth}
        \centering
        \includegraphics[width=0.9\linewidth]{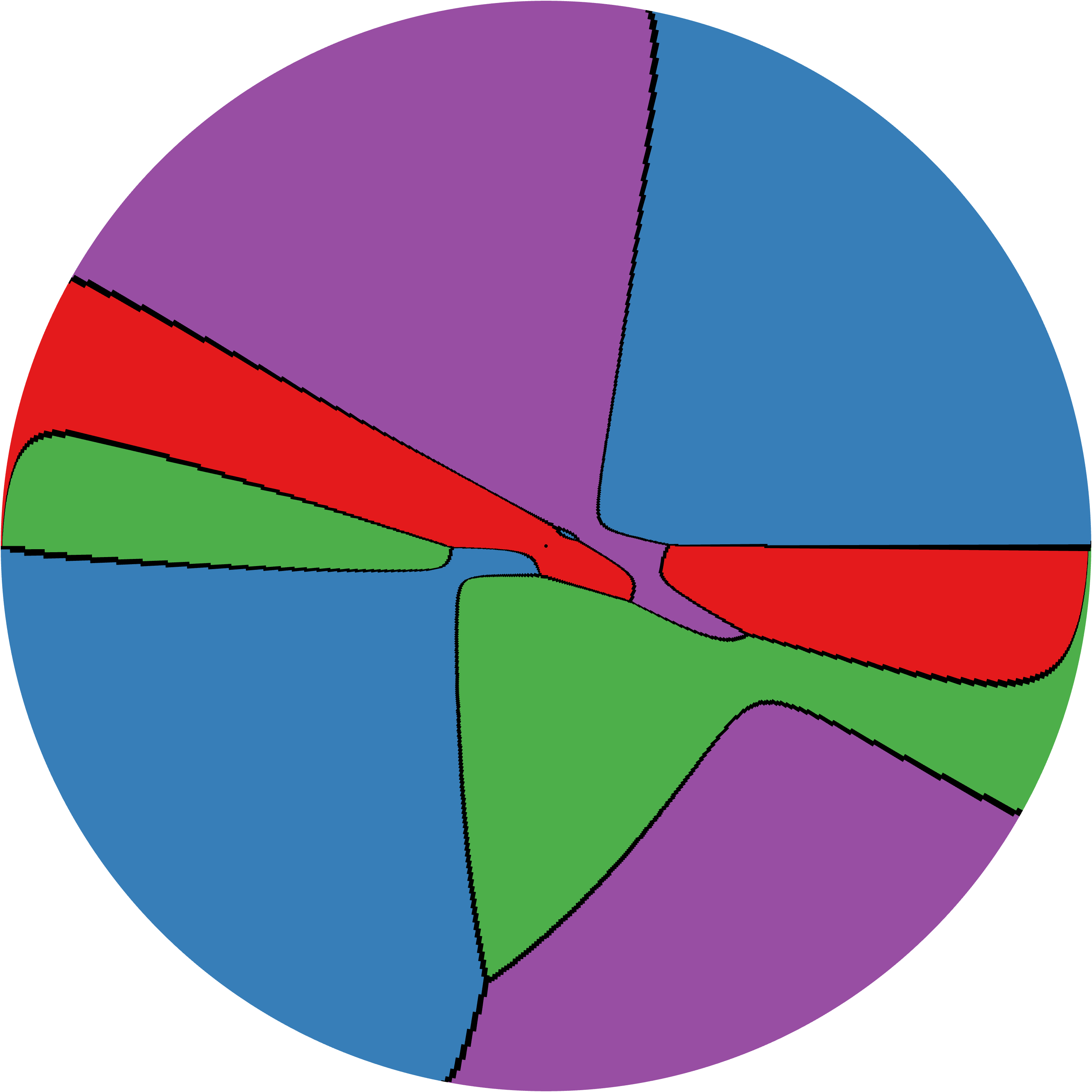}
     \end{minipage}
\caption{Top and bottom view of $\gmap(\fvd(L))$ of topology \Rom{2}. }\label{fig:gmapfvd2}
\end{figure}

\subsection{Topology \Rom{3}: 2 vertices}

The configuration tuple that induces topology \Rom{3} is shown in \cref{fig:topo3comb}. A set of lines that realizes this tuple, hence this topology, has the following parameters: $(a,b_{3},c_{3},d_{3},e_{3},b_{4},c_{4},d_{4},e_{4})=(-4,5,10,3,4,5,6,-1,1)$. \cref{fig:topo3verify} shows the six projected bisectors of these four lines, which match the configurations shown in \cref{fig:topo3comb}. The $\gmap(\fvd(L))$ is shown in \cref{fig:gmapfvd3}.

\begin{figure}[h]
        \centering
        \includegraphics[page=3,width=\textwidth]{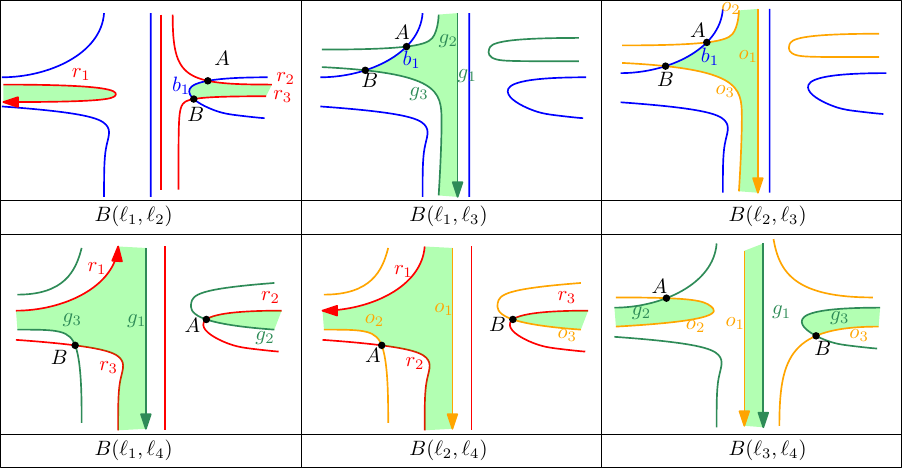}
        \caption{Combination that forms topology \Rom{3}. }\label{fig:topo3comb}
\end{figure}

\begin{figure}[!h]
\centering
\minipage{0.03\textwidth}
(a)	
\endminipage
\minipage{0.3\textwidth}
\includegraphics[width=\textwidth]{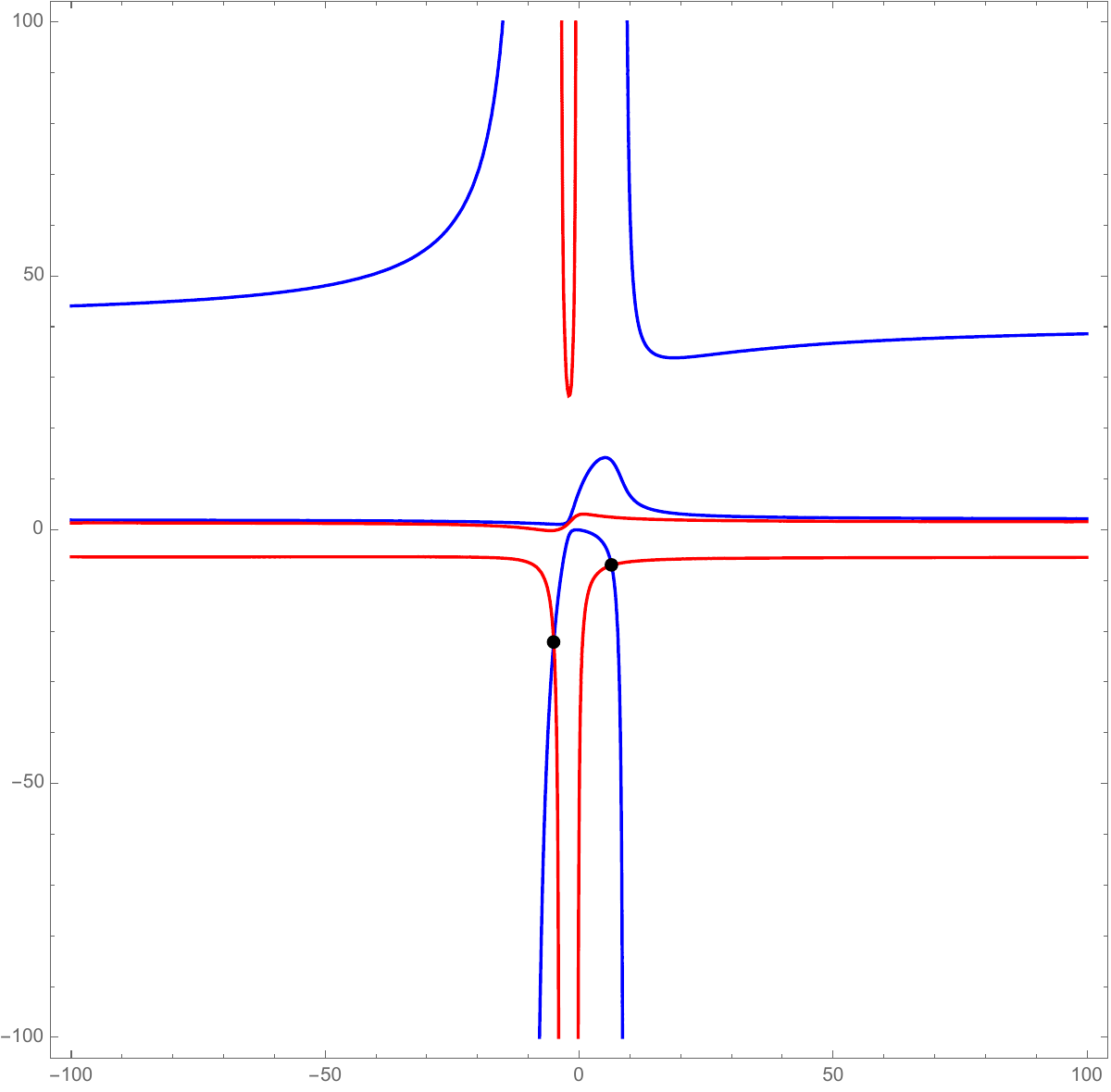}
\endminipage\hfill
\minipage{0.03\textwidth}
(b)	
\endminipage
\minipage{0.3\textwidth}
\includegraphics[width=\textwidth]{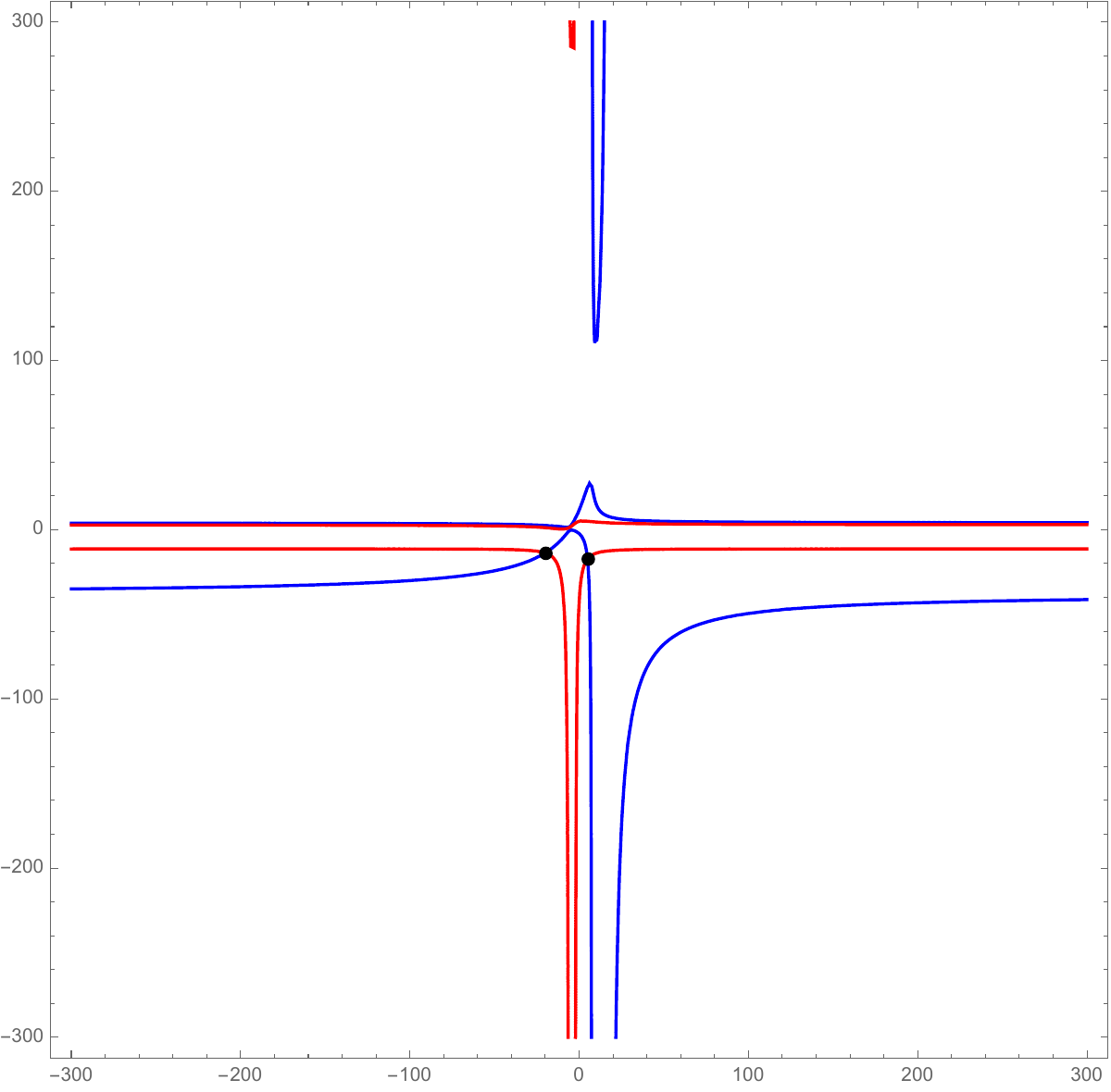}	
\endminipage\hfill
\minipage{0.03\textwidth}
(c)	
\endminipage
\minipage{0.3\textwidth}
\includegraphics[width=\textwidth]{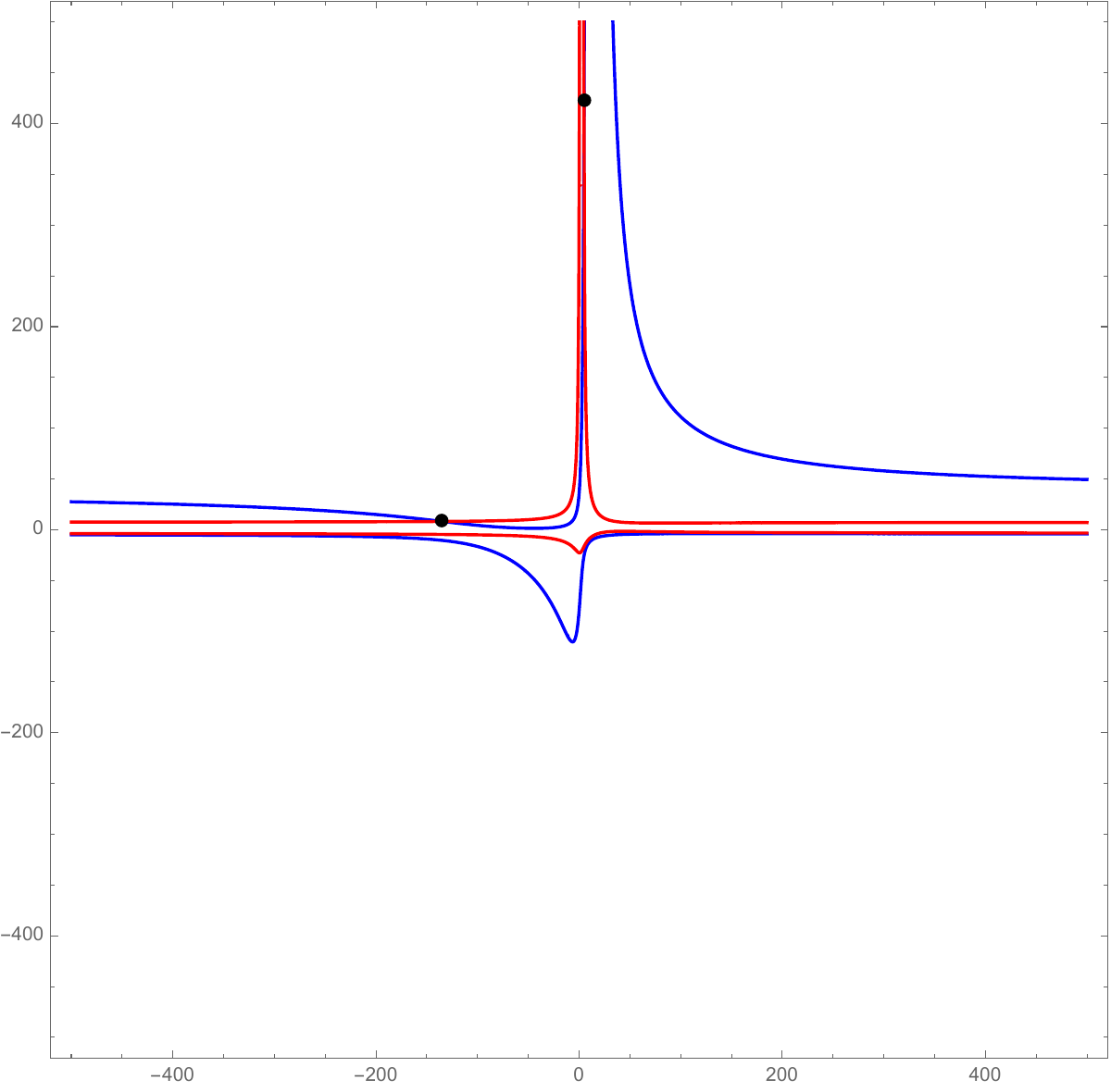}
\endminipage%
\newline
\centering
\minipage{0.03\textwidth}
(d)	
\endminipage
\minipage{0.3\textwidth}
\includegraphics[width=\textwidth]{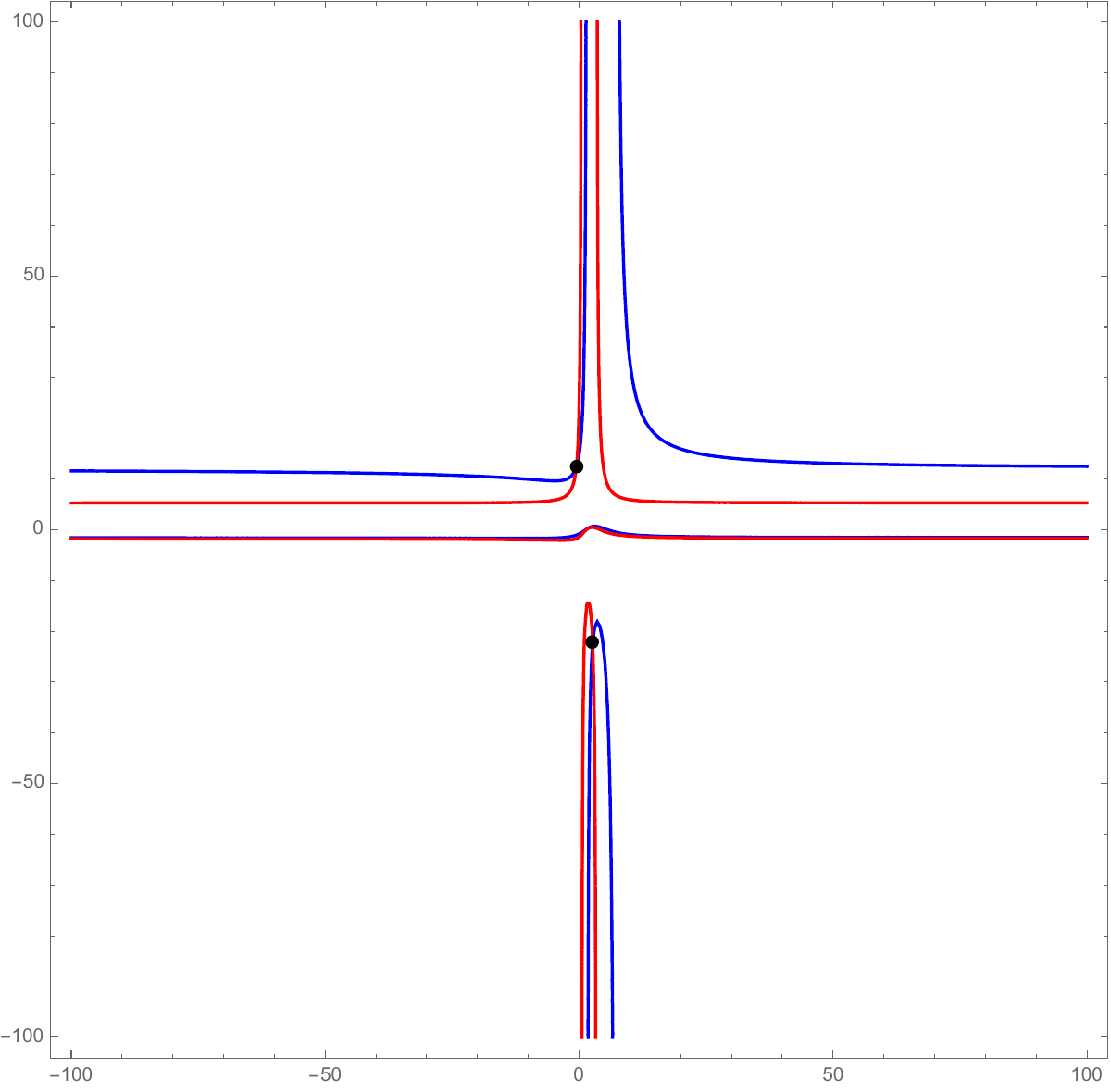}
\endminipage\hfill
\minipage{0.03\textwidth}
(e)	
\endminipage
\minipage{0.3\textwidth}
\includegraphics[width=\textwidth]{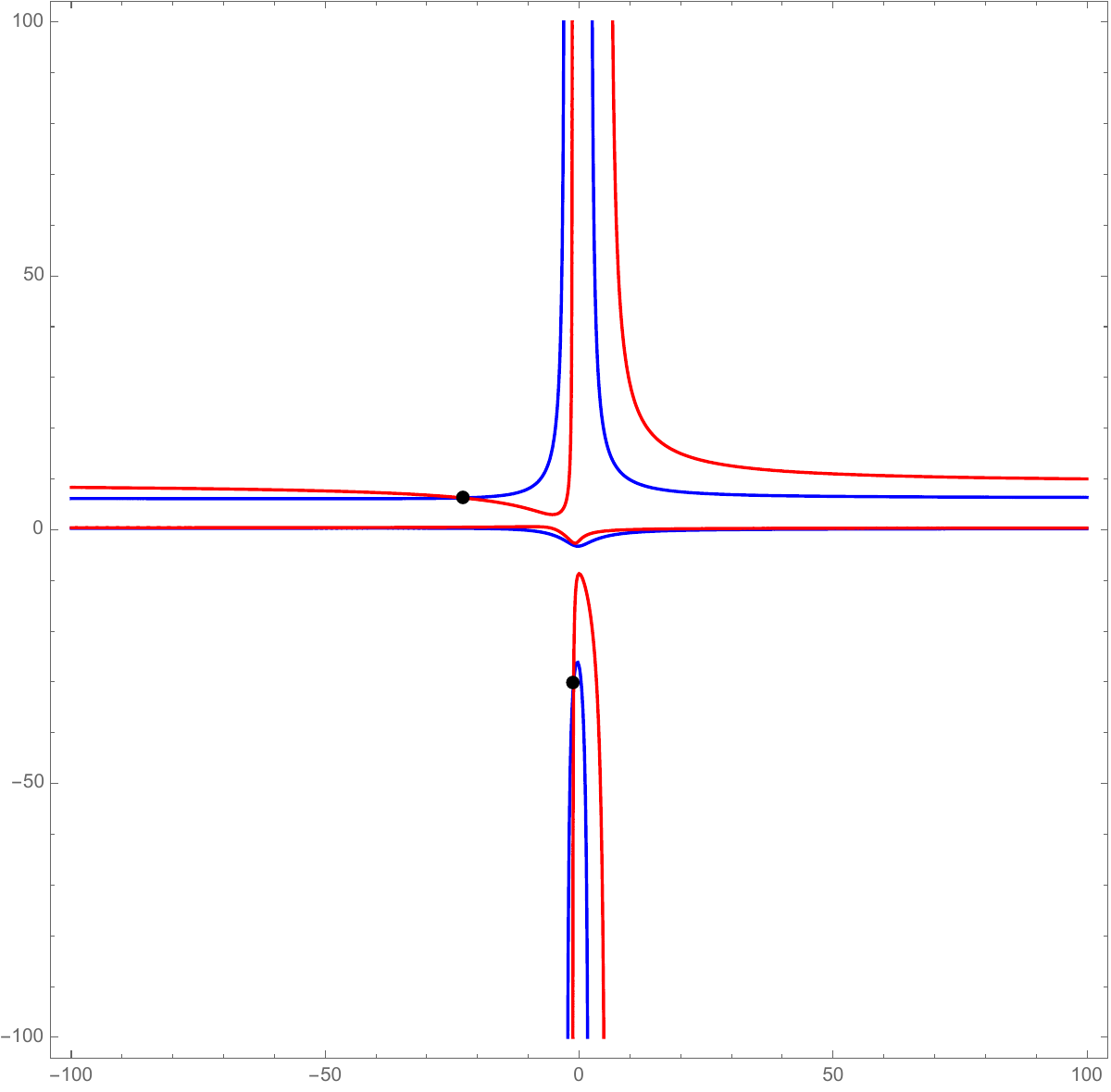}
\endminipage\hfill
\minipage{0.03\textwidth}
(f)	
\endminipage
\minipage{0.3\textwidth}
\includegraphics[width=\textwidth]{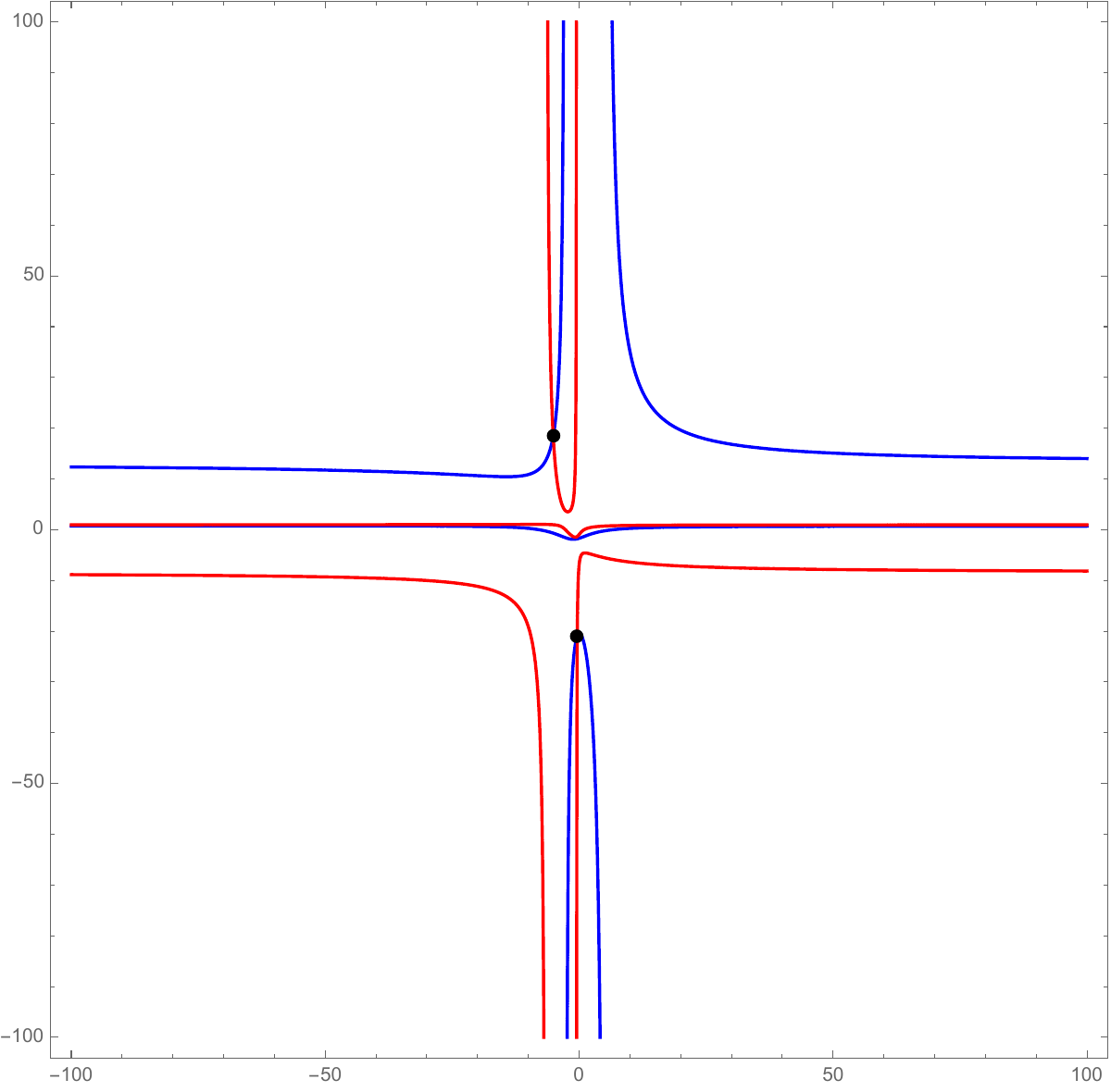}
\endminipage%
\caption{Six projected bisectors when the four lines have the following parameters: $(a,b_{3},c_{3},d_{3},e_{3},b_{4},c_{4},d_{4},e_{4})=(-4,5,10,3,4,5,6,-1,1)$. They match the configurations shown in \cref{fig:topo3comb} and form topology \Rom{3}. }\label{fig:topo3verify}
\end{figure}

\begin{figure}[h]
    \centering
    \begin{minipage}[b]{0.48\textwidth}
        \centering
        \includegraphics[width=0.9\linewidth]{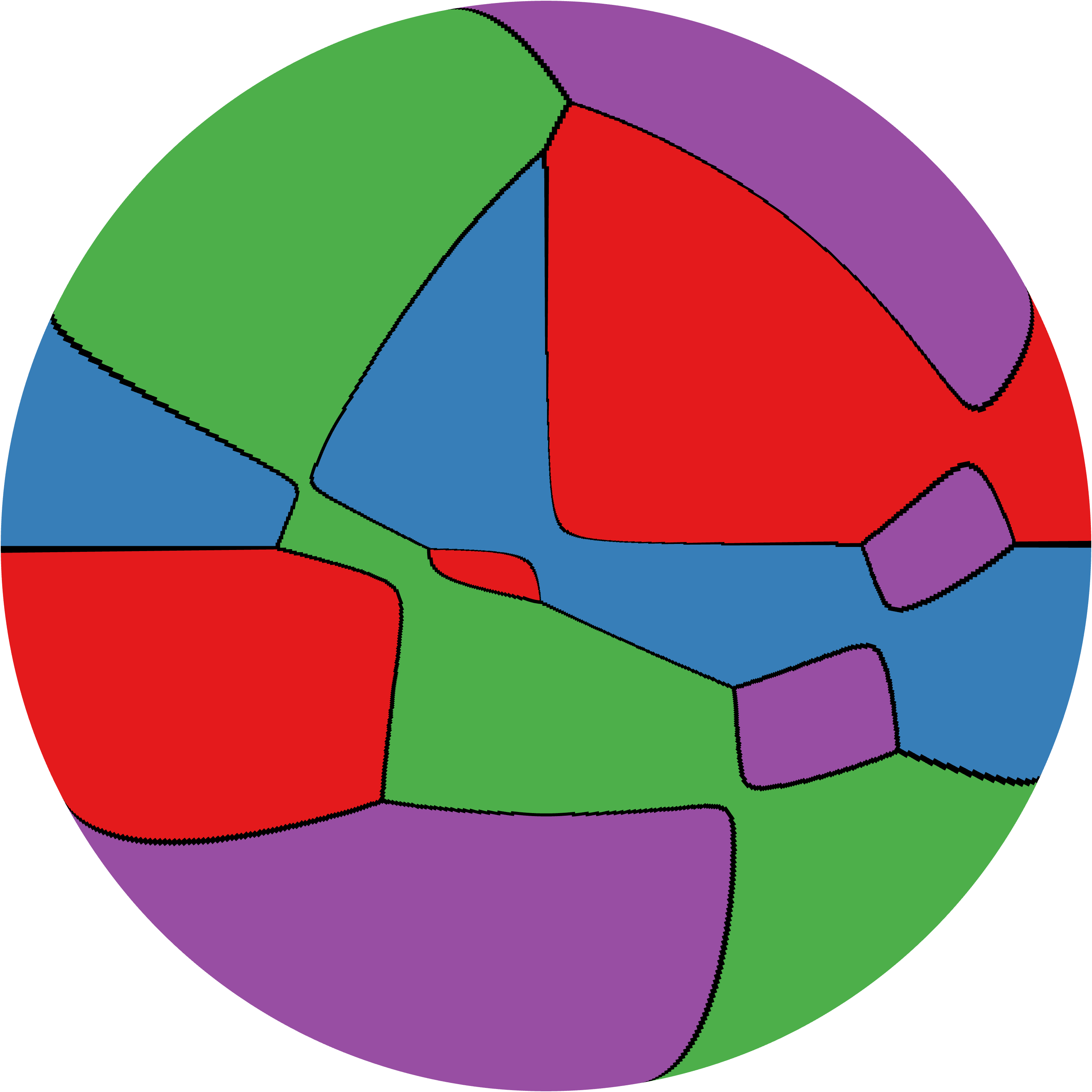}
    \end{minipage}
    \hfill
    \begin{minipage}[b]{0.48\textwidth}
        \centering
        \includegraphics[width=0.9\linewidth]{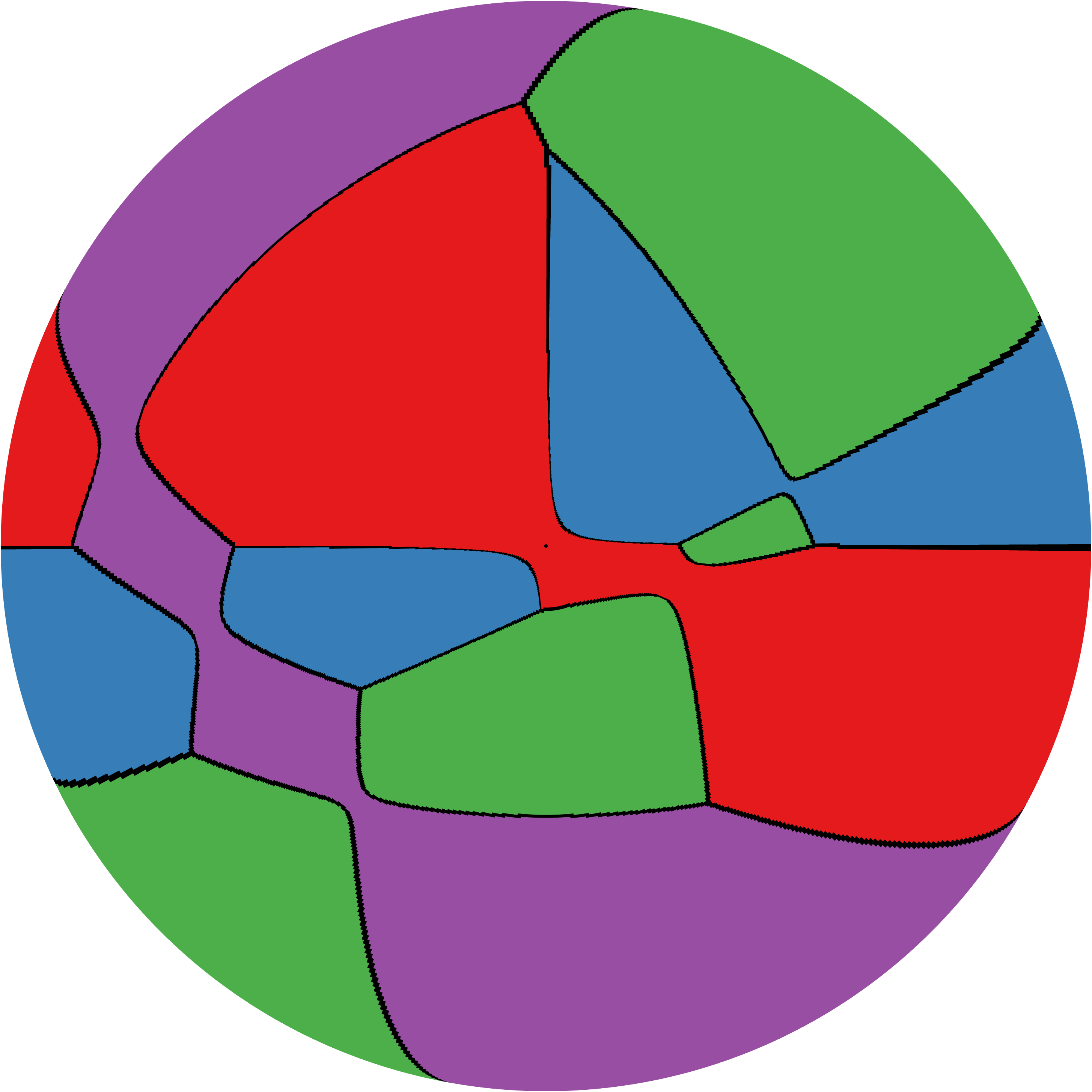}
     \end{minipage}
\caption{Top and bottom view of $\gmap(\fvd(L))$ of topology \Rom{3}. }\label{fig:gmapfvd3}
\end{figure}

\subsection{Topology \Rom{4}: 2 vertices}

The configuration tuple that induces topology \Rom{4} is shown in \cref{fig:topo4comb}. A set of lines that realizes this tuple, hence this topology, has the following parameters: $(a,b_{3},c_{3},d_{3},e_{3},b_{4},c_{4},d_{4},e_{4})=(4,-9,-8,-6,-1,10,-1,0,-1)$. \cref{fig:topo4verify} shows the six projected bisectors of these four lines, which match the configurations shown in \cref{fig:topo4comb}. The $\gmap(\fvd(L))$ is shown in \cref{fig:gmapfvd4}.

\begin{figure}[h]
        \centering
        \includegraphics[page=4,width=\textwidth]{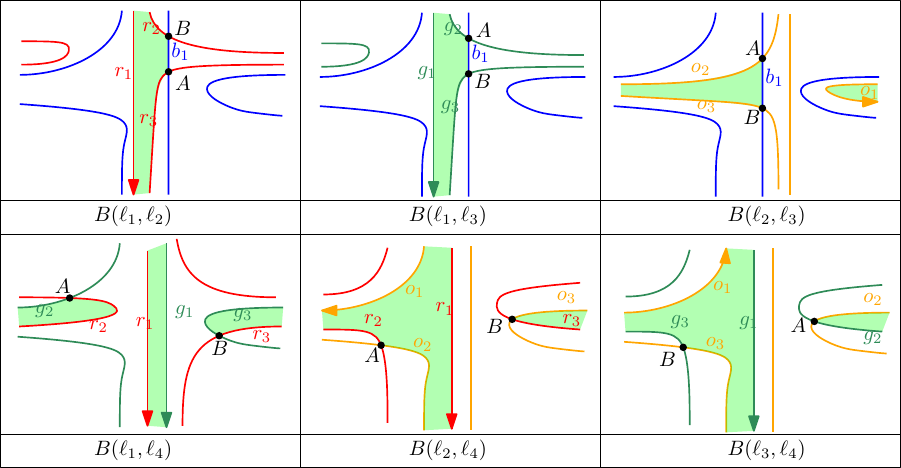}
        \caption{Combination that forms topology \Rom{4}. }\label{fig:topo4comb}
\end{figure}

\begin{figure}[!h]
\centering
\minipage{0.03\textwidth}
(a)	
\endminipage
\minipage{0.3\textwidth}
\includegraphics[width=\textwidth]{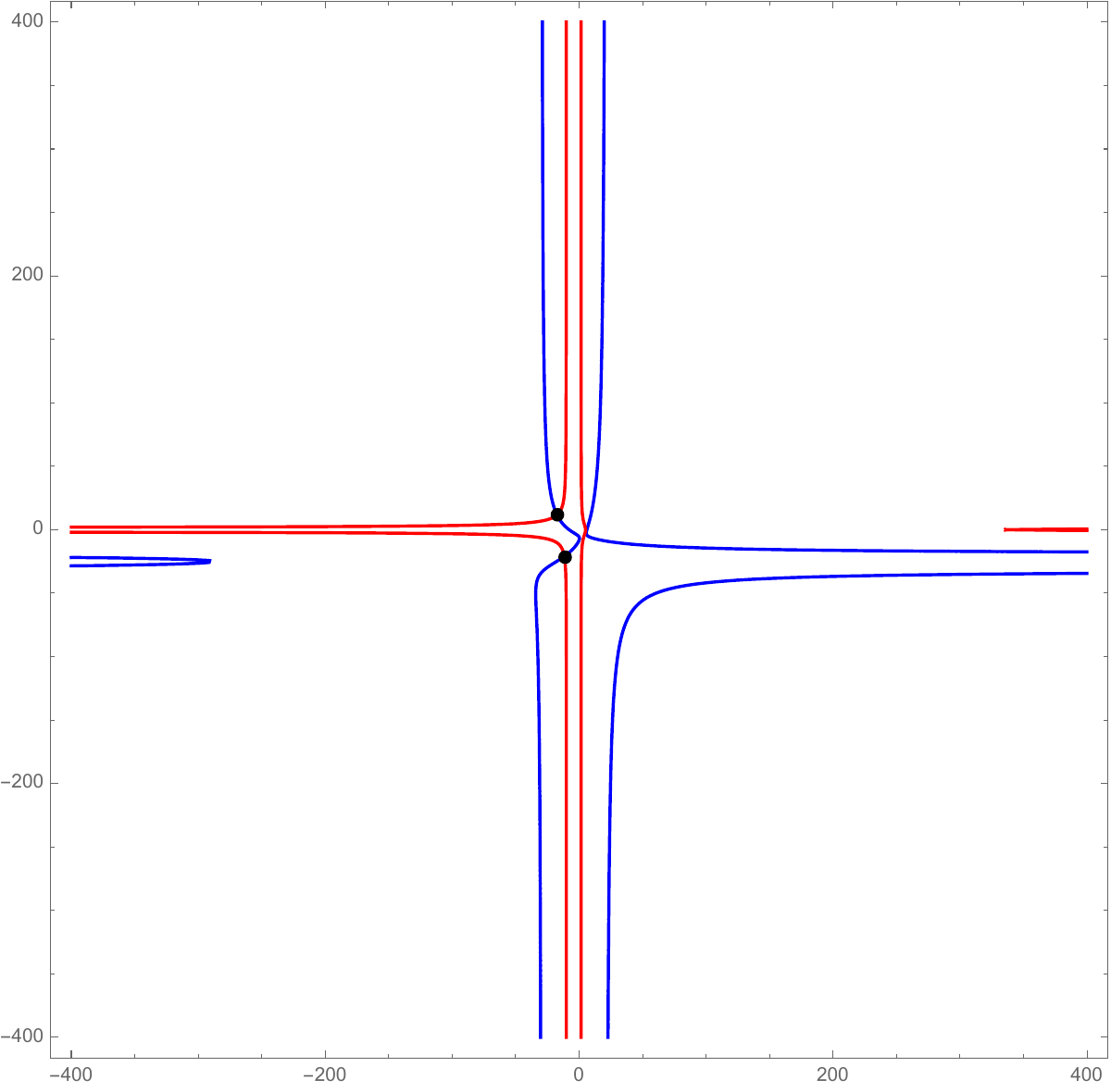}
\endminipage\hfill
\minipage{0.03\textwidth}
(b)	
\endminipage
\minipage{0.3\textwidth}
\includegraphics[width=\textwidth]{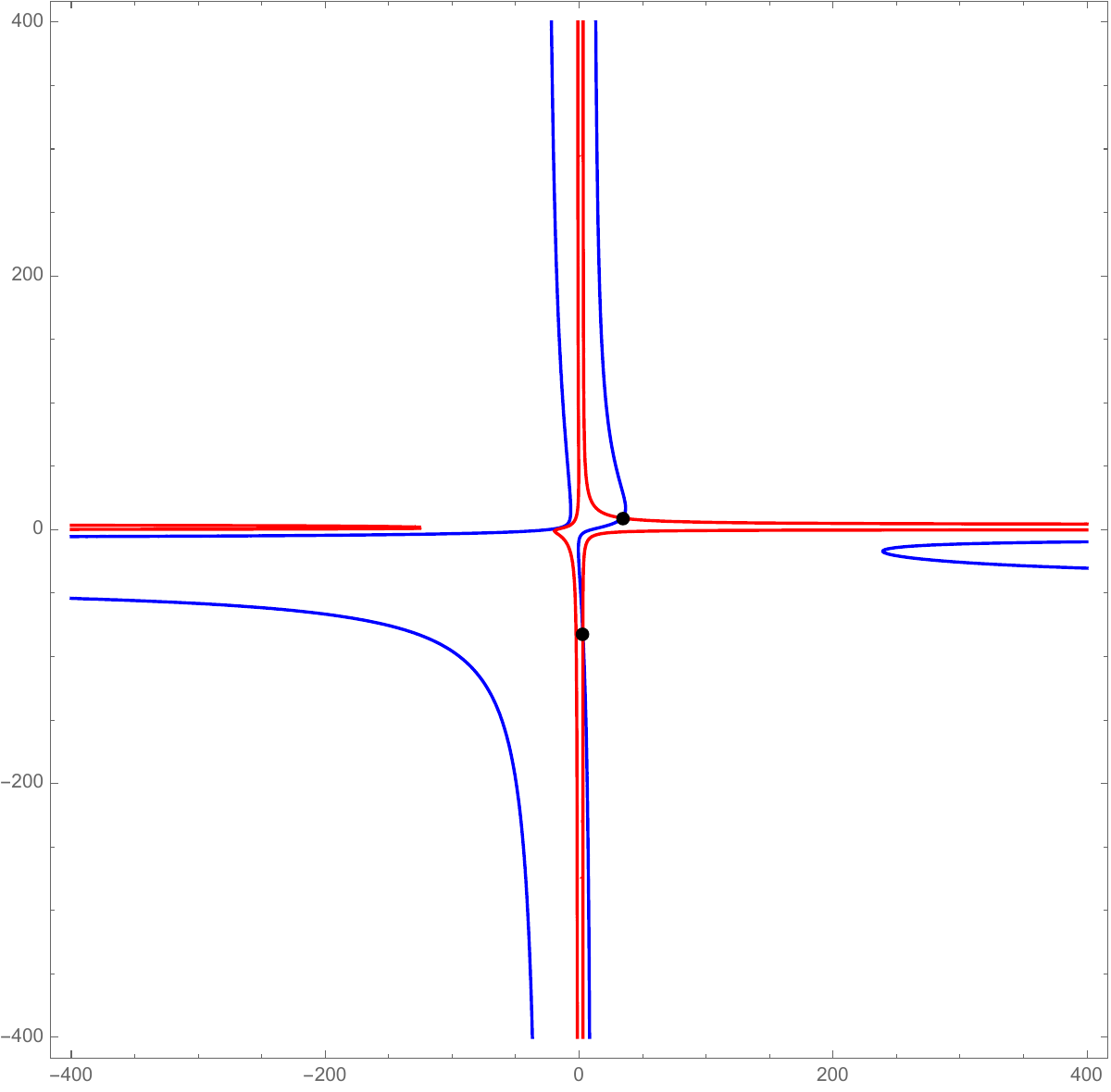}	
\endminipage\hfill
\minipage{0.03\textwidth}
(c)	
\endminipage
\minipage{0.3\textwidth}
\includegraphics[width=\textwidth]{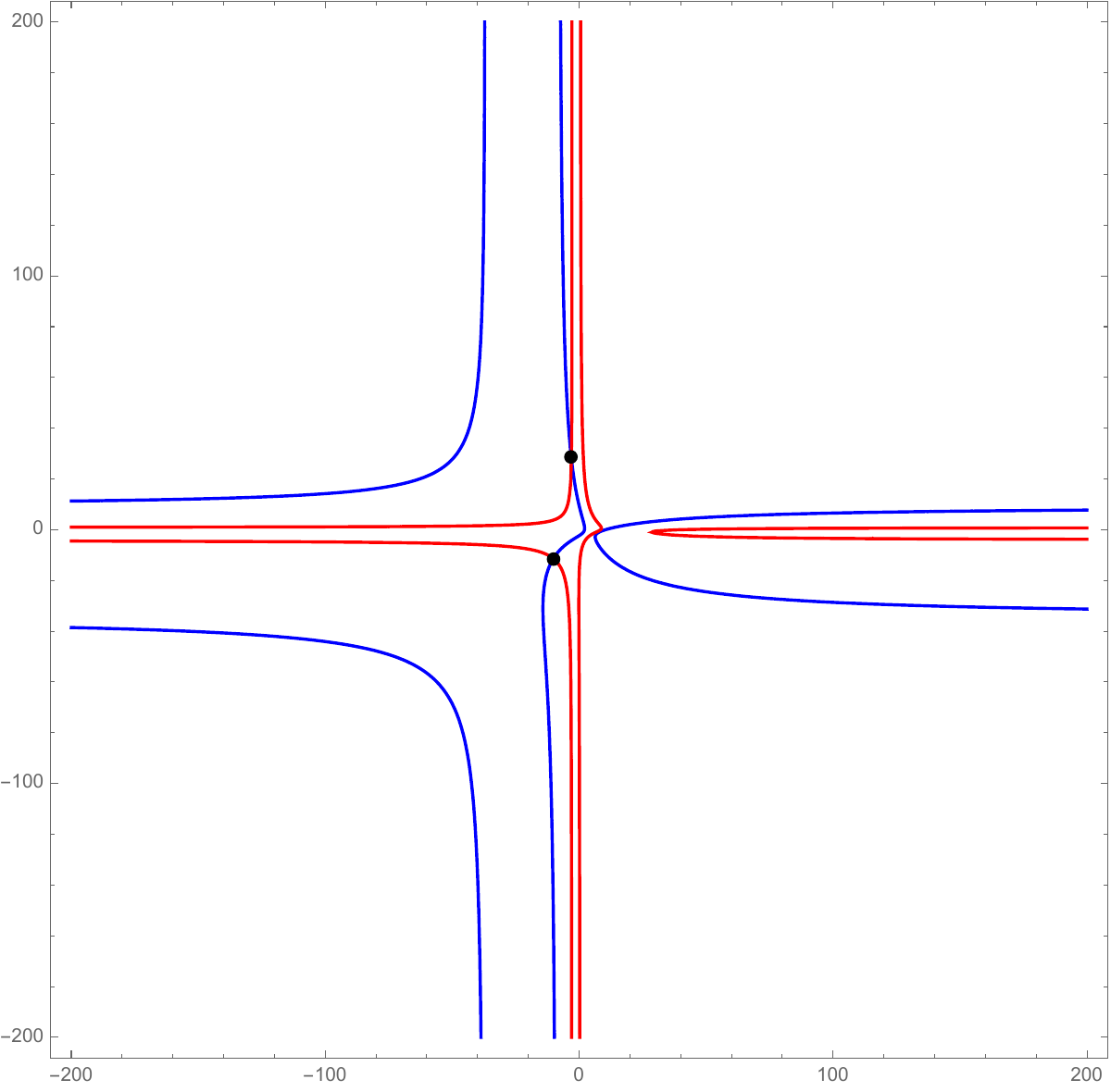}
\endminipage%
\newline
\centering
\minipage{0.03\textwidth}
(d)	
\endminipage
\minipage{0.3\textwidth}
\includegraphics[width=\textwidth]{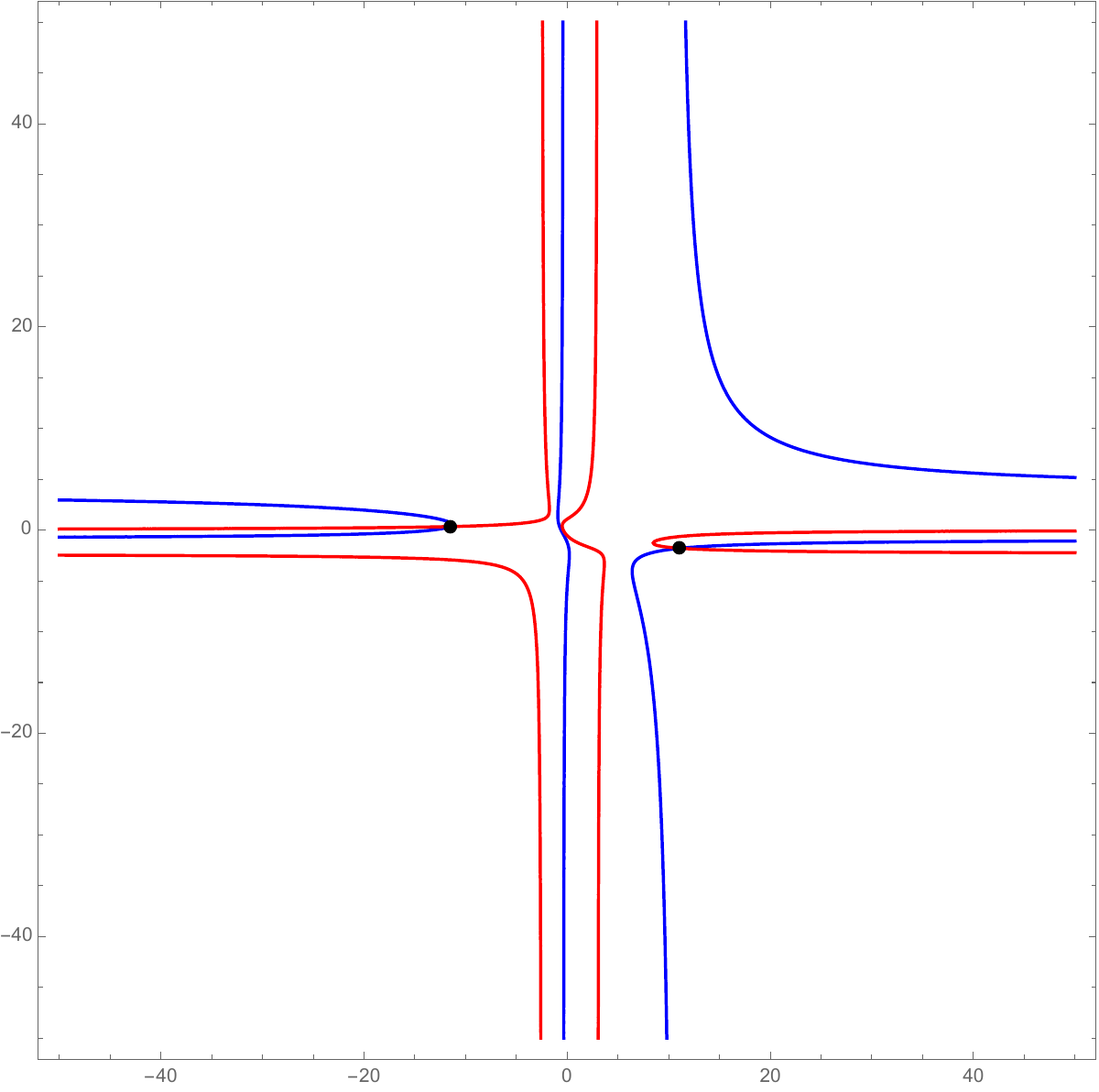}
\endminipage\hfill
\minipage{0.03\textwidth}
(e)	
\endminipage
\minipage{0.3\textwidth}
\includegraphics[width=\textwidth]{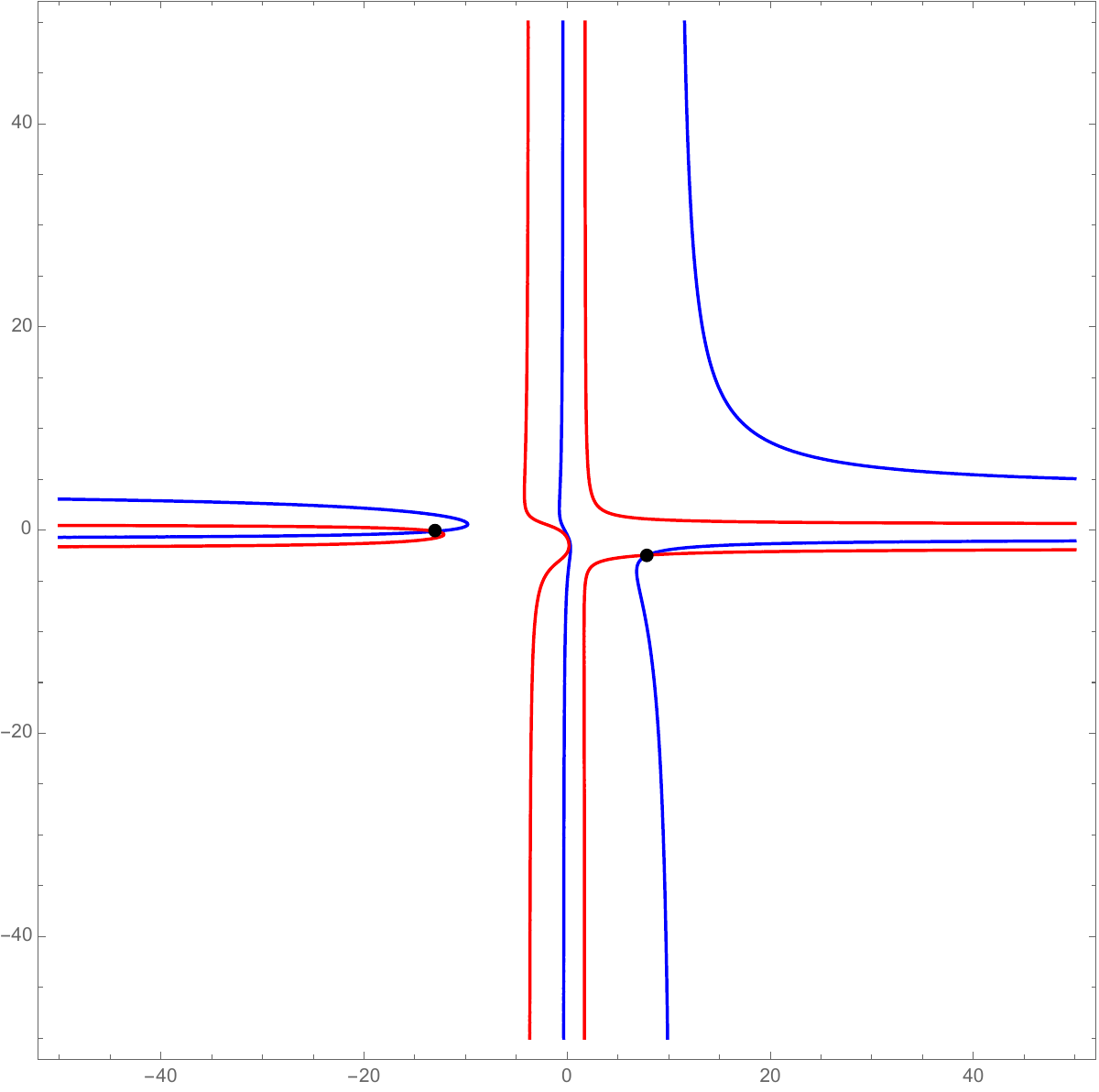}
\endminipage\hfill
\minipage{0.03\textwidth}
(f)	
\endminipage
\minipage{0.3\textwidth}
\includegraphics[width=\textwidth]{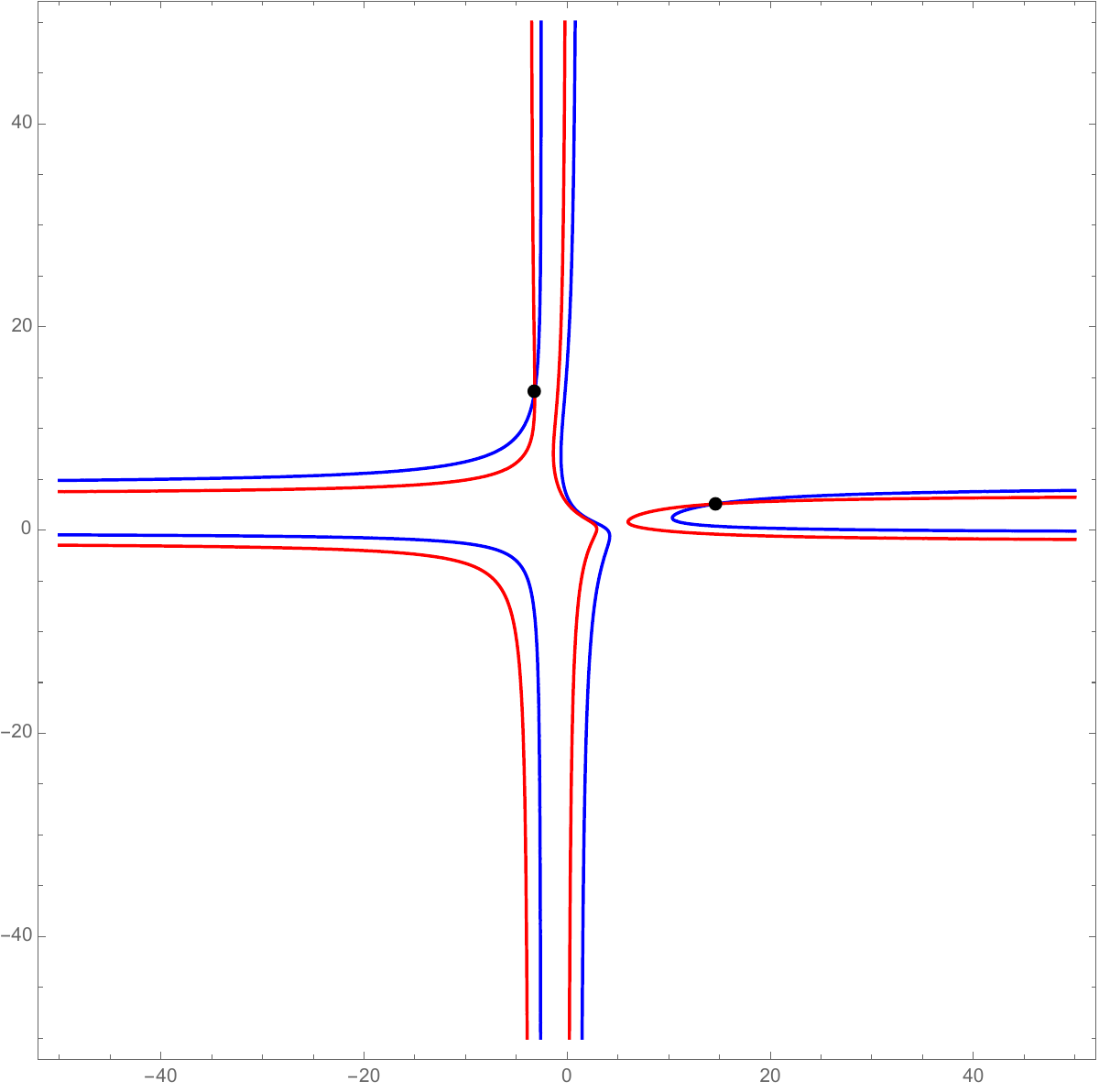}
\endminipage%
\caption{Six projected bisectors when the four lines have the following parameters: $(a,b_{3},c_{3},d_{3},e_{3},b_{4},c_{4},d_{4},e_{4})=(4,-9,-8,-6,-1,10,-1,0,-1)$. They match the configurations shown in \cref{fig:topo4comb} and form topology \Rom{4}. }\label{fig:topo4verify}
\end{figure}

\begin{figure}[h]
    \centering
    \begin{minipage}[b]{0.48\textwidth}
        \centering
        \includegraphics[width=0.9\linewidth]{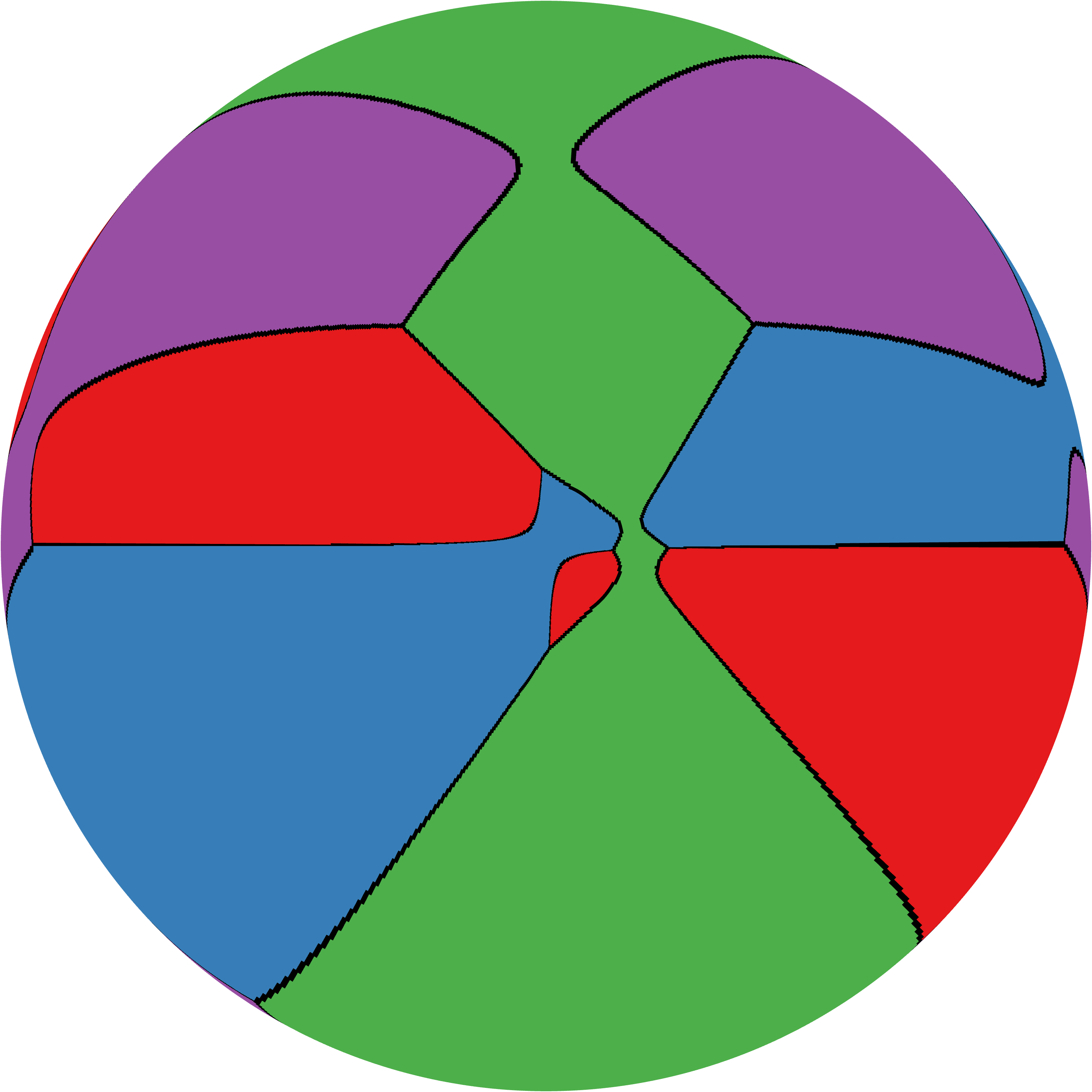}
    \end{minipage}
    \hfill
    \begin{minipage}[b]{0.48\textwidth}
        \centering
        \includegraphics[width=0.9\linewidth]{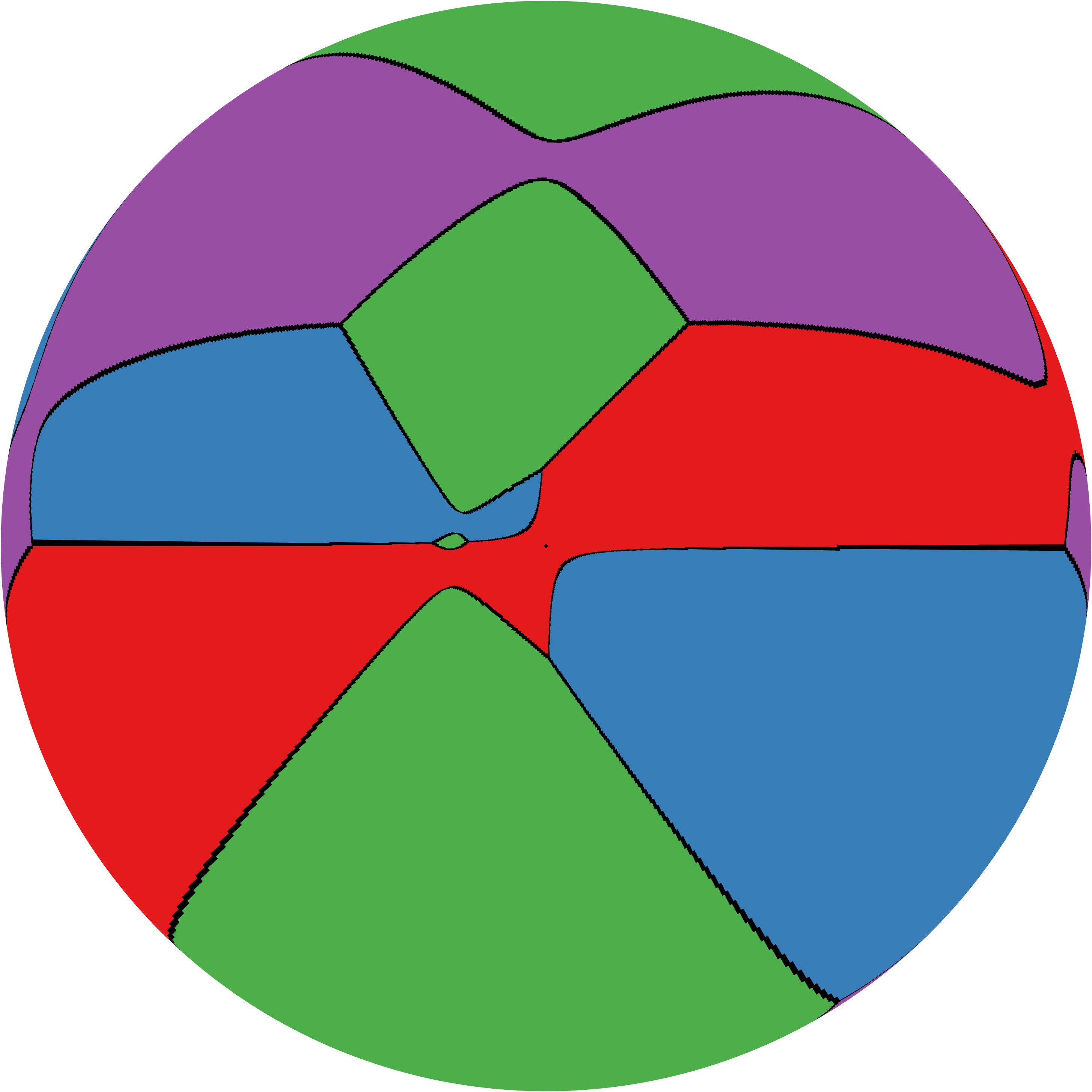}
     \end{minipage}
\caption{Top and bottom view of $\gmap(\fvd(L))$ of topology \Rom{4}. }\label{fig:gmapfvd4}
\end{figure}

\subsection{Topology \Rom{5}: 4 vertices}

The configuration tuple that induces topology \Rom{5} is shown in \cref{fig:topo5comb}. A set of lines that realizes this tuple, hence this topology, has the following parameters: $(a,b_{3},c_{3},d_{3},e_{3},b_{4},c_{4},d_{4},e_{4})=(1,-19,11,7,16,-11,1,-12,6)$. \cref{fig:topo5verify} shows the six projected bisectors of these four lines, which match the configurations shown in \cref{fig:topo5comb}. The $\gmap(\fvd(L))$ is shown in \cref{fig:gmapfvd5}.

\begin{figure}[h]
        \centering
        \includegraphics[page=5,width=\textwidth]{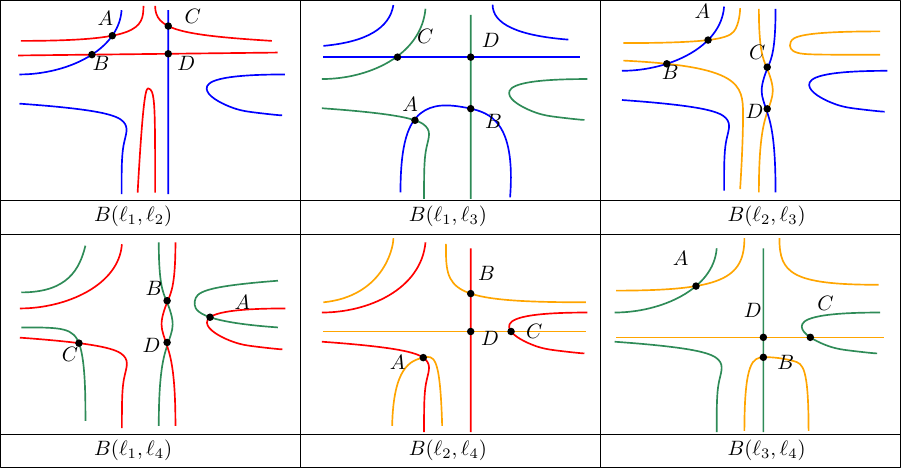}
        \caption{Combination that forms topology \Rom{5}. }\label{fig:topo5comb}
\end{figure}

\begin{figure}[!h]
\centering
\minipage{0.03\textwidth}
(a)	
\endminipage
\minipage{0.3\textwidth}
\includegraphics[width=\textwidth]{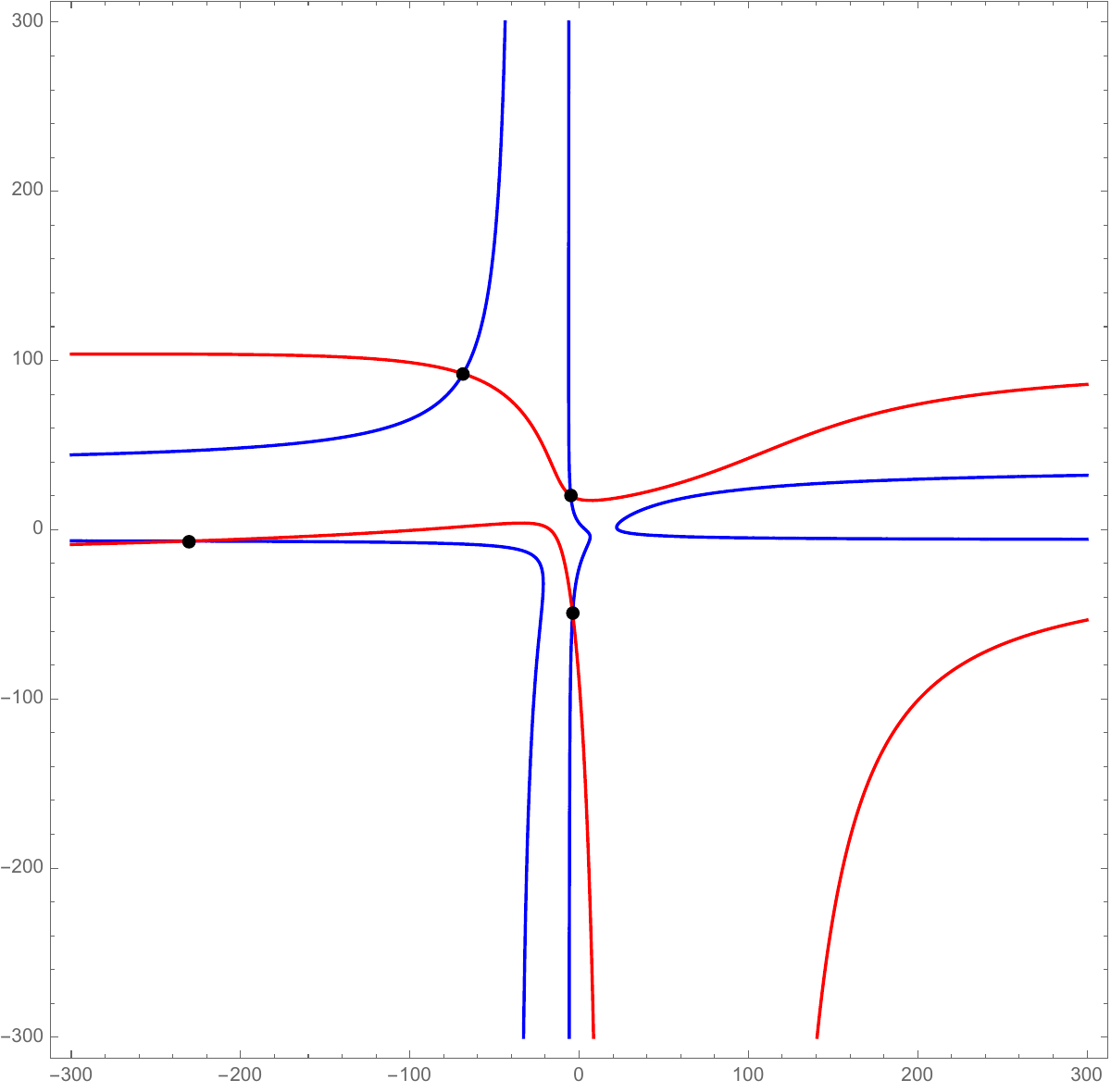}
\endminipage\hfill
\minipage{0.03\textwidth}
(b)	
\endminipage
\minipage{0.3\textwidth}
\includegraphics[width=\textwidth]{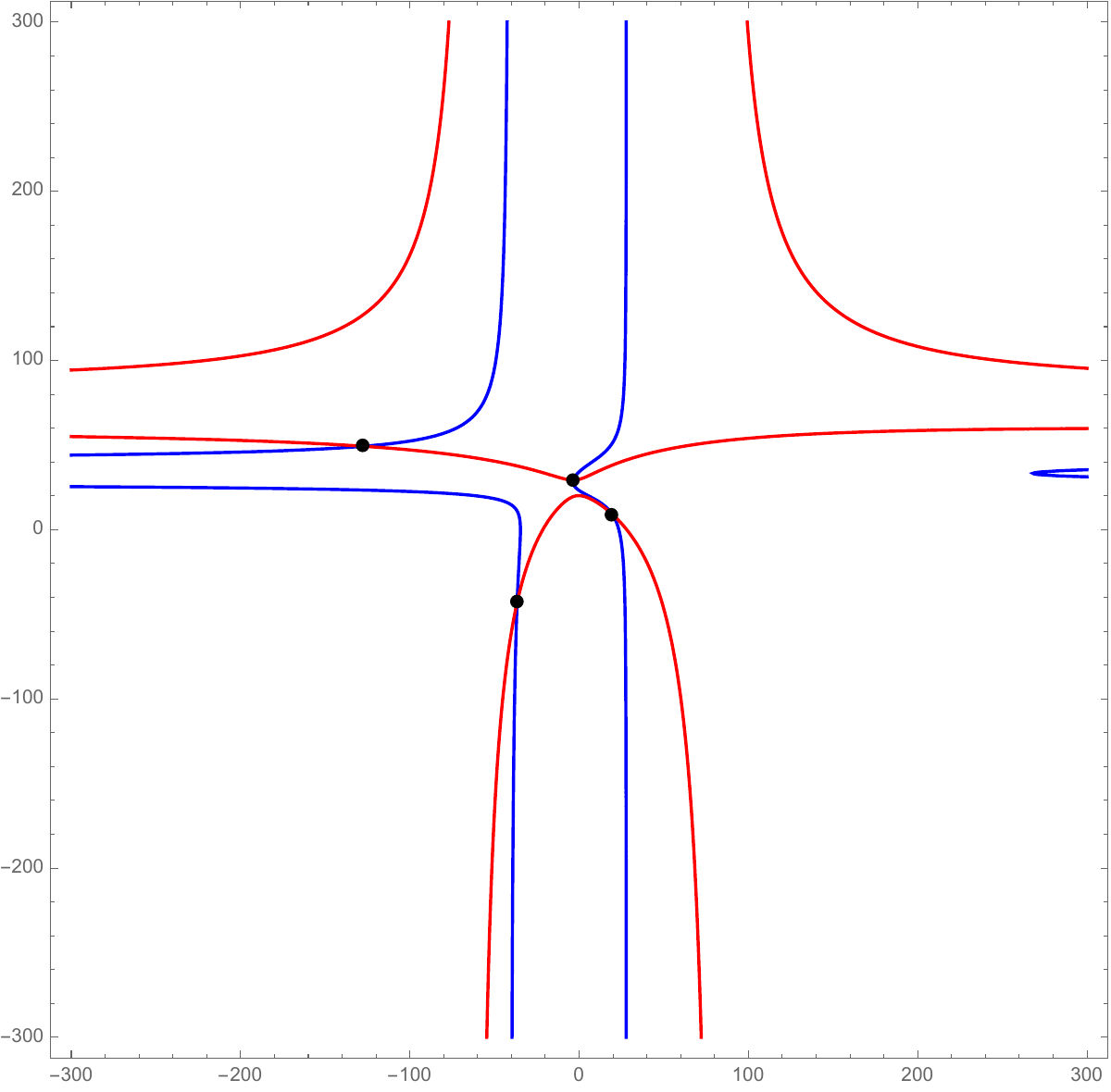}	
\endminipage\hfill
\minipage{0.03\textwidth}
(c)	
\endminipage
\minipage{0.3\textwidth}
\includegraphics[width=\textwidth]{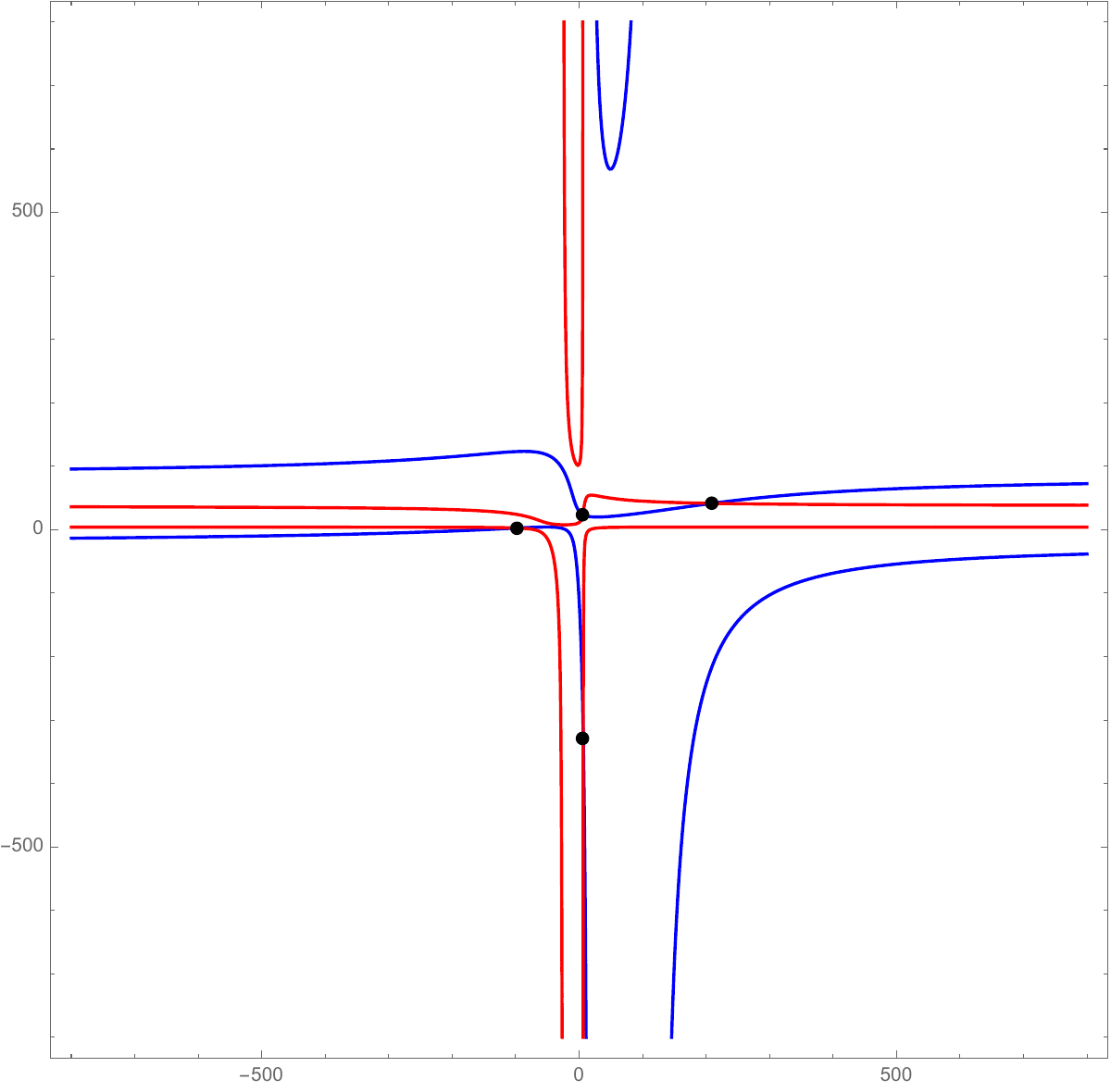}
\endminipage%
\newline
\centering
\minipage{0.03\textwidth}
(d)	
\endminipage
\minipage{0.3\textwidth}
\includegraphics[width=\textwidth]{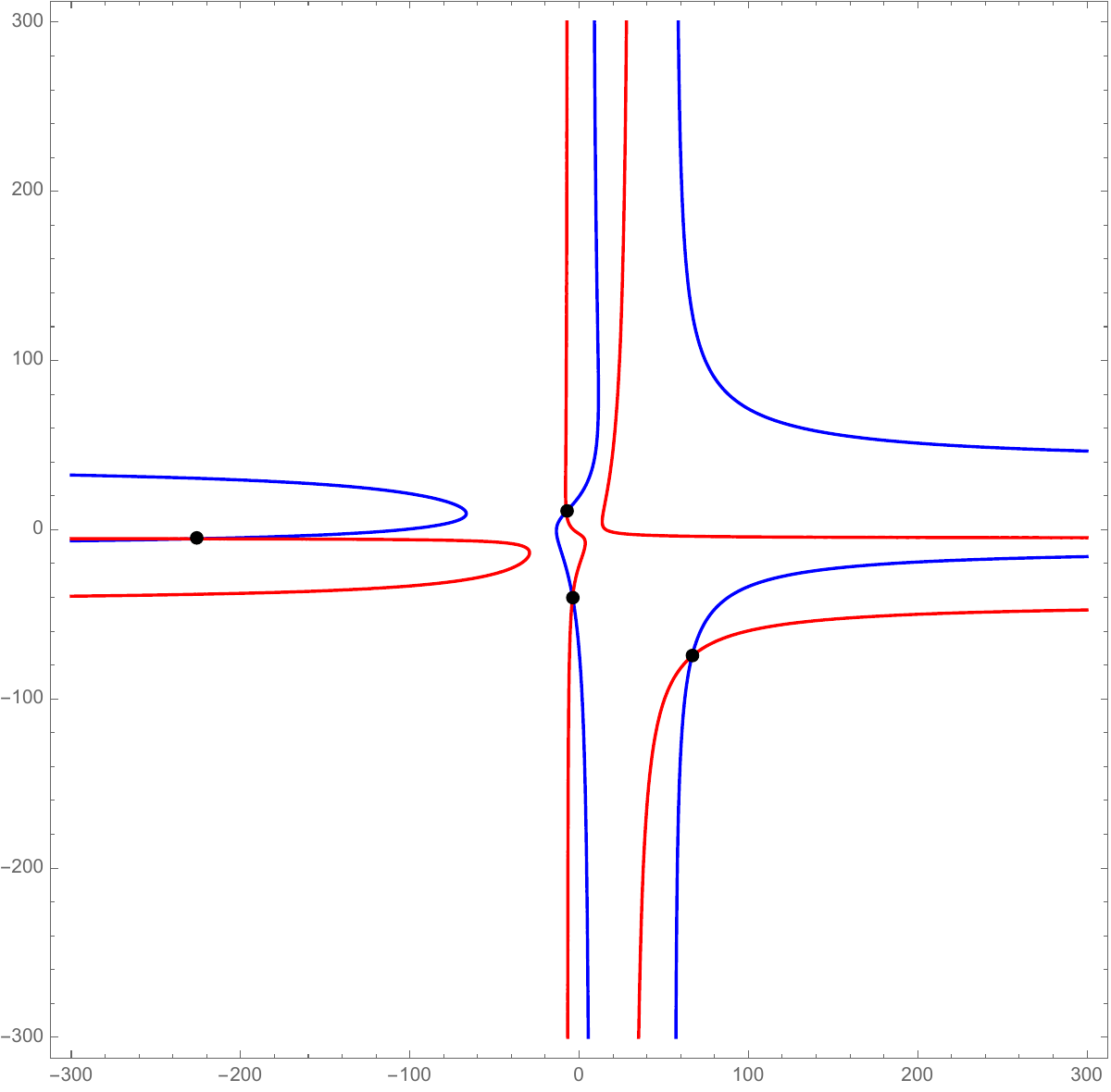}
\endminipage\hfill
\minipage{0.03\textwidth}
(e)	
\endminipage
\minipage{0.3\textwidth}
\includegraphics[width=\textwidth]{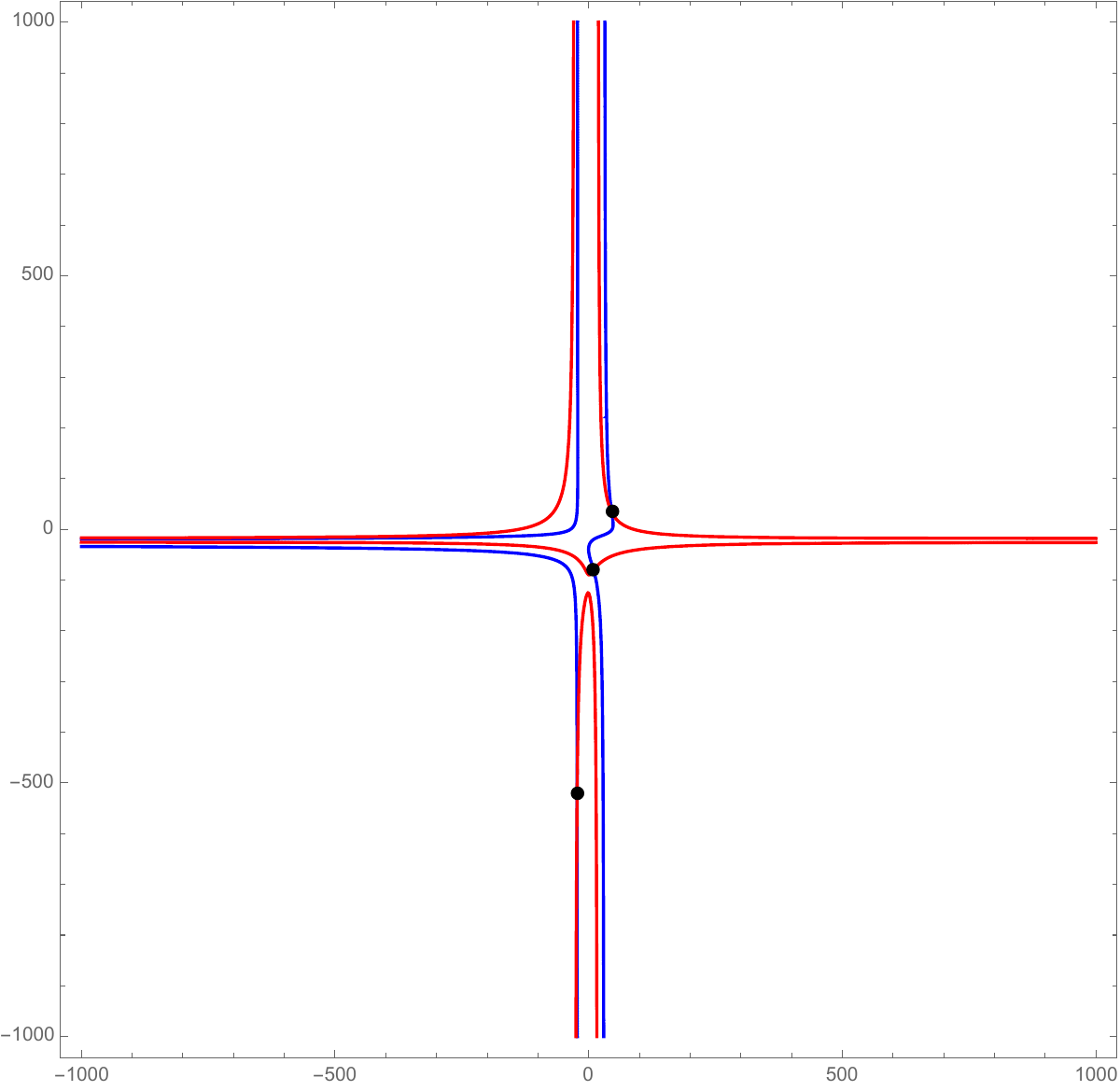}
\endminipage\hfill
\minipage{0.03\textwidth}
(f)	
\endminipage
\minipage{0.3\textwidth}
\includegraphics[width=\textwidth]{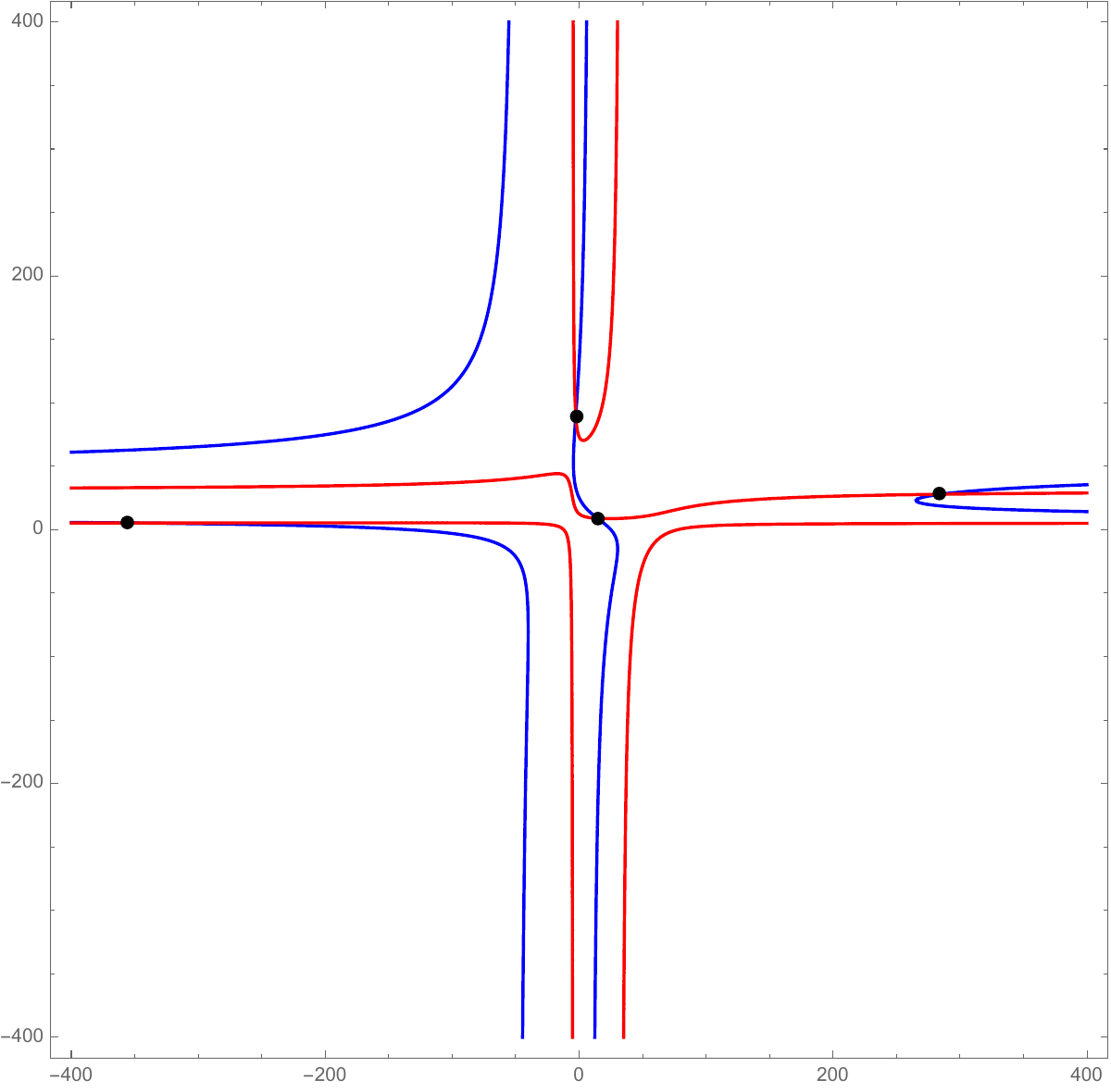}
\endminipage%
\caption{Six projected bisectors when the four lines have the following parameters: $(a,b_{3},c_{3},d_{3},e_{3},b_{4},c_{4},d_{4},e_{4})=(1,-19,11,7,16,-11,1,-12,6)$. They match the configurations shown in \cref{fig:topo5comb} and form topology \Rom{5}. }\label{fig:topo5verify}
\end{figure}

\begin{figure}[h]
    \centering
    \begin{minipage}[b]{0.48\textwidth}
        \centering
        \includegraphics[width=0.9\linewidth]{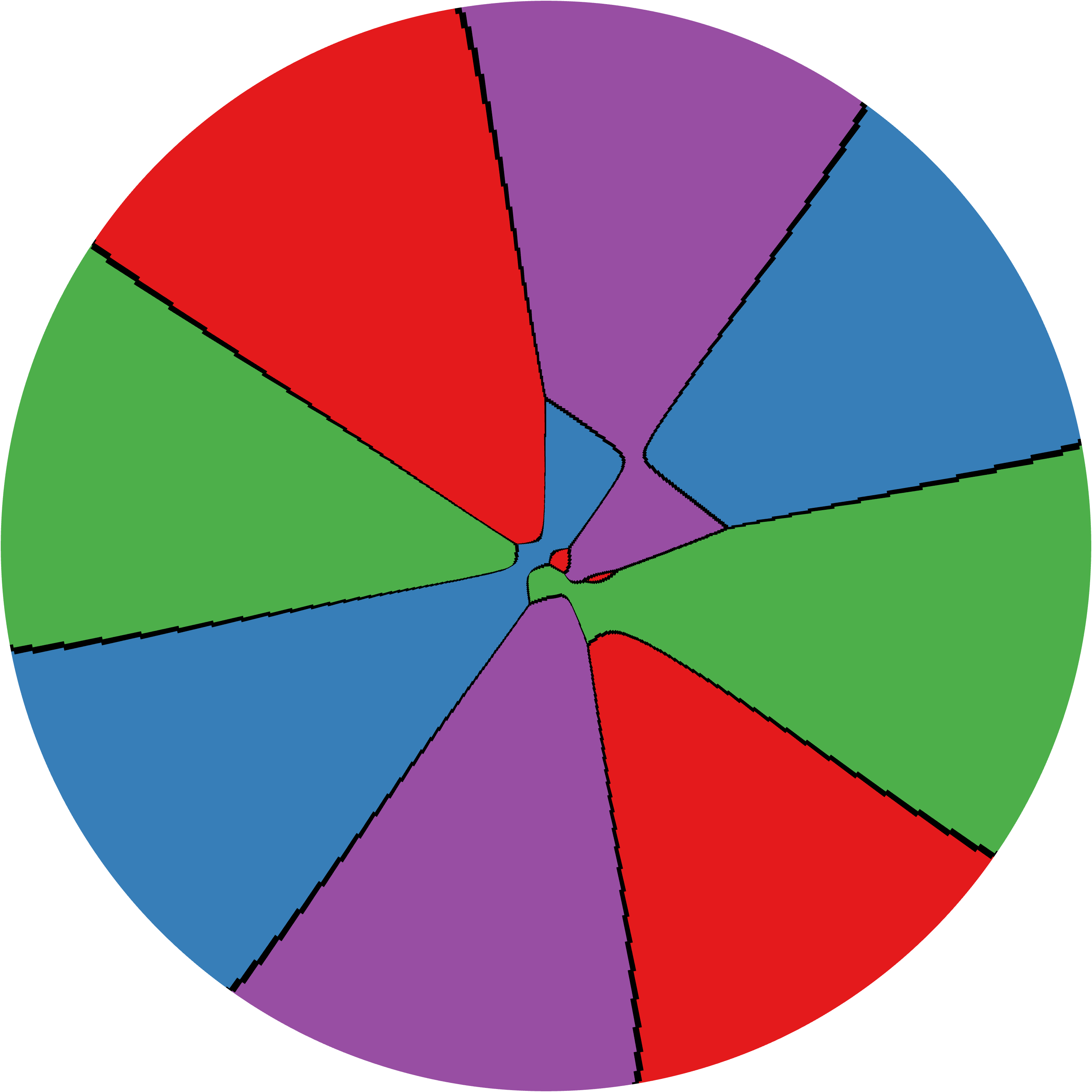}
    \end{minipage}
    \hfill
    \begin{minipage}[b]{0.48\textwidth}
        \centering
        \includegraphics[width=0.9\linewidth]{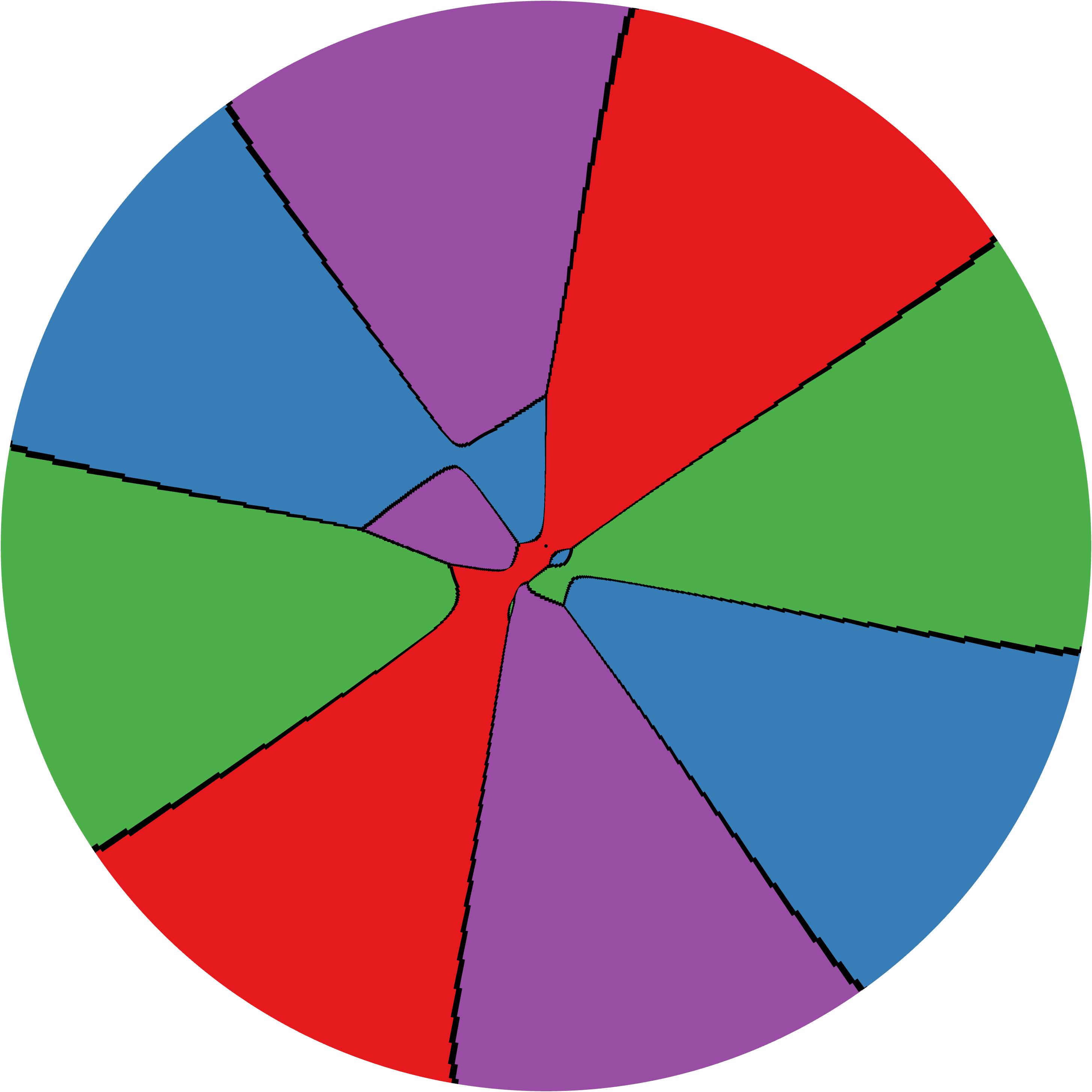}
     \end{minipage}
\caption{Top and bottom view of $\gmap(\fvd(L))$ of topology \Rom{5}. }\label{fig:gmapfvd5}
\end{figure}

\subsection{Topology \Rom{6}: 4 vertices}

The configuration tuple that induces topology \Rom{6} is shown in \cref{fig:topo6comb}. A set of lines that realizes this tuple, hence this topology, has the following parameters: $(a,b_{3},c_{3},d_{3},e_{3},b_{4},c_{4},d_{4},e_{4})=(16,-10,3,0,-1,5,1,2,5)$. \cref{fig:topo6verify} shows the six projected bisectors of these four lines, which match the configurations shown in \cref{fig:topo6comb}. The $\gmap(\fvd(L))$ is shown in \cref{fig:gmapfvd6}.

\begin{figure}[h]
    \centering
    \includegraphics[page=6,width=\textwidth]{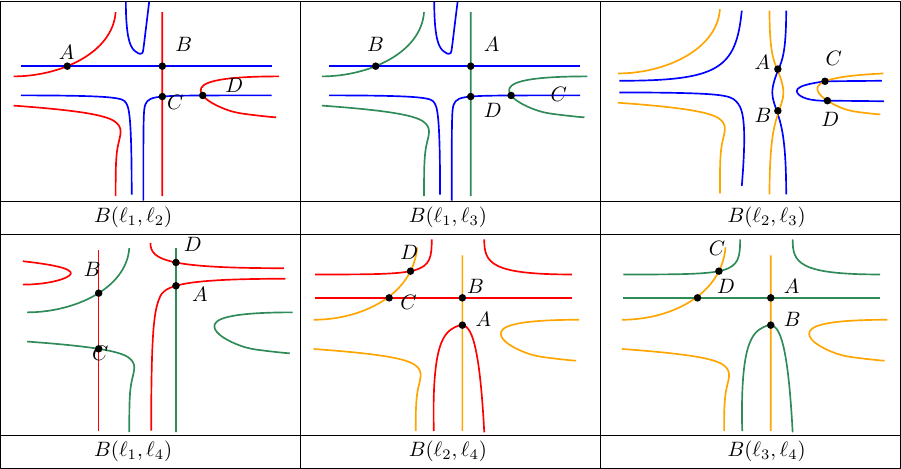}
    \caption{Combination that forms topology \Rom{6}. }\label{fig:topo6comb}
\end{figure}

\begin{figure}[!h]
\centering
\minipage{0.03\textwidth}
(a)	
\endminipage
\minipage{0.3\textwidth}
\includegraphics[width=\textwidth]{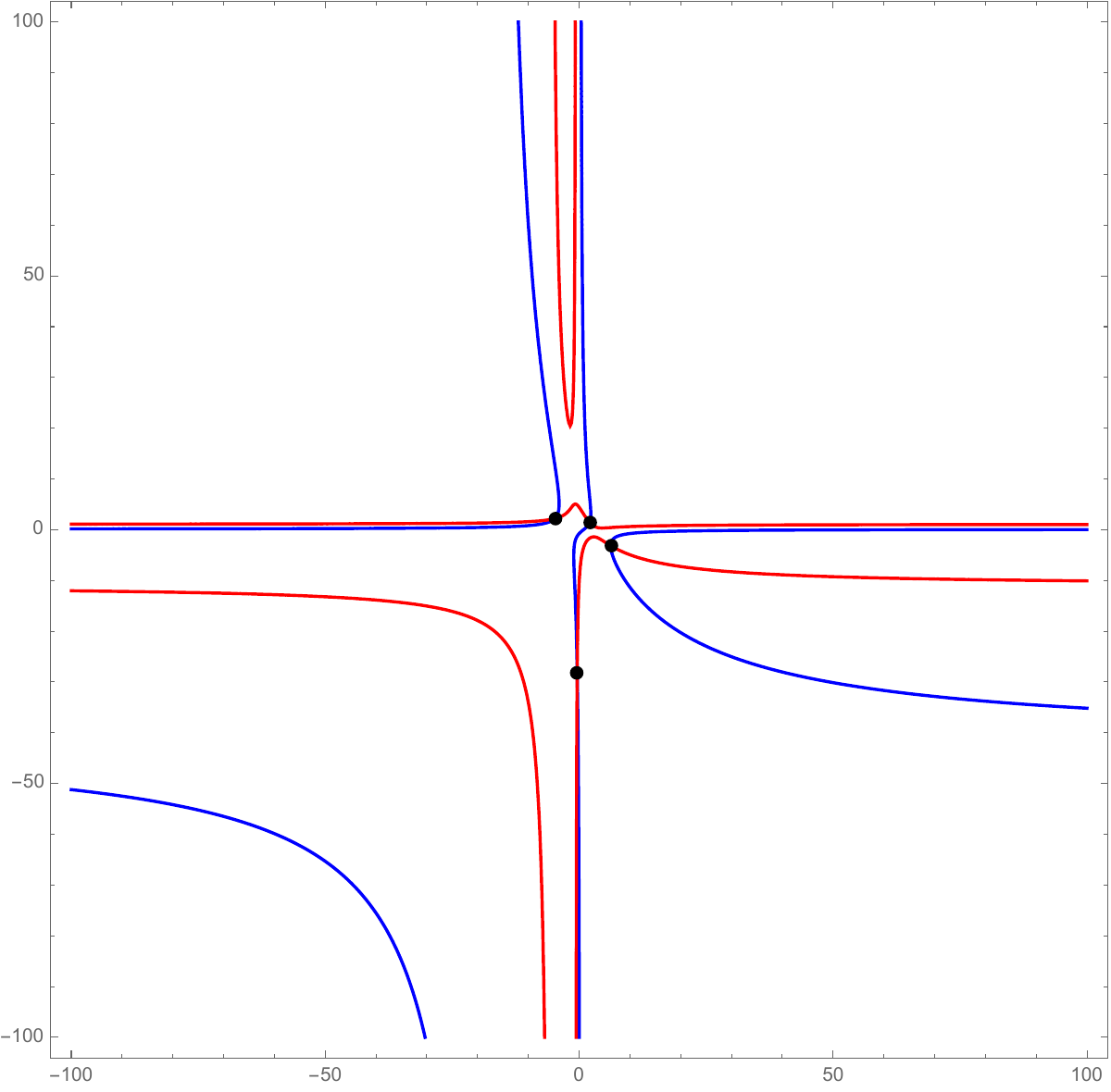}
\endminipage\hfill
\minipage{0.03\textwidth}
(b)	
\endminipage
\minipage{0.3\textwidth}
\includegraphics[width=\textwidth]{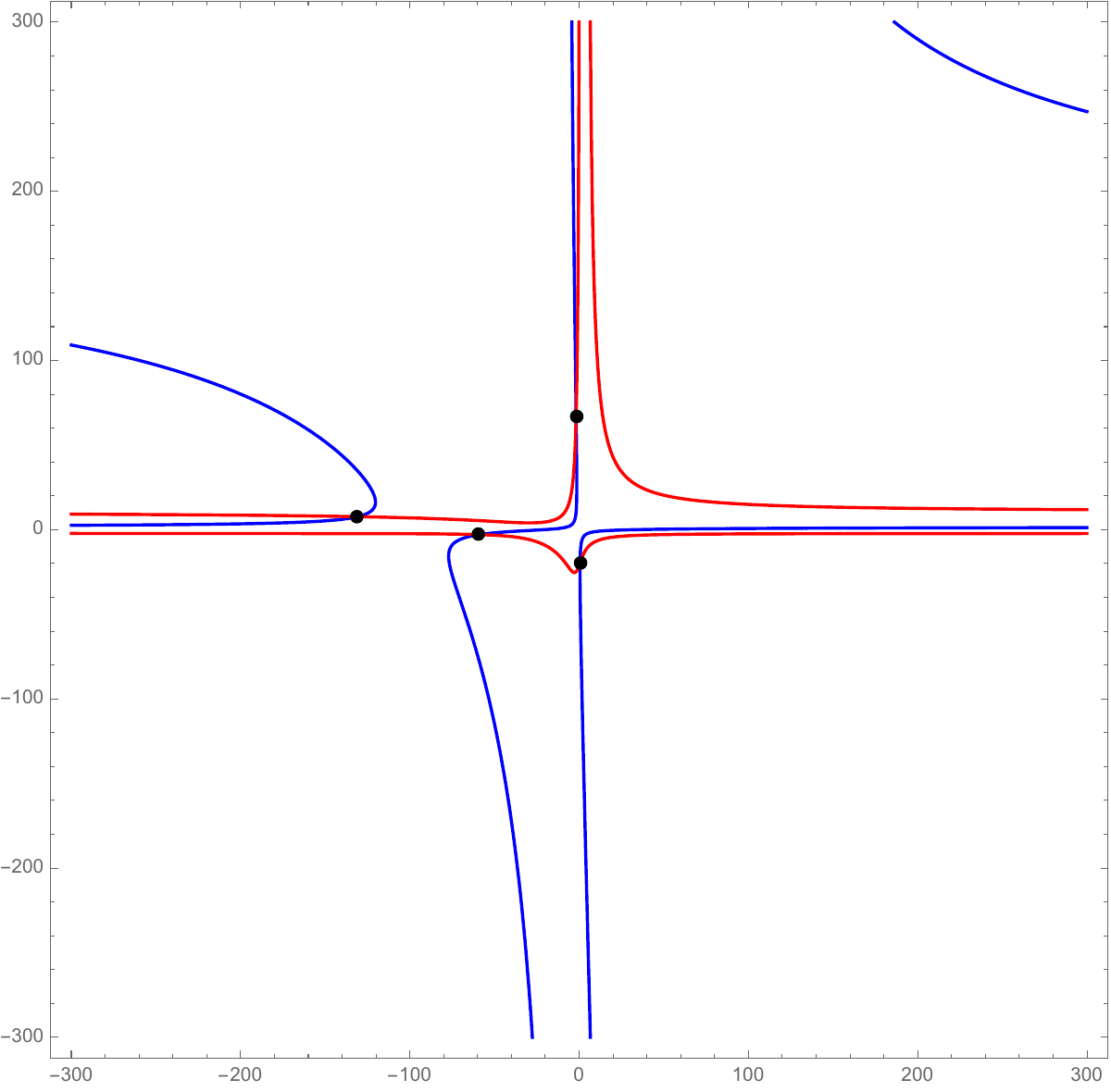}	
\endminipage\hfill
\minipage{0.03\textwidth}
(c)	
\endminipage
\minipage{0.3\textwidth}
\includegraphics[width=\textwidth]{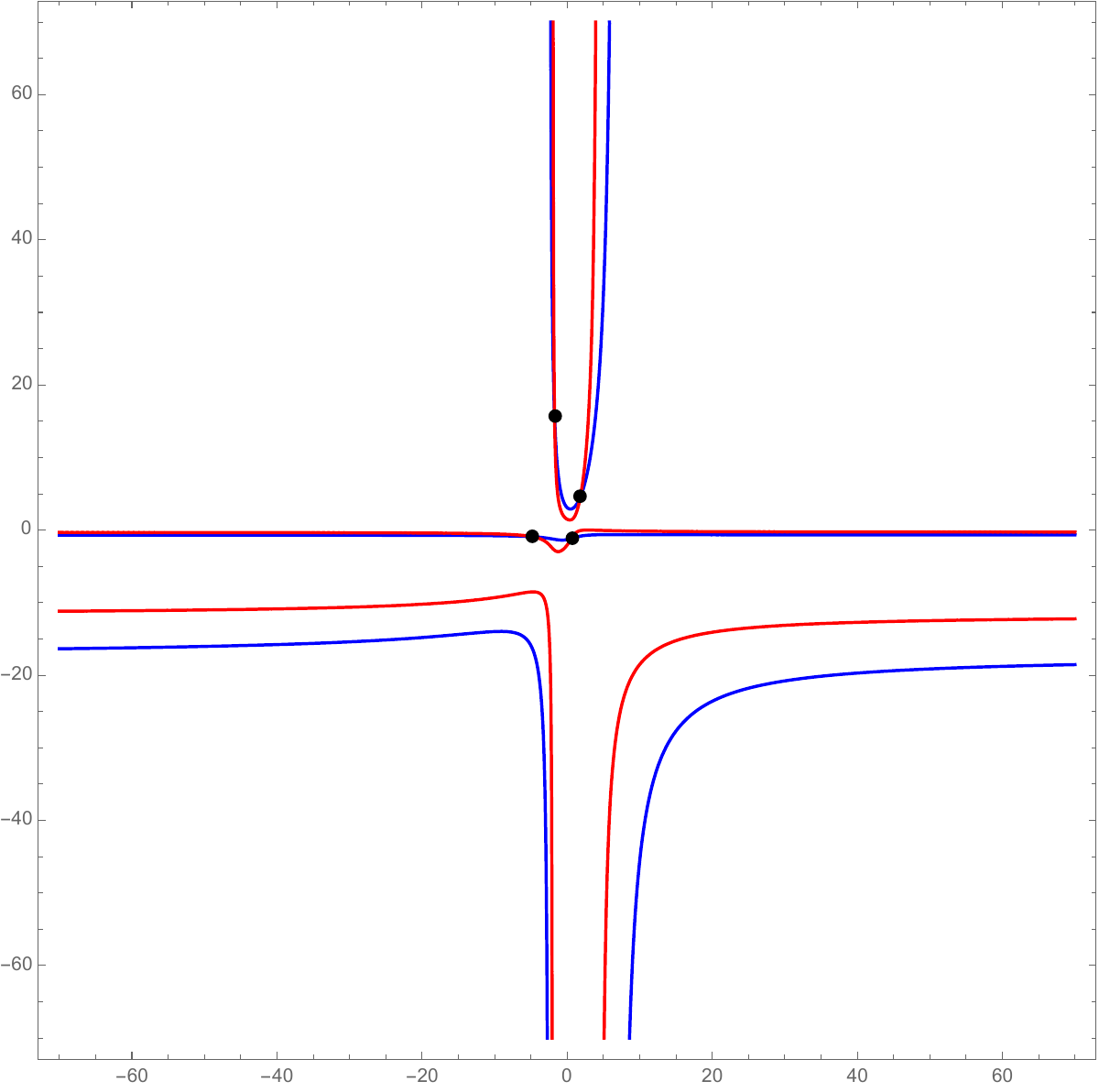}
\endminipage%
\newline
\centering
\minipage{0.03\textwidth}
(d)	
\endminipage
\minipage{0.3\textwidth}
\includegraphics[width=\textwidth]{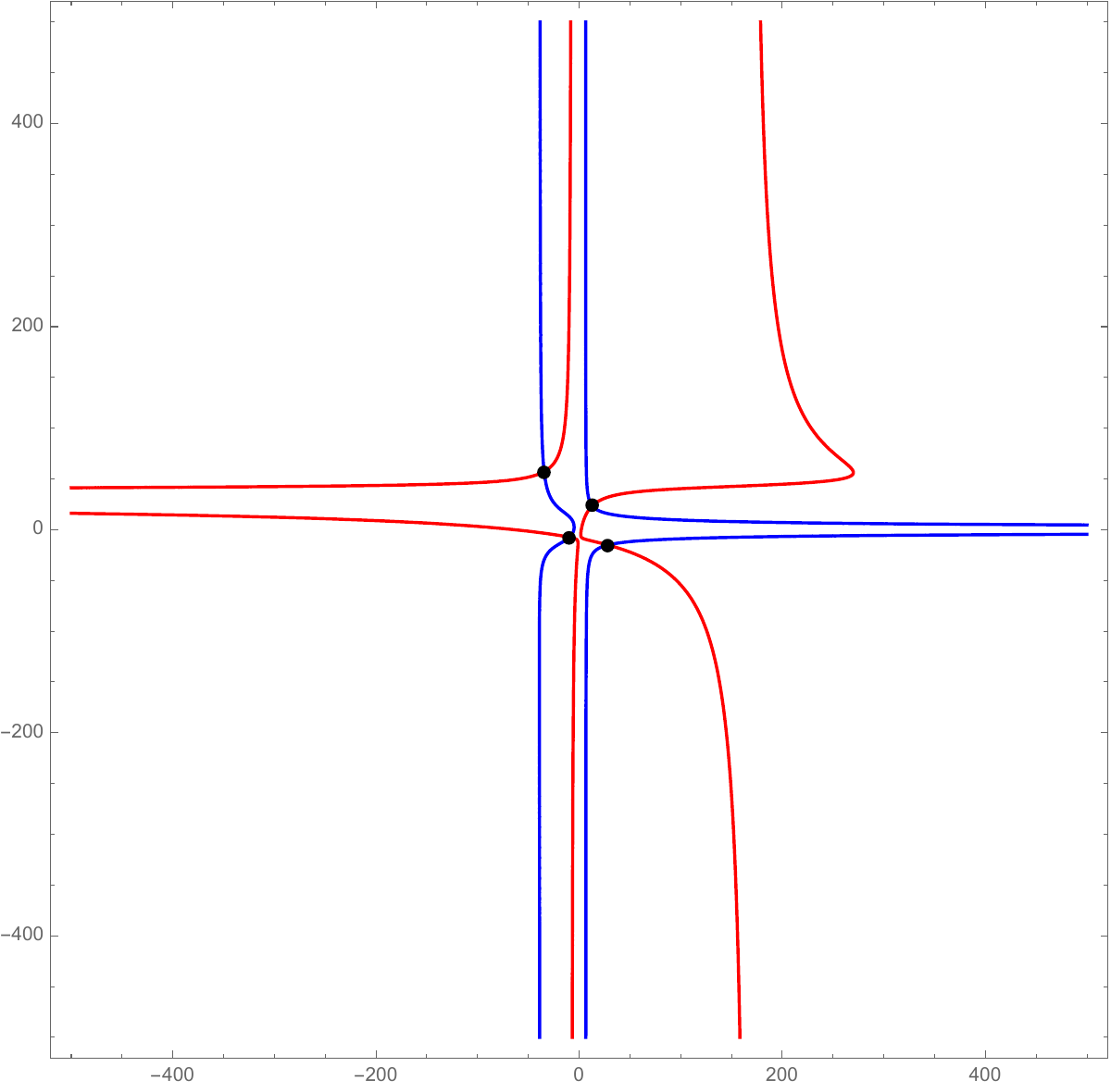}
\endminipage\hfill
\minipage{0.03\textwidth}
(e)	
\endminipage
\minipage{0.3\textwidth}
\includegraphics[width=\textwidth]{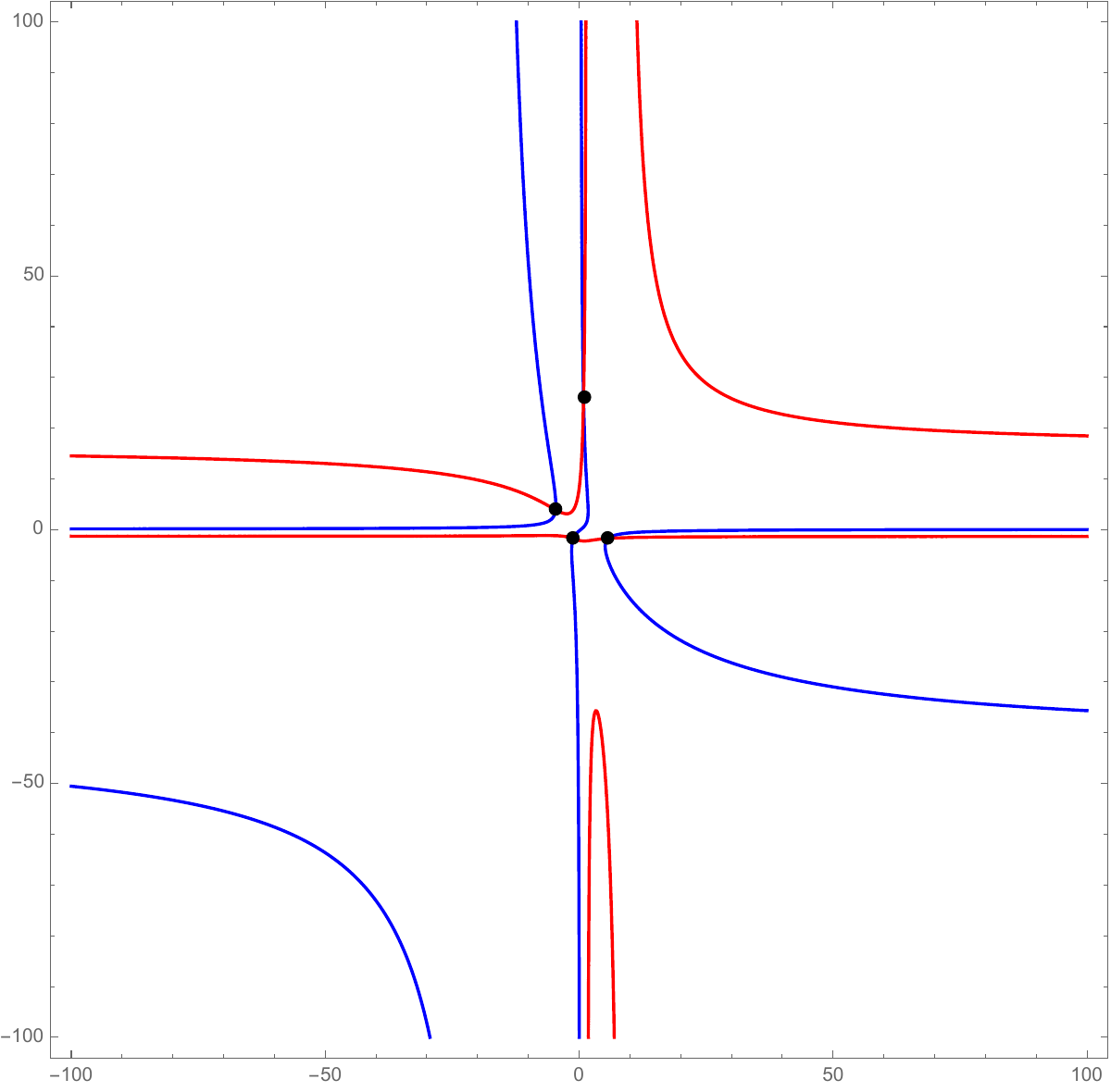}
\endminipage\hfill
\minipage{0.03\textwidth}
(f)	
\endminipage
\minipage{0.3\textwidth}
\includegraphics[width=\textwidth]{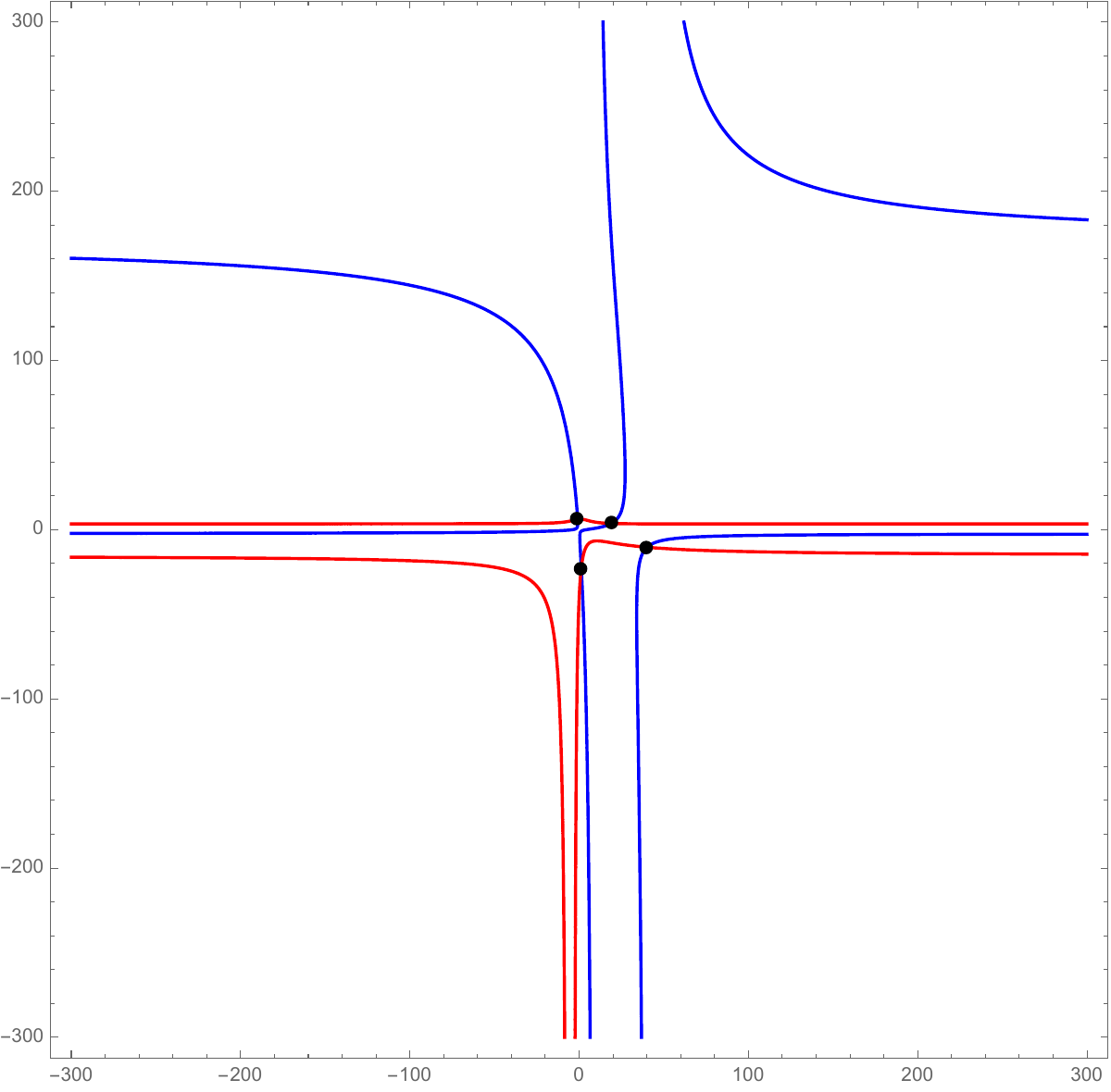}
\endminipage%
\caption{Six projected bisectors when the four lines have the following parameters: $(a,b_{3},c_{3},d_{3},e_{3},b_{4},c_{4},d_{4},e_{4})=(16,-10,3,0,-1,5,1,2,5)$. They match the configurations shown in \cref{fig:topo6comb} and form topology \Rom{6}. }\label{fig:topo6verify}
\end{figure}

\begin{figure}[h]
    \centering
    \begin{minipage}[b]{0.48\textwidth}
        \centering
        \includegraphics[width=0.9\linewidth]{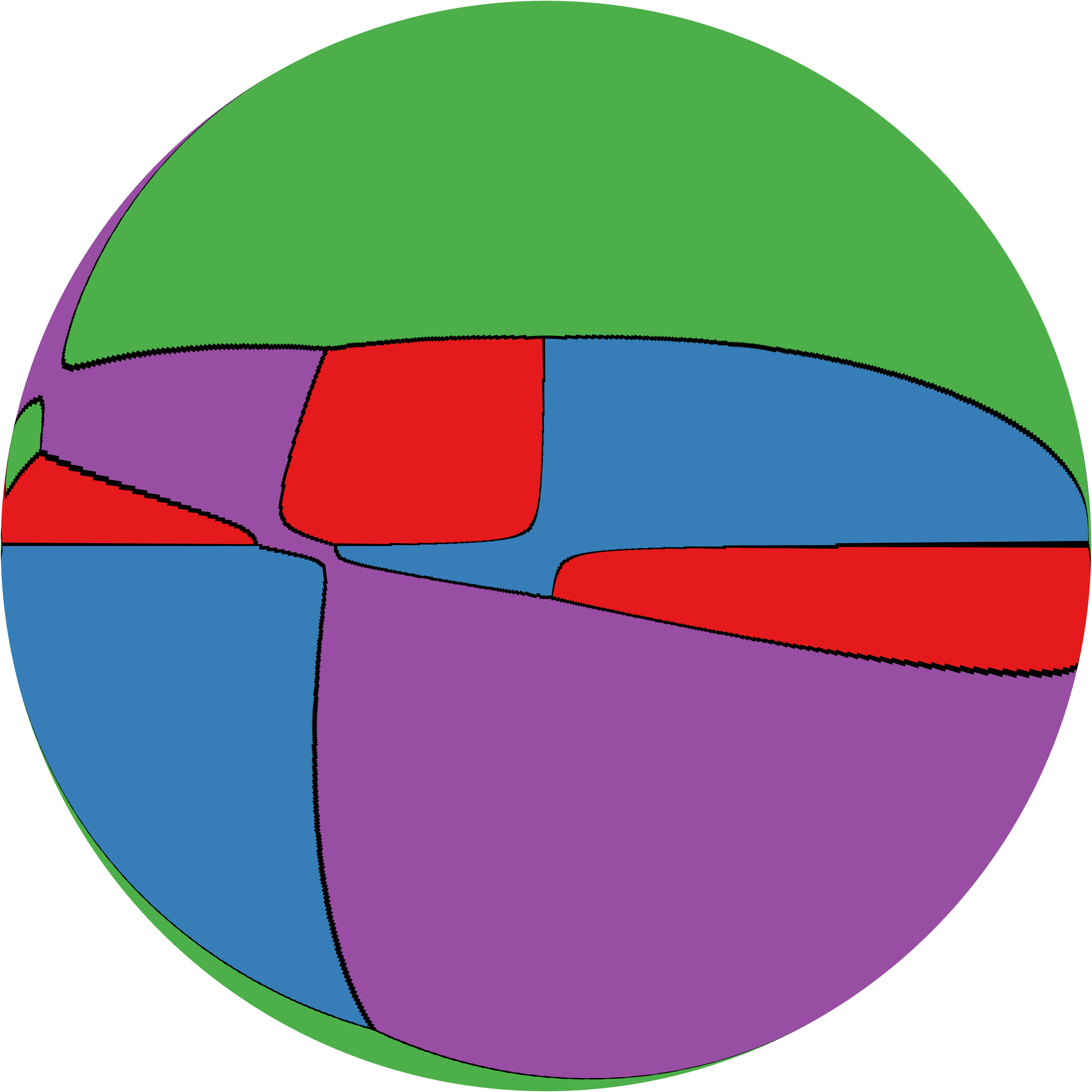}
    \end{minipage}
    \hfill
    \begin{minipage}[b]{0.48\textwidth}
        \centering
        \includegraphics[width=0.9\linewidth]{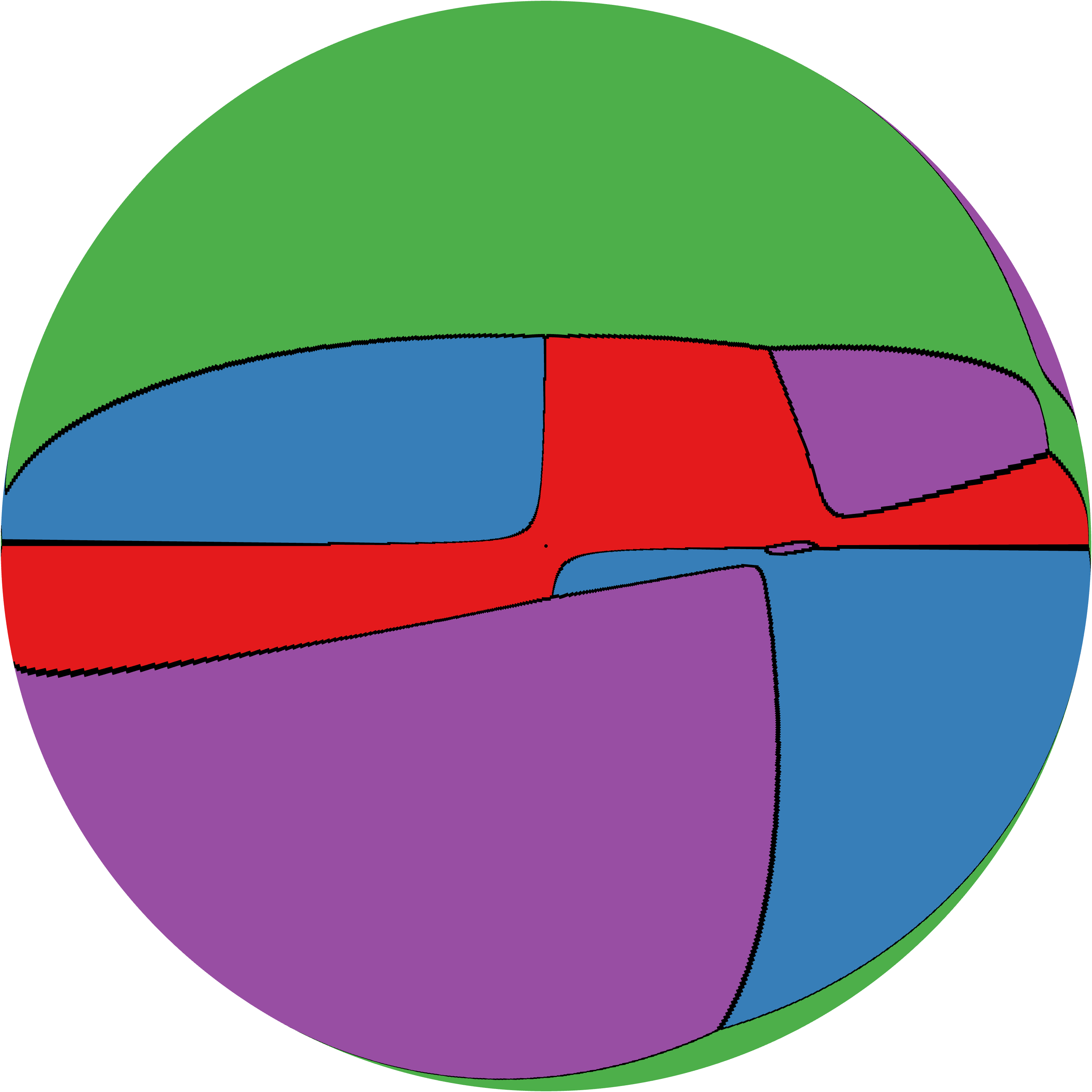}
     \end{minipage}
\caption{Top and bottom view of $\gmap(\fvd(L))$ of topology \Rom{6}. }\label{fig:gmapfvd6}
\end{figure}

\subsection{Topology \Rom{7}: 4 vertices}

The configuration tuple that induces topology \Rom{7} is shown in \cref{fig:topo7comb}. A set of lines that realizes this tuple, hence this topology, has the following parameters: $(a,b_{3},c_{3},d_{3},e_{3},b_{4},c_{4},d_{4},e_{4})=(4,7,9,1,10,-3,-10,6,-8)$. \cref{fig:topo7verify} shows the six projected bisectors of these four lines, which match the configurations shown in \cref{fig:topo7comb}. The $\gmap(\fvd(L))$ is shown in \cref{fig:gmapfvd7}.

\begin{figure}[h]
        \centering
        \includegraphics[page=7,width=\textwidth]{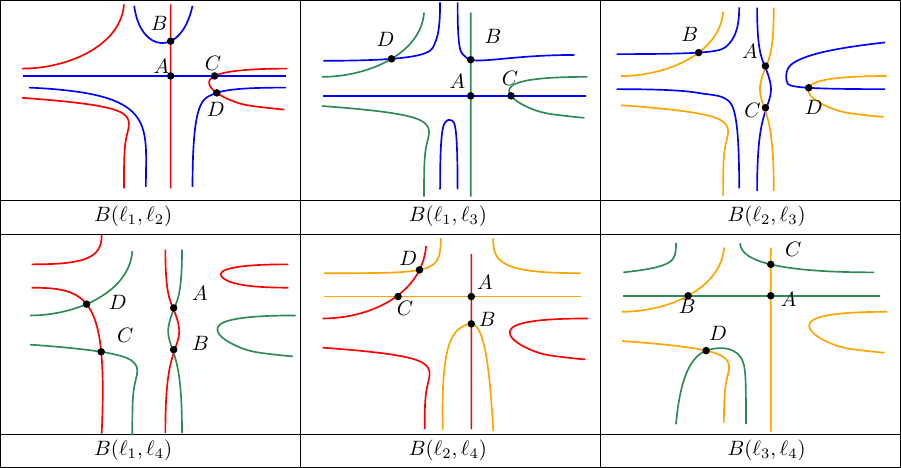}
        \caption{Combination that forms topology \Rom{7}. }\label{fig:topo7comb}
\end{figure}

\begin{figure}[!h]
\centering
\minipage{0.03\textwidth}
(a)	
\endminipage
\minipage{0.3\textwidth}
\includegraphics[width=\textwidth]{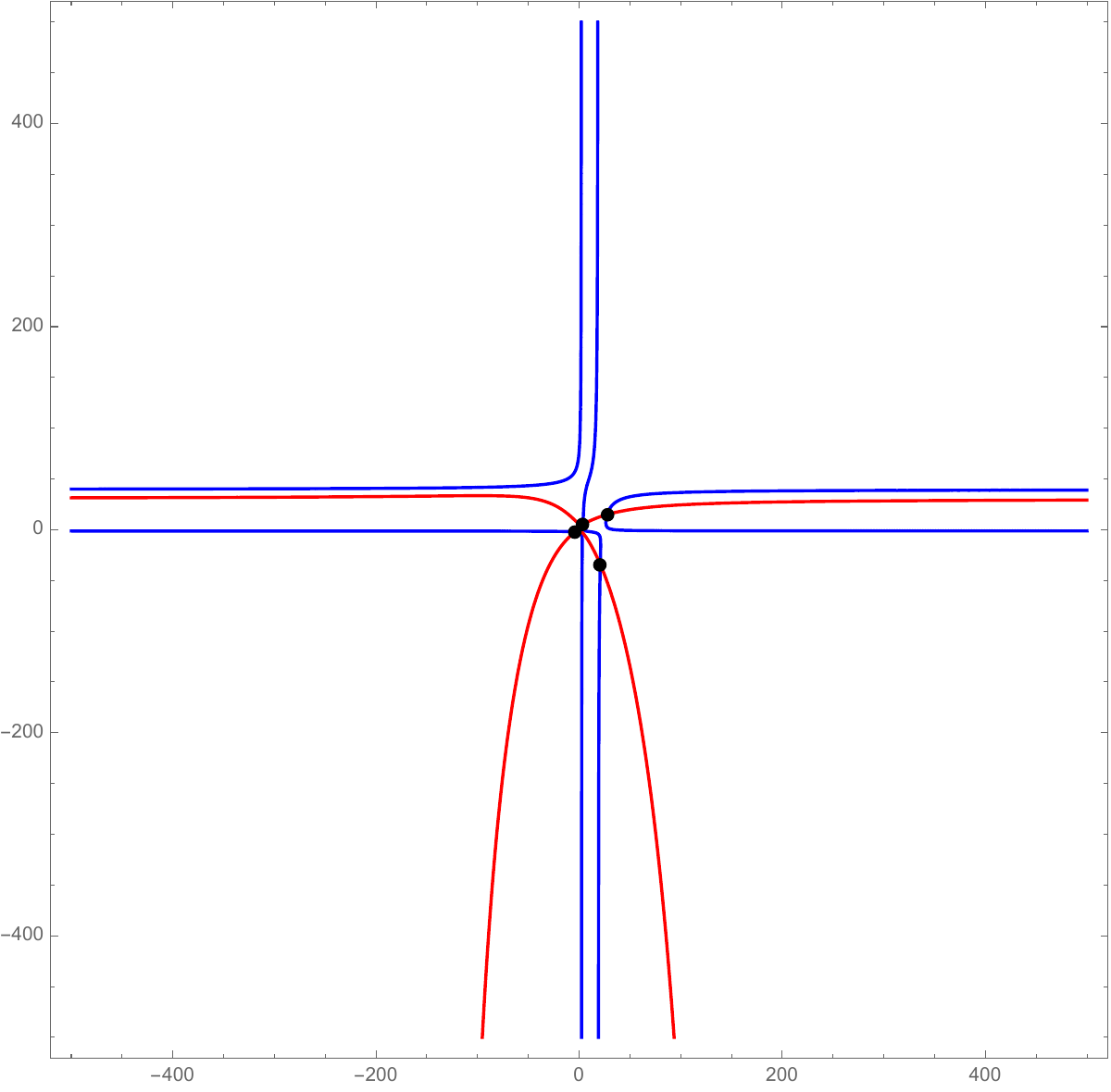}
\endminipage\hfill
\minipage{0.03\textwidth}
(b)	
\endminipage
\minipage{0.3\textwidth}
\includegraphics[width=\textwidth]{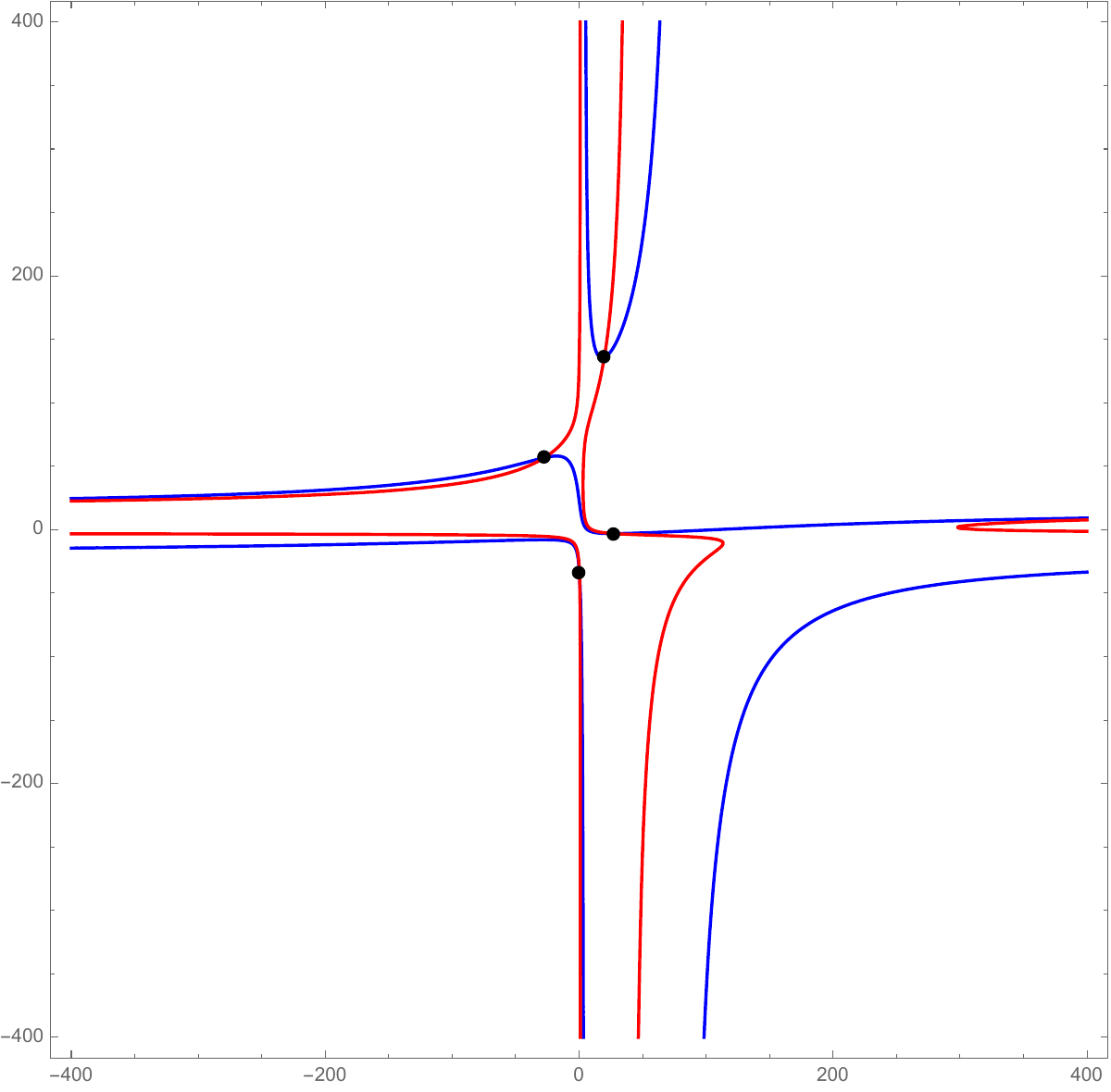}	
\endminipage\hfill
\minipage{0.03\textwidth}
(c)	
\endminipage
\minipage{0.3\textwidth}
\includegraphics[width=\textwidth]{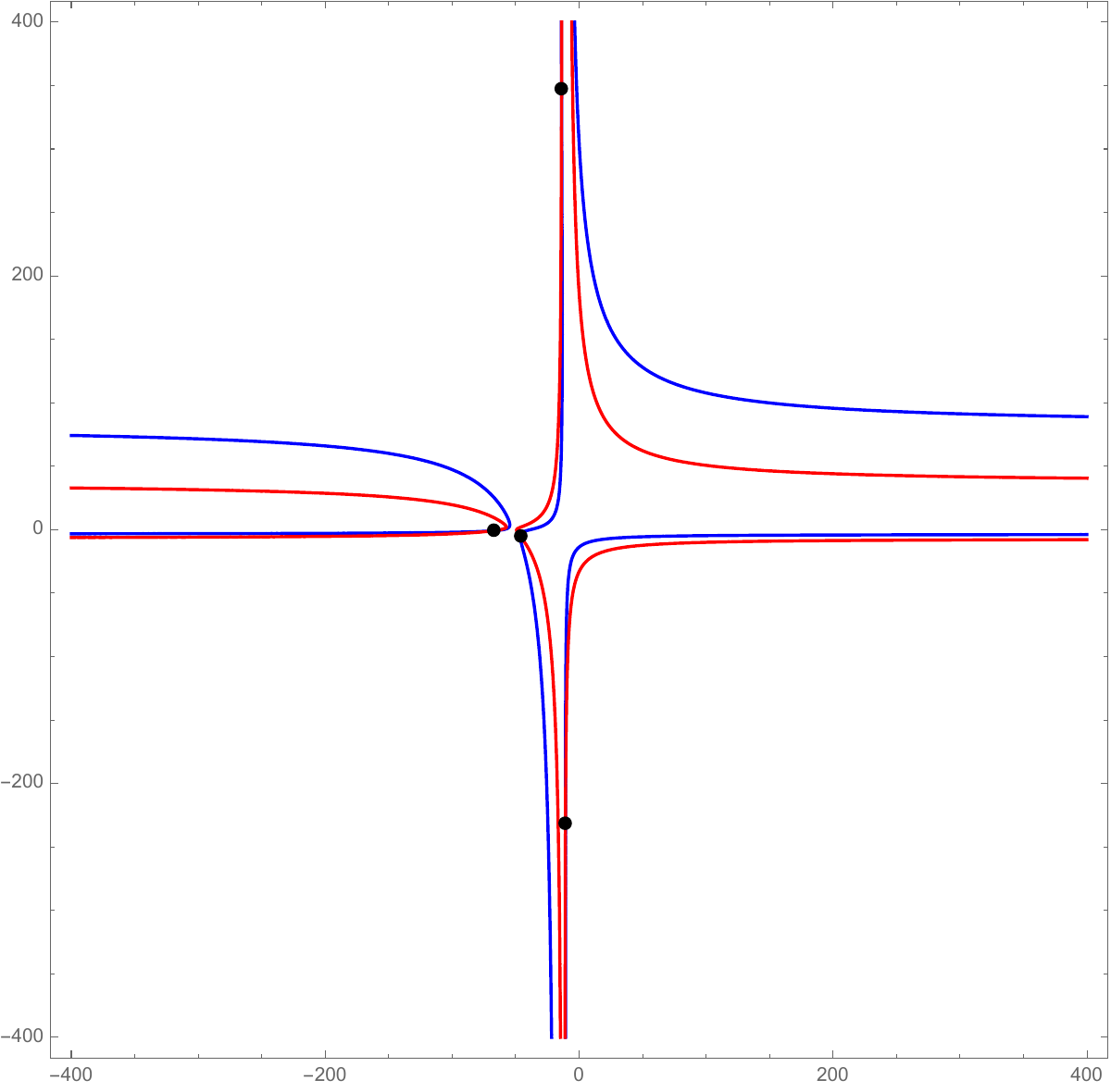}
\endminipage%
\newline
\centering
\minipage{0.03\textwidth}
(d)	
\endminipage
\minipage{0.3\textwidth}
\includegraphics[width=\textwidth]{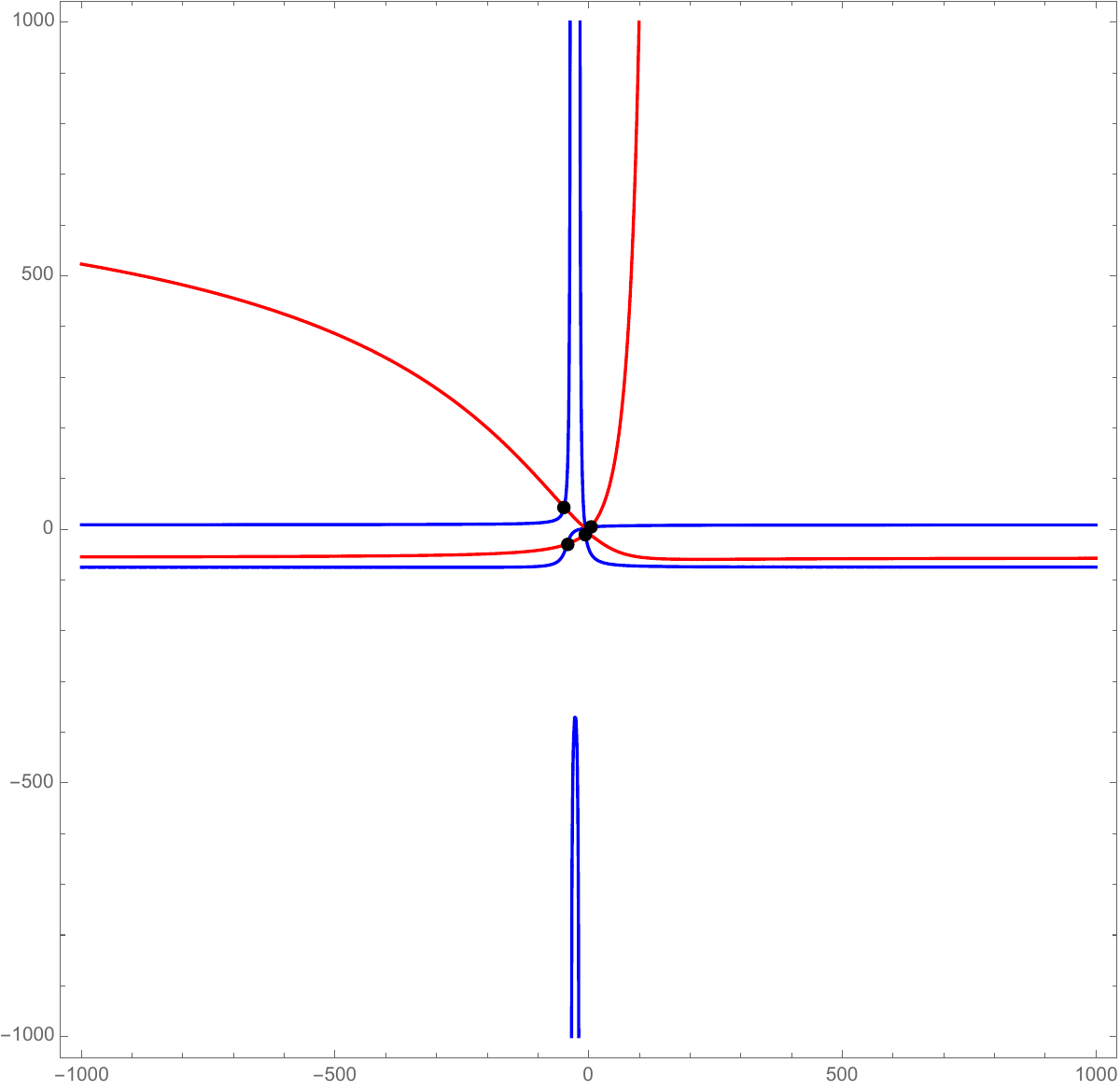}
\endminipage\hfill
\minipage{0.03\textwidth}
(e)	
\endminipage
\minipage{0.3\textwidth}
\includegraphics[width=\textwidth]{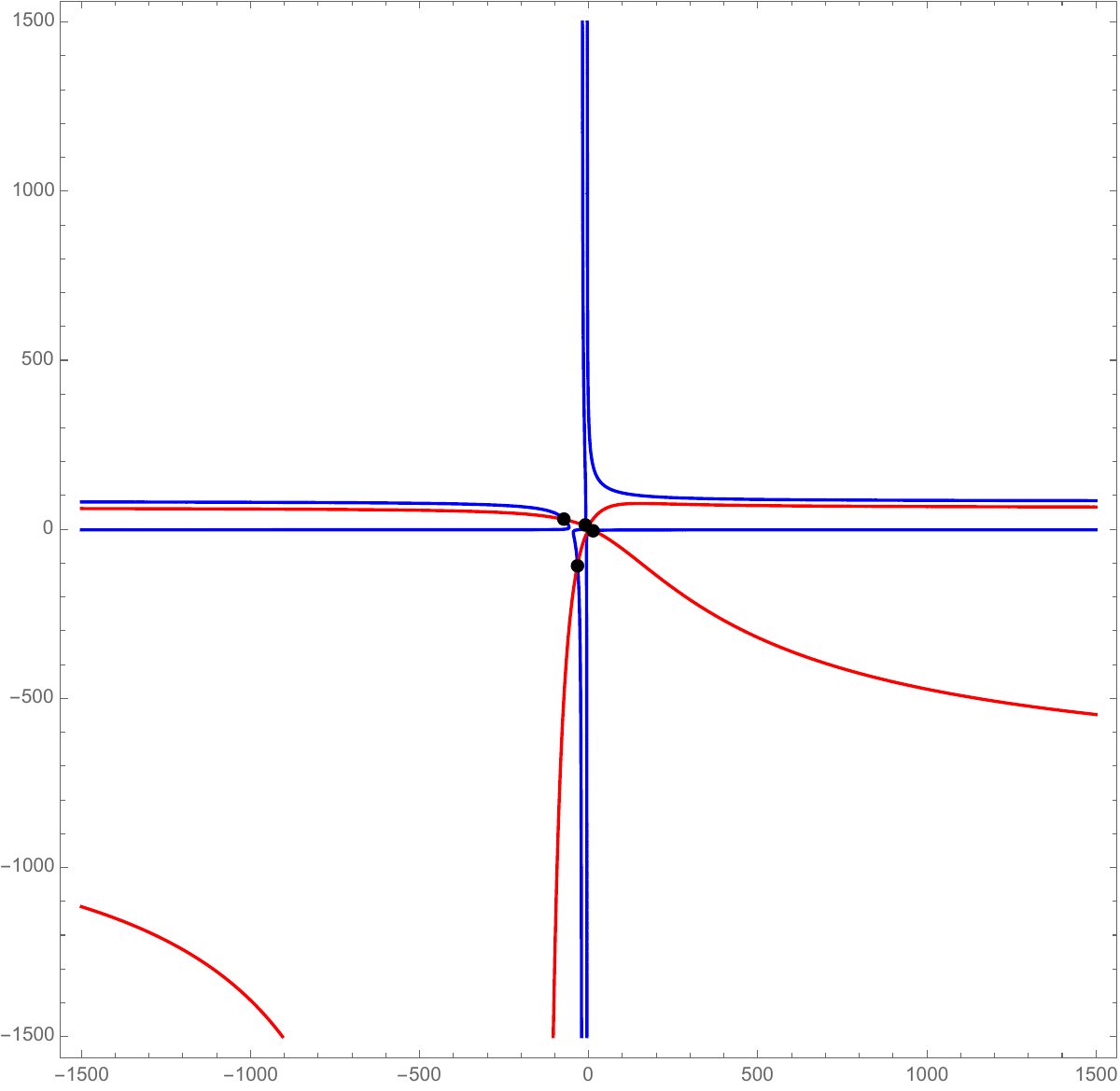}
\endminipage\hfill
\minipage{0.03\textwidth}
(f)	
\endminipage
\minipage{0.3\textwidth}
\includegraphics[width=\textwidth]{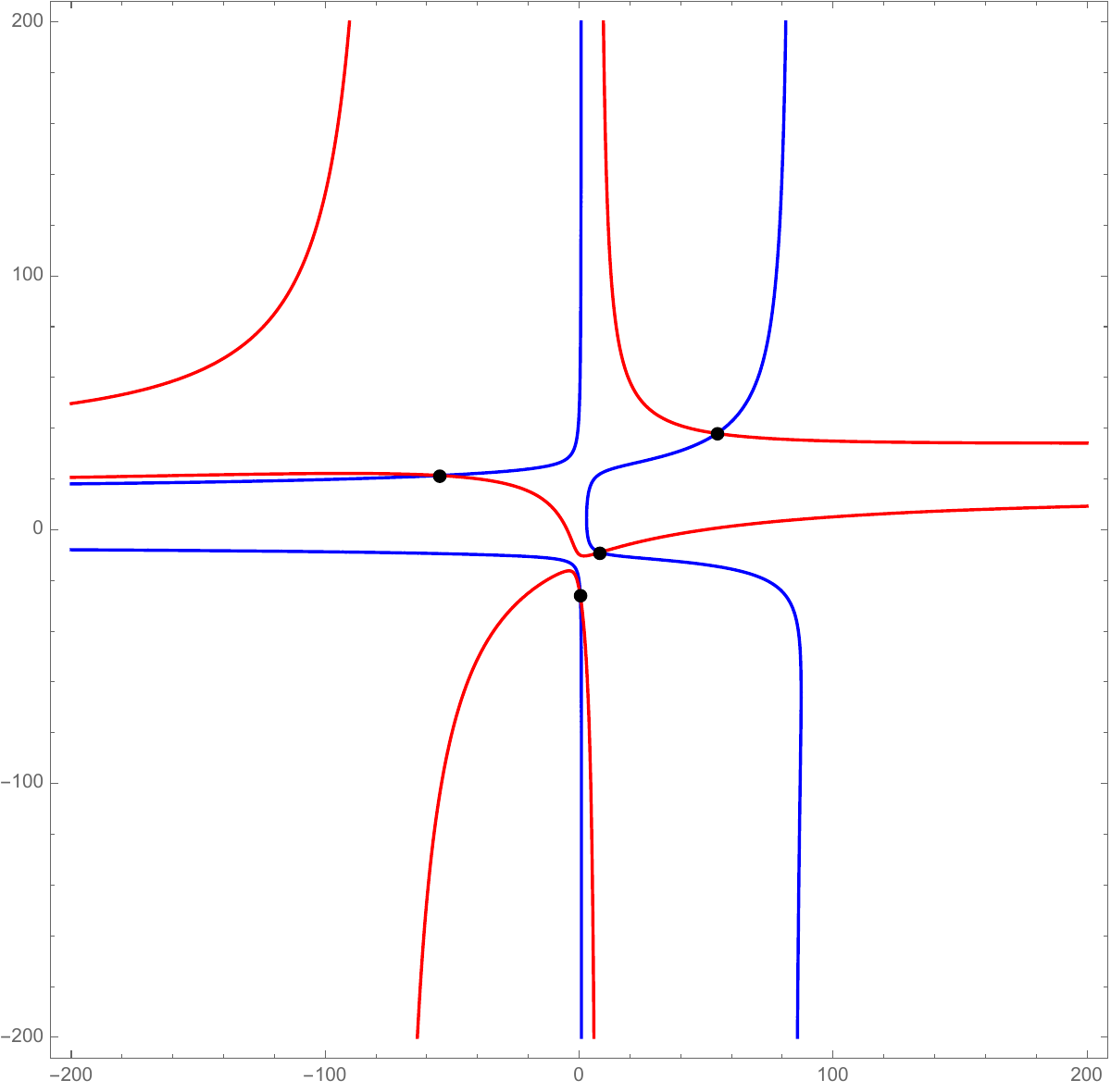}
\endminipage%
\caption{Six projected bisectors when the four lines have the following parameters: $(a,b_{3},c_{3},d_{3},e_{3},b_{4},c_{4},d_{4},e_{4})=(4,7,9,1,10,-3,-10,6,-8)$. They match the configurations shown in \cref{fig:topo7comb} and form topology \Rom{7}. }\label{fig:topo7verify}
\end{figure}

\begin{figure}[h]
    \centering
    \begin{minipage}[b]{0.48\textwidth}
        \centering
        \includegraphics[width=0.9\linewidth]{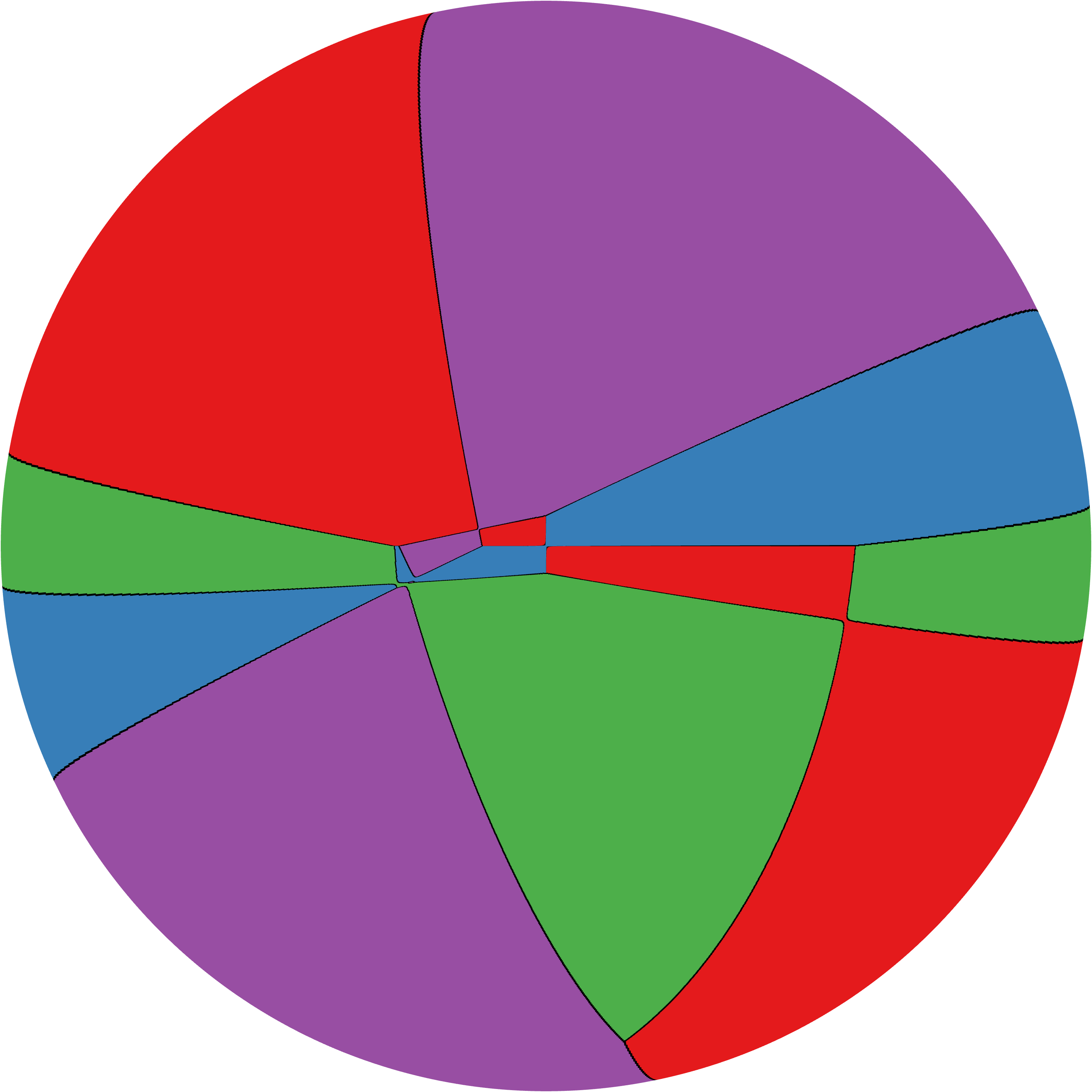}
    \end{minipage}
    \hfill
    \begin{minipage}[b]{0.48\textwidth}
        \centering
        \includegraphics[width=0.9\linewidth]{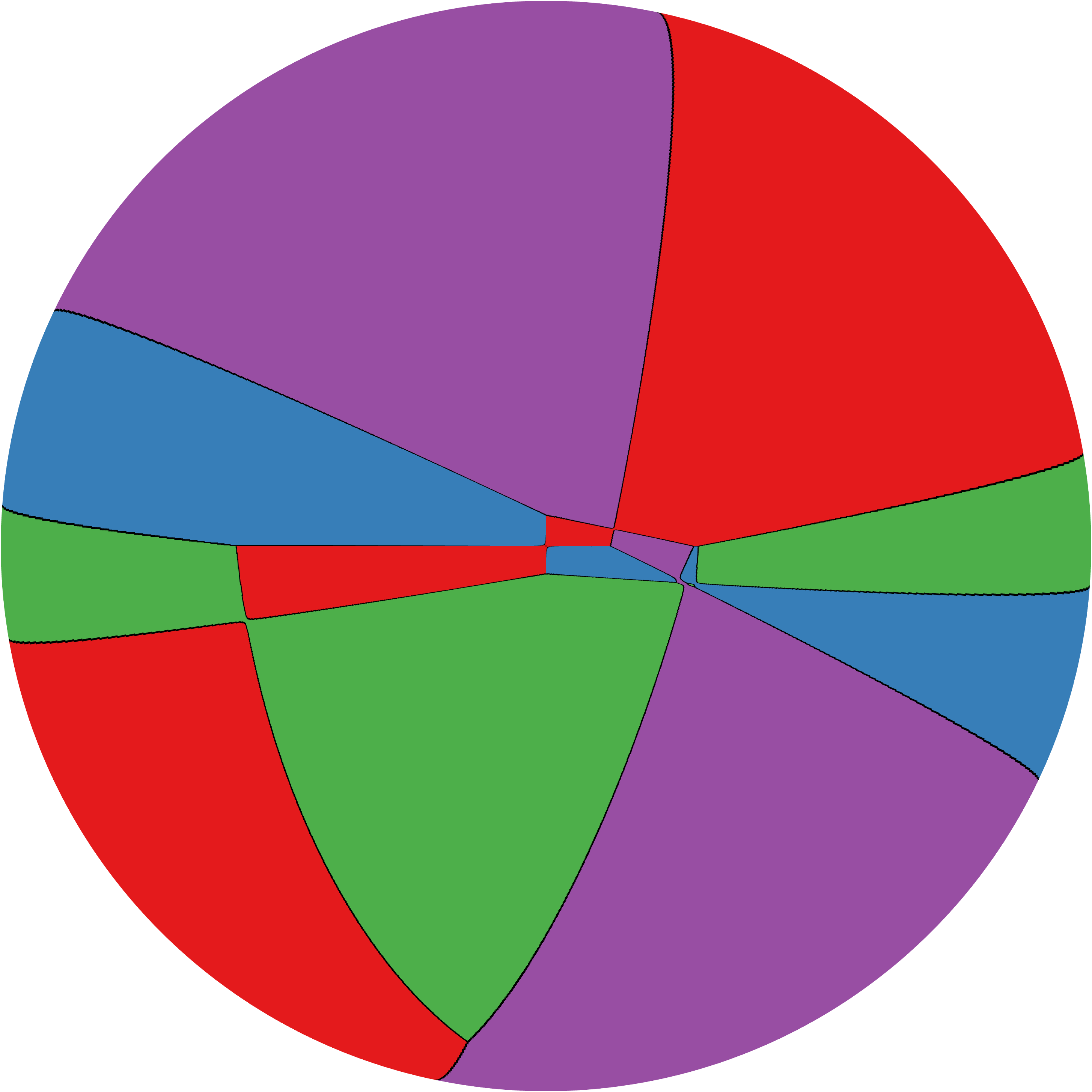}
     \end{minipage}
\caption{Top and bottom view of $\gmap(\fvd(L))$ of topology \Rom{7}. }\label{fig:gmapfvd7}
\end{figure}

\subsection{Topology \Rom{8}: 4 vertices}

The configuration tuple that induces topology \Rom{8} is shown in \cref{fig:topo8comb}. A set of lines that realizes this tuple, hence this topology, has the following parameters: $(a,b_{3},c_{3},d_{3},e_{3},b_{4},c_{4},d_{4},e_{4})=(2,-19,-4,-13,20,14,-18,-3,9)$. \cref{fig:topo8verify} shows the six projected bisectors of these four lines, which match the configurations shown in \cref{fig:topo8comb}. The $\gmap(\fvd(L))$ is shown in \cref{fig:gmapfvd8}.

\begin{figure}[h]
        \centering
        \includegraphics[page=8,width=\textwidth]{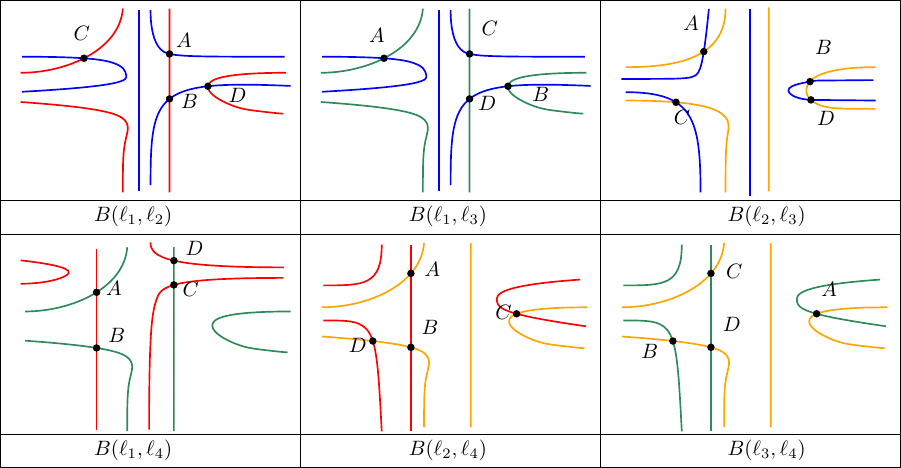}
        \caption{Combination that forms topology \Rom{8}. }\label{fig:topo8comb}
\end{figure}

\begin{figure}[!h]
\centering
\minipage{0.03\textwidth}
(a)	
\endminipage
\minipage{0.3\textwidth}
\includegraphics[width=\textwidth]{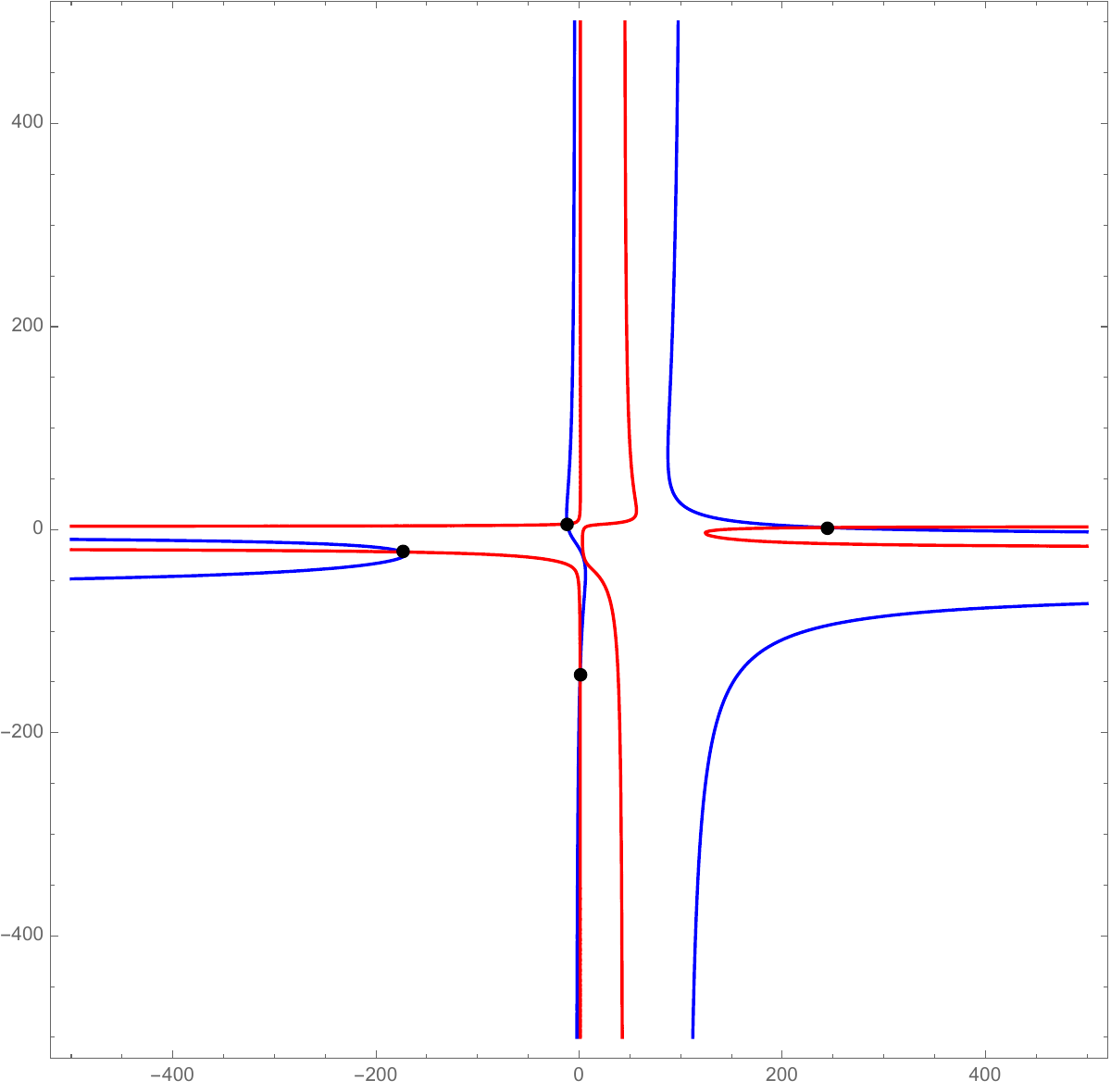}
\endminipage\hfill
\minipage{0.03\textwidth}
(b)	
\endminipage
\minipage{0.3\textwidth}
\includegraphics[width=\textwidth]{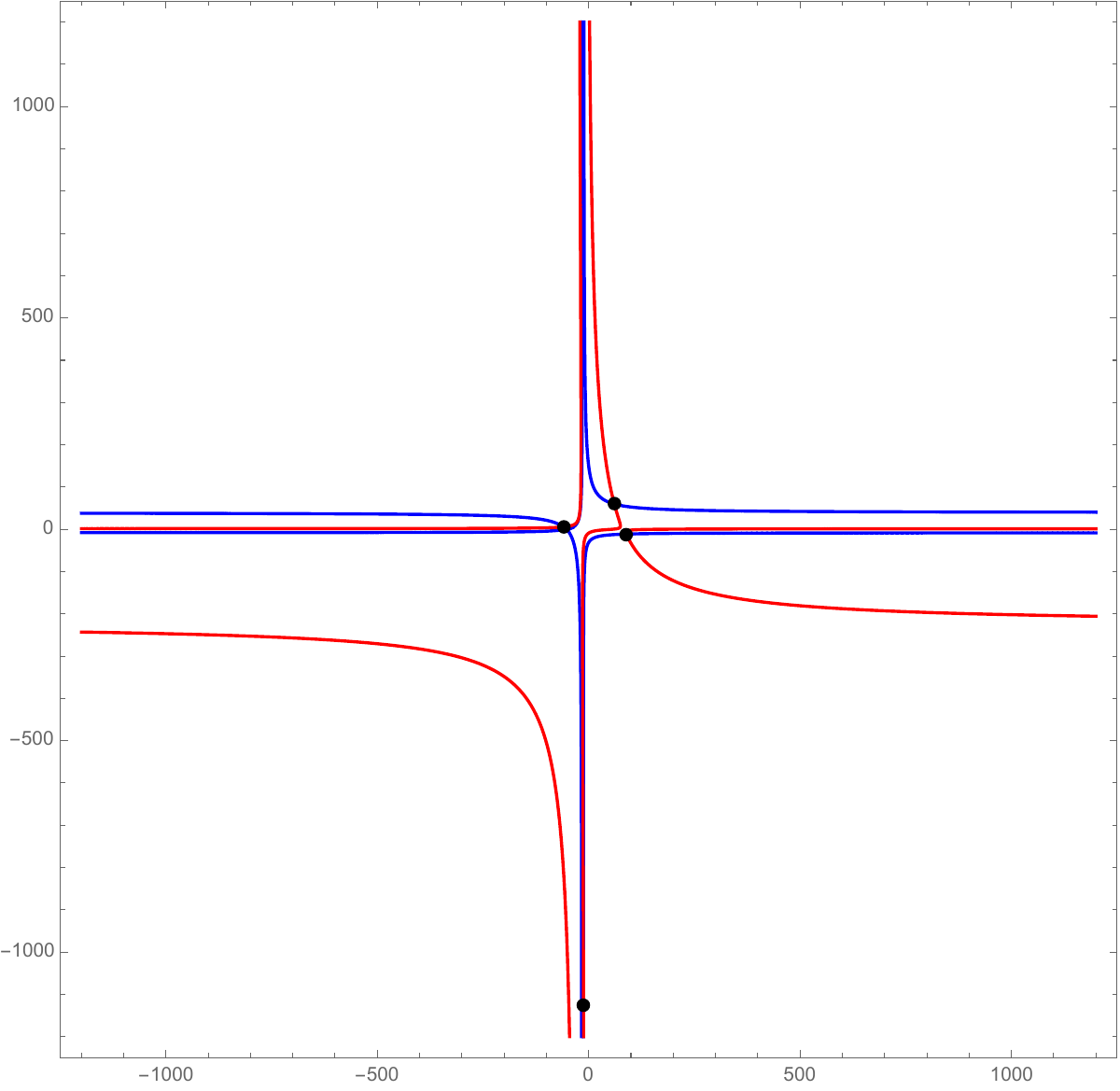}	
\endminipage\hfill
\minipage{0.03\textwidth}
(c)	
\endminipage
\minipage{0.3\textwidth}
\includegraphics[width=\textwidth]{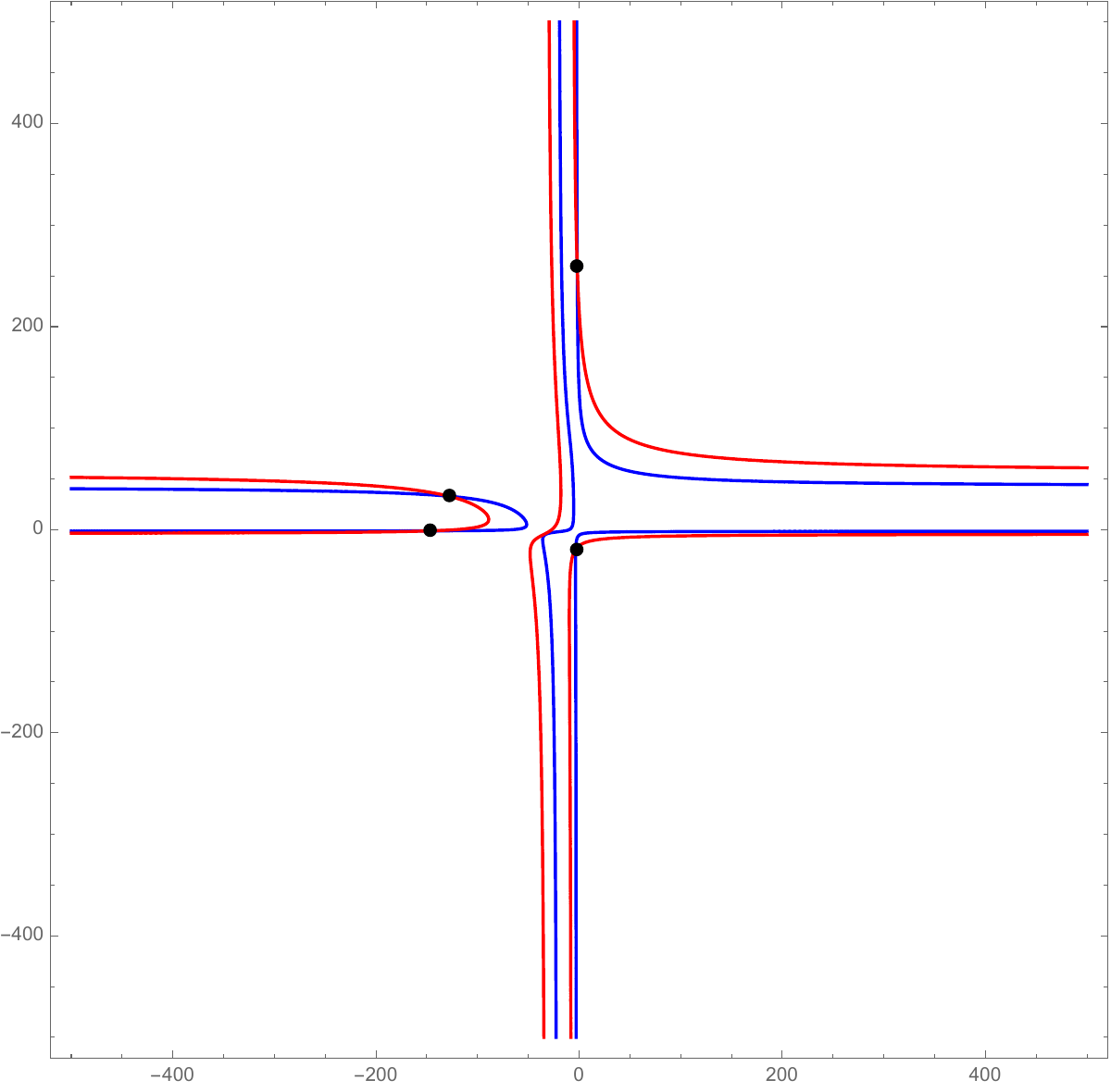}
\endminipage%
\newline
\centering
\minipage{0.03\textwidth}
(d)	
\endminipage
\minipage{0.3\textwidth}
\includegraphics[width=\textwidth]{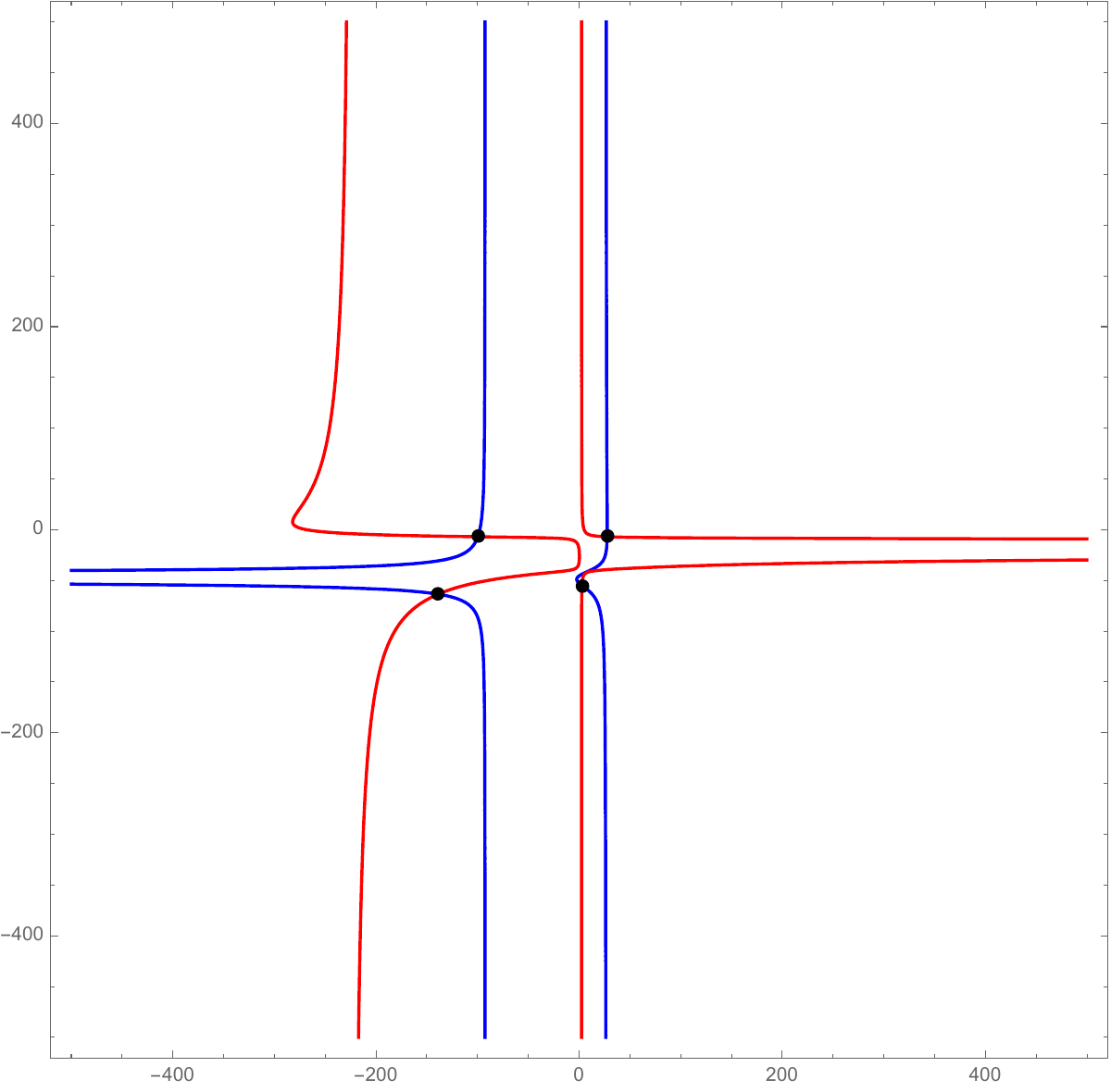}
\endminipage\hfill
\minipage{0.03\textwidth}
(e)	
\endminipage
\minipage{0.3\textwidth}
\includegraphics[width=\textwidth]{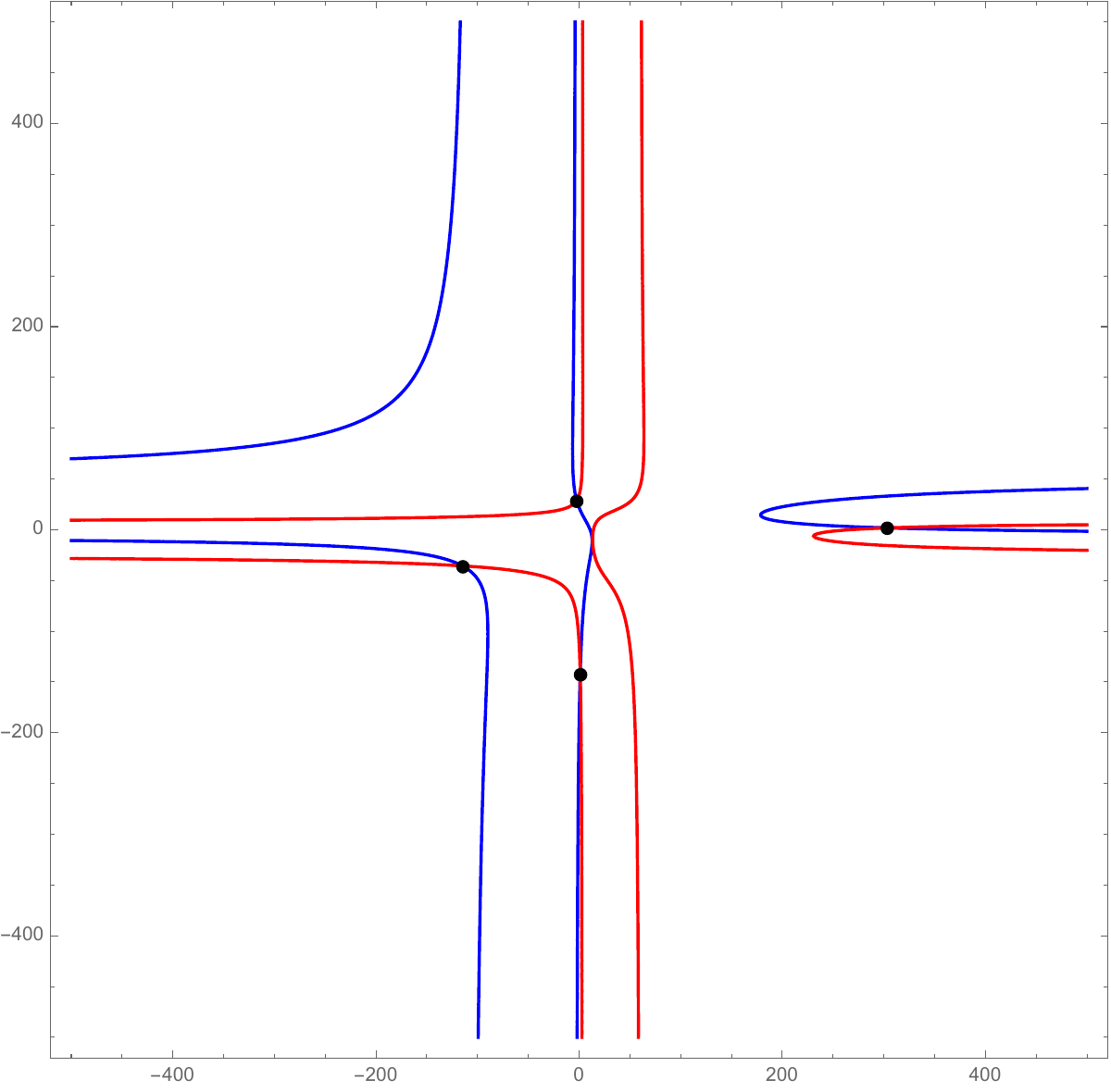}
\endminipage\hfill
\minipage{0.03\textwidth}
(f)	
\endminipage
\minipage{0.3\textwidth}
\includegraphics[width=\textwidth]{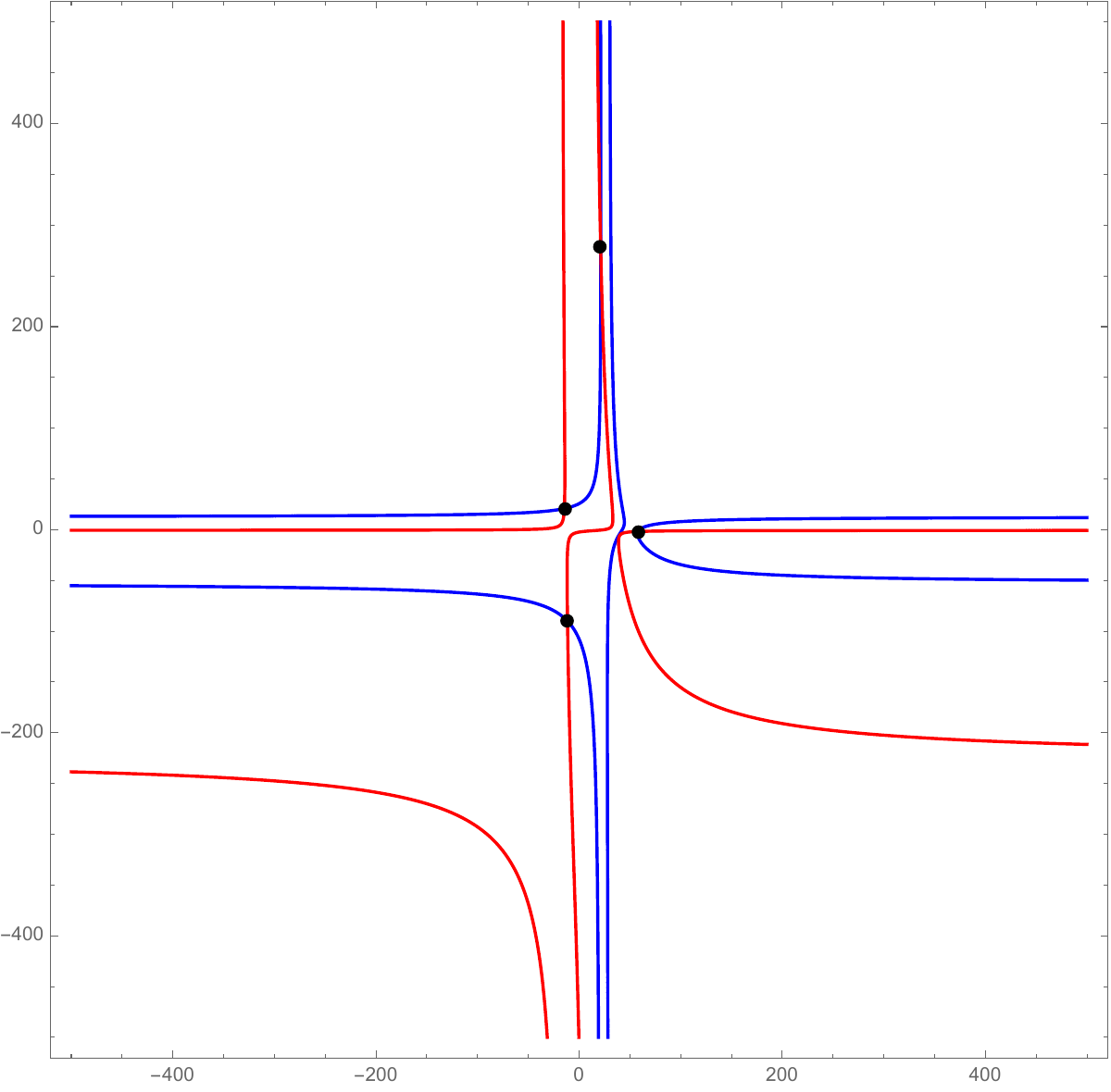}
\endminipage%
\caption{Six projected bisectors when the four lines have the following parameters: $(a,b_{3},c_{3},d_{3},e_{3},b_{4},c_{4},d_{4},e_{4})=(2,-19,-4,-13,20,14,-18,-3,9)$. They match the configurations shown in \cref{fig:topo8comb} and form topology \Rom{8}. }\label{fig:topo8verify}
\end{figure}

\begin{figure}[h]
    \centering
    \begin{minipage}[b]{0.48\textwidth}
        \centering
        \includegraphics[width=0.9\linewidth]{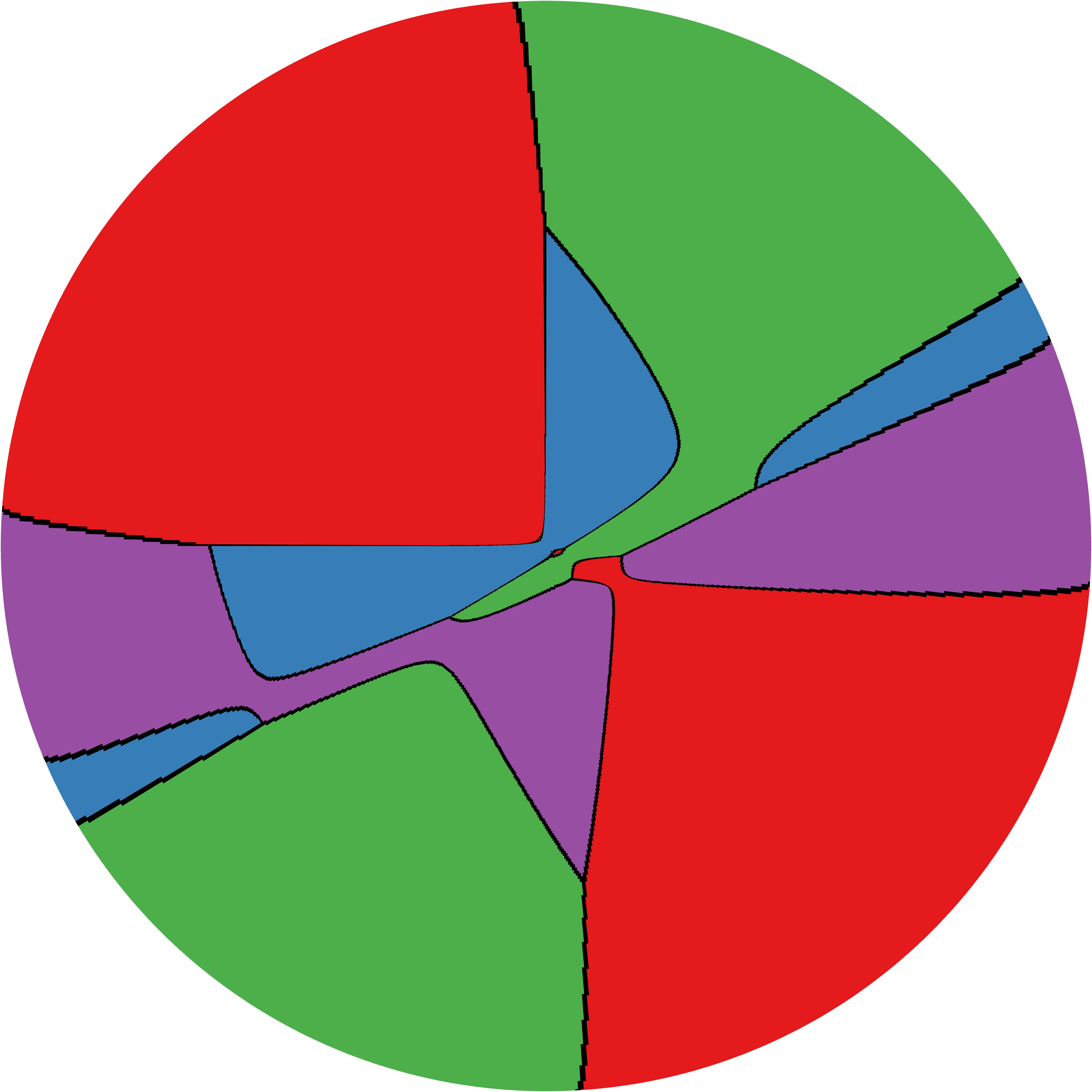}
    \end{minipage}
    \hfill
    \begin{minipage}[b]{0.48\textwidth}
        \centering
        \includegraphics[width=0.9\linewidth]{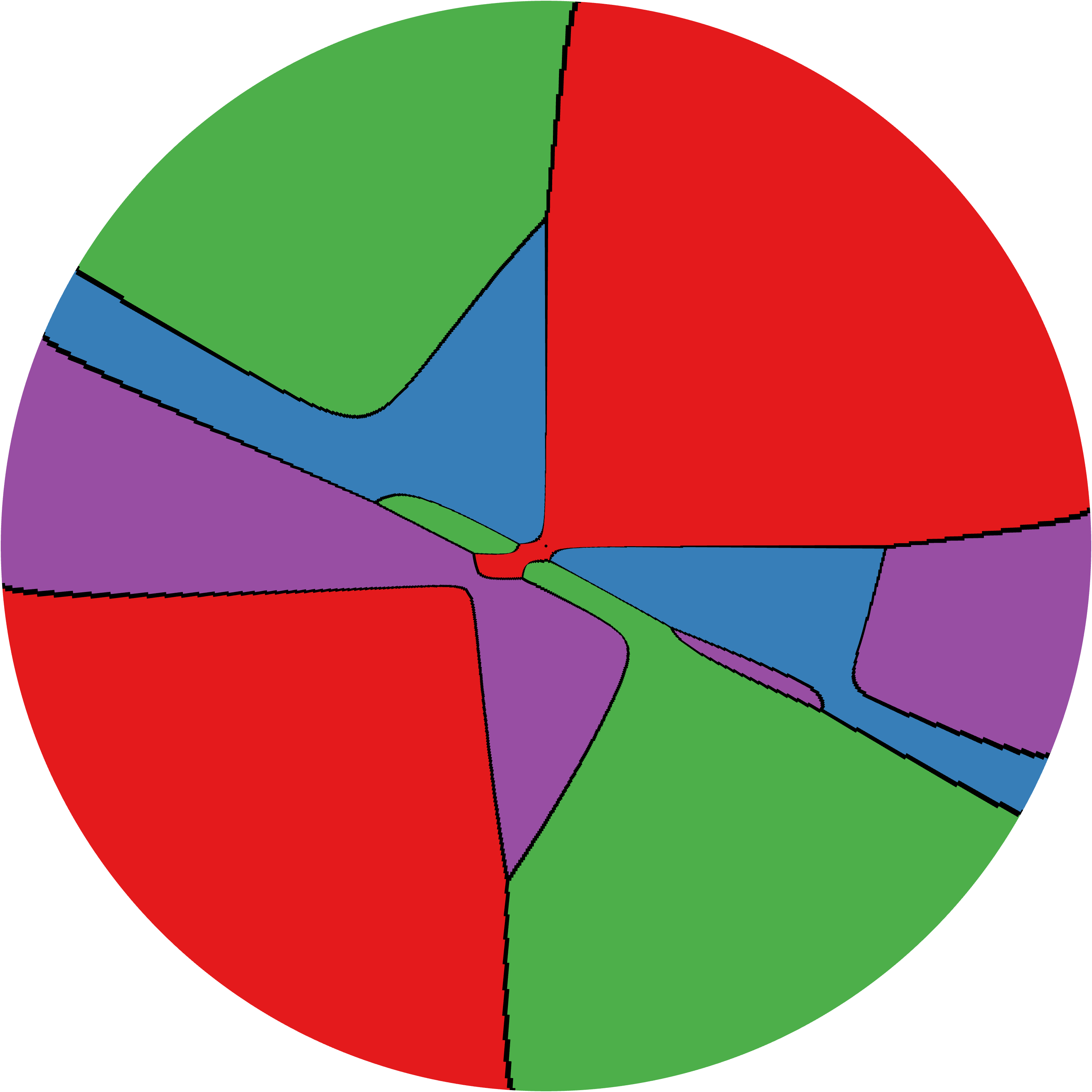}
     \end{minipage}
\caption{Top and bottom view of $\gmap(\fvd(L))$ of topology \Rom{8}. }\label{fig:gmapfvd8}
\end{figure}

\subsection{Topology \Rom{9}: 4 vertices}

The configuration tuple that induces topology \Rom{9} is shown in \cref{fig:topo9comb}. A set of lines that realizes this tuple, hence this topology, has the following parameters: $(a,b_{3},c_{3},d_{3},e_{3},b_{4},c_{4},d_{4},e_{4})=(10, 8, -15, -2, 15, 12, 9, -2, 4)$. \cref{fig:topo9verify} shows the six projected bisectors of these four lines, which match the configurations shown in \cref{fig:topo9comb}. The $\gmap(\fvd(L))$ is shown in \cref{fig:gmapfvd9}.

\begin{figure}[h]
        \centering
        \includegraphics[page=9,width=\textwidth]{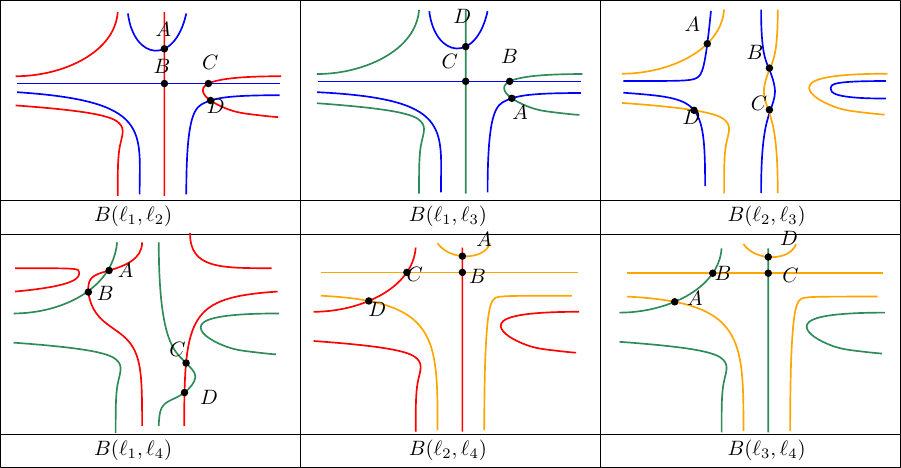}
        \caption{Combination that forms topology \Rom{9}. }\label{fig:topo9comb}
\end{figure}

\begin{figure}[!h]
\centering
\minipage{0.03\textwidth}
(a)	
\endminipage
\minipage{0.3\textwidth}
\includegraphics[width=\textwidth]{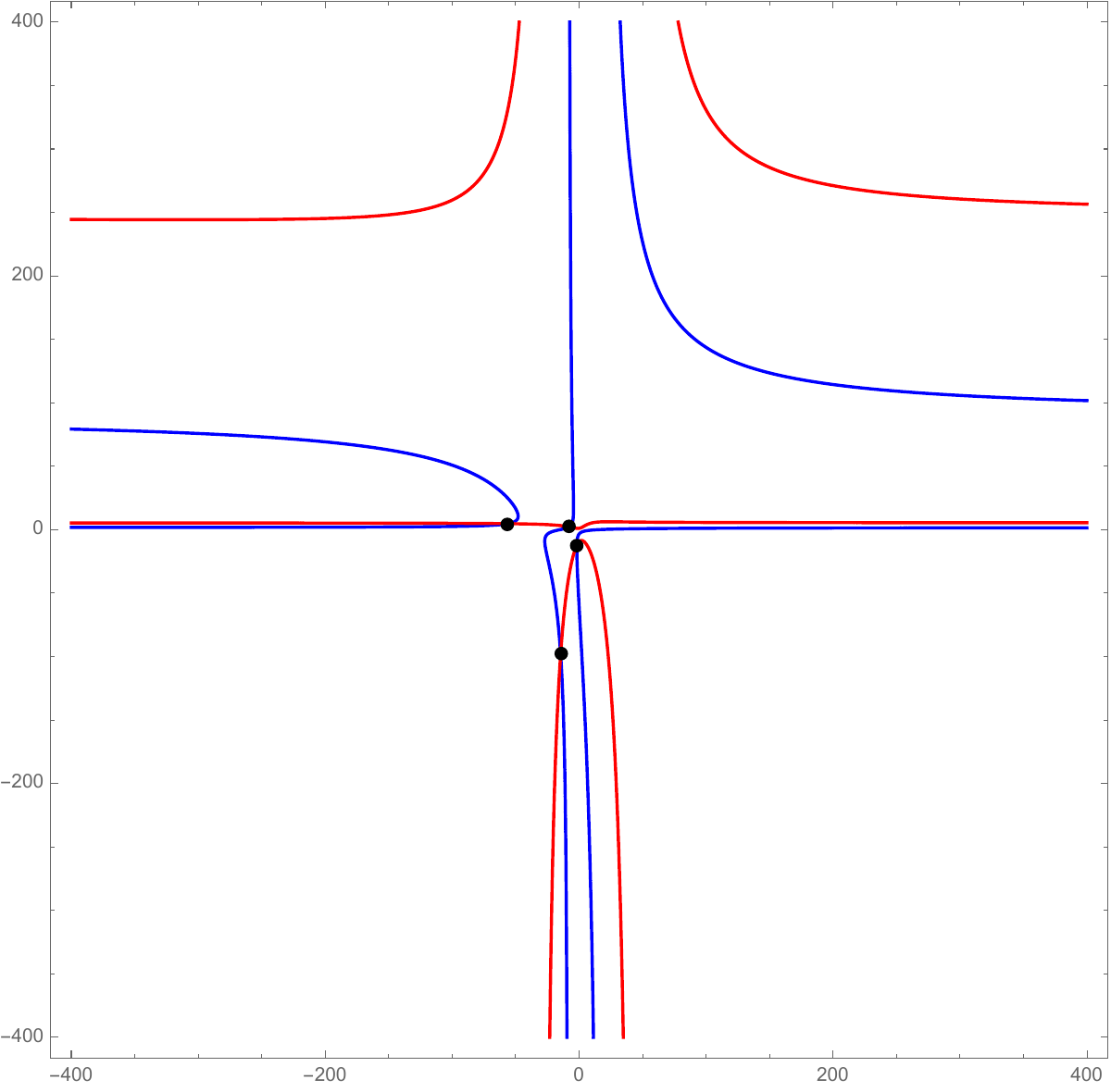}
\endminipage\hfill
\minipage{0.03\textwidth}
(b)	
\endminipage
\minipage{0.3\textwidth}
\includegraphics[width=\textwidth]{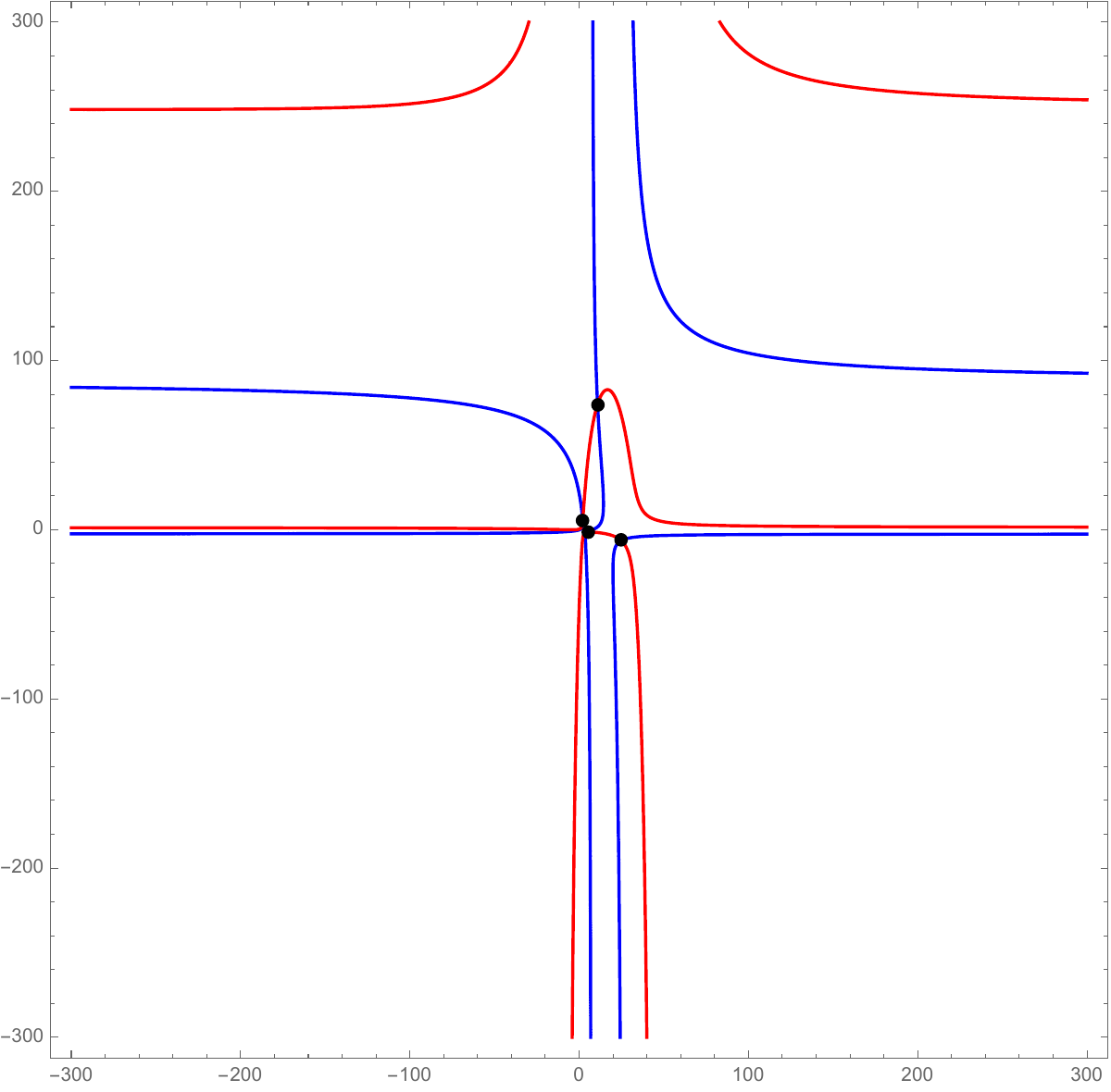}	
\endminipage\hfill
\minipage{0.03\textwidth}
(c)	
\endminipage
\minipage{0.3\textwidth}
\includegraphics[width=\textwidth]{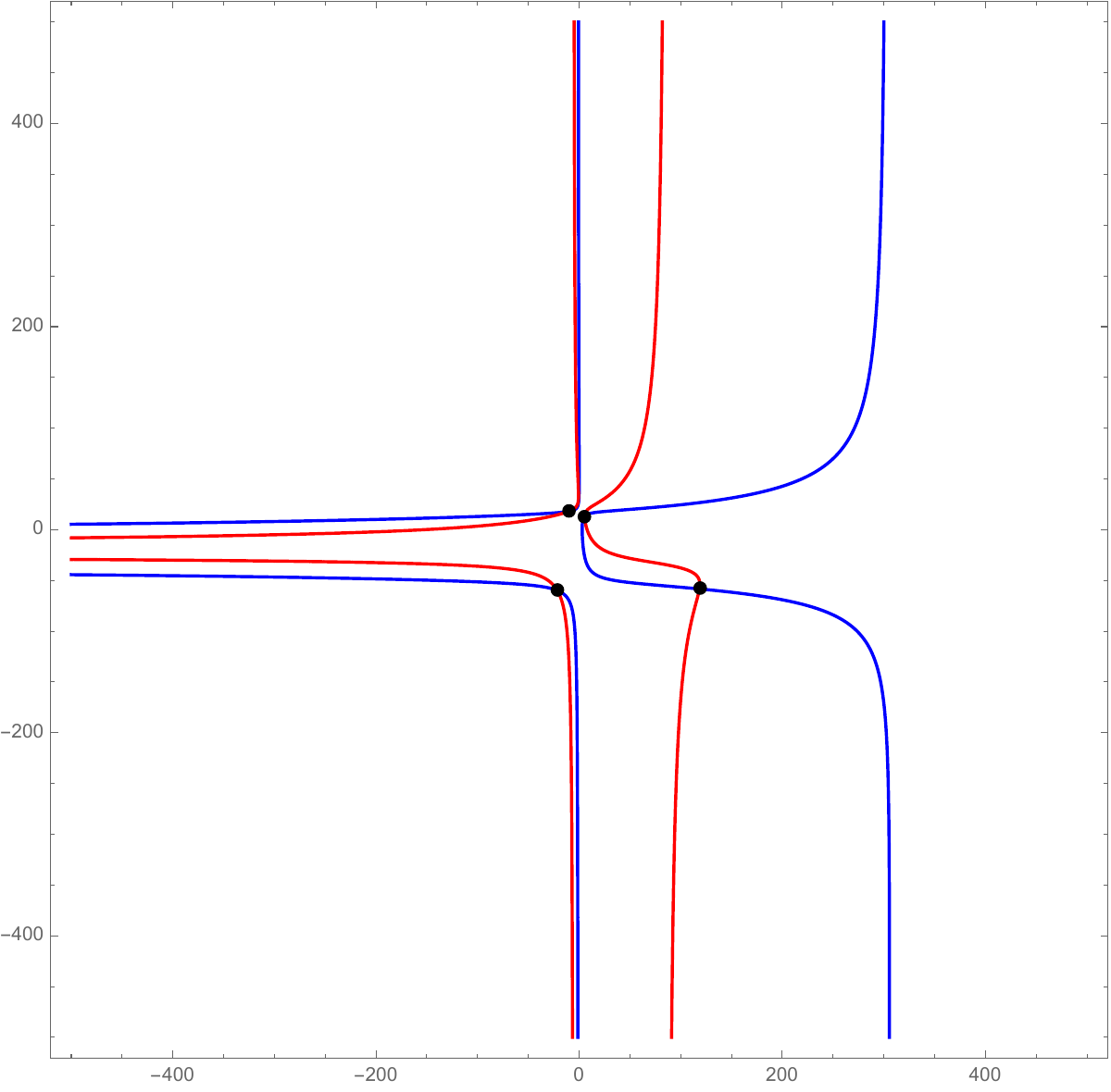}
\endminipage%
\newline
\centering
\minipage{0.03\textwidth}
(d)	
\endminipage
\minipage{0.3\textwidth}
\includegraphics[width=\textwidth]{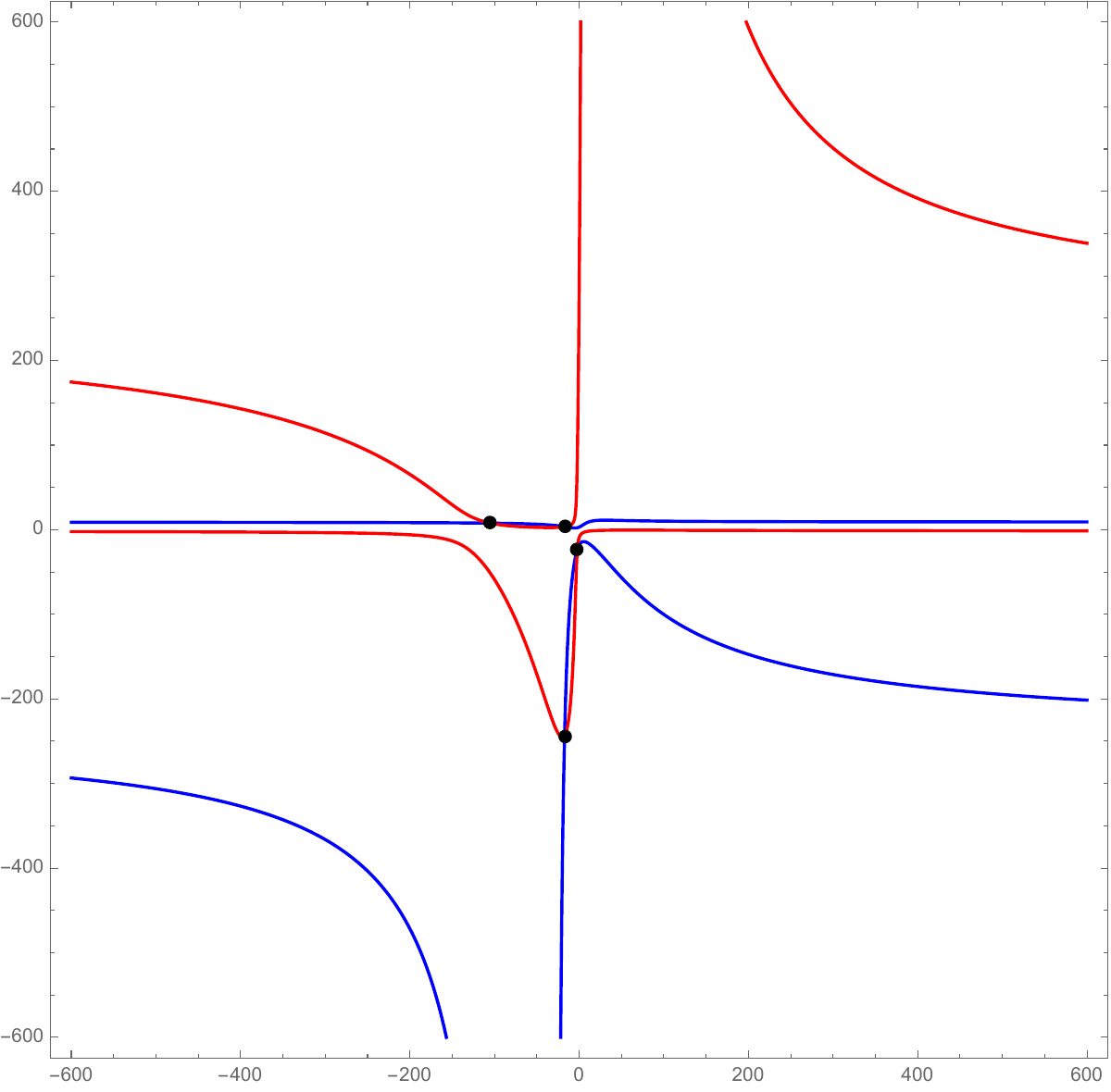}
\endminipage\hfill
\minipage{0.03\textwidth}
(e)	
\endminipage
\minipage{0.3\textwidth}
\includegraphics[width=\textwidth]{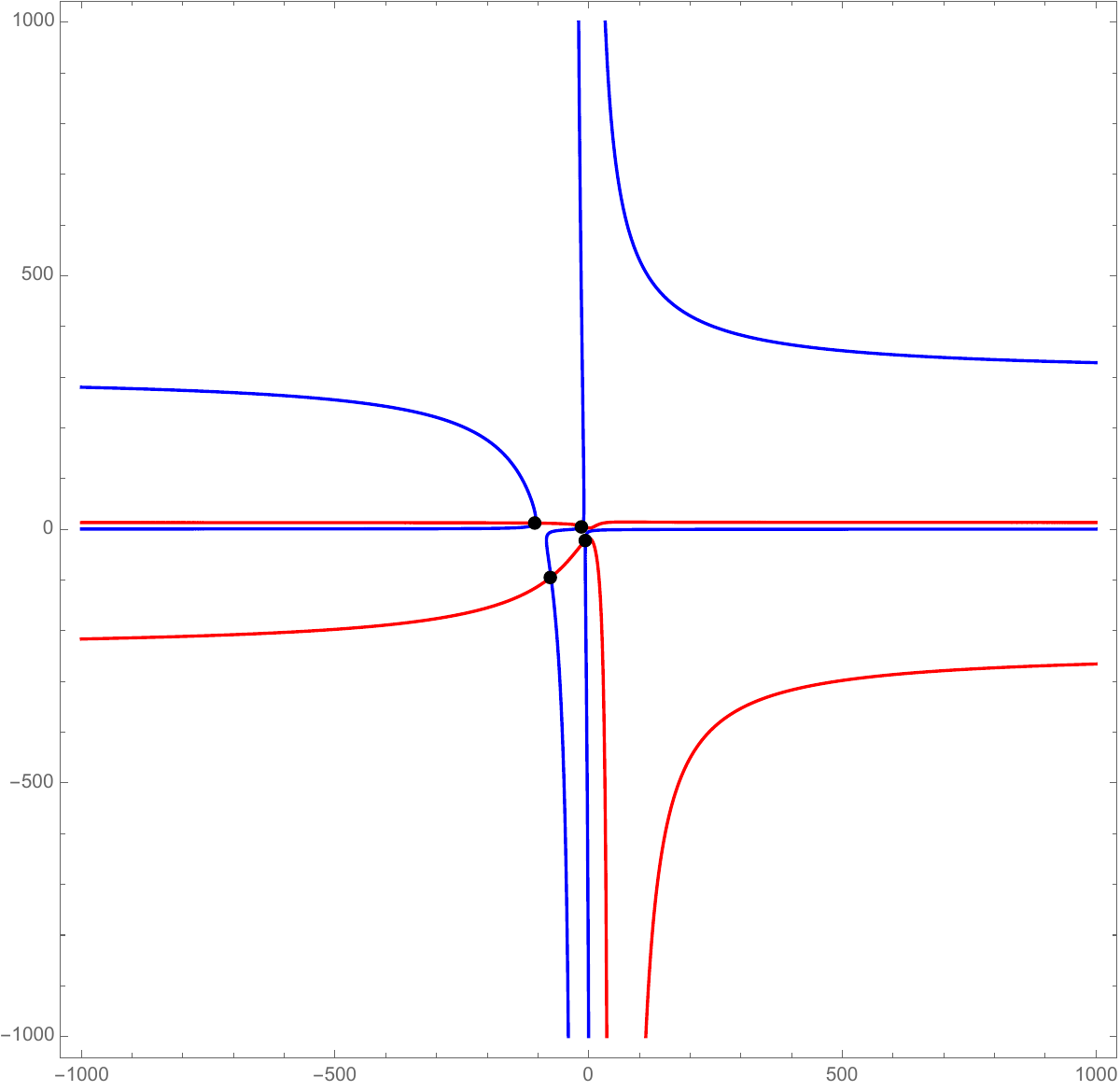}
\endminipage\hfill
\minipage{0.03\textwidth}
(f)	
\endminipage
\minipage{0.3\textwidth}
\includegraphics[width=\textwidth]{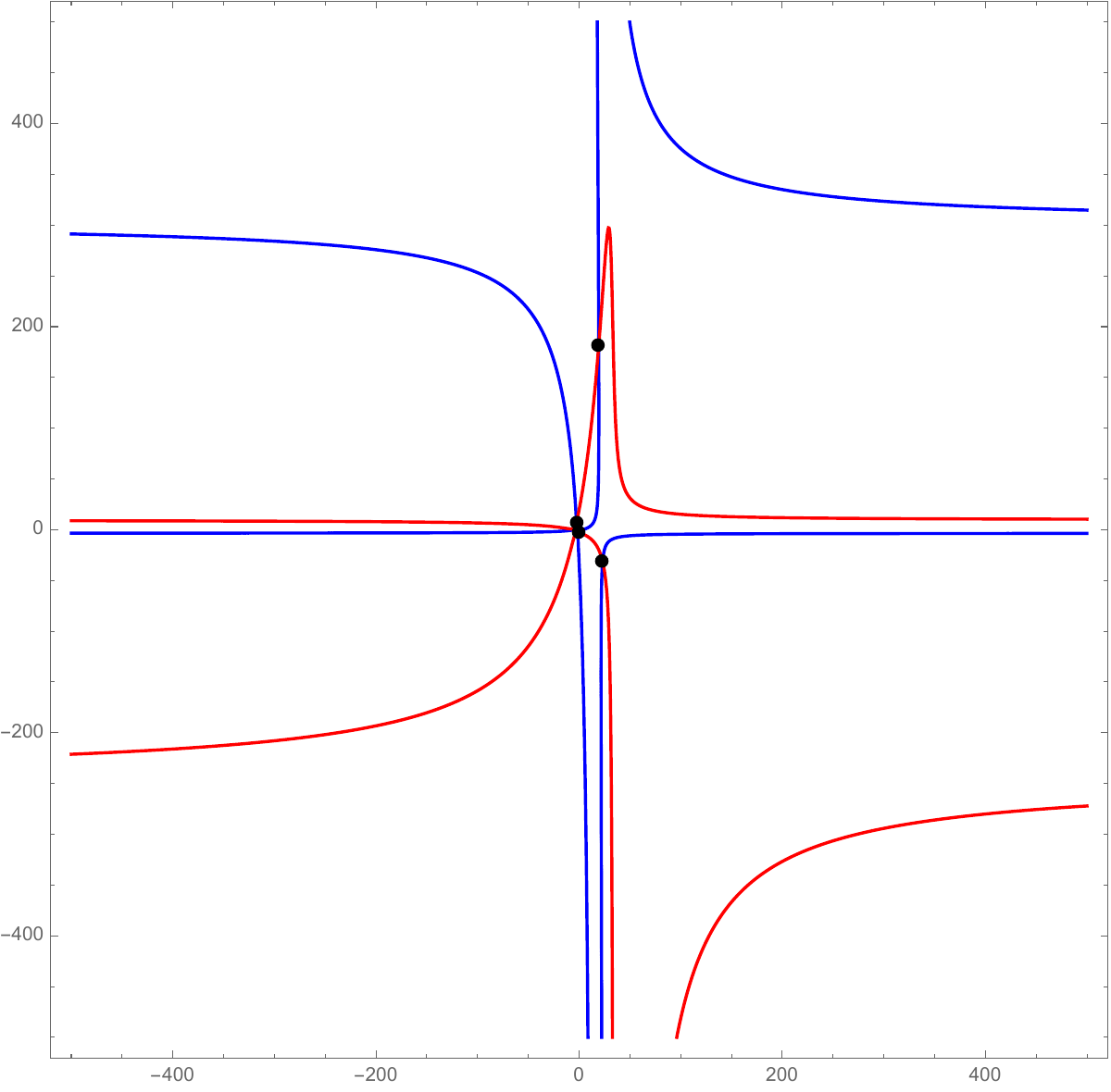}
\endminipage%
\caption{Six projected bisectors when the four lines have the following parameters: $(a,b_{3},c_{3},d_{3},e_{3},b_{4},c_{4},d_{4},e_{4})=(10, 8, -15, -2, 15, 12, 9, -2, 4)$. They match the configurations shown in \cref{fig:topo9comb} and form topology \Rom{9}. }\label{fig:topo9verify}
\end{figure}

\begin{figure}[h]
    \centering
    \begin{minipage}[b]{0.48\textwidth}
        \centering
        \includegraphics[width=0.9\linewidth]{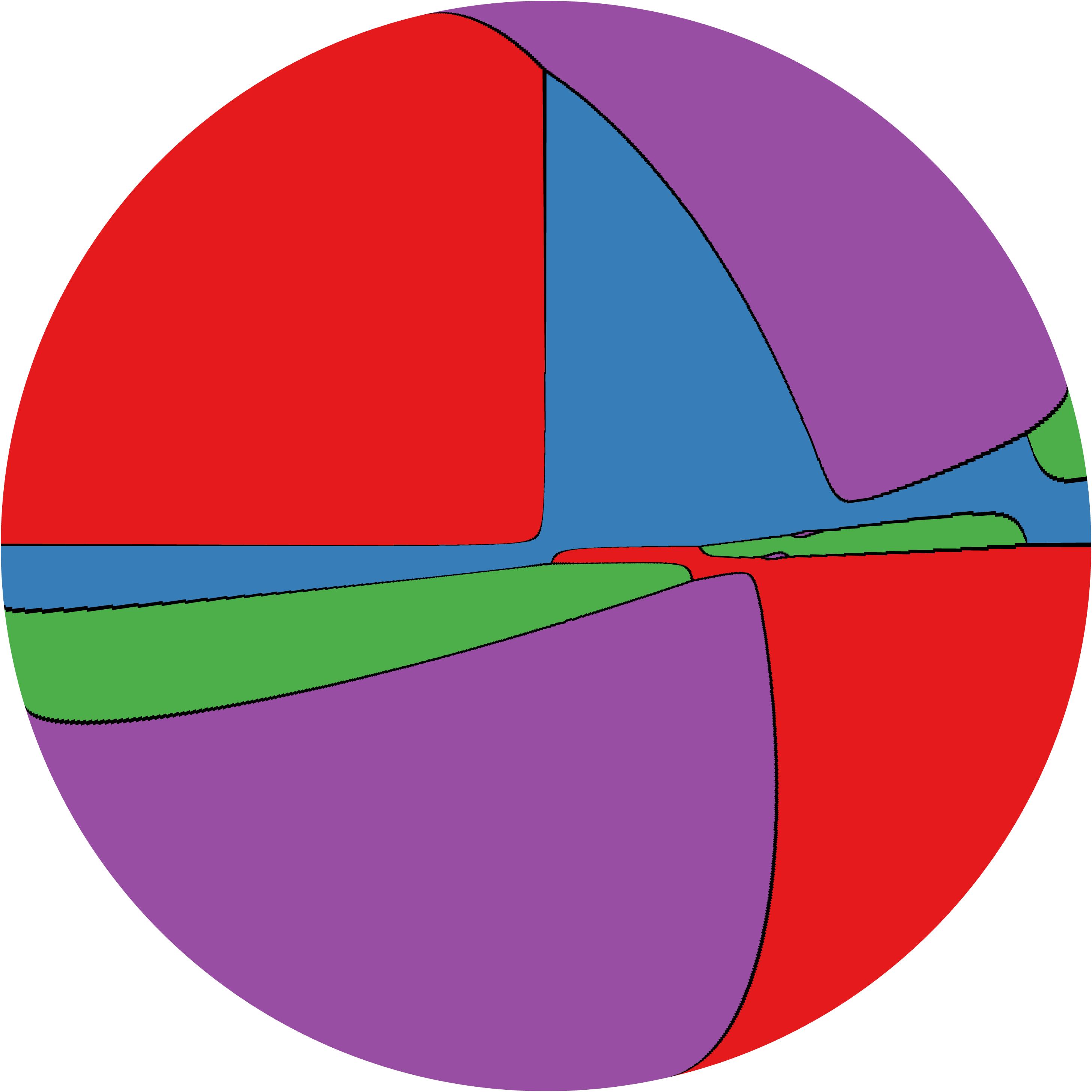}
    \end{minipage}
    \hfill
    \begin{minipage}[b]{0.48\textwidth}
        \centering
        \includegraphics[width=0.9\linewidth]{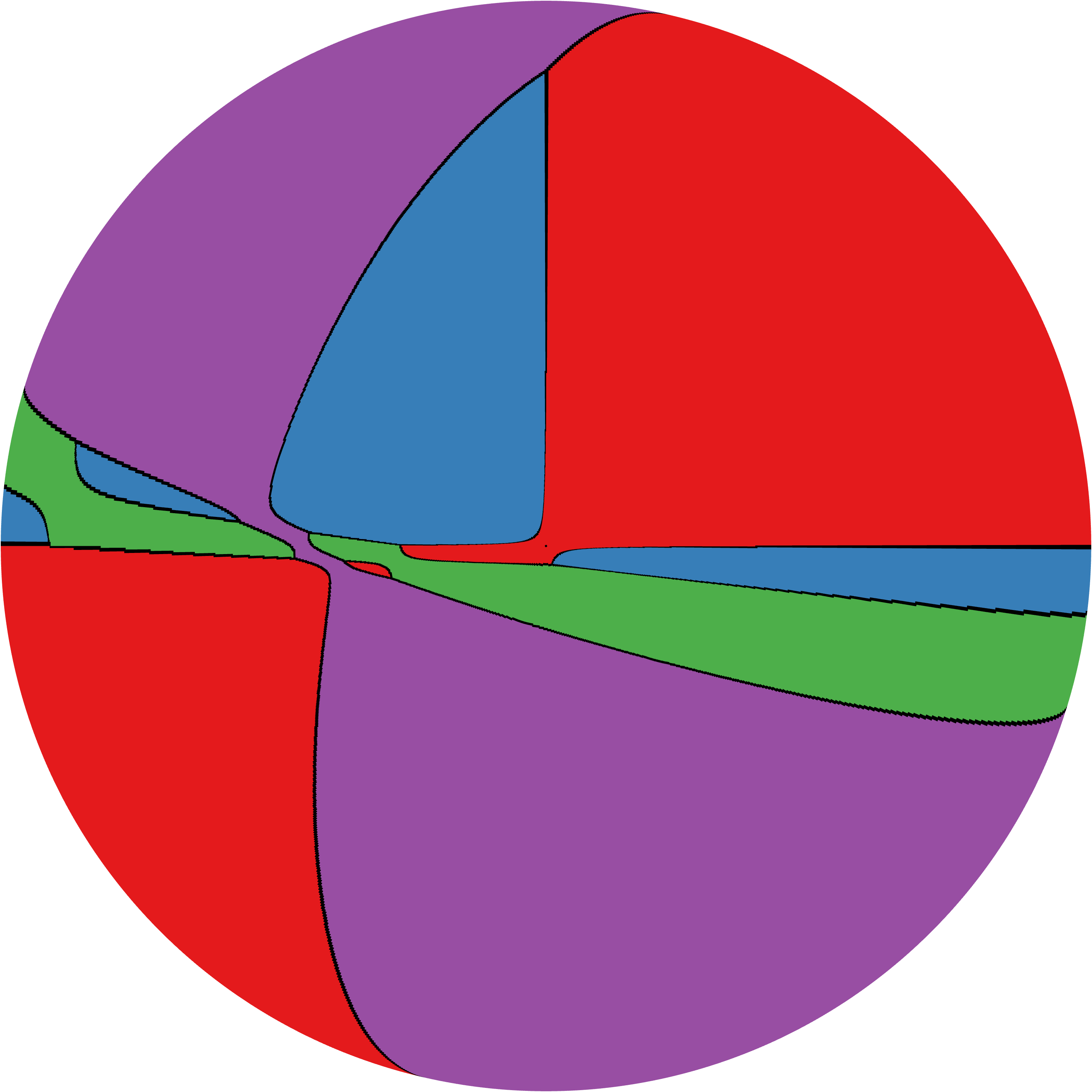}
     \end{minipage}
\caption{Top and bottom view of $\gmap(\fvd(L))$ of topology \Rom{9}. }\label{fig:gmapfvd9}
\end{figure}

\subsection{Topology \Rom{10}: 6 vertices}

The configuration tuple that induces topology \Rom{10} is shown in \cref{fig:topo10comb}. A set of lines that realizes this tuple, hence this topology, has the following parameters: $(a,b_{3},c_{3},d_{3},e_{3},b_{4},c_{4},d_{4},e_{4})=(1,-18,-5,-4,-16,11,-18,-6,-3)$. \cref{fig:topo10verify} shows the six projected bisectors of these four lines, which match the configurations shown in \cref{fig:topo10comb}. The $\gmap(\fvd(L))$ is shown in \cref{fig:gmapfvd10}.

\begin{figure}[h]
        \centering
        \includegraphics[page=10,width=\textwidth]{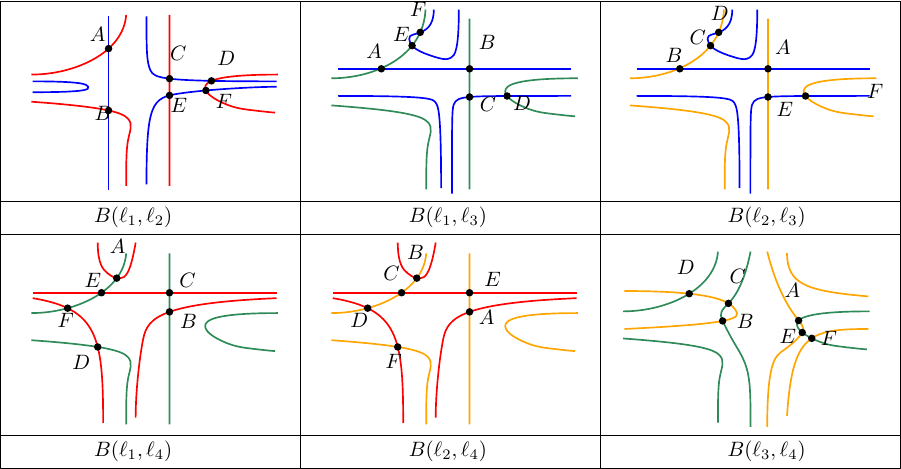}
        \caption{Combination that forms topology \Rom{10}. }\label{fig:topo10comb}
\end{figure}

\begin{figure}[!h]
\centering
\minipage{0.03\textwidth}
(a)	
\endminipage
\minipage{0.3\textwidth}
\includegraphics[width=\textwidth]{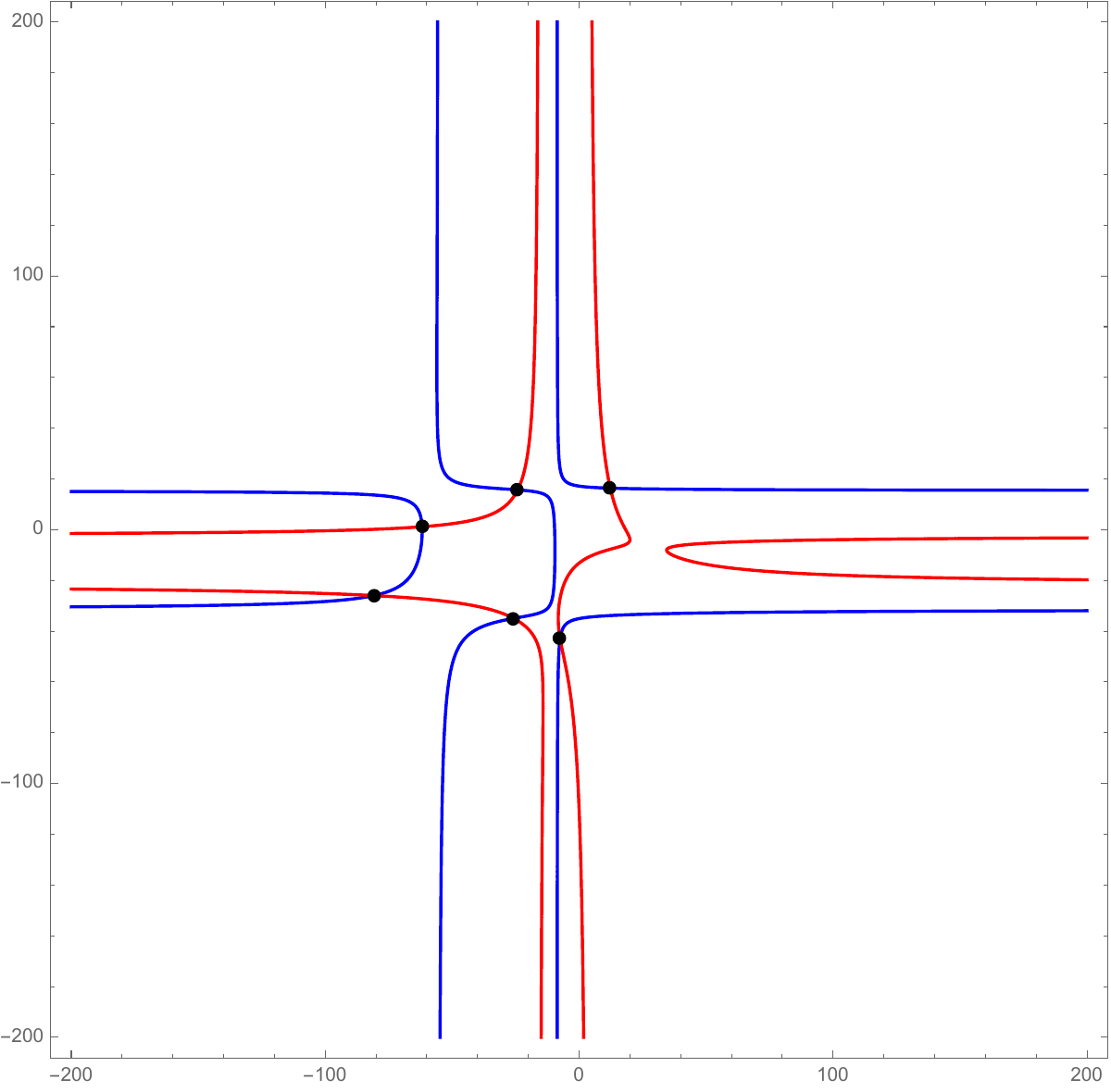}
\endminipage\hfill
\minipage{0.03\textwidth}
(b)	
\endminipage
\minipage{0.3\textwidth}
\includegraphics[width=\textwidth]{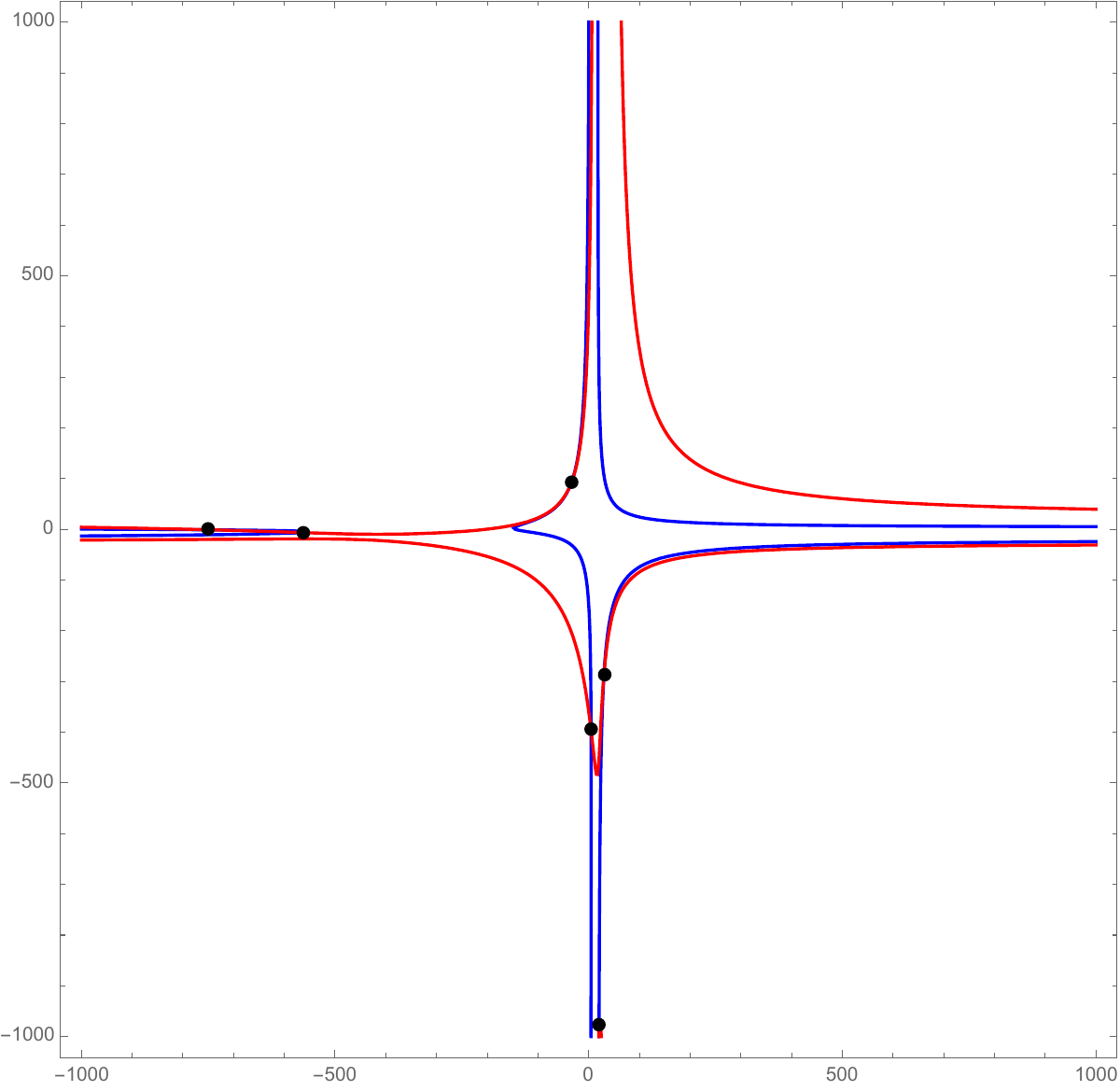}	
\endminipage\hfill
\minipage{0.03\textwidth}
(c)	
\endminipage
\minipage{0.3\textwidth}
\includegraphics[width=\textwidth]{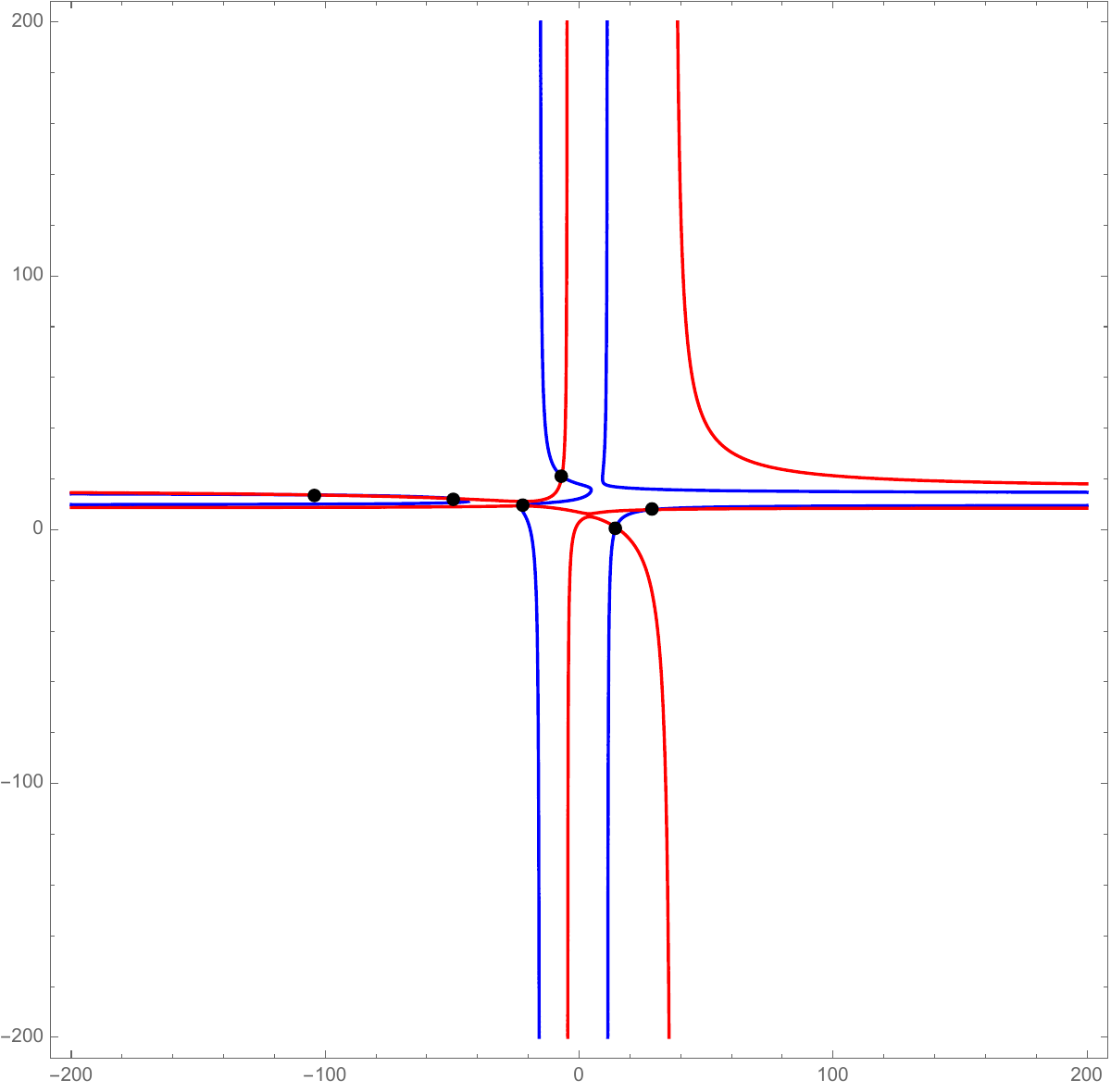}
\endminipage%
\newline
\centering
\minipage{0.03\textwidth}
(d)	
\endminipage
\minipage{0.3\textwidth}
\includegraphics[width=\textwidth]{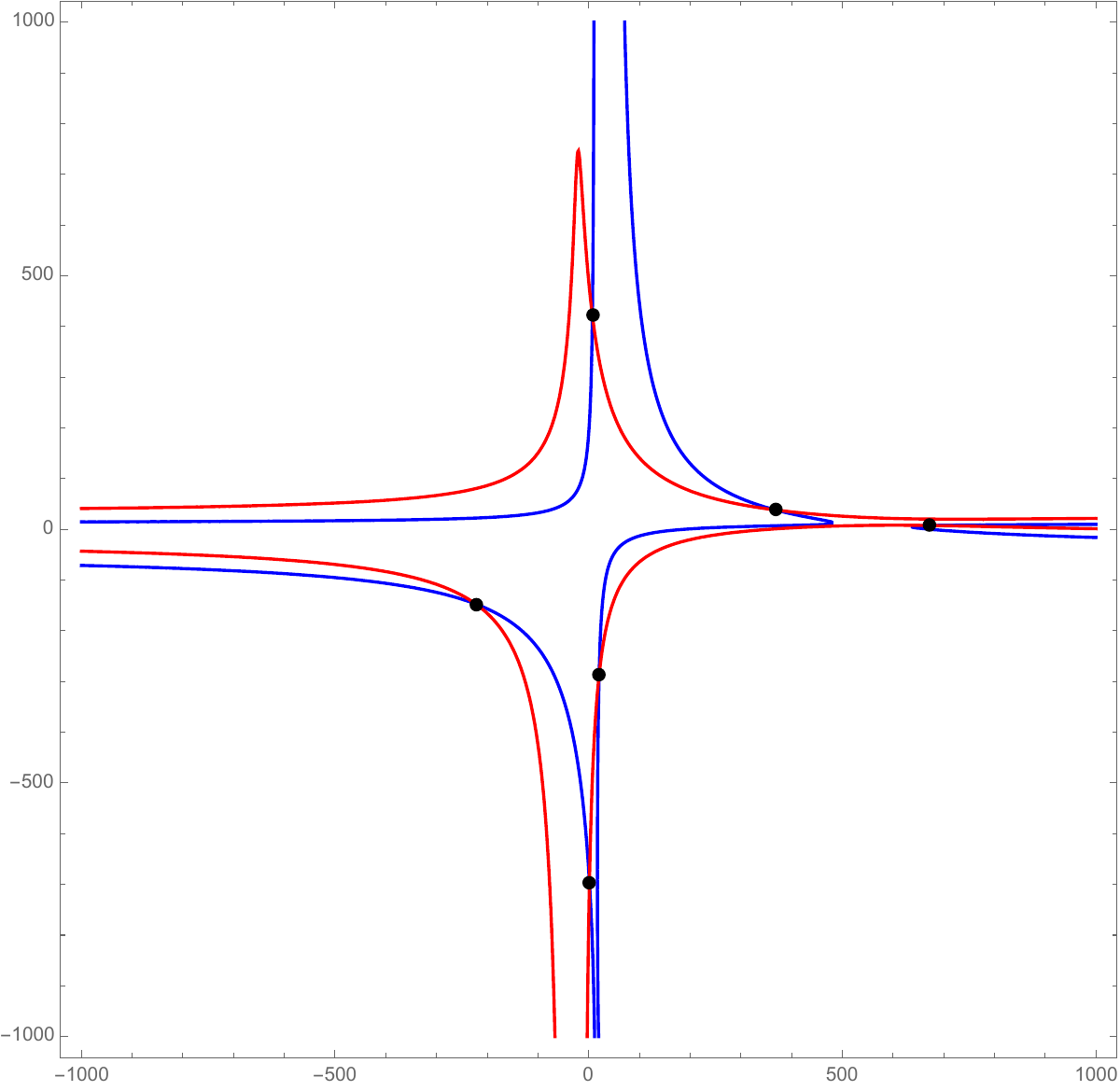}
\endminipage\hfill
\minipage{0.03\textwidth}
(e)	
\endminipage
\minipage{0.3\textwidth}
\includegraphics[width=\textwidth]{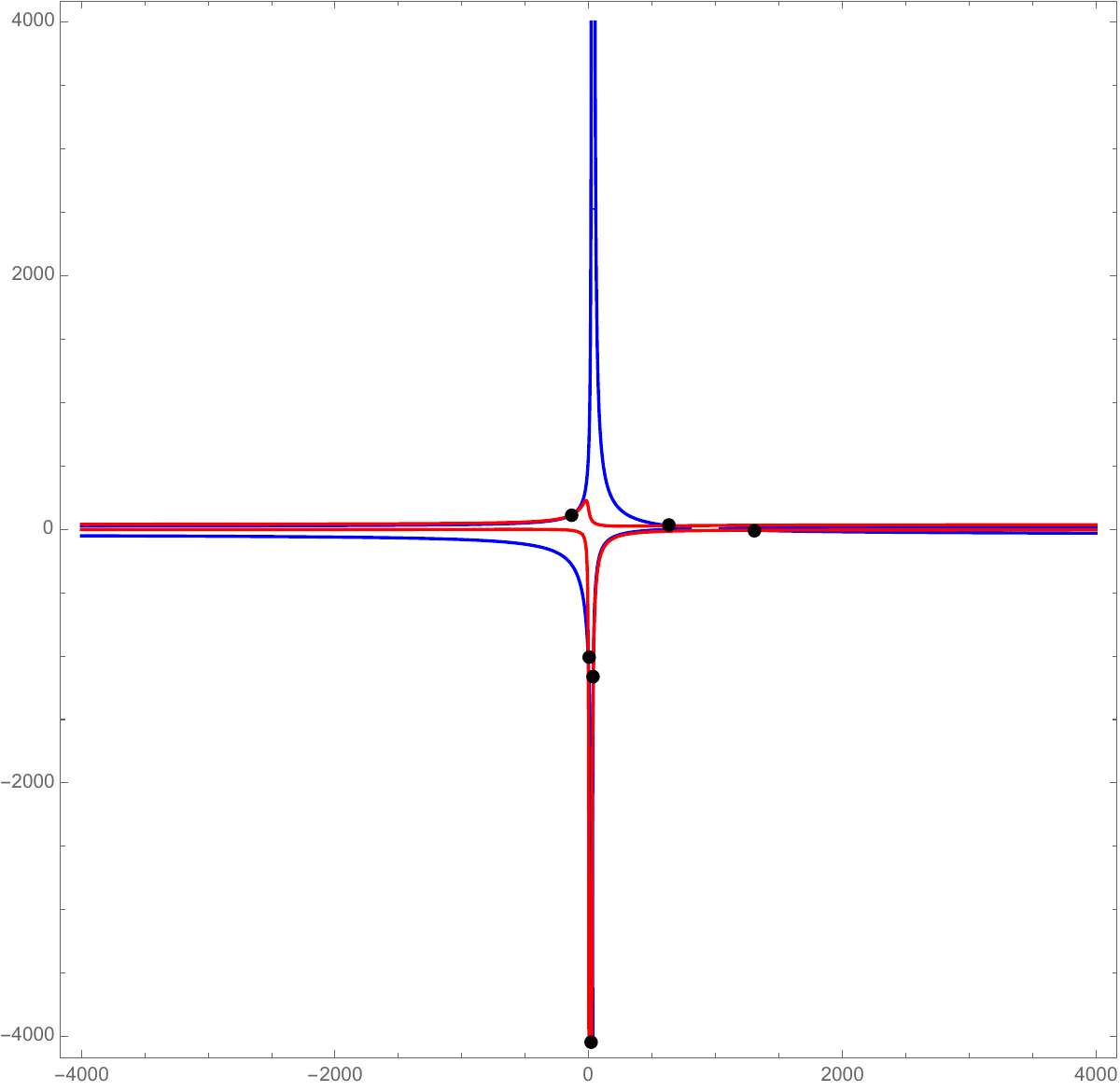}
\endminipage\hfill
\minipage{0.03\textwidth}
(f)	
\endminipage
\minipage{0.3\textwidth}
\includegraphics[width=\textwidth]{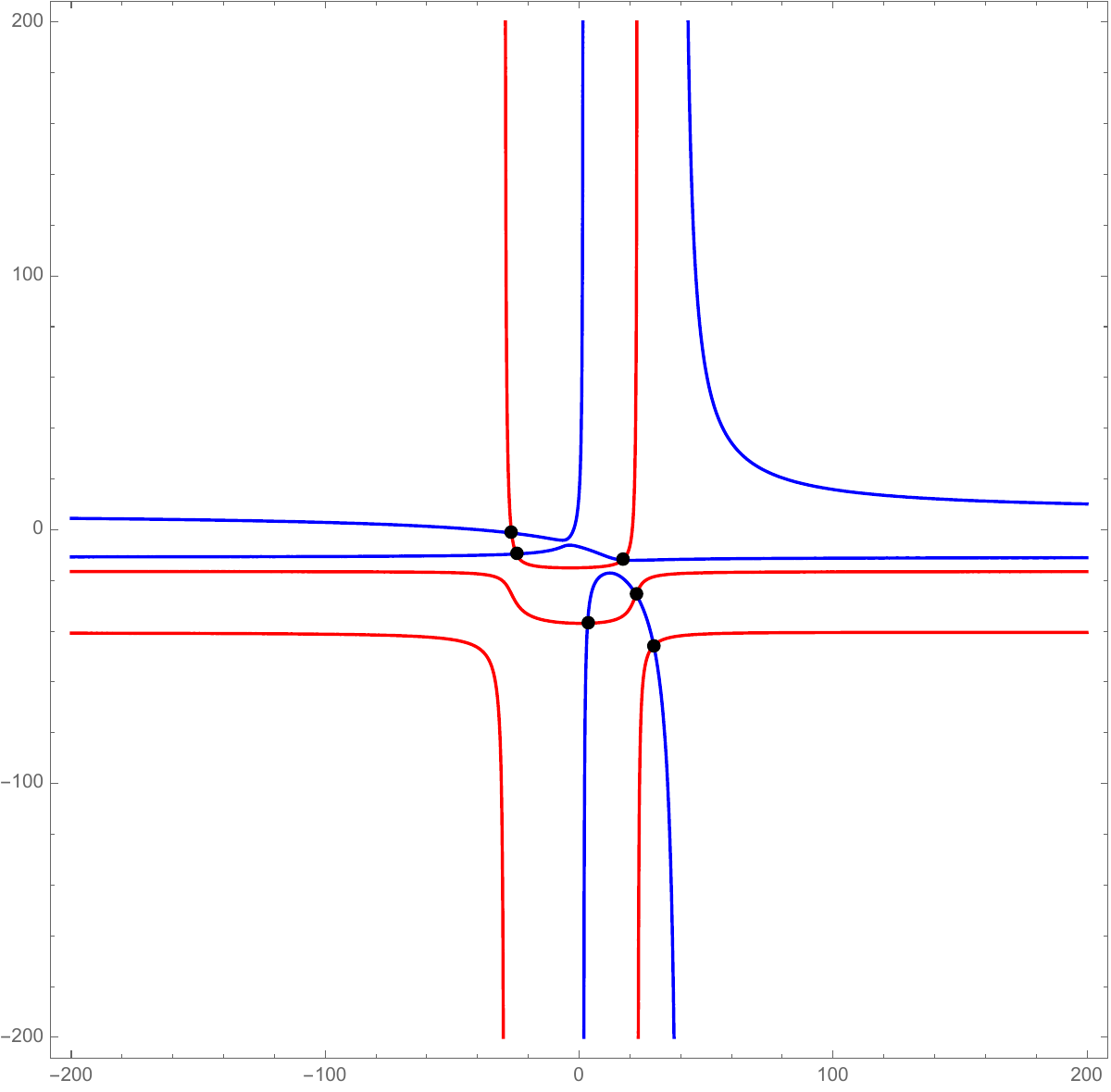}
\endminipage%
\caption{Six projected bisectors when the four lines have the following parameters: $(a,b_{3},c_{3},d_{3},e_{3},b_{4},c_{4},d_{4},e_{4})=(1,-18,-5,-4,-16,11,-18,-6,-3)$. They match the configurations shown in \cref{fig:topo10comb} and form topology \Rom{10}. }\label{fig:topo10verify}
\end{figure}

\begin{figure}[h]
    \centering
    \begin{minipage}[b]{0.48\textwidth}
        \centering
        \includegraphics[width=0.9\linewidth]{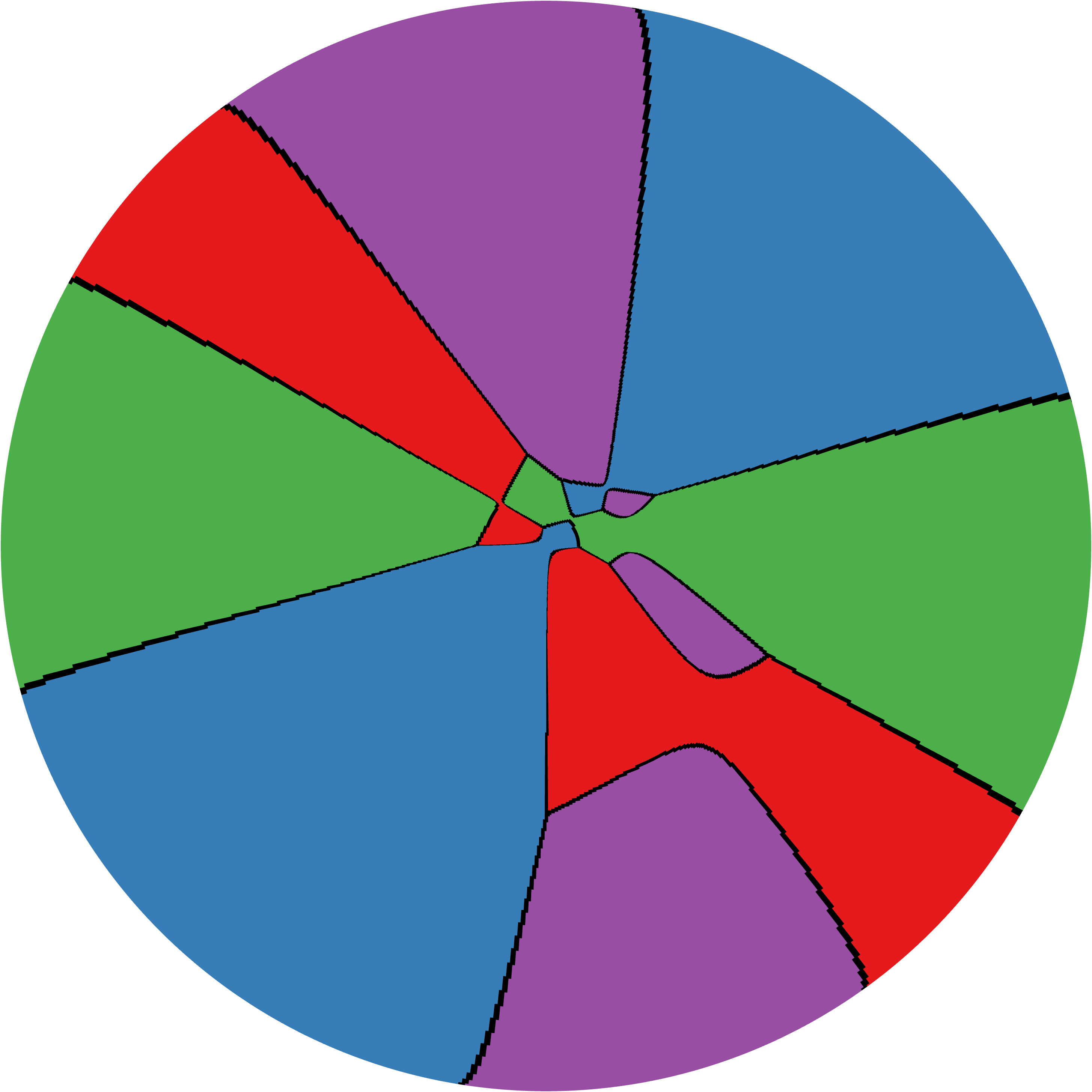}
    \end{minipage}
    \hfill
    \begin{minipage}[b]{0.48\textwidth}
        \centering
        \includegraphics[width=0.9\linewidth]{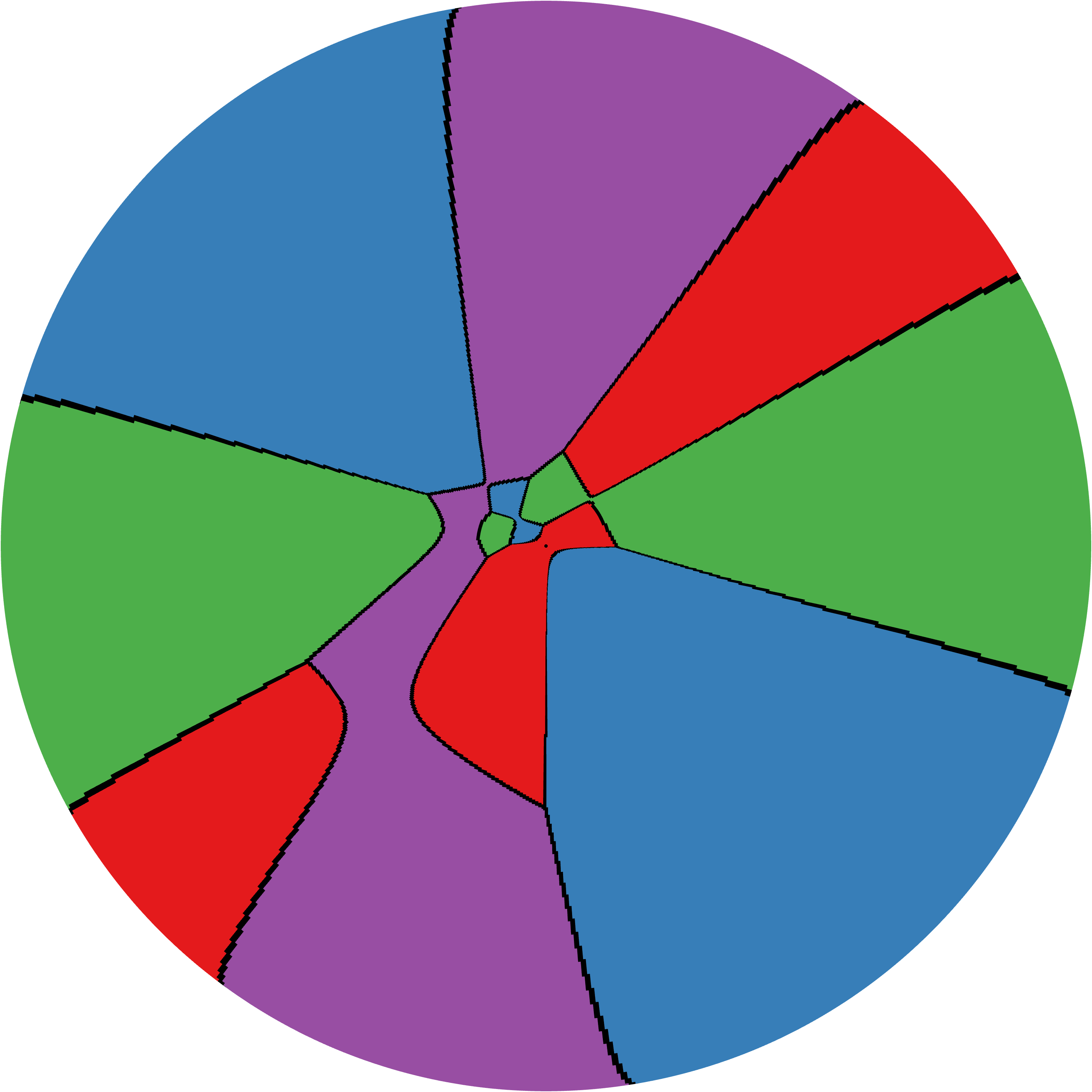}
     \end{minipage}
\caption{Top and bottom view of $\gmap(\fvd(L))$ of topology \Rom{10}. }\label{fig:gmapfvd10}
\end{figure}

\subsection{Topology \Rom{11}: 6 vertices}

The configuration tuple that induces topology \Rom{11} is shown in \cref{fig:topo11comb}. A set of lines that realizes this tuple, hence this topology, has the following parameters: $(a,b_{3},c_{3},d_{3},e_{3},b_{4},c_{4},d_{4},e_{4})=(1,29,4,2,20,11,25,0,-3)$. \cref{fig:topo11verify} shows the six projected bisectors of these four lines, which match the configurations shown in \cref{fig:topo11comb}. The $\gmap(\fvd(L))$ is shown in \cref{fig:gmapfvd11}.

\begin{figure}[h]
        \centering
        \includegraphics[page=11,width=\textwidth]{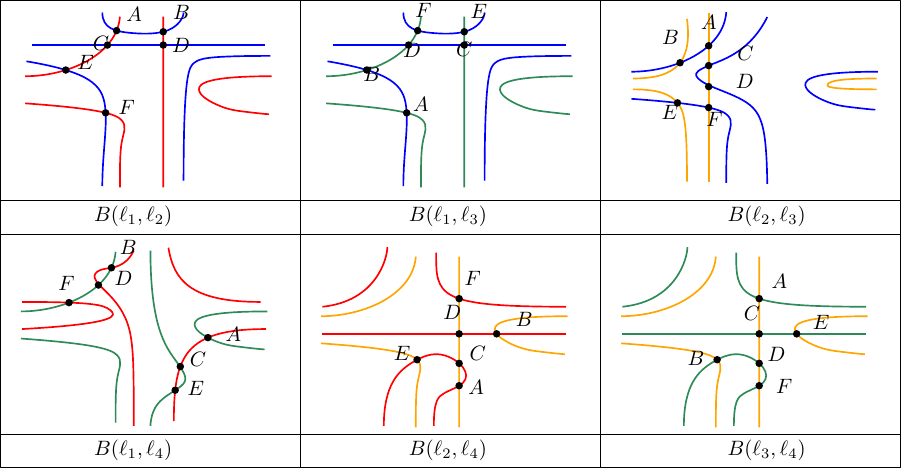}
        \caption{Combination that forms topology \Rom{11}. }\label{fig:topo11comb}
\end{figure}

\begin{figure}[!h]
\centering
\minipage{0.03\textwidth}
(a)	
\endminipage
\minipage{0.3\textwidth}
\includegraphics[width=\textwidth]{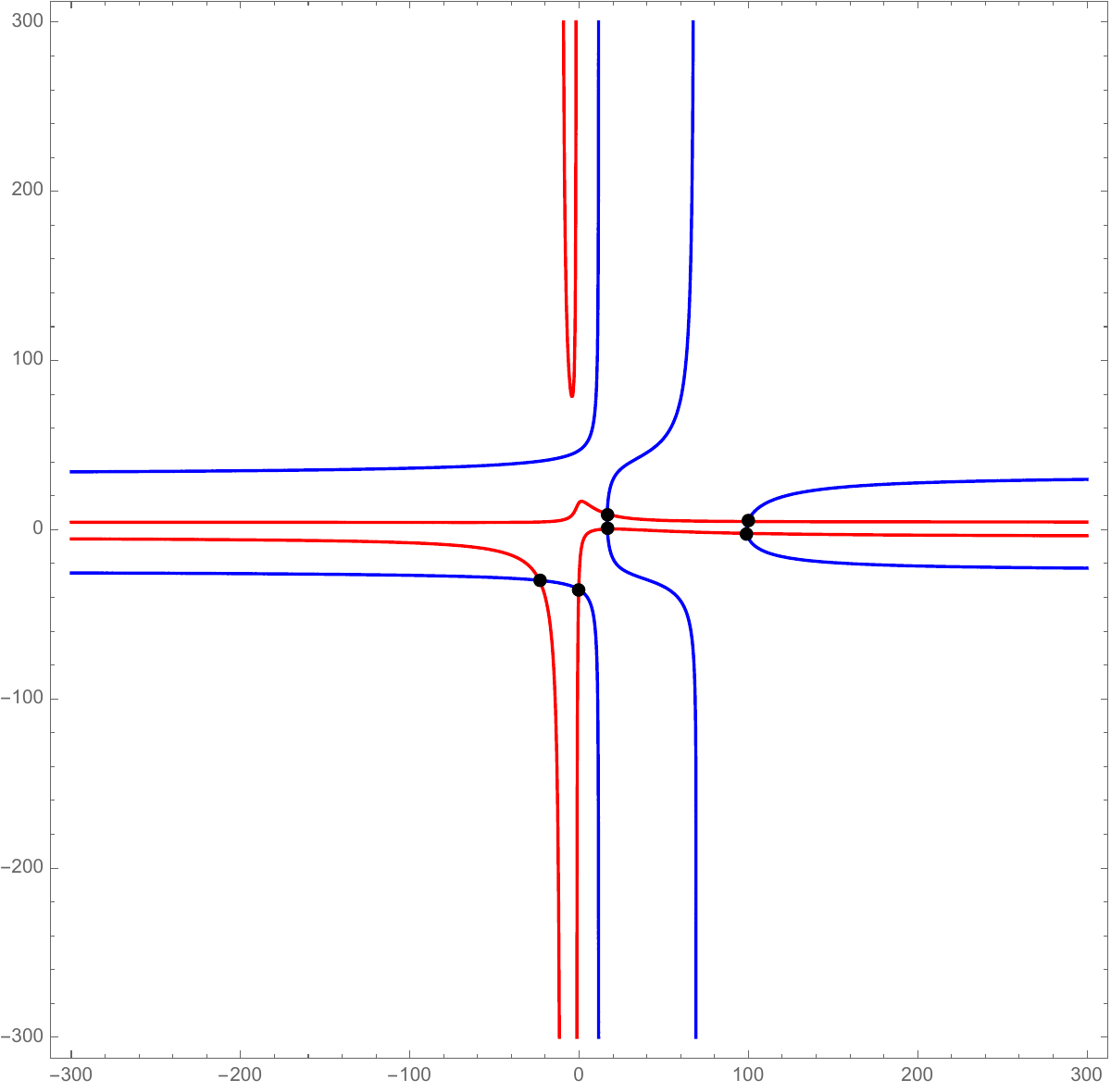}
\endminipage\hfill
\minipage{0.03\textwidth}
(b)	
\endminipage
\minipage{0.3\textwidth}
\includegraphics[width=\textwidth]{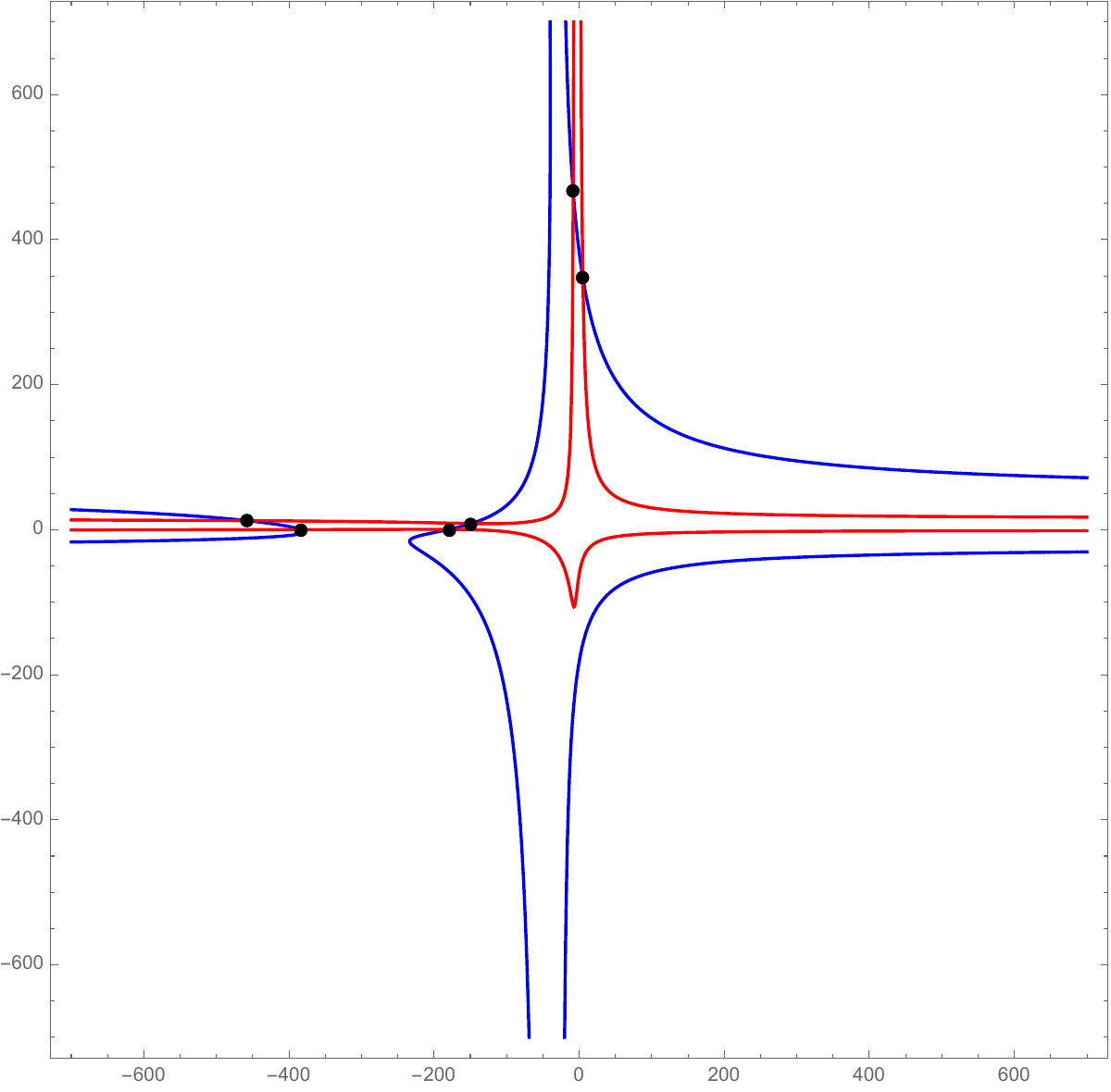}	
\endminipage\hfill
\minipage{0.03\textwidth}
(c)	
\endminipage
\minipage{0.3\textwidth}
\includegraphics[width=\textwidth]{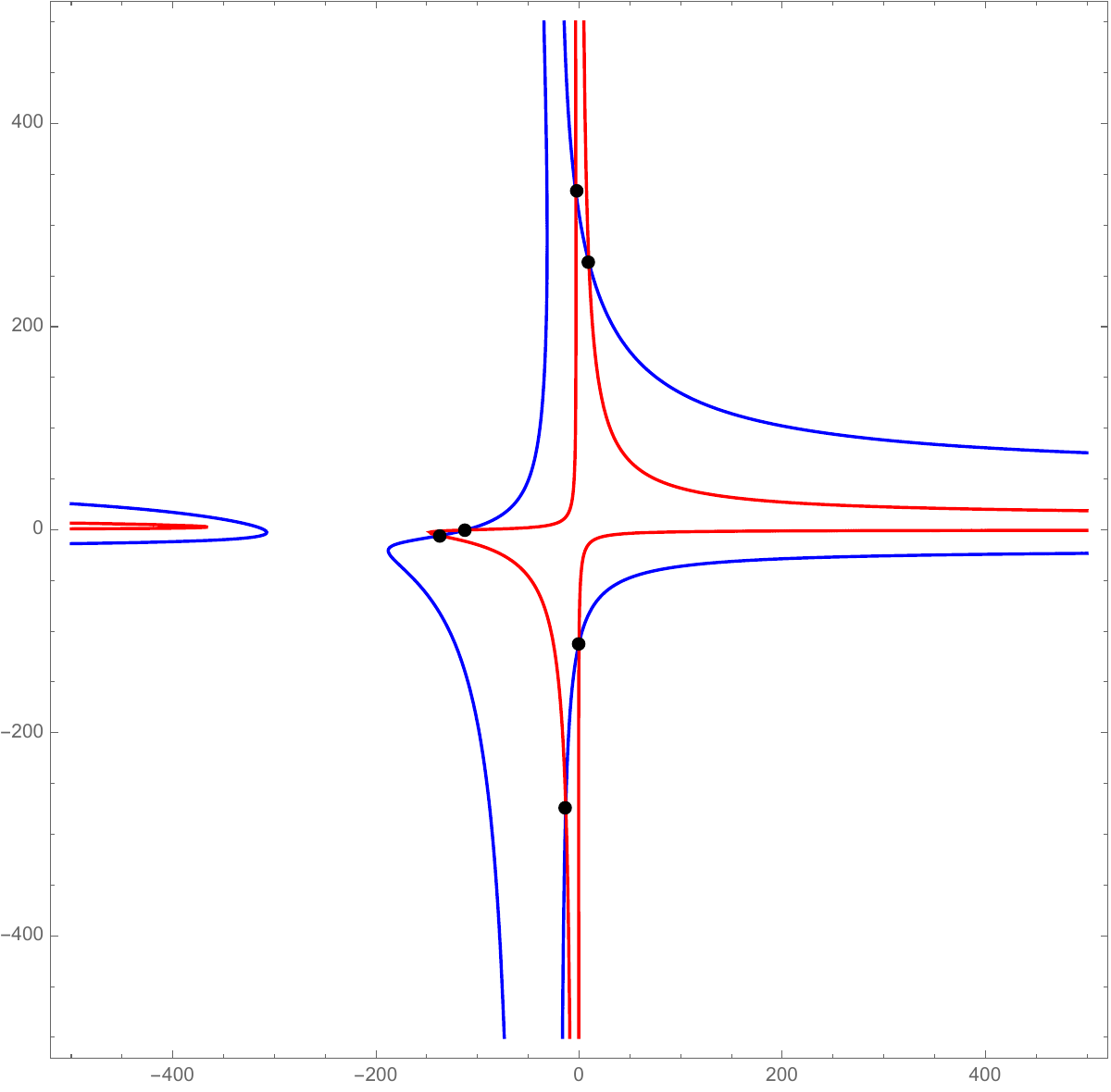}
\endminipage%
\newline
\centering
\minipage{0.03\textwidth}
(d)	
\endminipage
\minipage{0.3\textwidth}
\includegraphics[width=\textwidth]{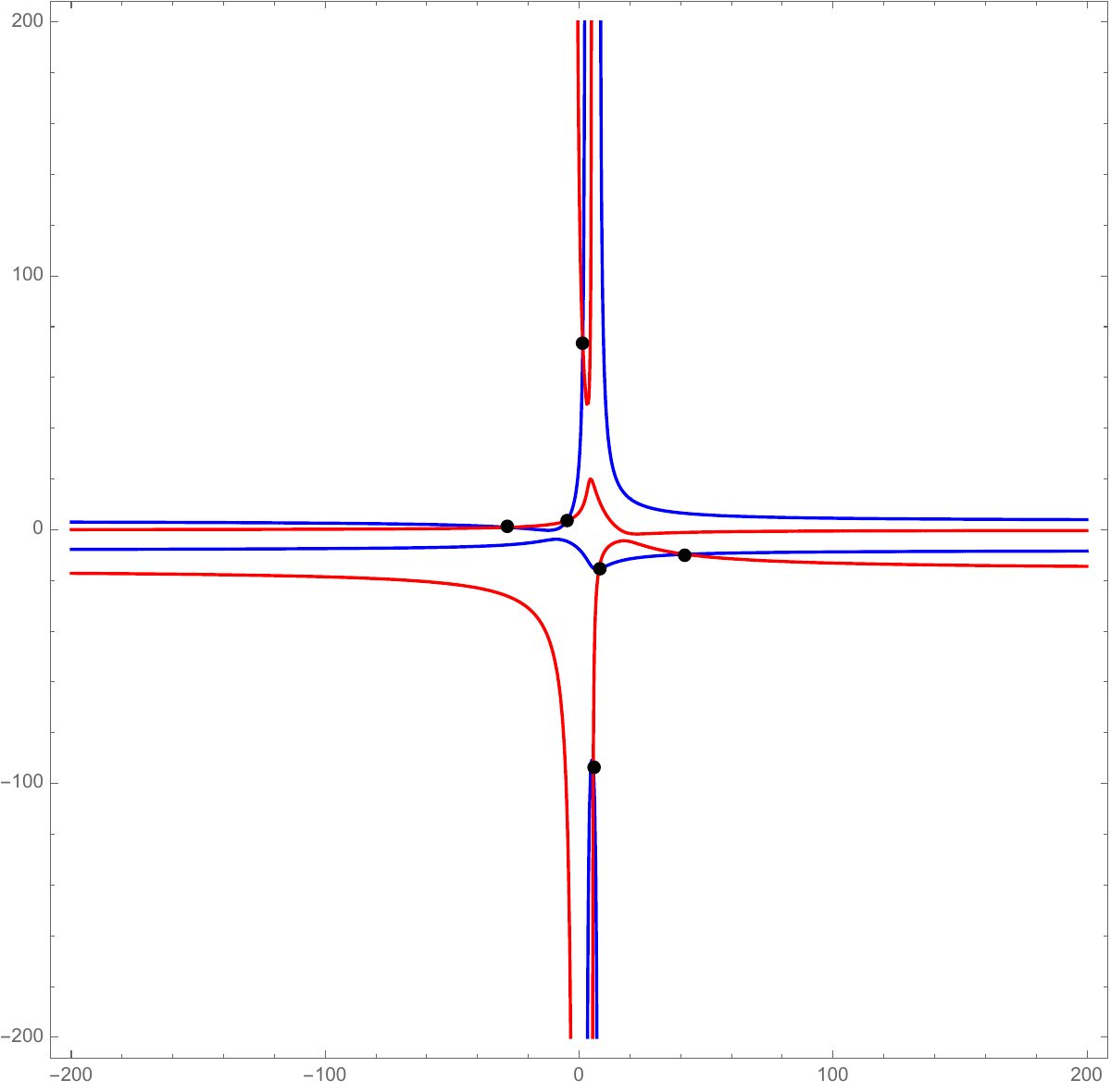}
\endminipage\hfill
\minipage{0.03\textwidth}
(e)	
\endminipage
\minipage{0.3\textwidth}
\includegraphics[width=\textwidth]{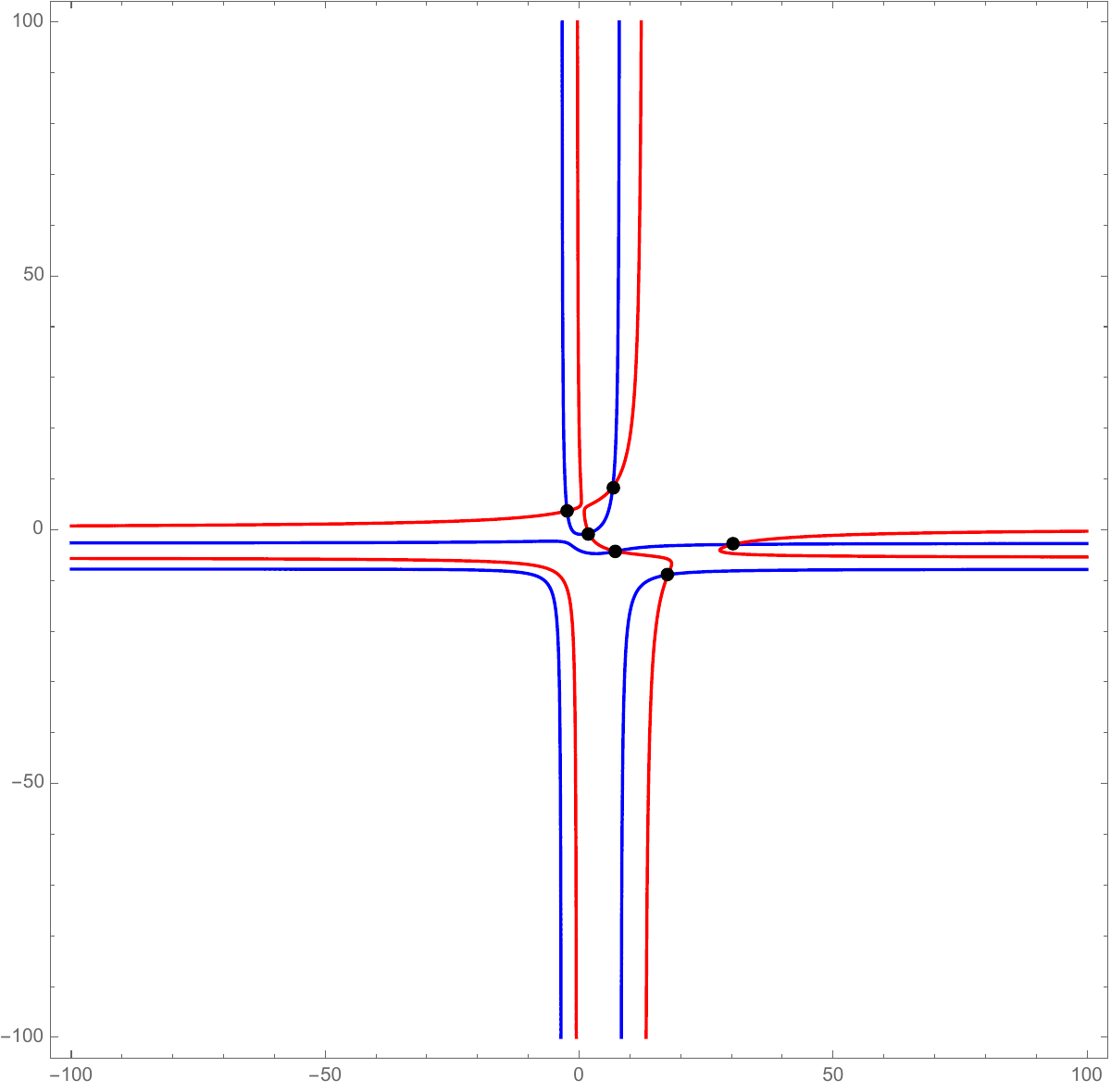}
\endminipage\hfill
\minipage{0.03\textwidth}
(f)	
\endminipage
\minipage{0.3\textwidth}
\includegraphics[width=\textwidth]{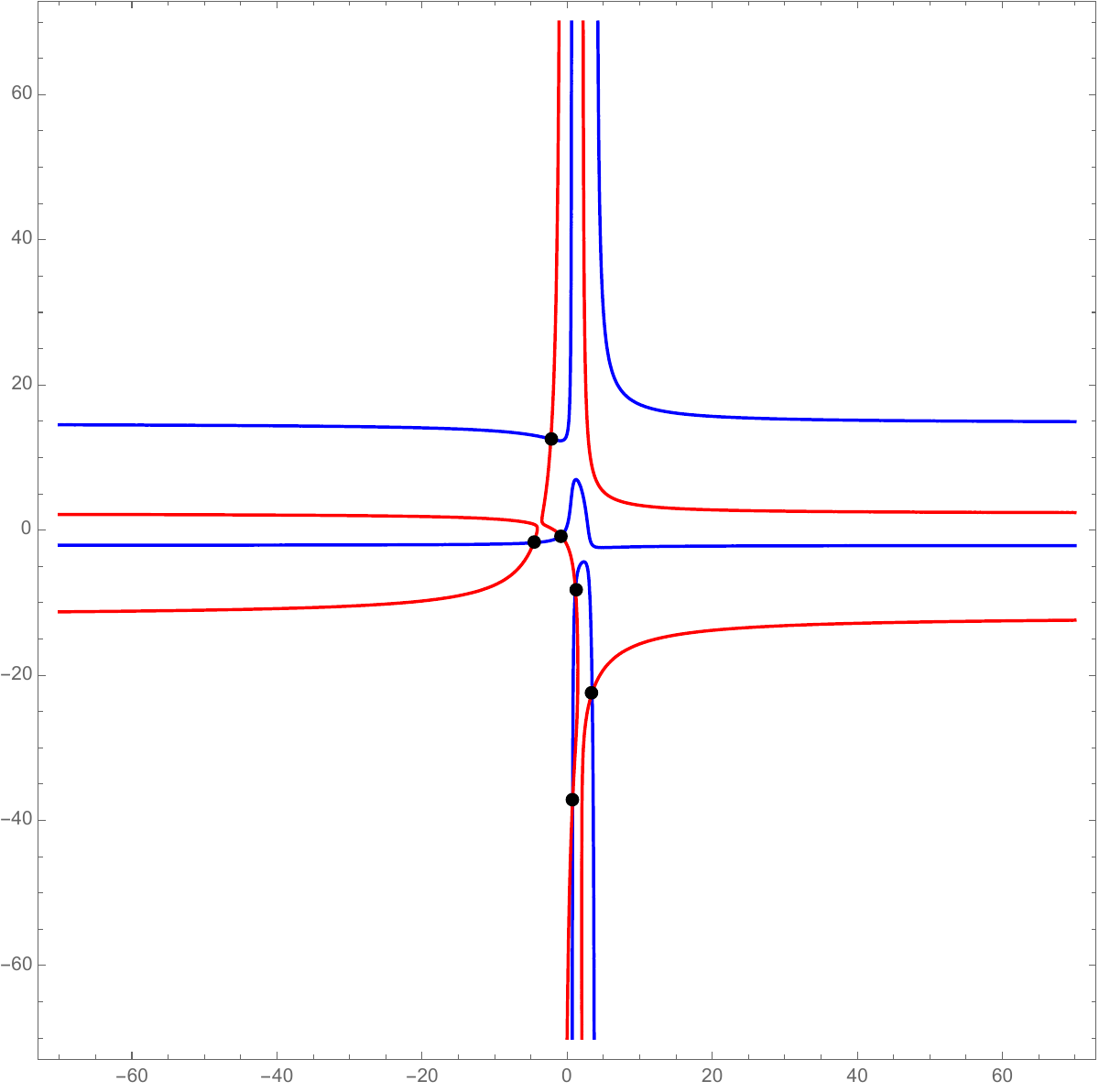}
\endminipage%
\caption{Six projected bisectors when the four lines have the following parameters: $(a,b_{3},c_{3},d_{3},e_{3},b_{4},c_{4},d_{4},e_{4})=(1,29,4,2,20,11,25,0,-3)$. They match the configurations shown in \cref{fig:topo11comb} and form topology \Rom{11}. }\label{fig:topo11verify}
\end{figure}

\begin{figure}[h]
    \centering
    \begin{minipage}[b]{0.48\textwidth}
        \centering
        \includegraphics[width=0.9\linewidth]{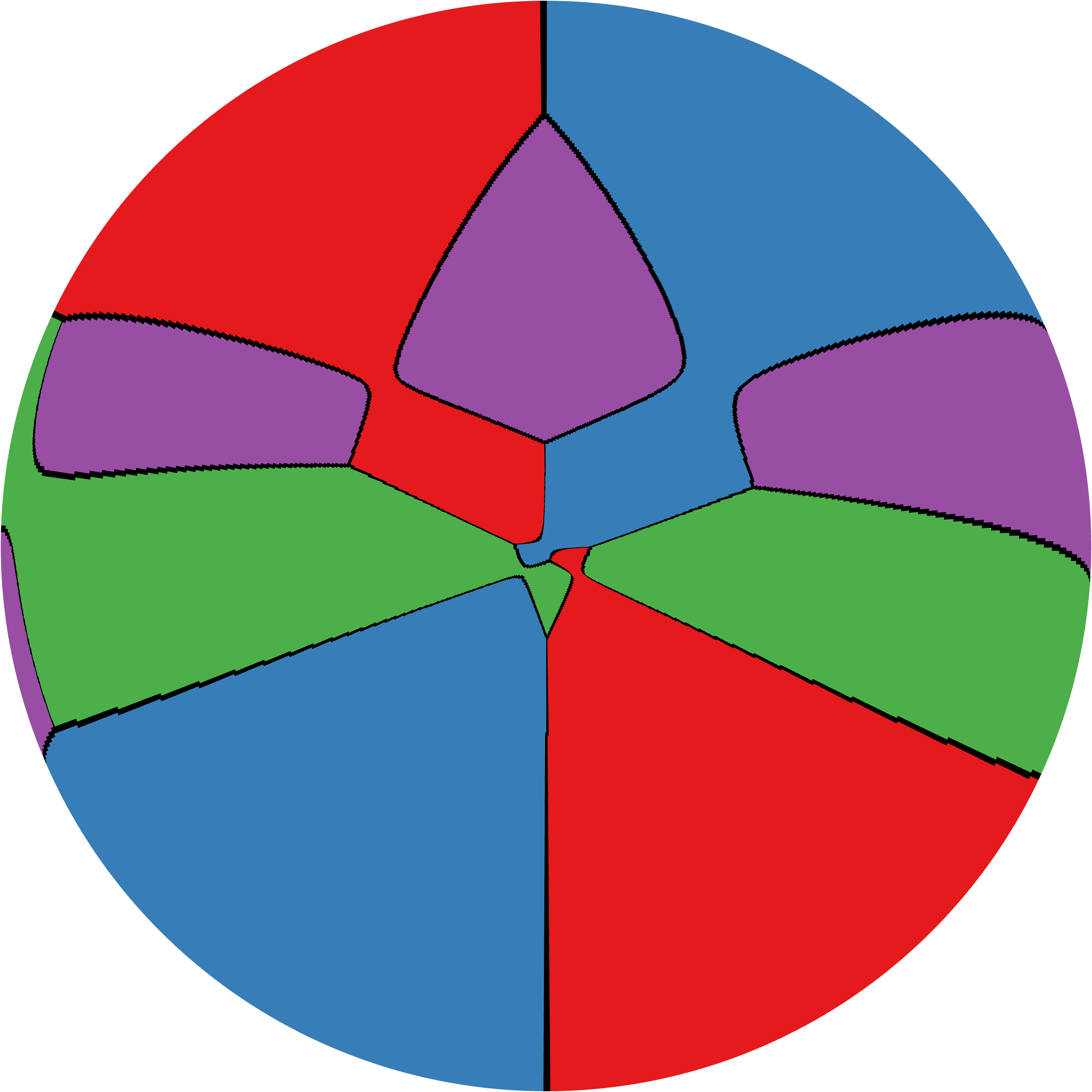}
    \end{minipage}
    \hfill
    \begin{minipage}[b]{0.48\textwidth}
        \centering
        \includegraphics[width=0.9\linewidth]{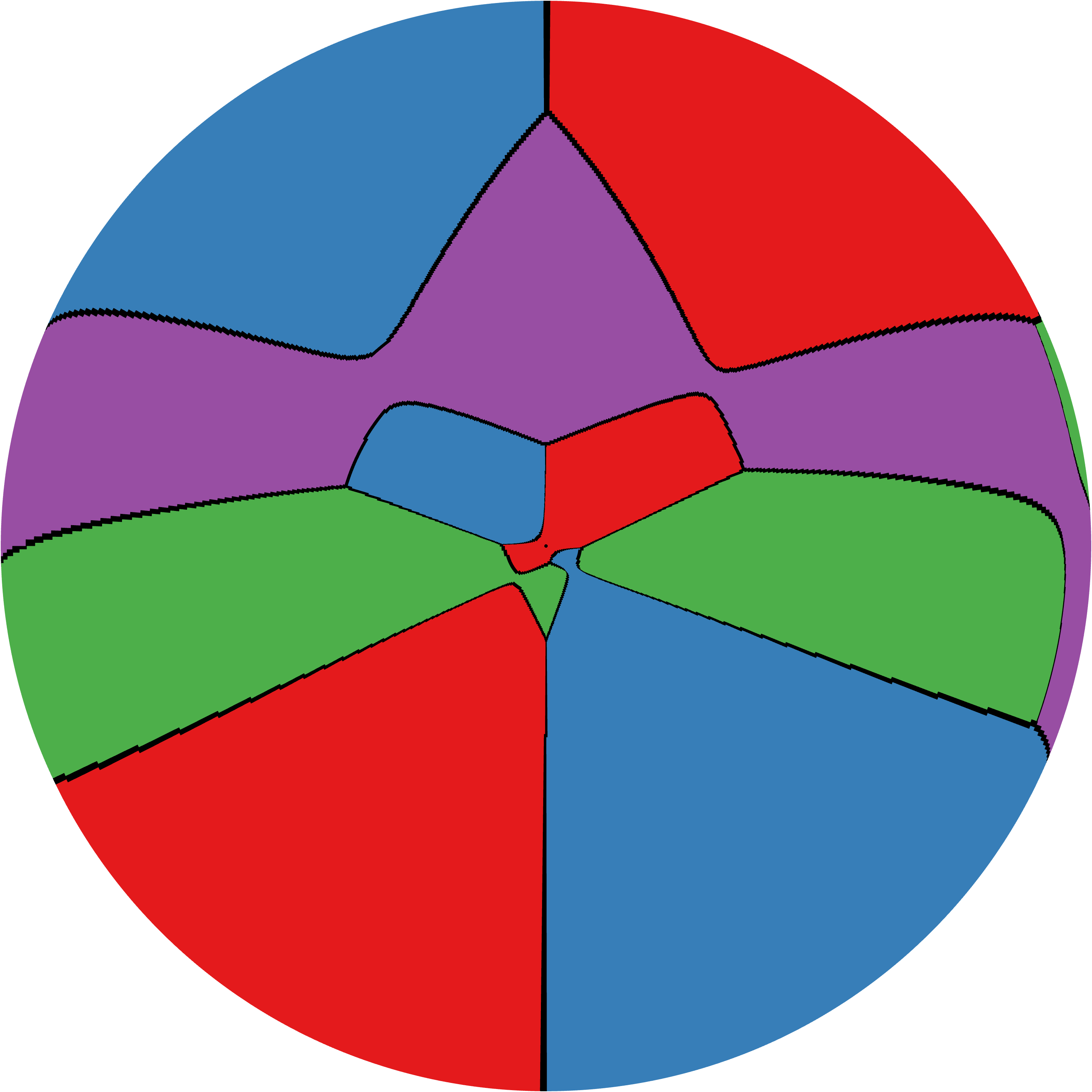}
     \end{minipage}
\caption{Top and bottom view of $\gmap(\fvd(L))$ of topology \Rom{11}. }\label{fig:gmapfvd11}
\end{figure}

\subsection{Topology \Rom{12}: 6 vertices}

The configuration tuple that induces topology \Rom{12} is shown in \cref{fig:topo12comb}. A set of lines that realizes this tuple, hence this topology, has the following parameters: $(a,b_{3},c_{3},d_{3},e_{3},b_{4},c_{4},d_{4},e_{4})=(2, 3, -6, -3, 16, -1, -12, -3, -19)$. \cref{fig:topo12verify} shows the six projected bisectors of these four lines, which match the configurations shown in \cref{fig:topo12comb}. The $\gmap(\fvd(L))$ is shown in \cref{fig:gmapfvd12}.

\begin{figure}[h]
        \centering
        \includegraphics[page=12,width=\textwidth]{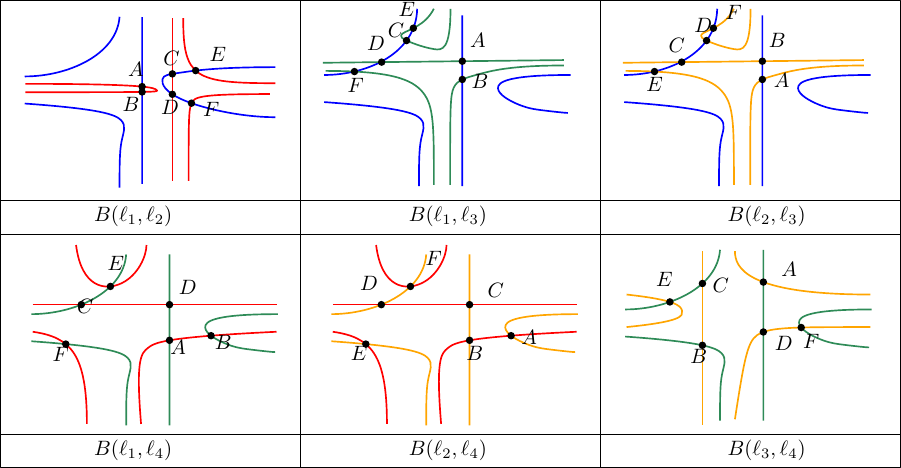}
        \caption{Combination that forms topology \Rom{12}. }\label{fig:topo12comb}
\end{figure}

\begin{figure}[!h]
\centering
\minipage{0.03\textwidth}
(a)	
\endminipage
\minipage{0.3\textwidth}
\includegraphics[width=\textwidth]{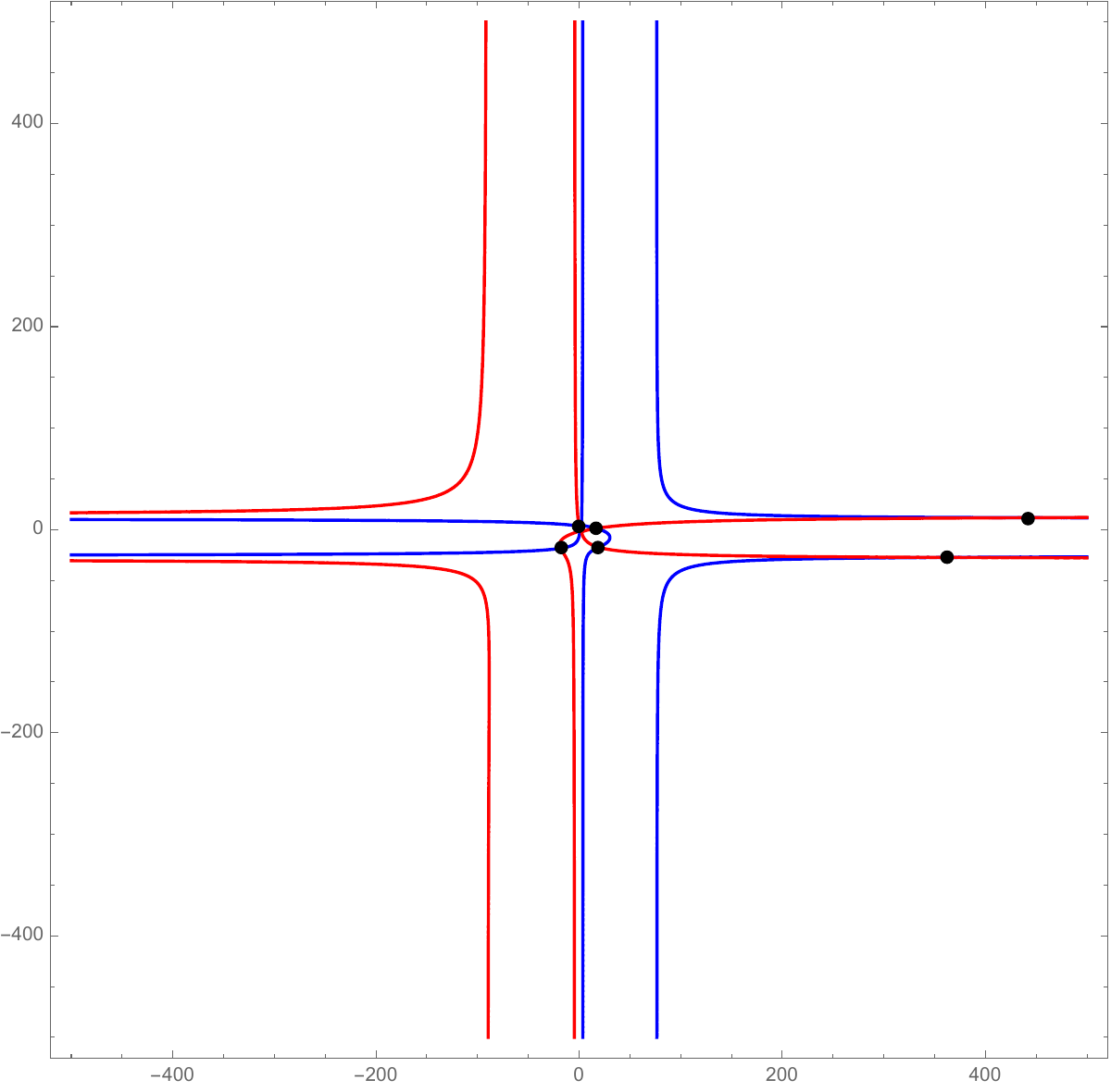}
\endminipage\hfill
\minipage{0.03\textwidth}
(b)	
\endminipage
\minipage{0.3\textwidth}
\includegraphics[width=\textwidth]{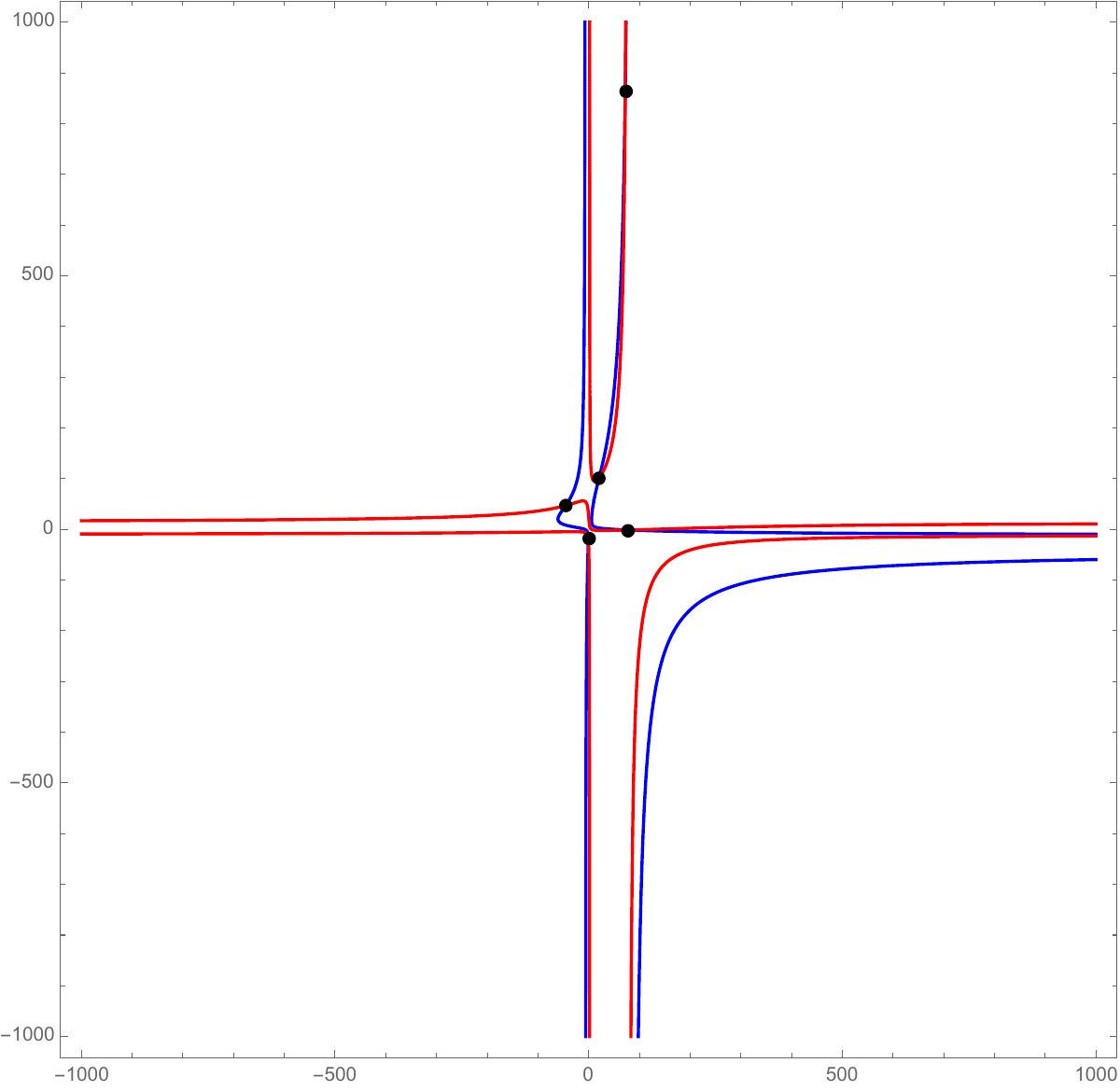}	
\endminipage\hfill
\minipage{0.03\textwidth}
(c)	
\endminipage
\minipage{0.3\textwidth}
\includegraphics[width=\textwidth]{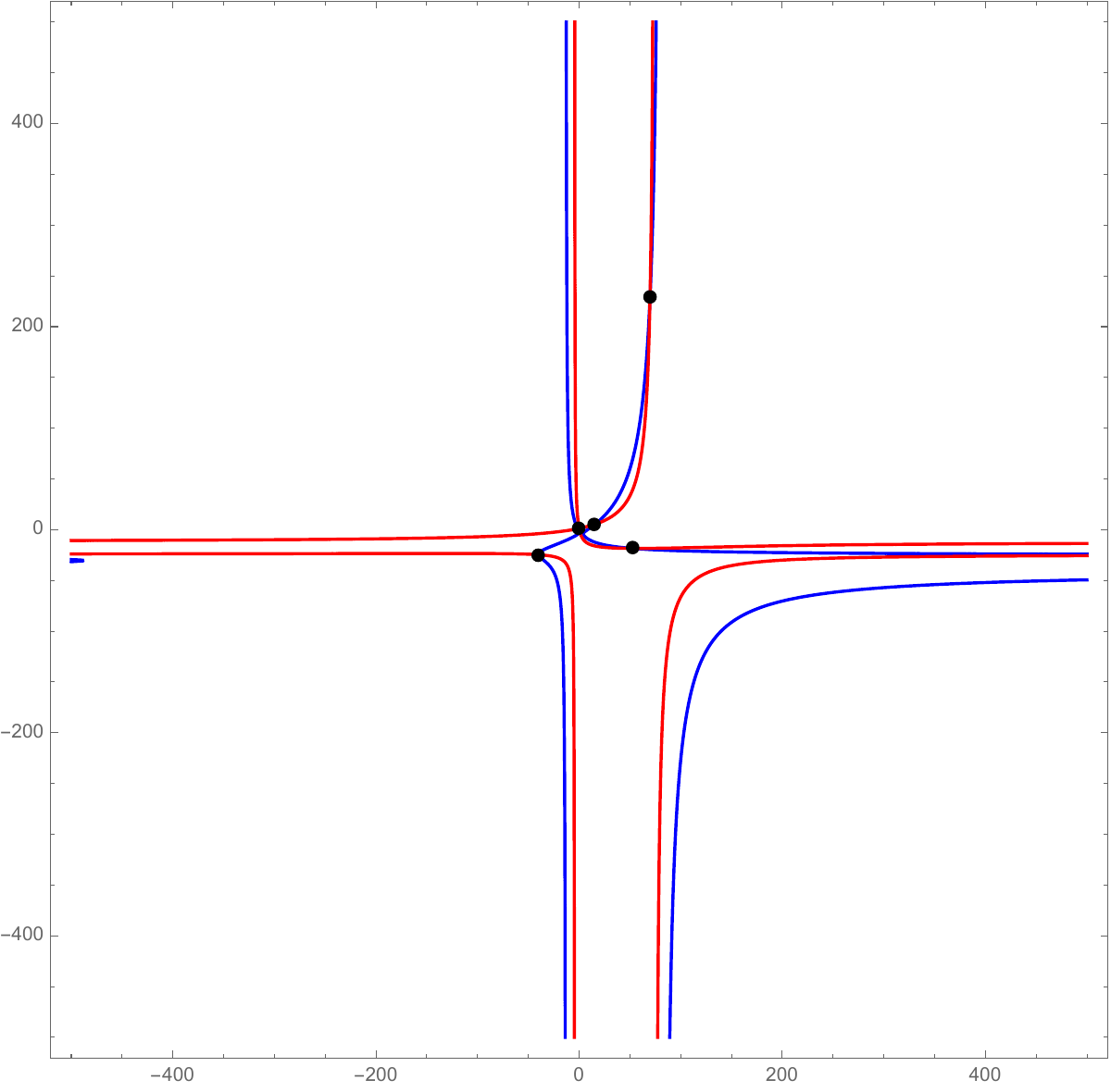}
\endminipage%
\newline
\centering
\minipage{0.03\textwidth}
(d)	
\endminipage
\minipage{0.3\textwidth}
\includegraphics[width=\textwidth]{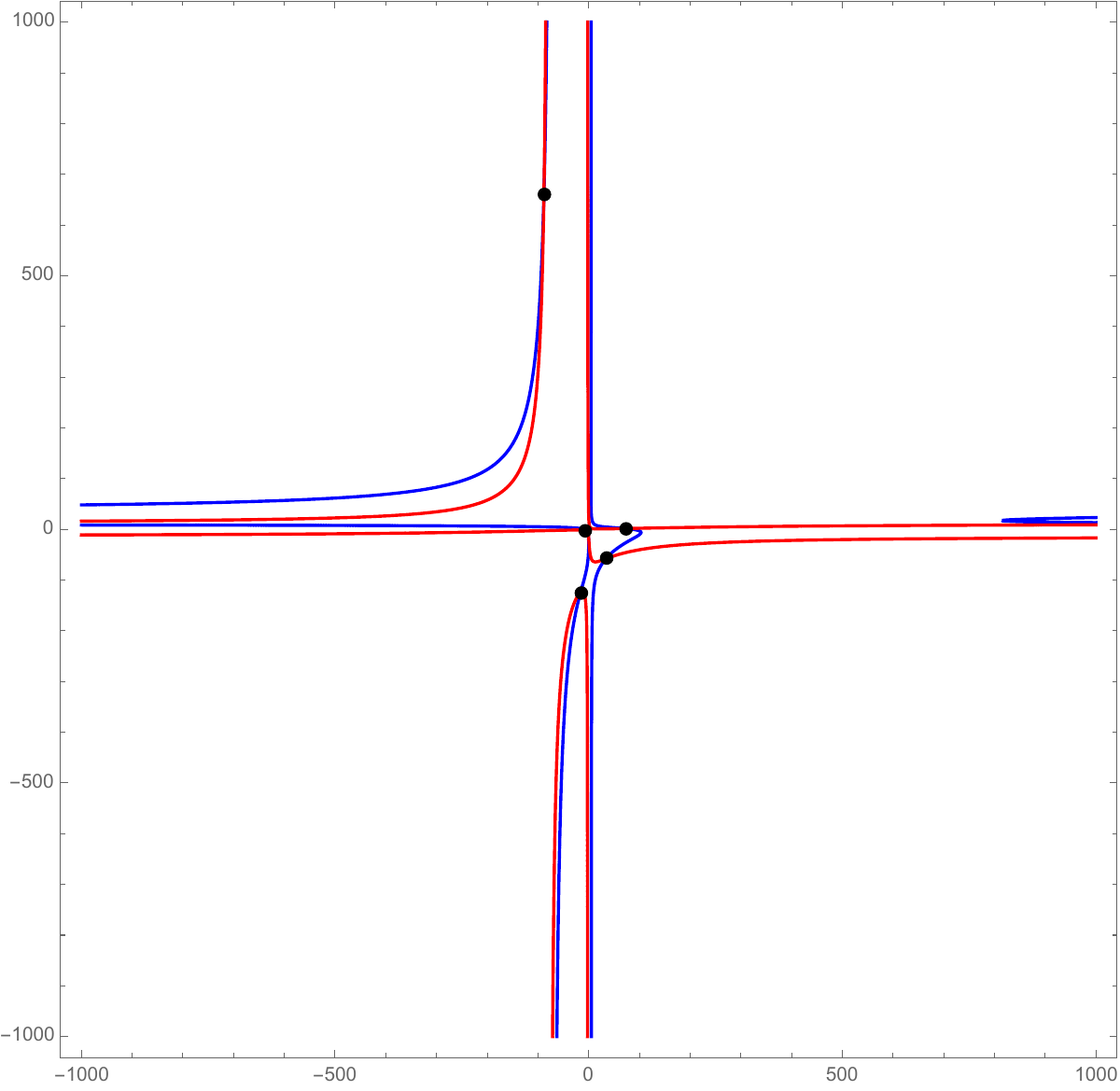}
\endminipage\hfill
\minipage{0.03\textwidth}
(e)	
\endminipage
\minipage{0.3\textwidth}
\includegraphics[width=\textwidth]{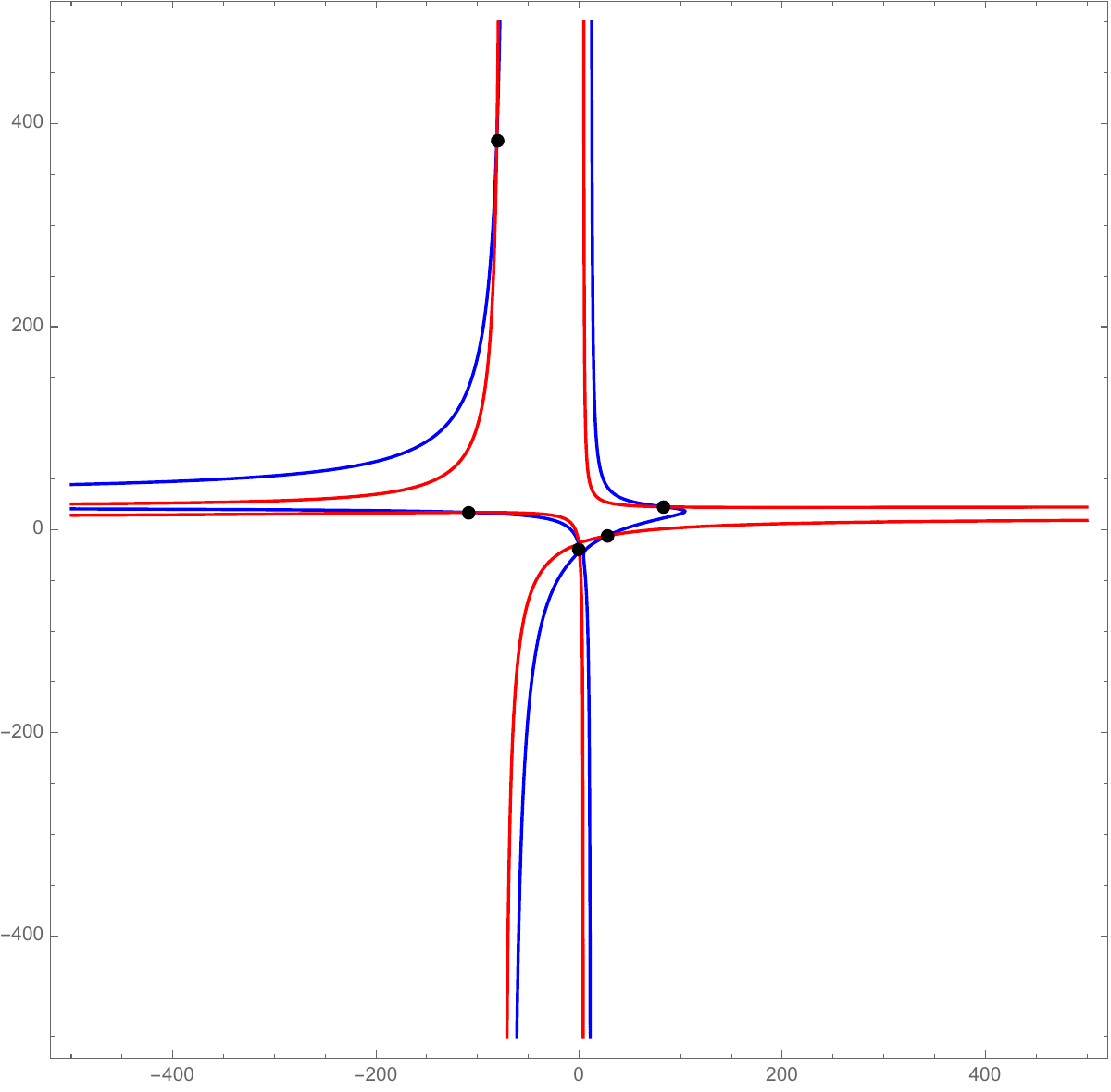}
\endminipage\hfill
\minipage{0.03\textwidth}
(f)	
\endminipage
\minipage{0.3\textwidth}
\includegraphics[width=\textwidth]{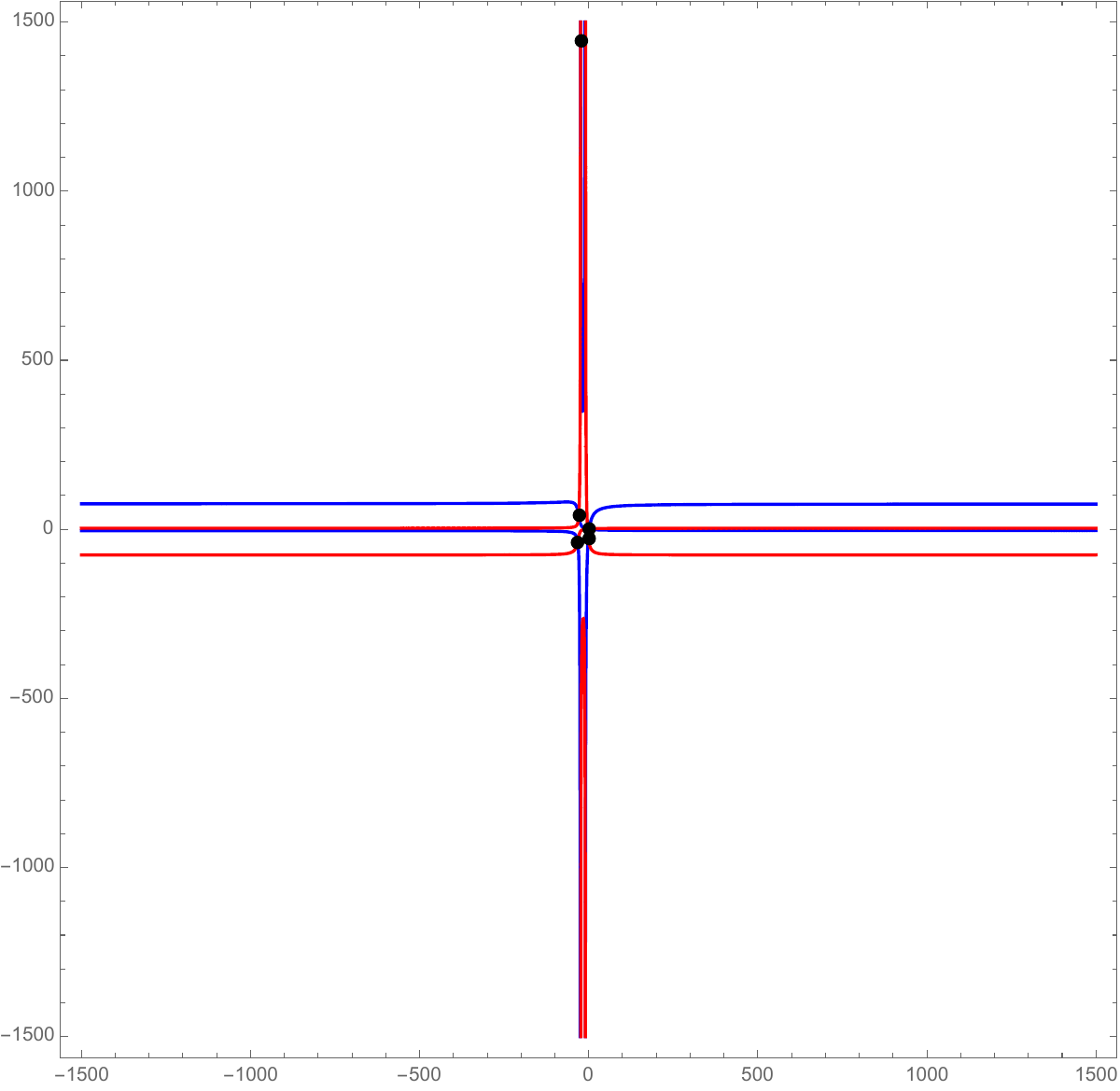}
\endminipage%
\caption{Six projected bisectors when the four lines have the following parameters: $(a,b_{3},c_{3},d_{3},e_{3},b_{4},c_{4},d_{4},e_{4})=(2, 3, -6, -3, 16, -1, -12, -3, -19)$. They match the configurations shown in \cref{fig:topo12comb} and form topology \Rom{12}. }\label{fig:topo12verify}
\end{figure}

\begin{figure}[h]
    \centering
    \begin{minipage}[b]{0.48\textwidth}
        \centering
        \includegraphics[width=0.9\linewidth]{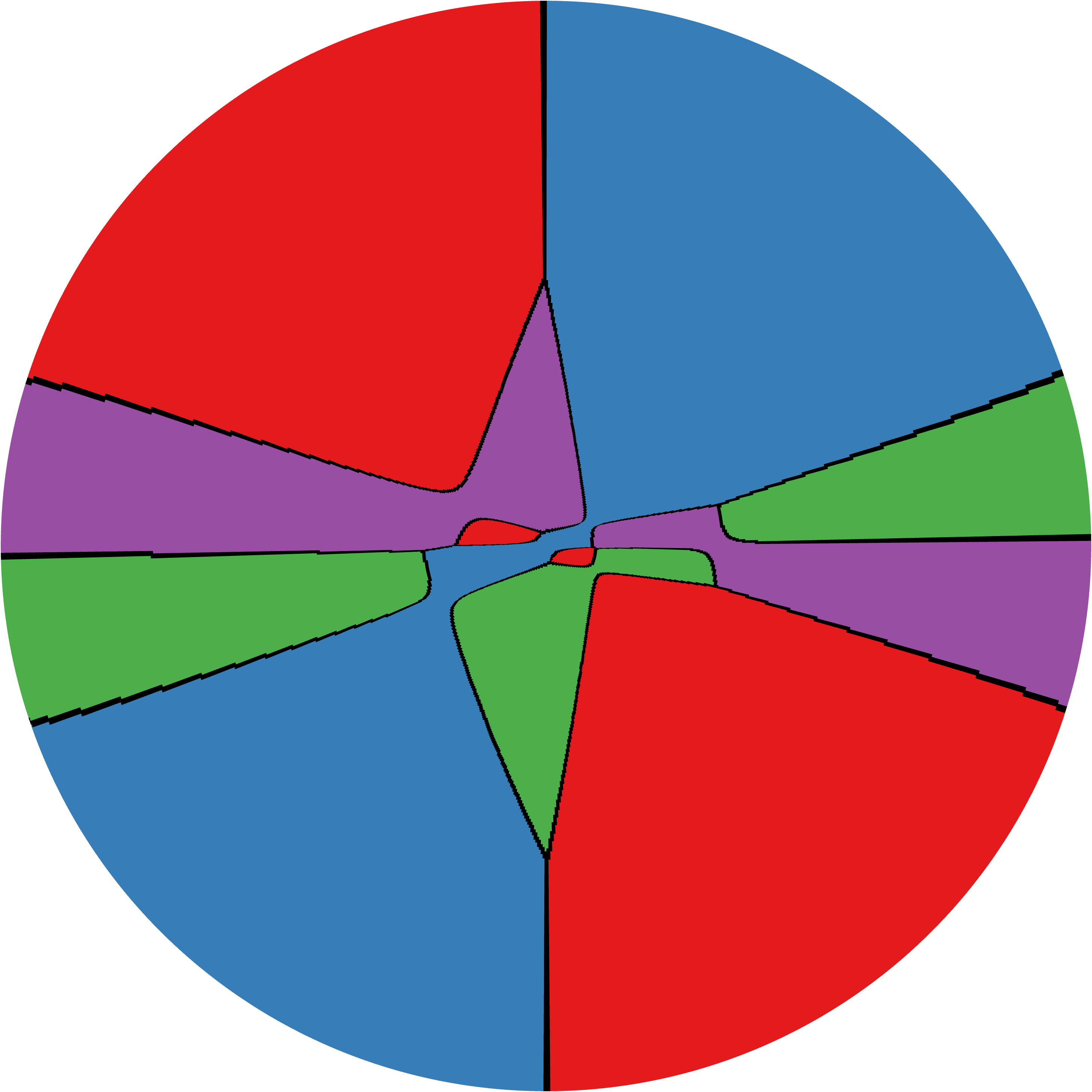}
    \end{minipage}
    \hfill
    \begin{minipage}[b]{0.48\textwidth}
        \centering
        \includegraphics[width=0.9\linewidth]{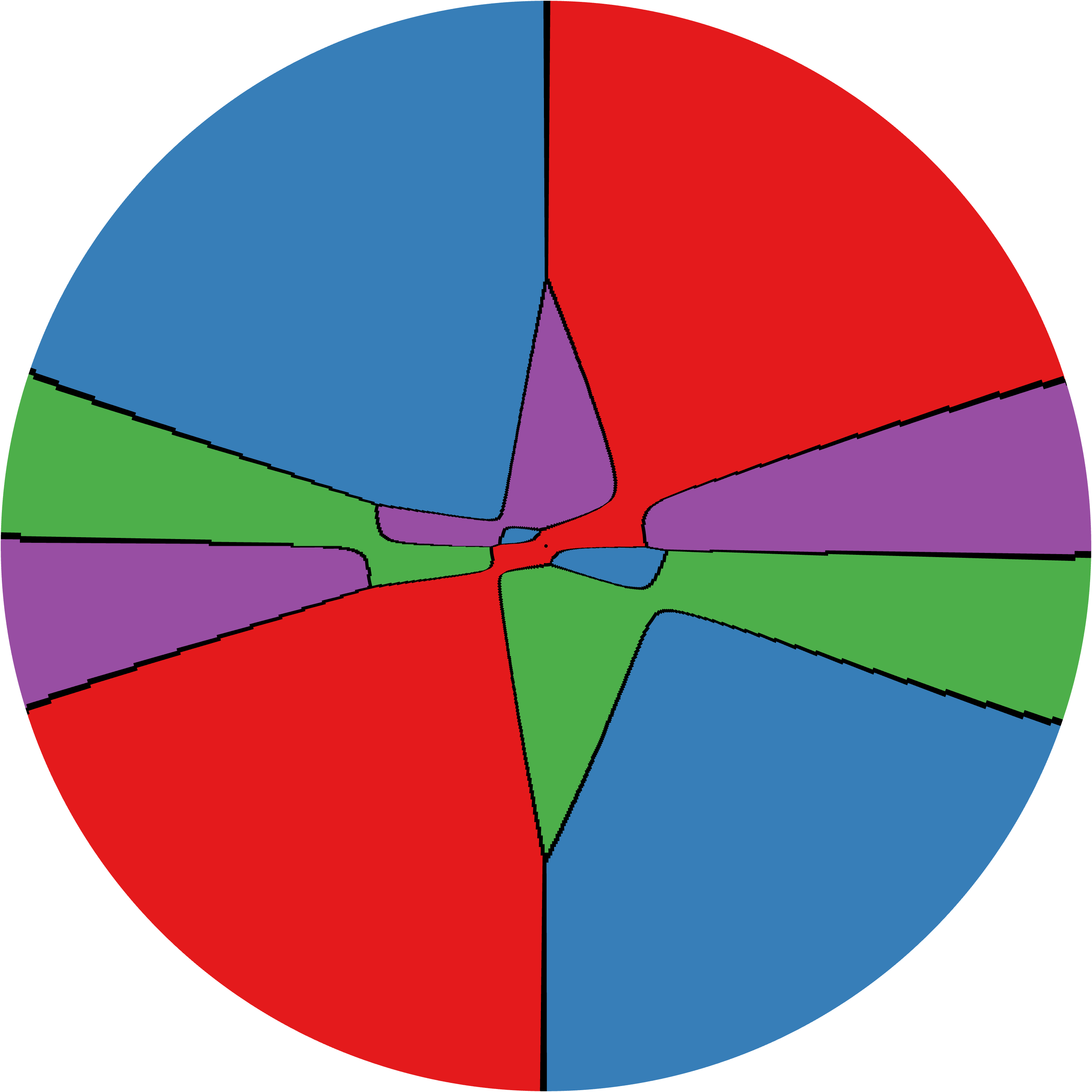}
     \end{minipage}
\caption{Top and bottom view of $\gmap(\fvd(L))$ of topology \Rom{12}. }\label{fig:gmapfvd12}
\end{figure}

\subsection{Topology \Rom{13}: 6 vertices}

The configuration tuple that induces topology \Rom{13} is shown in \cref{fig:topo13comb}. A set of lines that realizes this tuple, hence this topology, has the following parameters: $(a,b_{3},c_{3},d_{3},e_{3},b_{4},c_{4},d_{4},e_{4})=(5,7,-10,0,2,16,7,0,-16)$. \cref{fig:topo13verify} shows the six projected bisectors of these four lines, which match the configurations shown in \cref{fig:topo13comb}. The $\gmap(\fvd(L))$ is shown in \cref{fig:gmapfvd13}.

\begin{figure}[h]
        \centering
        \includegraphics[page=13,width=\textwidth]{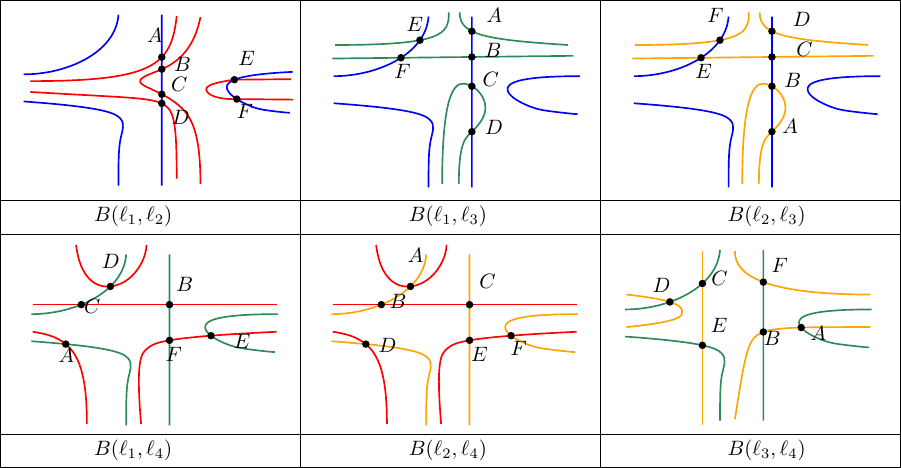}
        \caption{Combination that forms topology \Rom{13}. }\label{fig:topo13comb}
\end{figure}

\begin{figure}[!h]
\centering
\minipage{0.03\textwidth}
(a)	
\endminipage
\minipage{0.3\textwidth}
\includegraphics[width=\textwidth]{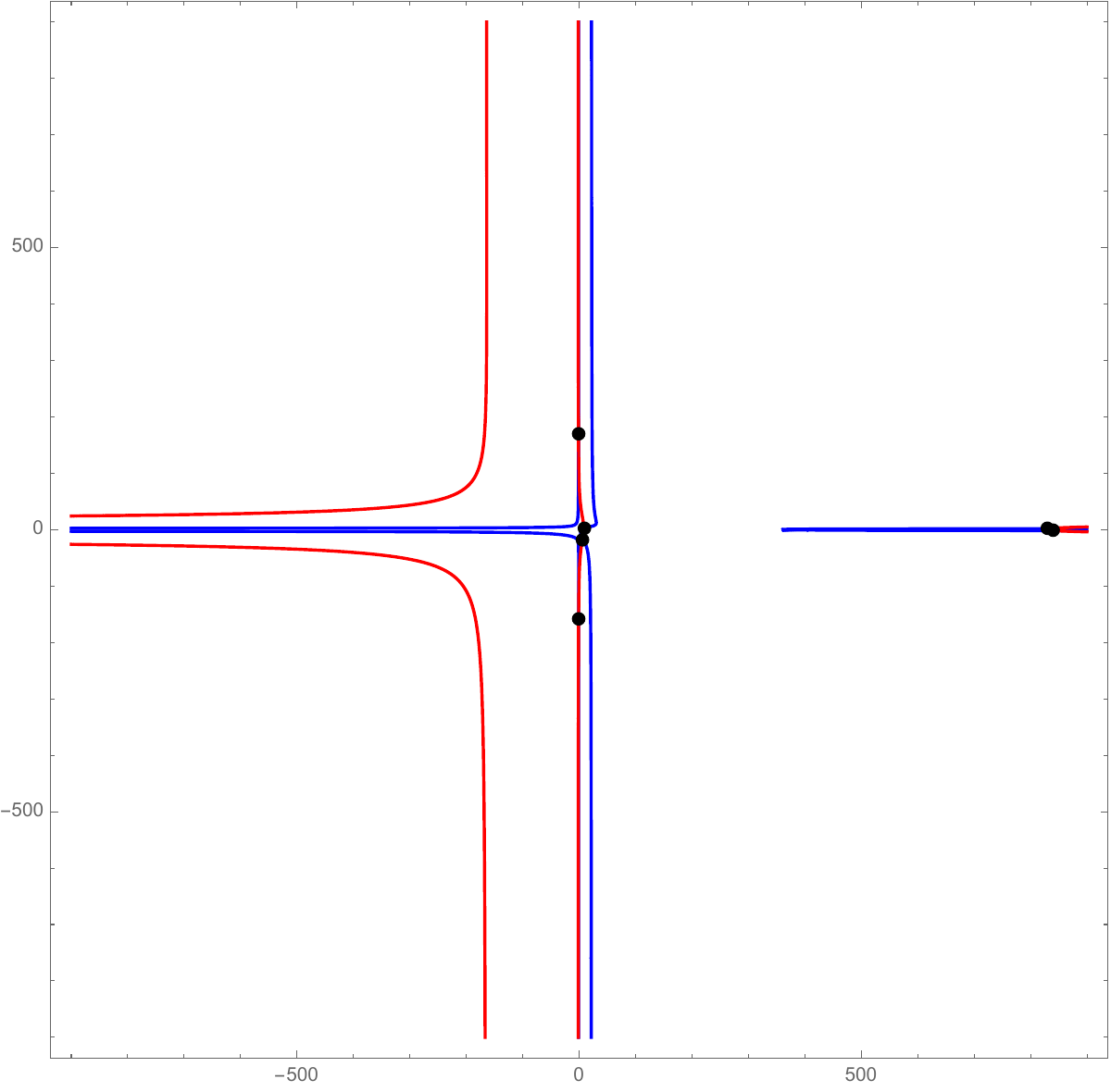}
\endminipage\hfill
\minipage{0.03\textwidth}
(b)	
\endminipage
\minipage{0.3\textwidth}
\includegraphics[width=\textwidth]{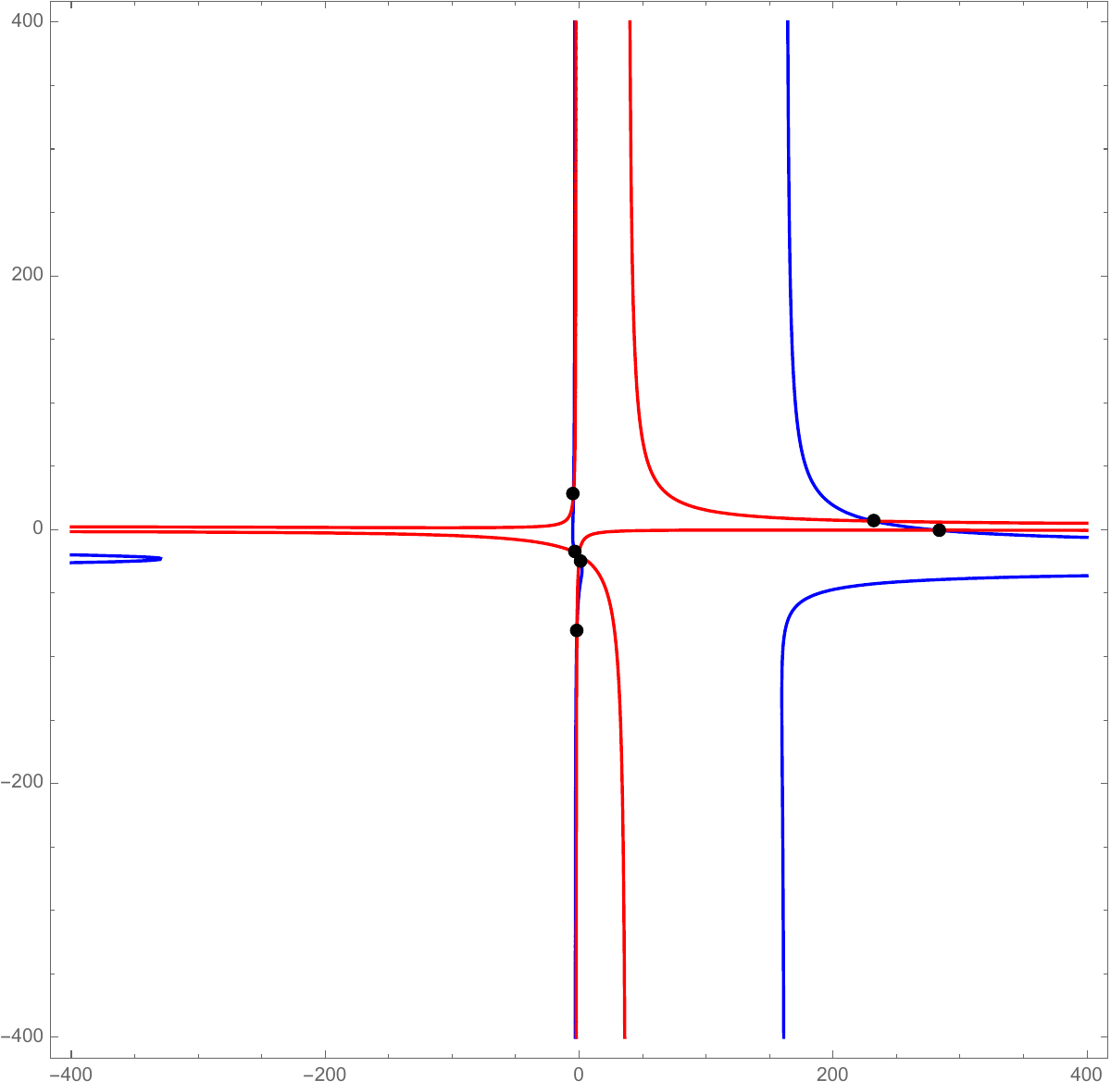}	
\endminipage\hfill
\minipage{0.03\textwidth}
(c)	
\endminipage
\minipage{0.3\textwidth}
\includegraphics[width=\textwidth]{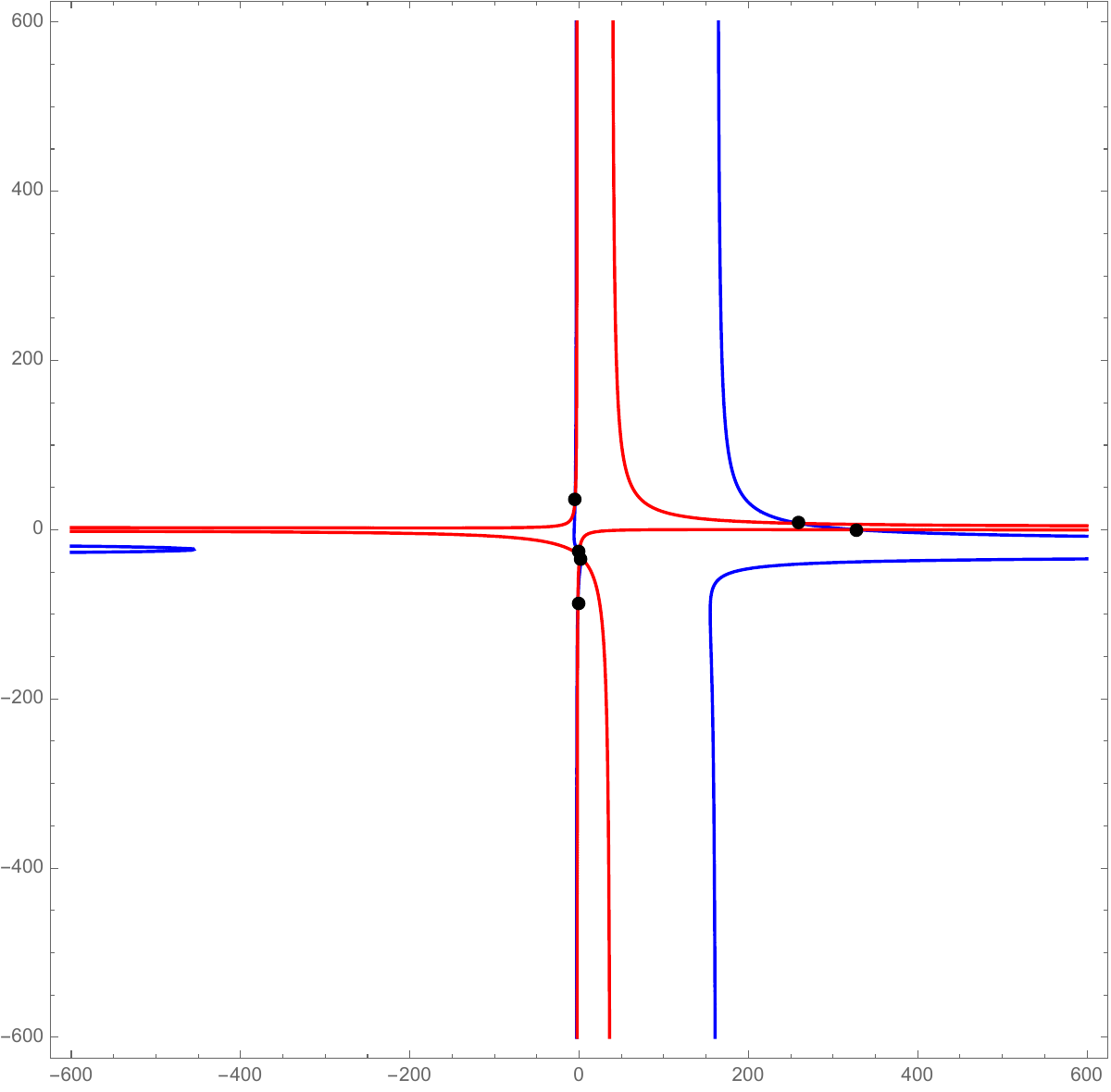}
\endminipage%
\newline
\centering
\minipage{0.03\textwidth}
(d)	
\endminipage
\minipage{0.3\textwidth}
\includegraphics[width=\textwidth]{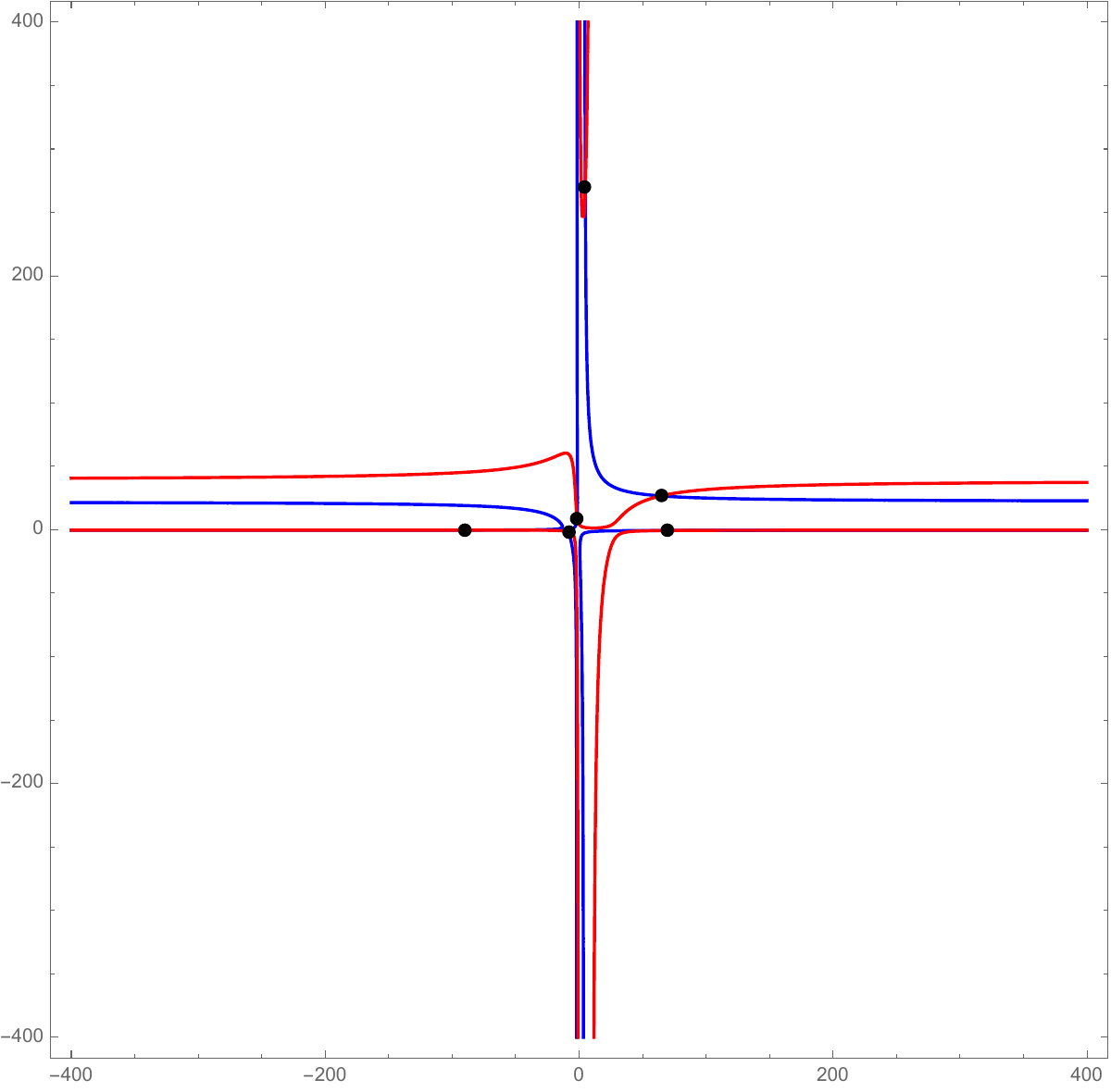}
\endminipage\hfill
\minipage{0.03\textwidth}
(e)	
\endminipage
\minipage{0.3\textwidth}
\includegraphics[width=\textwidth]{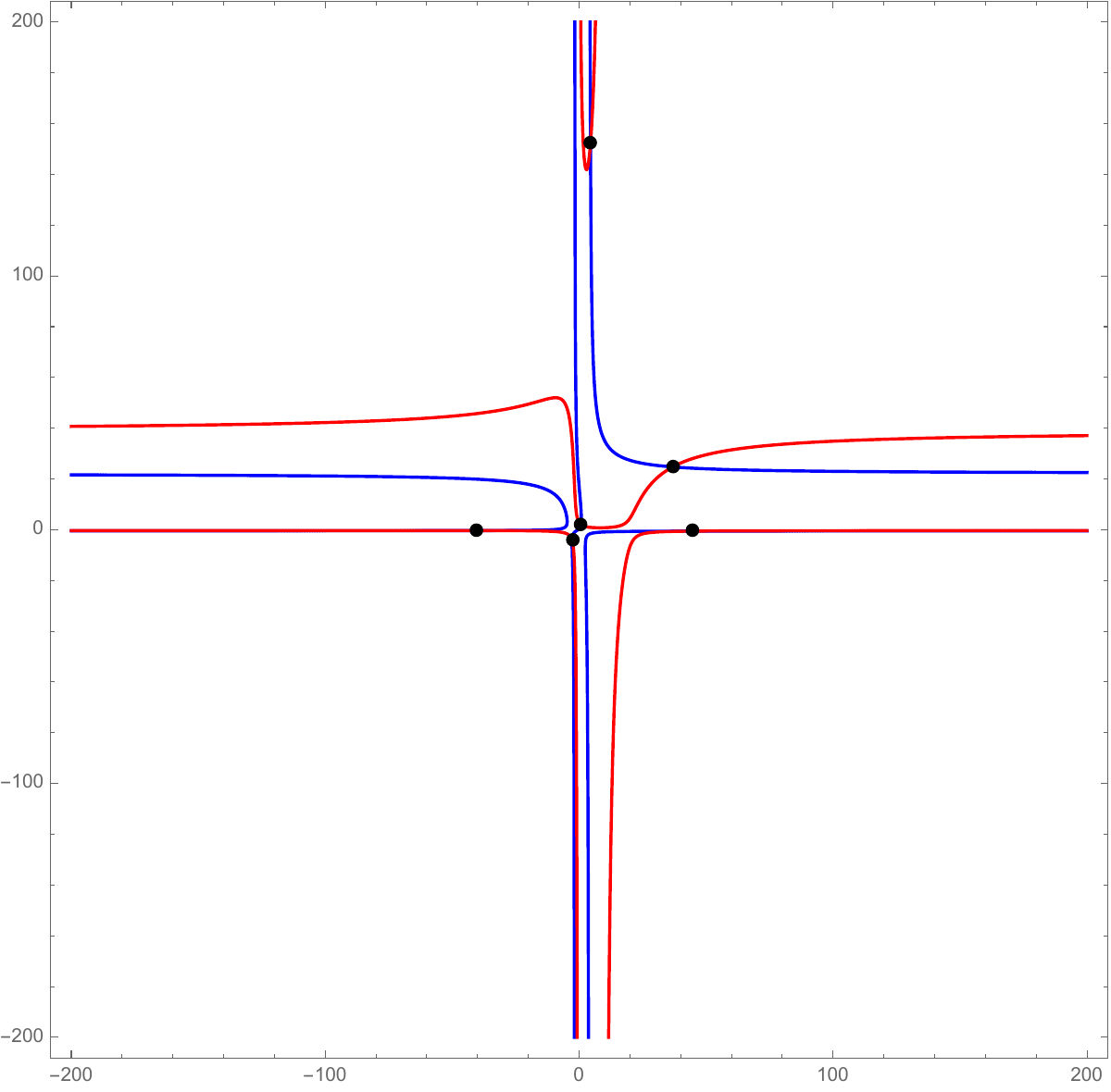}
\endminipage\hfill
\minipage{0.03\textwidth}
(f)	
\endminipage
\minipage{0.3\textwidth}
\includegraphics[width=\textwidth]{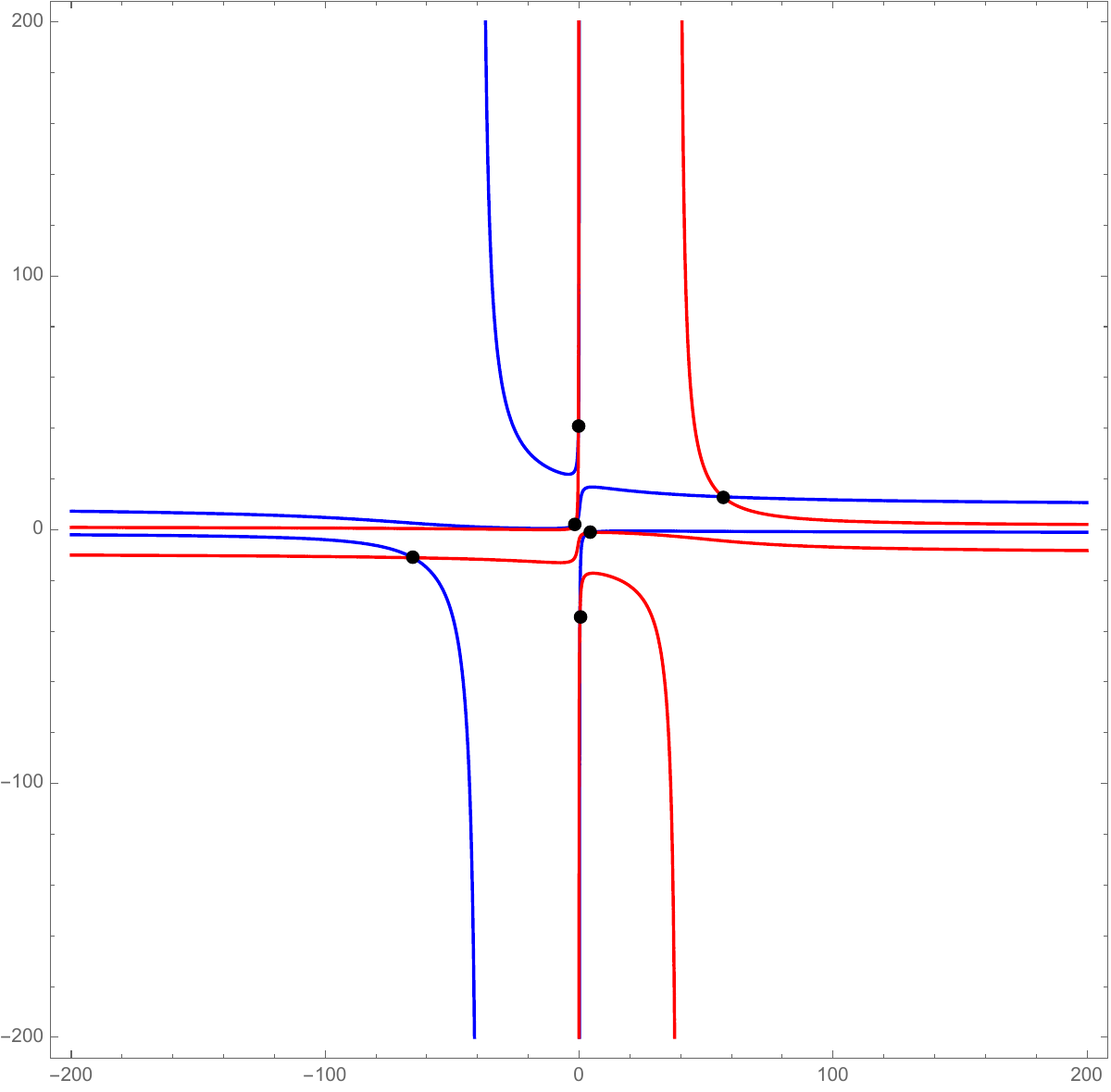}
\endminipage%
\caption{Six projected bisectors when the four lines have the following parameters: $(a,b_{3},c_{3},d_{3},e_{3},b_{4},c_{4},d_{4},e_{4})=(5,7,-10,0,2,16,7,0,-16)$. They match the configurations shown in \cref{fig:topo13comb} and form topology \Rom{13}. }\label{fig:topo13verify}
\end{figure}

\begin{figure}[h]
    \centering
    \begin{minipage}[b]{0.48\textwidth}
        \centering
        \includegraphics[width=0.9\linewidth]{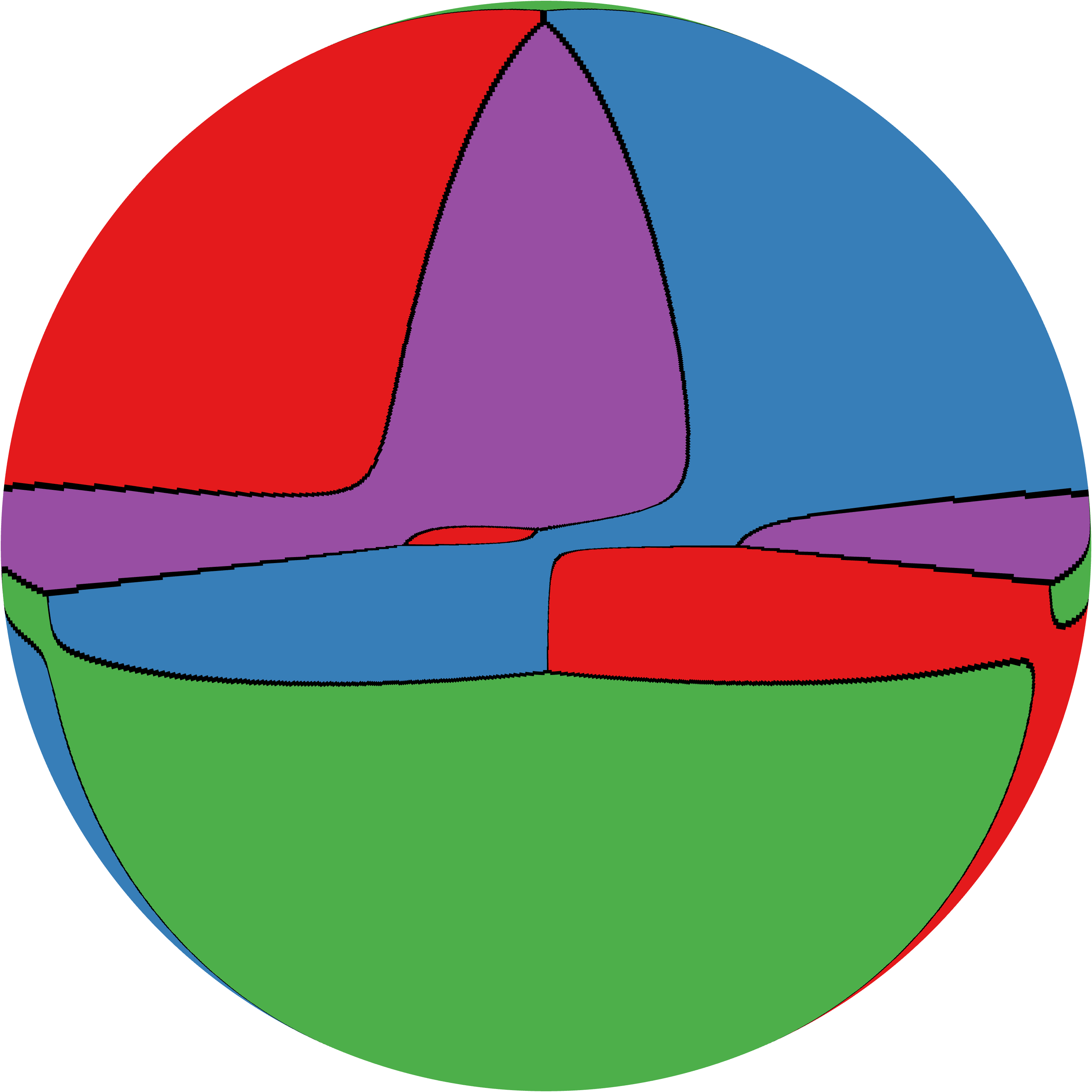}
    \end{minipage}
    \hfill
    \begin{minipage}[b]{0.48\textwidth}
        \centering
        \includegraphics[width=0.9\linewidth]{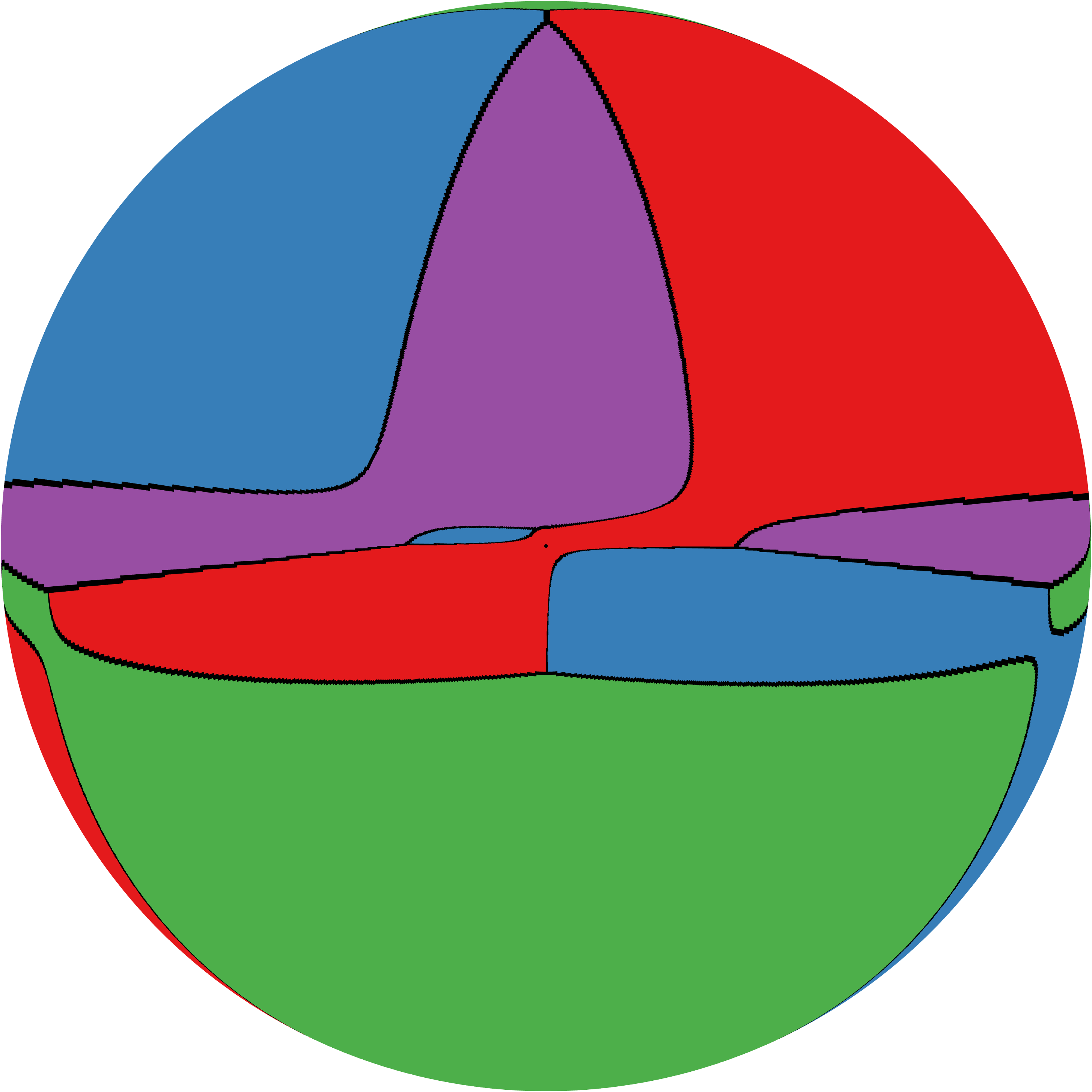}
     \end{minipage}
\caption{Top and bottom view of $\gmap(\fvd(L))$ of topology \Rom{13}. }\label{fig:gmapfvd13}
\end{figure}

\subsection{Topology \Rom{14}: 6 vertices}

The configuration tuple that induces topology \Rom{14} is shown in \cref{fig:topo14comb}. A set of lines that realizes this tuple, hence this topology, has the following parameters: $(a,b_{3},c_{3},d_{3},e_{3},b_{4},c_{4},d_{4},e_{4})=(3,17,9,0,19,11,-16,2,-8)$. \cref{fig:topo14verify} shows the six projected bisectors of these four lines, which match the configurations shown in \cref{fig:topo14comb}. The $\gmap(\fvd(L))$ is shown in \cref{fig:gmapfvd14}.

\begin{figure}[h]
        \centering
        \includegraphics[page=14,width=\textwidth]{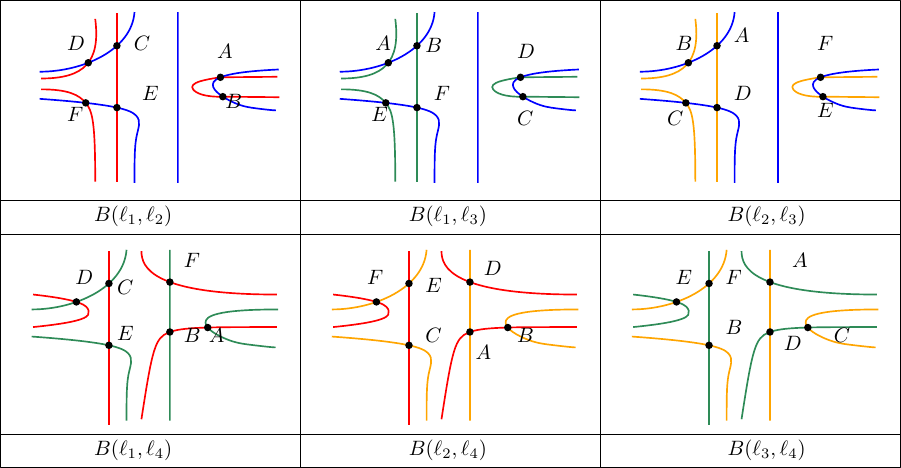}
        \caption{Combination that forms topology \Rom{14}. }\label{fig:topo14comb}
\end{figure}

\begin{figure}[!h]
\centering
\minipage{0.03\textwidth}
(a)	
\endminipage
\minipage{0.3\textwidth}
\includegraphics[width=\textwidth]{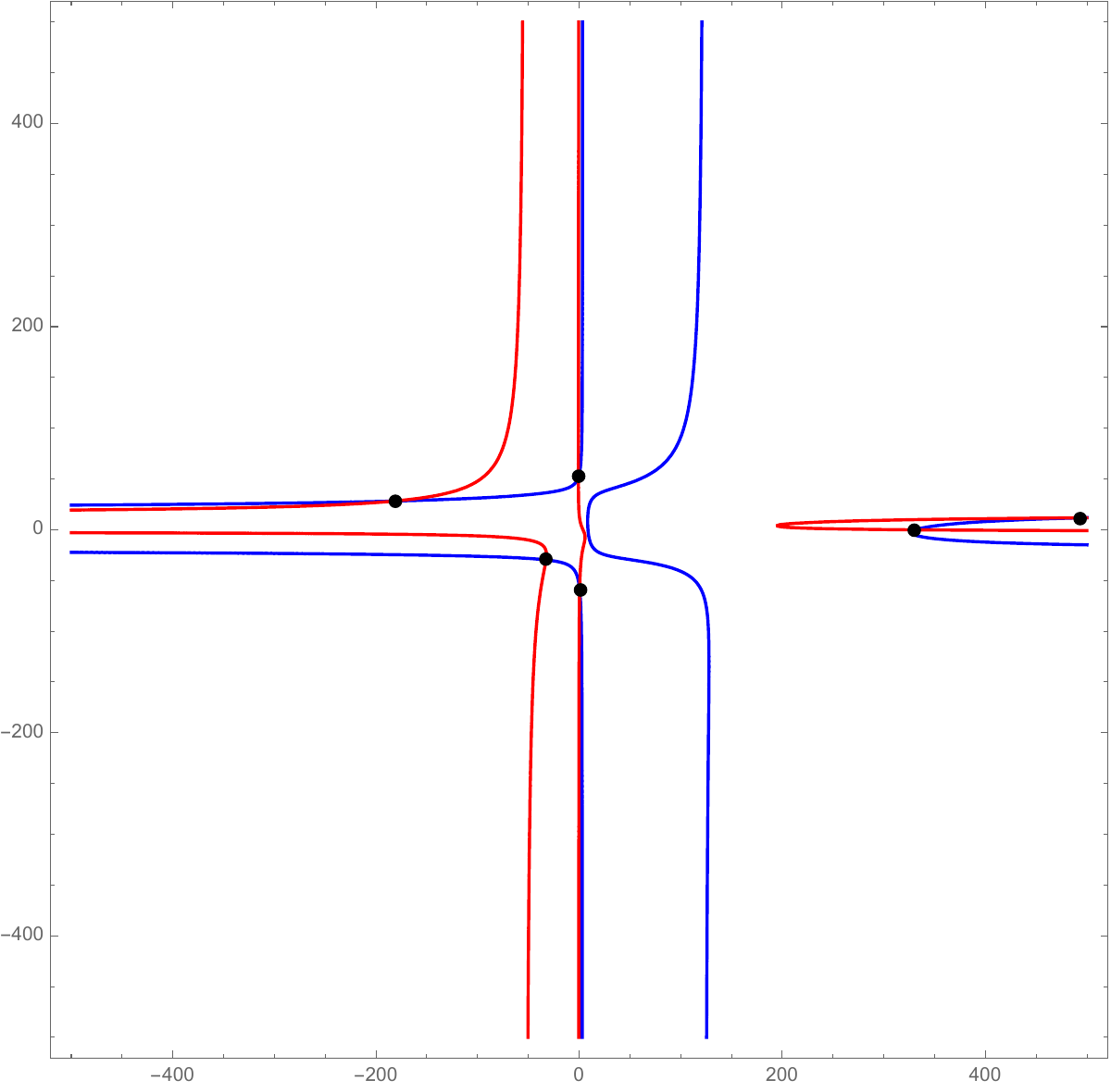}
\endminipage\hfill
\minipage{0.03\textwidth}
(b)	
\endminipage
\minipage{0.3\textwidth}
\includegraphics[width=\textwidth]{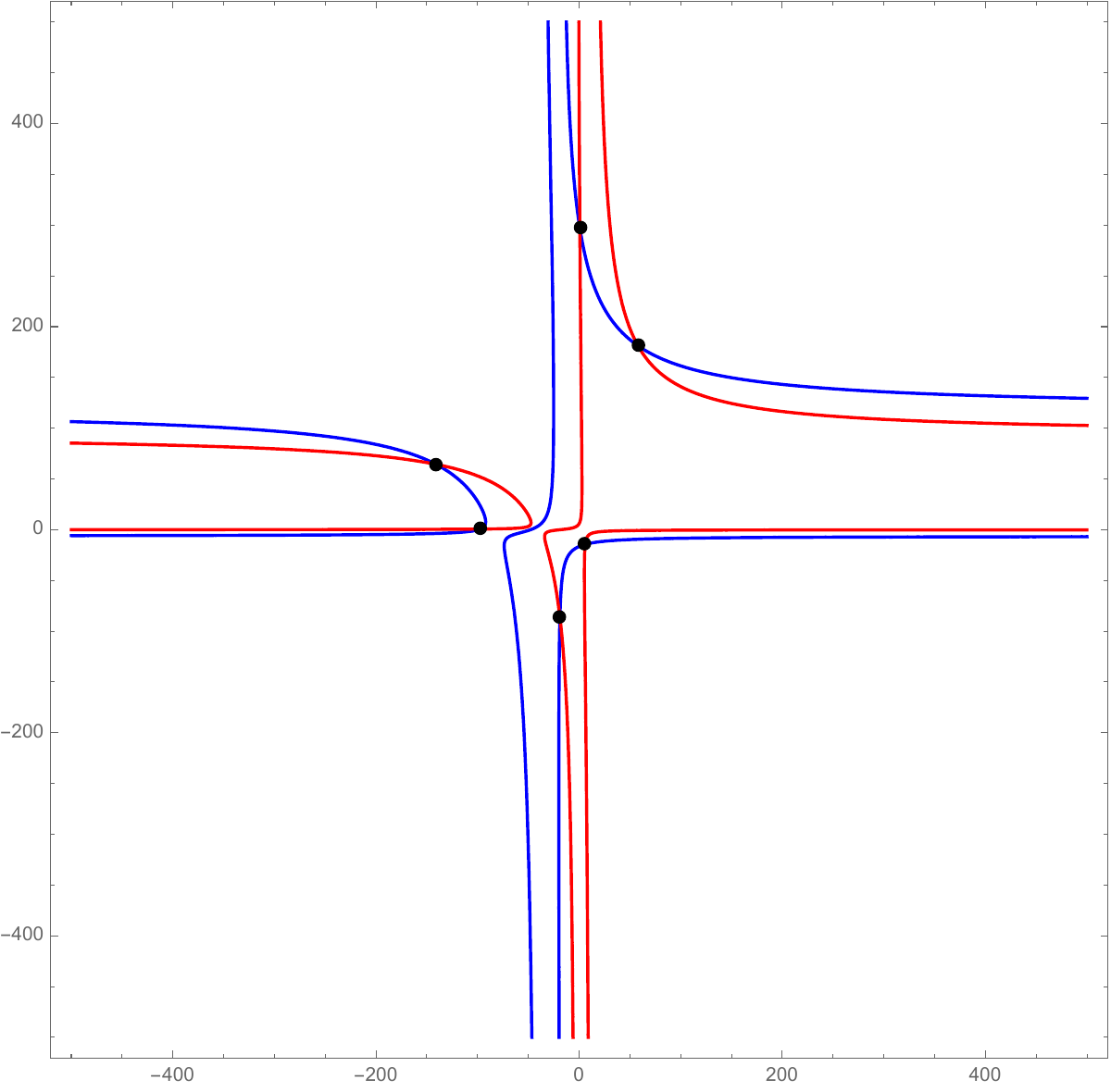}	
\endminipage\hfill
\minipage{0.03\textwidth}
(c)	
\endminipage
\minipage{0.3\textwidth}
\includegraphics[width=\textwidth]{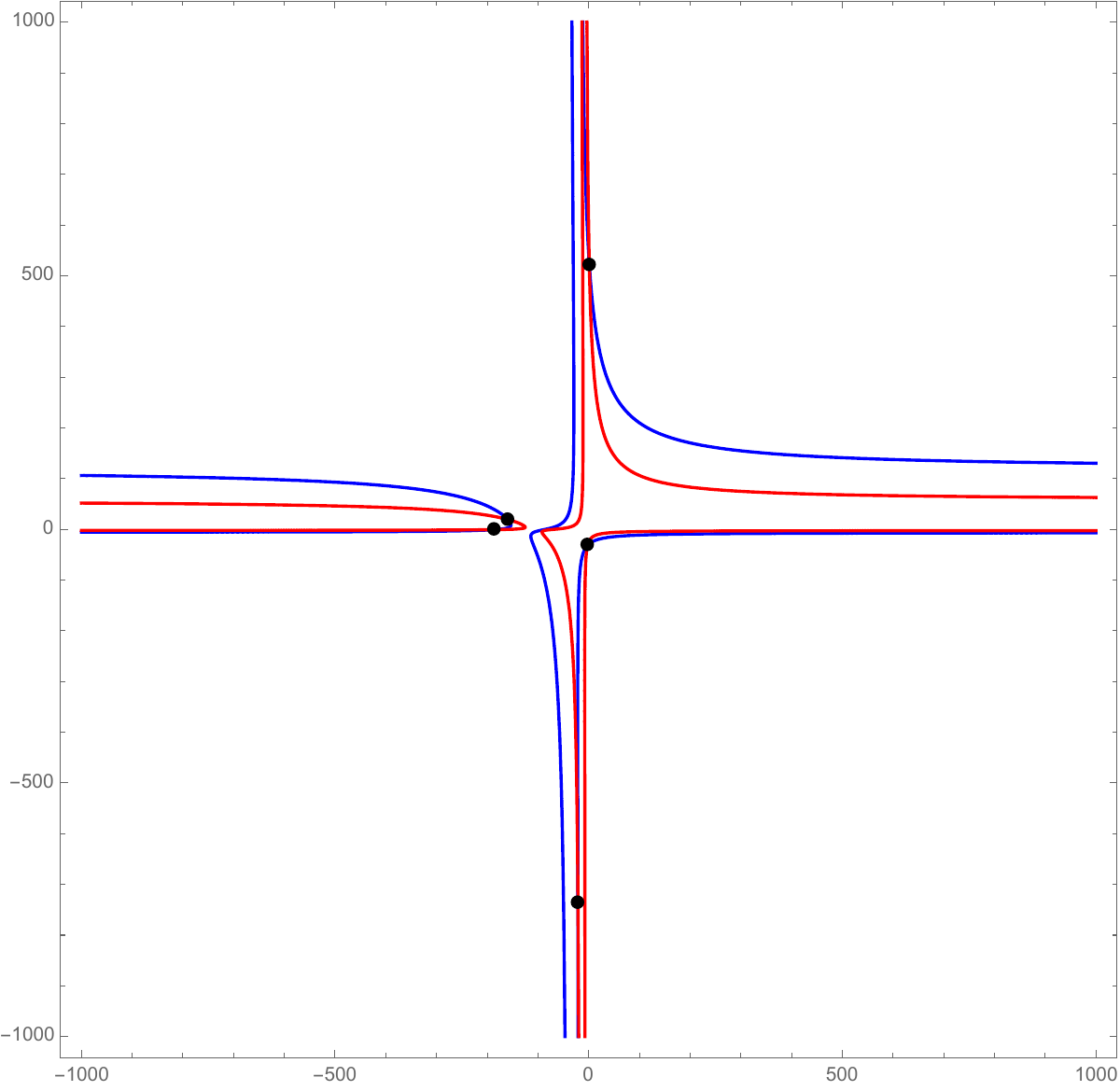}
\endminipage%
\newline
\centering
\minipage{0.03\textwidth}
(d)	
\endminipage
\minipage{0.3\textwidth}
\includegraphics[width=\textwidth]{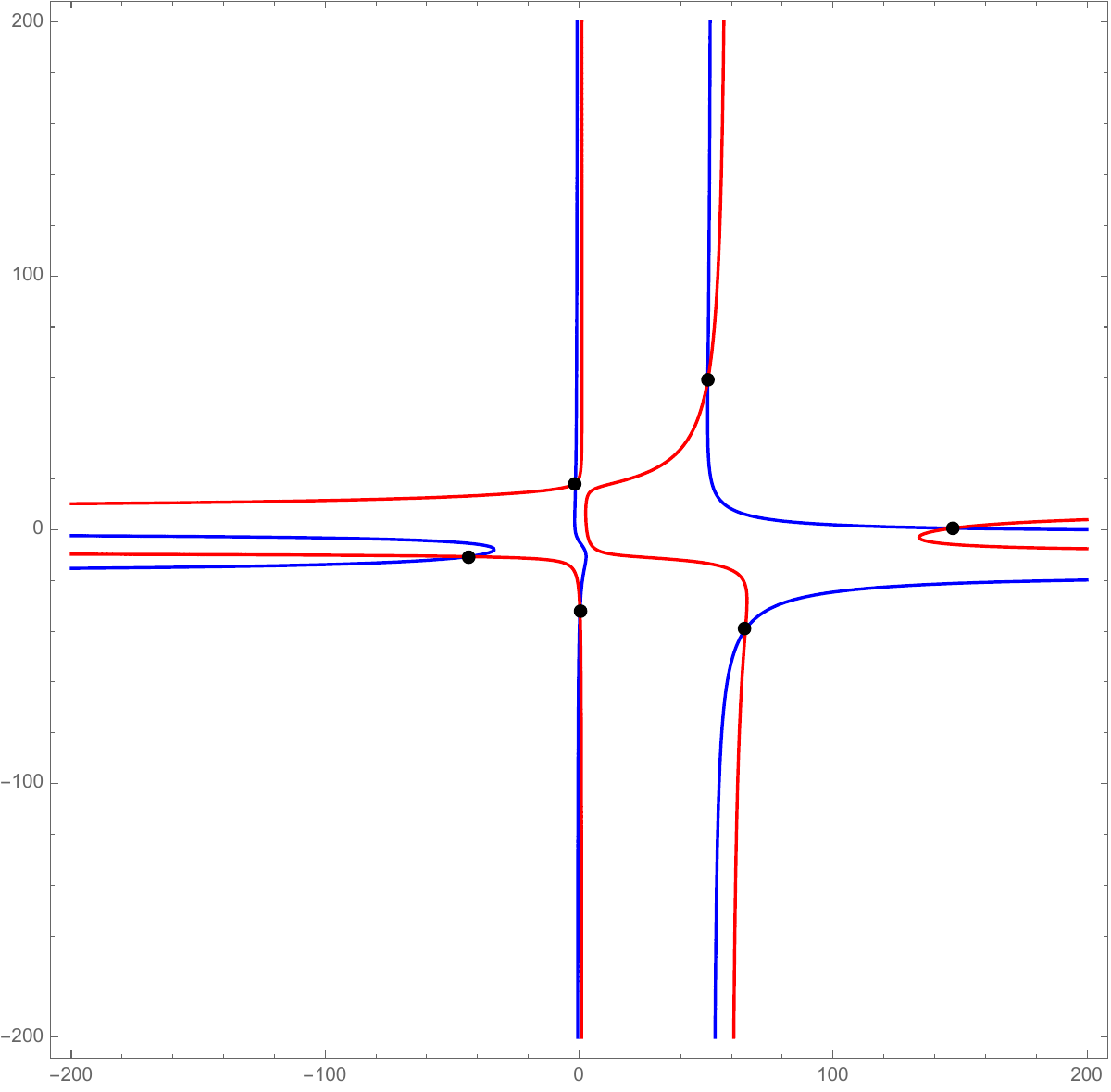}
\endminipage\hfill
\minipage{0.03\textwidth}
(e)	
\endminipage
\minipage{0.3\textwidth}
\includegraphics[width=\textwidth]{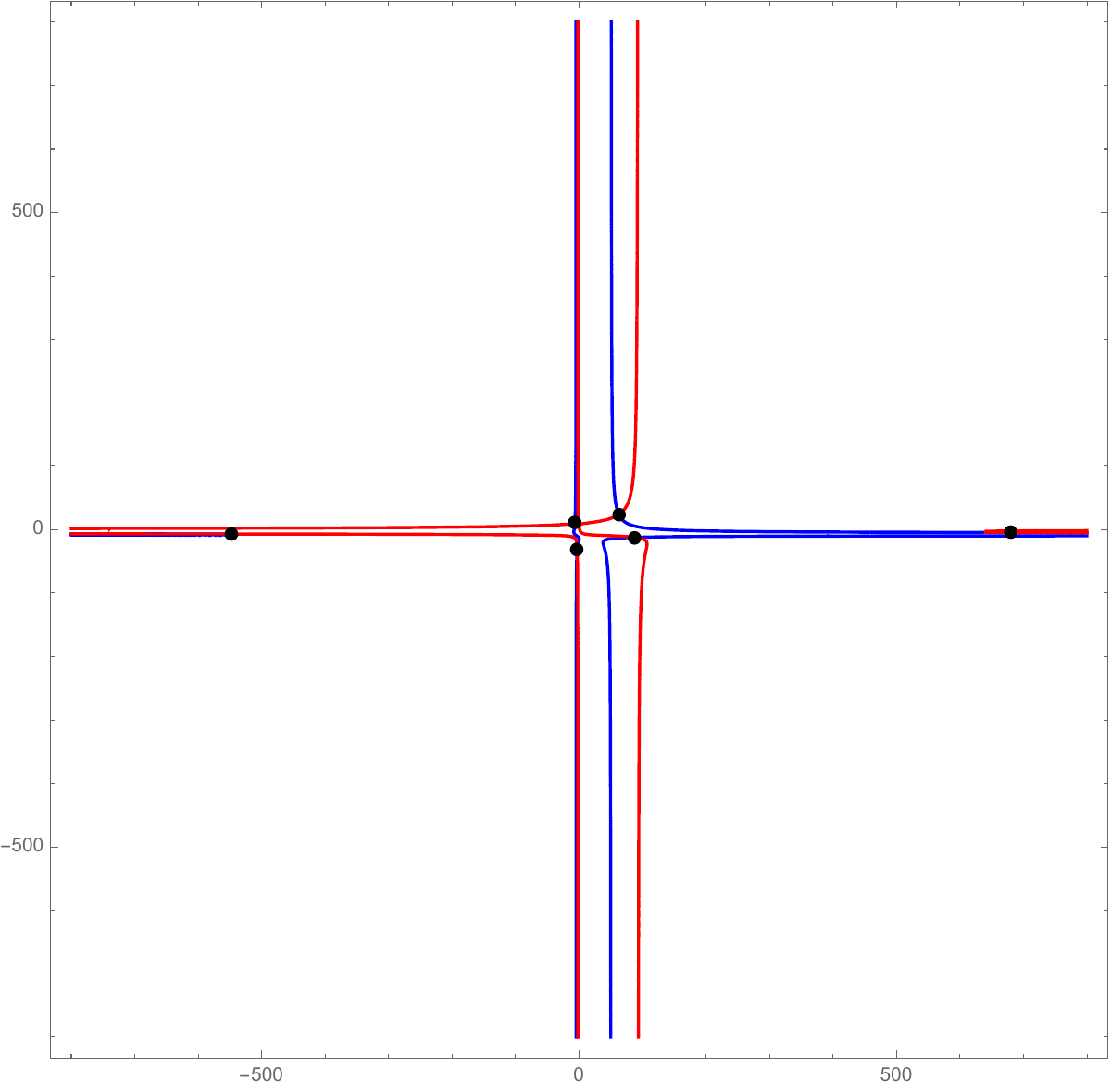}
\endminipage\hfill
\minipage{0.03\textwidth}
(f)	
\endminipage
\minipage{0.3\textwidth}
\includegraphics[width=\textwidth]{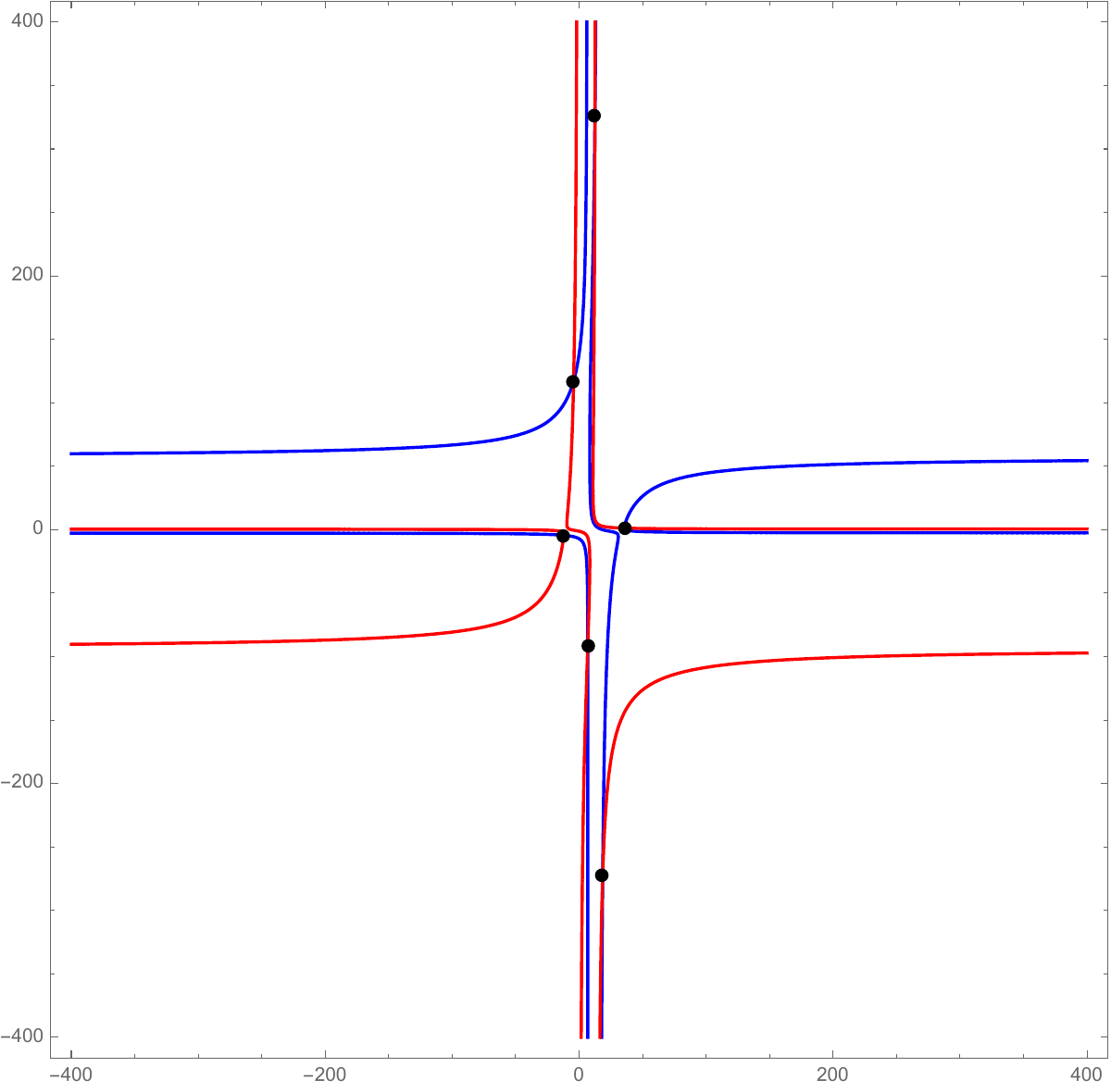}
\endminipage%
\caption{Six projected bisectors when the four lines have the following parameters: $(a,b_{3},c_{3},d_{3},e_{3},b_{4},c_{4},d_{4},e_{4})=(3,17,9,0,19,11,-16,2,-8)$. They match the configurations shown in \cref{fig:topo14comb} and form topology \Rom{14}. }\label{fig:topo14verify}
\end{figure}

\begin{figure}[h]
    \centering
    \begin{minipage}[b]{0.48\textwidth}
        \centering
        \includegraphics[width=0.9\linewidth]{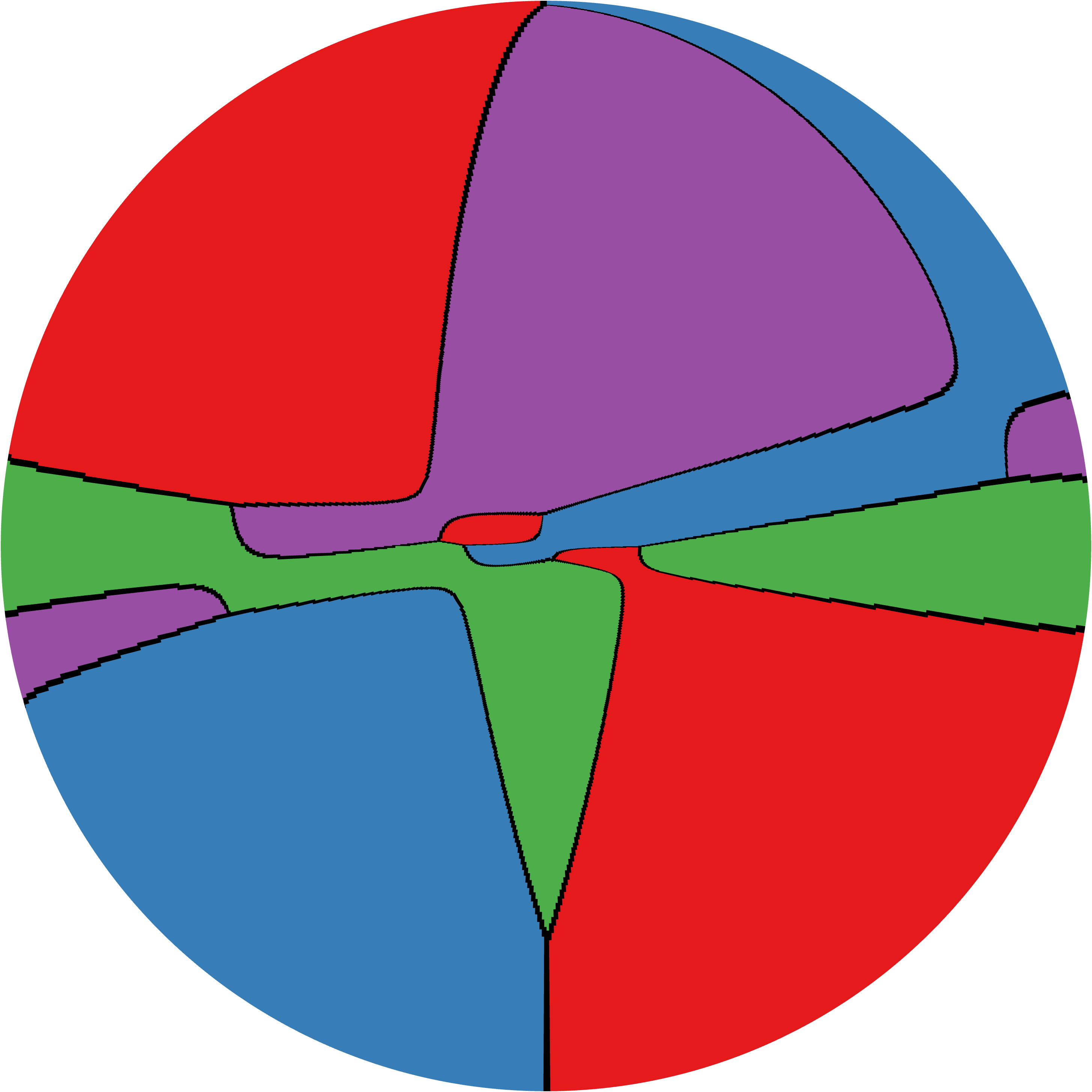}
    \end{minipage}
    \hfill
    \begin{minipage}[b]{0.48\textwidth}
        \centering
        \includegraphics[width=0.9\linewidth]{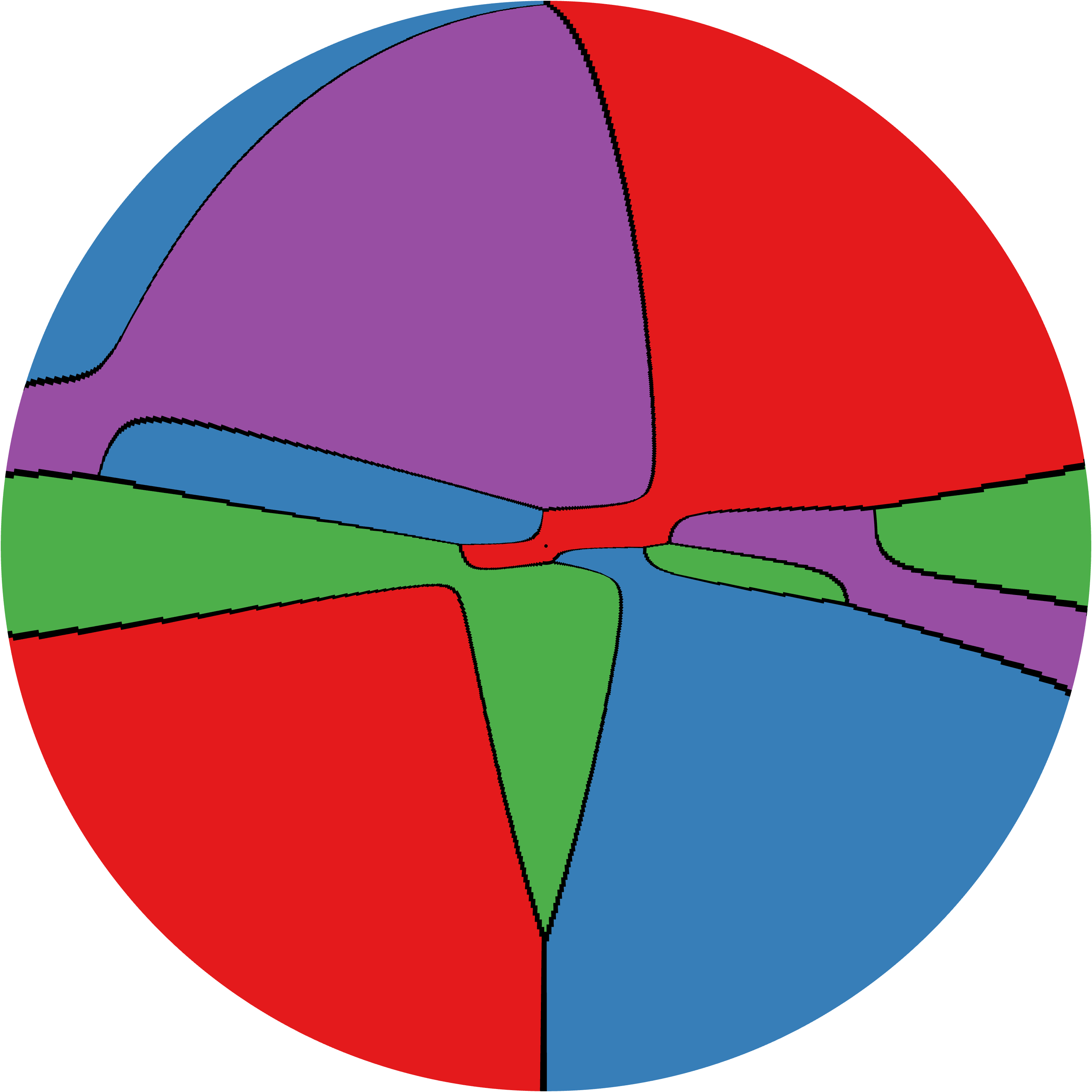}
     \end{minipage}
\caption{Top and bottom view of $\gmap(\fvd(L))$ of topology \Rom{14}. }\label{fig:gmapfvd14}
\end{figure}

\subsection{Topology \Rom{15}: 6 vertices}

The configuration tuple that induces topology \Rom{15} is shown in \cref{fig:topo15comb}. A set of lines that realizes this tuple, hence this topology, has the following parameters: $(a,b_{3},c_{3},d_{3},e_{3},b_{4},c_{4},d_{4},e_{4})=(9,32,-26,14,38,11,-40,12,-15)$. \cref{fig:topo15verify} shows the six projected bisectors of these four lines, which match the configurations shown in \cref{fig:topo15comb}. The $\gmap(\fvd(L))$ is shown in \cref{fig:gmapfvd15}.

\begin{figure}[h]
        \centering
        \includegraphics[page=15,width=\textwidth]{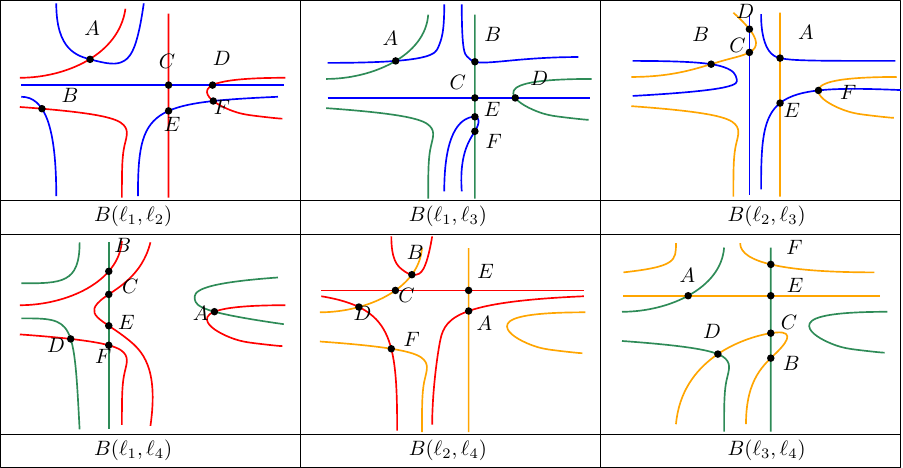}
        \caption{Combination that forms topology \Rom{15}. }\label{fig:topo15comb}
\end{figure}

\begin{figure}[!h]
\centering
\minipage{0.03\textwidth}
(a)	
\endminipage
\minipage{0.3\textwidth}
\includegraphics[width=\textwidth]{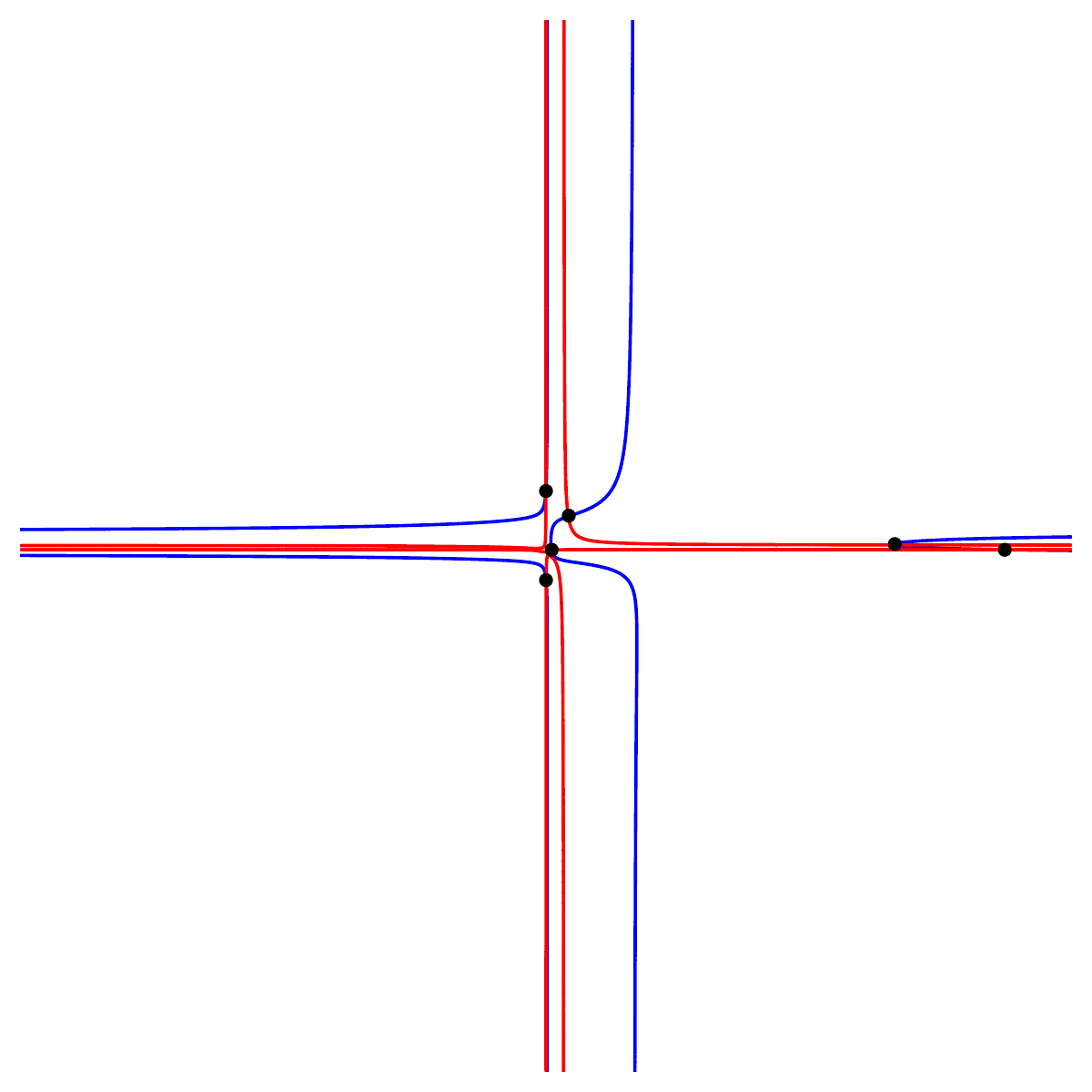}
\endminipage\hfill
\minipage{0.03\textwidth}
(b)	
\endminipage
\minipage{0.3\textwidth}
\includegraphics[width=\textwidth]{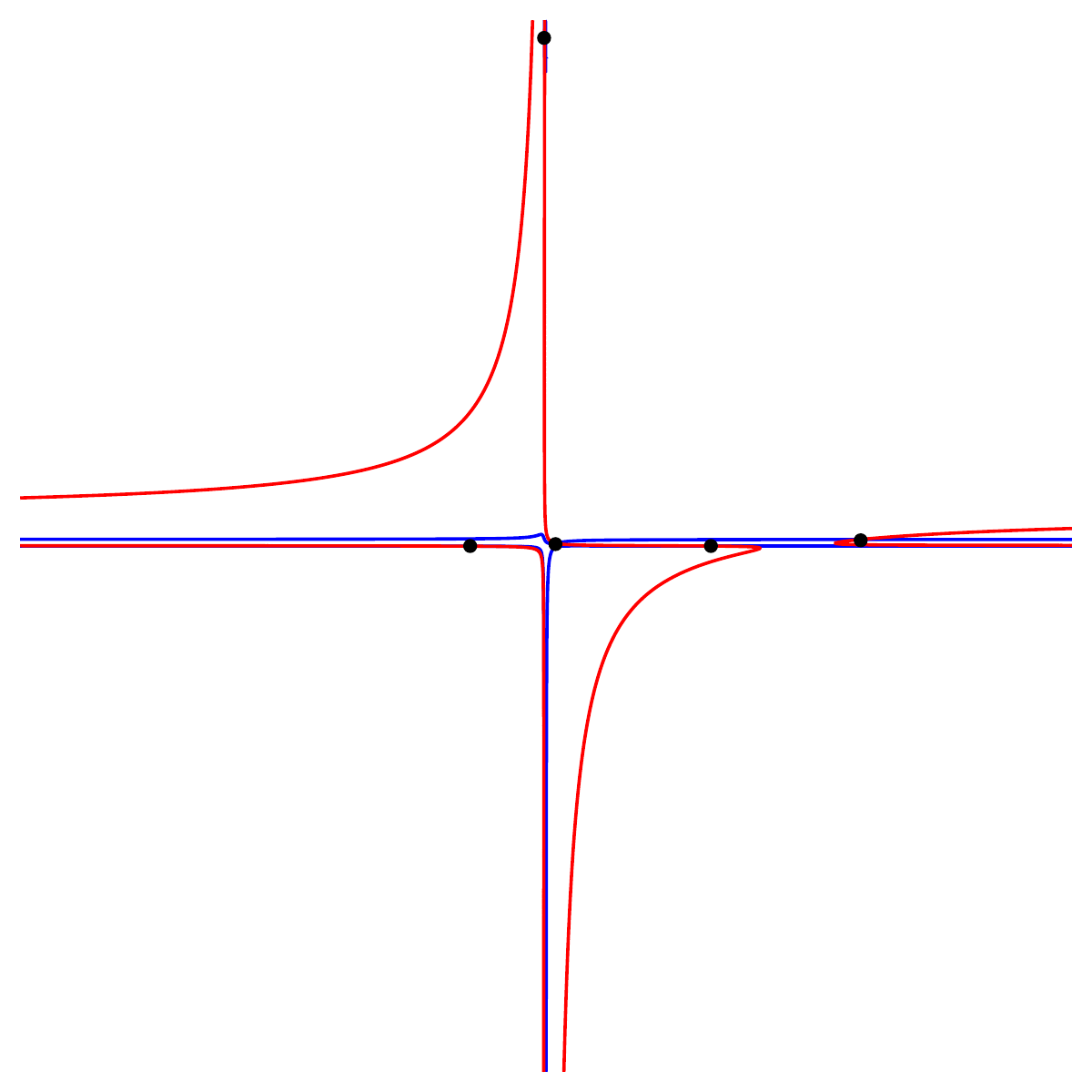}	
\endminipage\hfill
\minipage{0.03\textwidth}
(c)	
\endminipage
\minipage{0.3\textwidth}
\includegraphics[width=\textwidth]{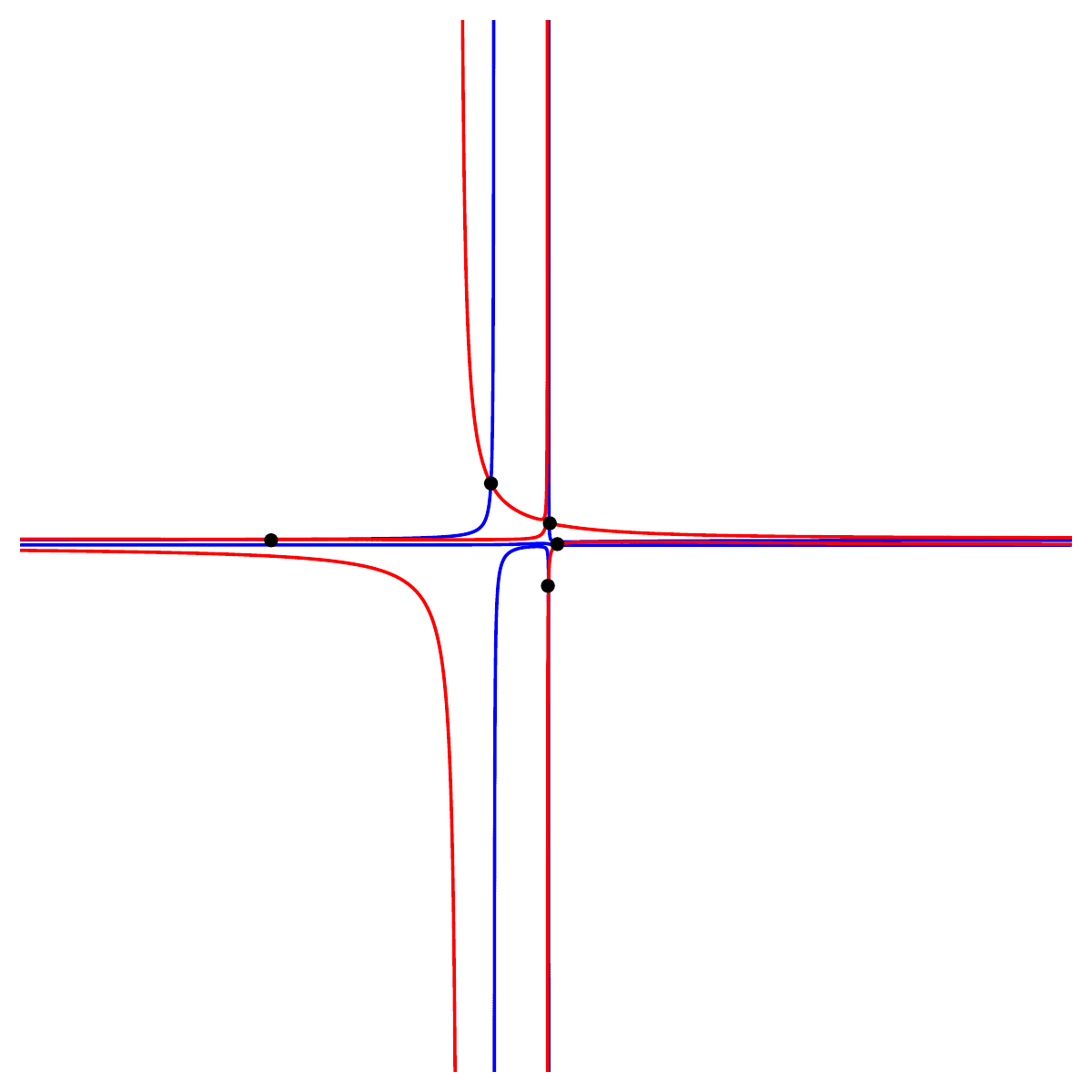}
\endminipage%
\newline
\centering
\minipage{0.03\textwidth}
(d)	
\endminipage
\minipage{0.3\textwidth}
\includegraphics[width=\textwidth]{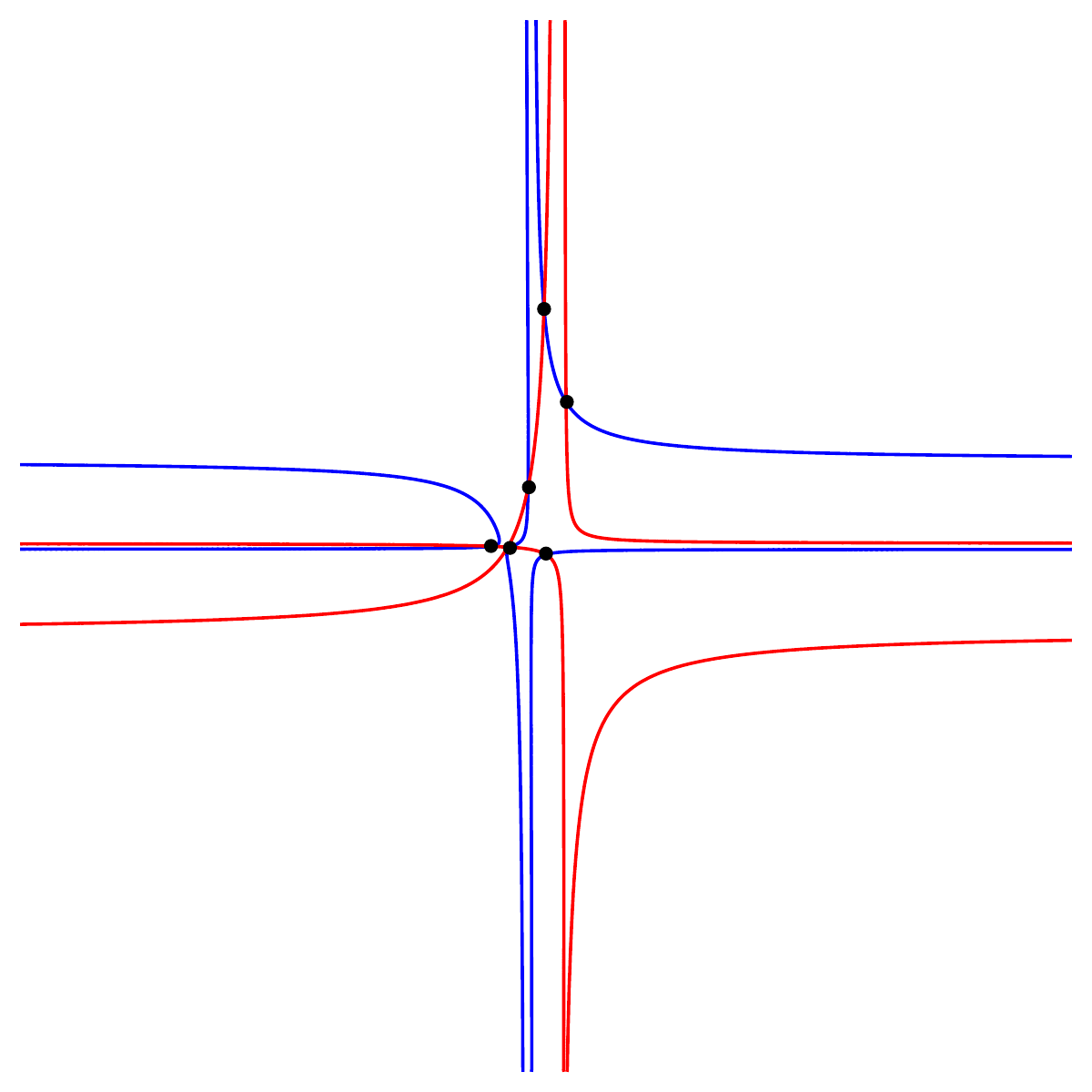}
\endminipage\hfill
\minipage{0.03\textwidth}
(e)	
\endminipage
\minipage{0.3\textwidth}
\includegraphics[width=\textwidth]{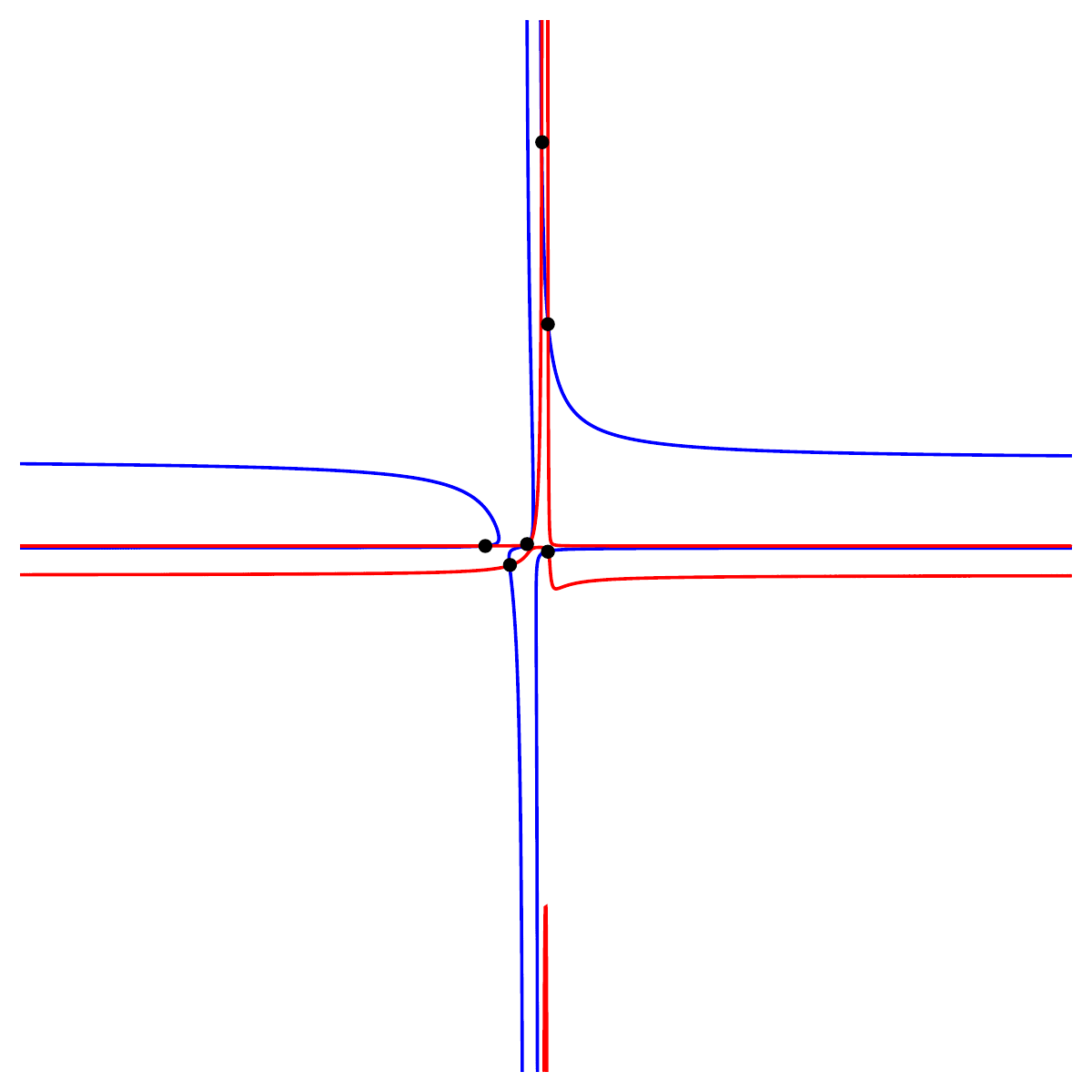}
\endminipage\hfill
\minipage{0.03\textwidth}
(f)	
\endminipage
\minipage{0.3\textwidth}
\includegraphics[width=\textwidth]{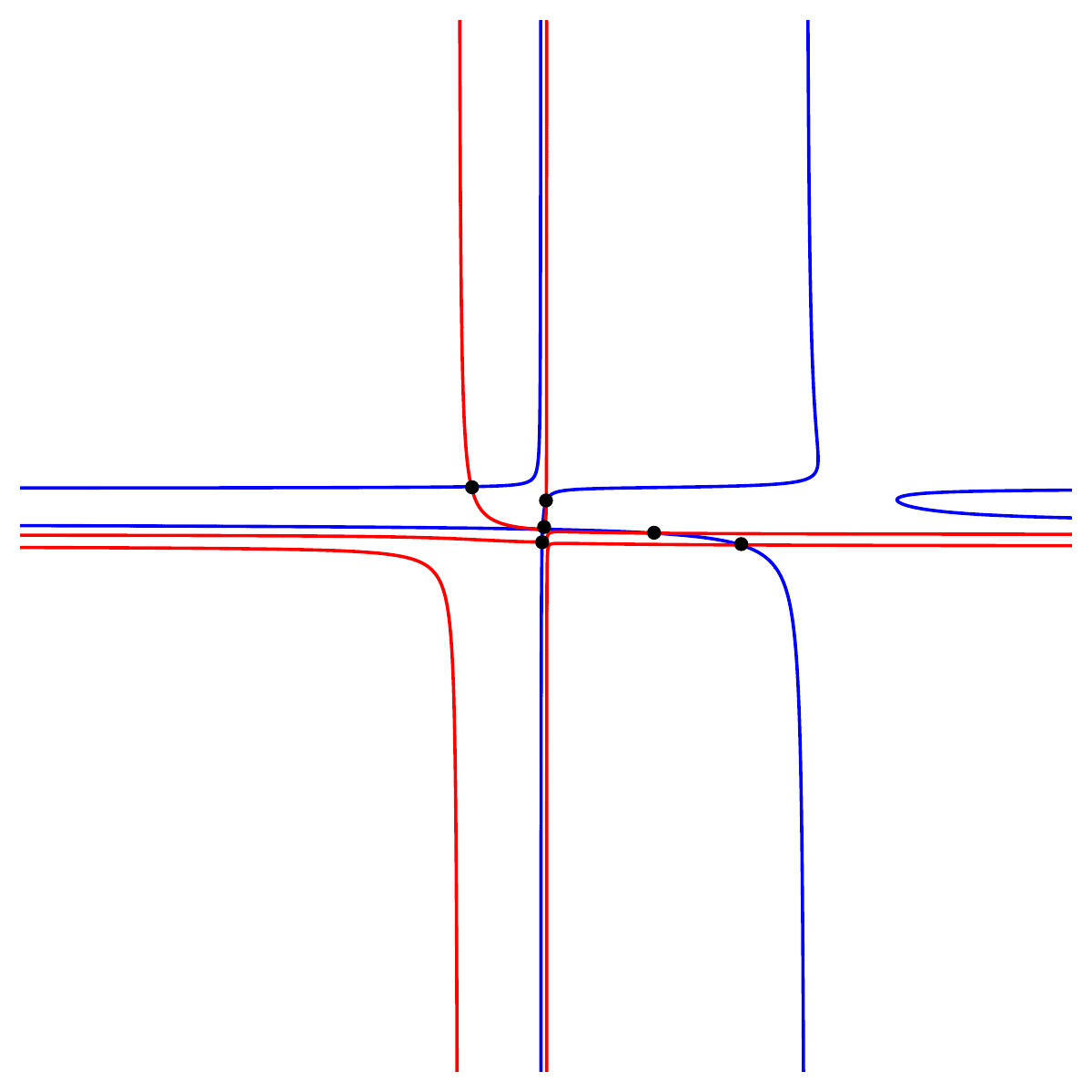}
\endminipage%
\caption{Six projected bisectors when the four lines have the following parameters: $(a,b_{3},c_{3},d_{3},e_{3},b_{4},c_{4},d_{4},e_{4})=(9,32,-26,14,38,11,-40,12,-15)$. They match the configurations shown in \cref{fig:topo15comb} and form topology \Rom{15}. }\label{fig:topo15verify}
\end{figure}

\begin{figure}[h]
    \centering
    \begin{minipage}[b]{0.48\textwidth}
        \centering
        \includegraphics[width=0.9\linewidth]{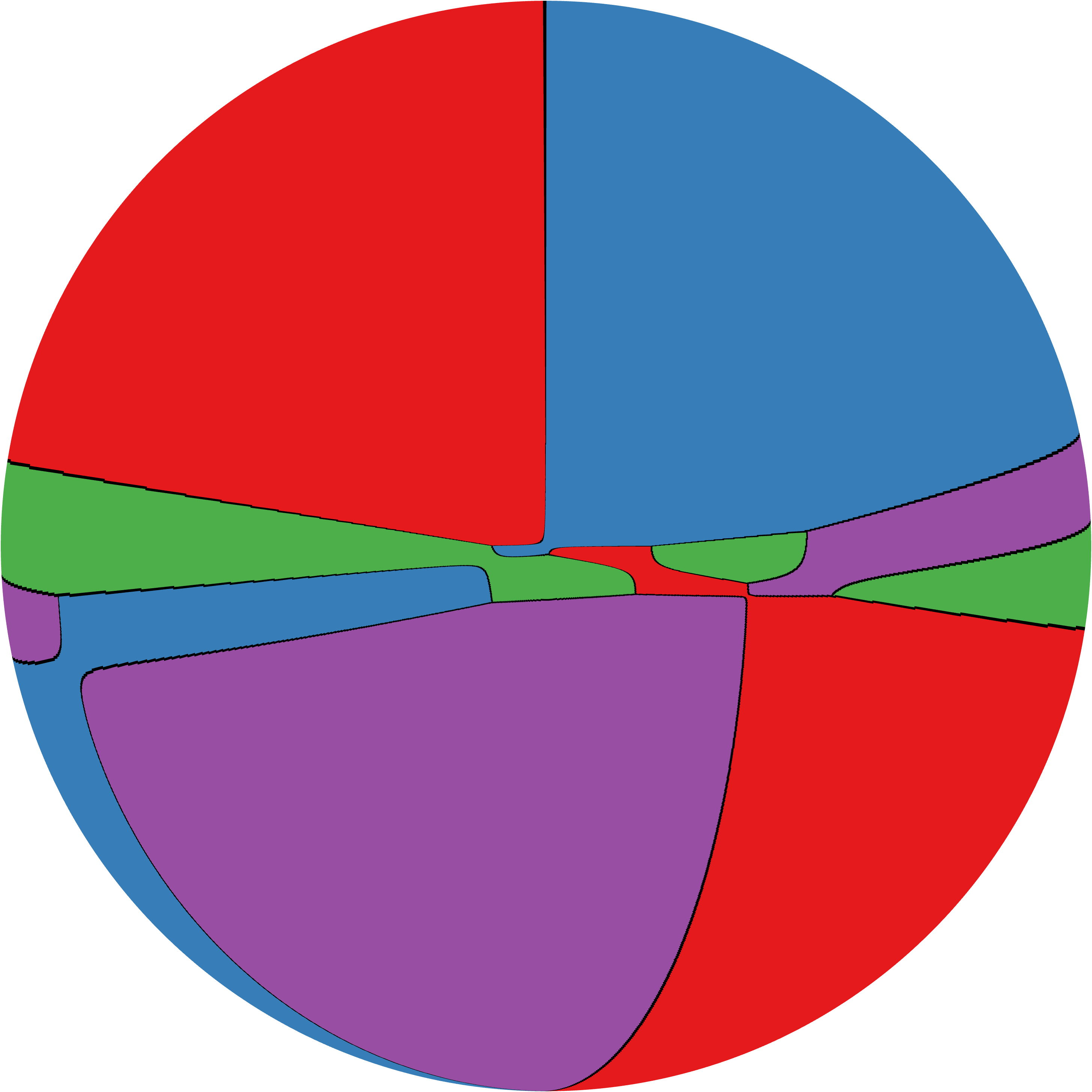}
    \end{minipage}
    \hfill
    \begin{minipage}[b]{0.48\textwidth}
        \centering
        \includegraphics[width=0.9\linewidth]{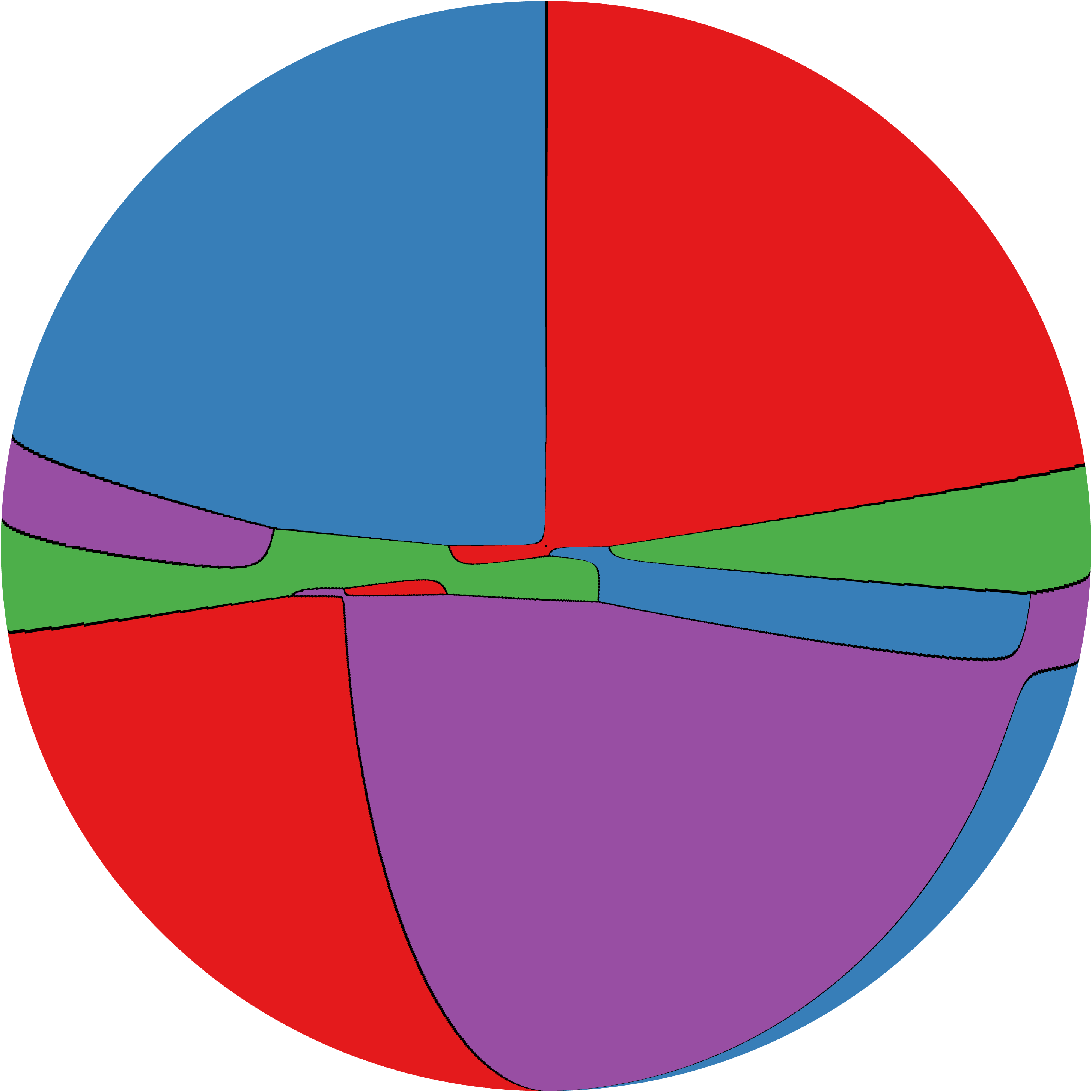}
     \end{minipage}
\caption{Top and bottom view of $\gmap(\fvd(L))$ of topology \Rom{15}. }\label{fig:gmapfvd15}
\end{figure}

\end{document}